\documentclass[11pt]{article}
\usepackage{fullpage}
\usepackage{amsmath}
\usepackage{amssymb}
\usepackage{amsthm}
\usepackage{bbm}
\usepackage{cite}
\usepackage{chngcntr}
\usepackage{color}
\usepackage{float}
\usepackage{graphicx}
\usepackage[linktocpage]{hyperref}
\usepackage{paralist}
\usepackage{thmtools,thm-restate}
\usepackage{xstring}
\usepackage[norelsize,ruled,vlined,linesnumbered]{algorithm2e}

\counterwithin{figure}{section}
\counterwithin{table}{section}

\newtheoremstyle{theoremdd}
{0.15in}
{0.15in}
{\it}
{}
{\bf}
{.}
{ }
{\thmname{#1}\thmnumber{ #2}\thmnote{ (#3)}}

\theoremstyle{theoremdd}

\newtheorem{theorem}{Theorem}[section]
\newtheorem{lemma}[theorem]{Lemma}
\newtheorem{definition}[theorem]{Definition}
\newtheorem{corollary}[theorem]{Corollary}

\newtheorem{proposition}[theorem]{Proposition}
\newtheorem{invariant}[theorem]{Invariant}

\newtheorem{remark}[theorem]{Remark}
\newtheorem{keyidea}[theorem]{Key Idea}

\newcommand{\ep}{\epsilon}
\newcommand{\mc}{\mathcal}

\newcommand{\CreateEmpire}{\ensuremath{\texttt{RebuildEmpire}}\xspace}
\newcommand{\ConditioningVerts}{\ensuremath{\texttt{ConditioningVerts}}\xspace}
\newcommand{\TwoLevelConditioningVerts}{\ensuremath{\texttt{ExtendHierarchy}}\xspace}
\newcommand{\SimpleConditioningVerts}{\ensuremath{\texttt{SimpleConditioningVerts}}\xspace}
\newcommand{\PartialSample}{\ensuremath{\texttt{PartialSample}}\xspace}
\newcommand{\FixShortcutters}{\ensuremath{\texttt{FixShortcutters}}\xspace}
\newcommand{\FastFix}{\ensuremath{\texttt{FastFix}}\xspace}
\newcommand{\Fix}{\ensuremath{\texttt{Fix}}\xspace}
\newcommand{\Oracle}{\ensuremath{\texttt{Oracle}}\xspace}
\newcommand{\SlowOracle}{\ensuremath{\texttt{SlowOracle}}\xspace}
\newcommand{\FastOracle}{\ensuremath{\texttt{FastOracle}}\xspace}
\newcommand{\VeryStable}{\ensuremath{\texttt{VeryStable}}\xspace}

\newcommand{\Shortcut}{\ensuremath{\texttt{Shortcut}}\xspace}
\newcommand{\CoveringCommunity}{\ensuremath{\texttt{CoveringCommunity}}\xspace}
\newcommand{\BallGrow}{\ensuremath{\texttt{BallGrow}}\xspace}

\newcommand{\Voronoi}{\ensuremath{\texttt{Voronoi}}\xspace}
\newcommand{\Split}{\ensuremath{\texttt{Split}}\xspace}
\newcommand{\SpecialFix}{\ensuremath{\texttt{SpecialFix}}\xspace}
\newcommand{\Clamp}{\ensuremath{\text{clamp}}\xspace}
\newcommand{\CG}{\ensuremath{\texttt{ConditioningDigraph}}\xspace}
\newcommand{\PCG}{\ensuremath{\texttt{PermanentConditioningDigraph}}\xspace}
\newcommand{\MakeNonedgesPermanent}{\ensuremath{\texttt{MakeNonedgesPermanent}}\xspace}
\newcommand{\Postprocess}{\ensuremath{\texttt{Postprocess}}}

\newcommand{\ExactTree}{\ensuremath{\texttt{ExactTree}}}
\newcommand{\ApxTree}{\ensuremath{\texttt{ApxTree}}}
\newcommand{\DSize}{\ensuremath{\texttt{DSize}}\xspace}
\newcommand{\ColumnApxPreproc}{\ensuremath{\texttt{ColumnApxPreprocessing}}\xspace}
\newcommand{\ColumnApx}{\ensuremath{\texttt{ColumnApx}}\xspace}
\newcommand{\ApxPreproc}{\ensuremath{\texttt{ApxPreprocessing}}\xspace}
\newcommand{\ApxQuery}{\ensuremath{\texttt{ApxQuery}}\xspace}
\newcommand{\SketchMatrix}{\ensuremath{\texttt{SketchMatrix}}\xspace}
\newcommand{\RecoverNorm}{\ensuremath{\texttt{RecoverNorm}}\xspace}
\newcommand{\InterclusterEdges}{\ensuremath{\texttt{InterclusterEdges}}\xspace}
\newcommand{\PartitionBall}{\ensuremath{\texttt{PartitionBall}}\xspace}
\newcommand{\PreprocANN}{\ensuremath{\texttt{PreprocANN}}\xspace}
\newcommand{\ANN}{\ensuremath{\texttt{ANN}}\xspace}

\newcommand{\R}[1]{
    \StrLen{#1}[\MyStrLen]
    \ifthenelse{\equal{\MyStrLen}{1}}{\ensuremath{\alpha^{#1/(\sigma_0+1)}r_{min}}\xspace}{\ensuremath{\alpha^{(#1)/(\sigma_0+1)}r_{min}}\xspace}}
\newcommand{\mucarve}{\ensuremath{\mu_{\text{carve}}}}
\newcommand{\muapp}{\ensuremath{\mu_{\text{app}}}}
\newcommand{\mucon}{\ensuremath{\mu_{\text{con}}}}
\newcommand{\murad}{\ensuremath{\mu_{\text{rad}}}}
\newcommand{\mumod}{\ensuremath{\mu_{\text{mod}}}}
\newcommand{\gammadel}{\ensuremath{\gamma_{\text{del}}}}
\newcommand{\gammatemp}{\ensuremath{\gamma_{\text{temp}}}}
\newcommand{\gammads}{\ensuremath{\gamma_{\text{ds}}}}
\newcommand{\gammaann}{\ensuremath{\gamma_{\text{ann}}}}
\newcommand{\xibuc}{\ensuremath{\xi_{\text{buckets}}}}
\newcommand{\zetamax}{\ensuremath{\zeta_{\max}}}
\newcommand{\taumax}{\ensuremath{\tau_{\max}}}
\newcommand{\kappamax}{\ensuremath{\kappa_{\max}}}
\newcommand{\ellmax}{\ensuremath{\ell_{\max}}}

\newcommand{\linestart}{}
\newcommand{\lineend}{}
\newcommand{\partspace}{\vspace{.4 in}}
\newcommand{\introstartspace}{\vspace{.2 in}}
\newcommand{\introendspace}{\vspace{.3 in}}

\newcommand{\highlight}[1]{\colorbox{yellow}{#1}}

\title{An almost-linear time algorithm for uniform random spanning tree generation}
\author{Aaron Schild\footnote{Supported by NSF grant CCF-1535977.} \\EECS, UC
Berkeley \\ \texttt{aschild@berkeley.edu}}

\begin{document}
\maketitle

\begin{abstract}
We give an $m^{1+o(1)}\beta^{o(1)}$-time algorithm for generating a uniformly random spanning tree in an undirected, weighted graph with max-to-min weight ratio $\beta$. We also give an $m^{1+o(1)}\ep^{-o(1)}$-time algorithm for generating a random spanning tree with total variation distance $\ep$ from the true uniform distribution. Our second algorithm's runtime does not depend on the edge weights. Our $m^{1+o(1)}\beta^{o(1)}$-time algorithm is the first almost-linear time algorithm for the problem --- even on unweighted graphs --- and is the first subquadratic time algorithm for sparse weighted graphs.

Our algorithms improve on the random walk-based approach given in \cite{KM09} and \cite{MST15}. We introduce a new way of using Laplacian solvers to shortcut a random walk. In order to fully exploit this shortcutting technique, we prove a number of new facts about electrical flows in graphs. These facts seek to better understand sets of vertices that are well-separated in the effective resistance metric in connection with Schur complements, concentration phenomena for electrical flows after conditioning on partial samples of a random spanning tree, and more.
\end{abstract}

\newpage

\section{Introduction}

In this paper, we give the first almost-linear time algorithm for the following problem:

\begin{quote}
Given an undirected graph $G$ with weights (conductances) $\{c_e\}_{e\in E(G)}$ on its edges, generate a spanning tree $T$ of $G$ with probability proportional to $\prod_{e\in E(T)} c_e$.
\end{quote}

Random spanning tree generation has been studied for a long time \cite{K47}, has many connections to probability theory (for example \cite{BS13}), and is a special case of determinantal point processes \cite{AGR16}. They also have found applications in constructing cut sparsifiers \cite{FH10,GRV09} and have played crucial roles in obtaining better approximation algorithms for both the symmetric \cite{GSS11} and asymmetric \cite{AGMGS10} traveling salesman problem.

The uniform random spanning tree distribution is also one of the simplest examples of a negatively-correlated probability distribution that is nontrivial to sample from. Much work has gone into efficiently sampling from the uniform spanning tree distribution in the past forty years \cite{G83}. This work falls into three categories:

\begin{itemize}
\item Approaches centered around fast exact computation of effective resistances \cite{G83,CDN89,K90,CMN96,HX16}. The fastest algorithm among these takes $\tilde{O}(n^{\omega})$ time for undirected, weighted graphs \cite{CDN89}.
\item Approaches that approximate and sparsify the input graph using Schur complements \cite{DKPRS17,DPPR17}. \cite{DKPRS17} samples a truly uniform tree in $\tilde{O}(n^{4/3}m^{1/2} + n^2)$ time, while \cite{DPPR17} samples a random tree in $\tilde{O}(n^2\delta^{-2})$ time from a distribution with total variation distance $\delta$ from the real uniform distribution for undirected, weighted graphs.
\item Random-walk based approaches \cite{A90,B89,W96,KM09,MST15}. \cite{MST15} takes $\tilde{O}(m^{4/3})$ time for undirected, unweighted graphs.
\end{itemize}

Our main result is an algorithm for sampling a uniformly random spanning tree from a weighted graph with polynomial ratio of maximum to minimum weight in almost-linear time:

\begin{theorem}\label{thm:main-result-aspect}
Given a graph $G$ with edge weights $\{c_e\}_e$ and $\beta = (\max_{e\in E(G)} c_e)/(\min_{e\in E(G)} c_e)$, a uniformly random spanning tree of $G$ can be sampled in $m^{1+o(1)}\beta^{o(1)}$ time.
\end{theorem}

We also give a result whose runtime does not depend on the edge weights, but samples from a distribution that is approximately uniform rather than exactly uniform. However, the runtime dependence on the error is small enough to achieve $1/\text{poly}(n)$ error in almost-linear time, so it suffices for all known applications:

\begin{restatable}{theorem}{thmmainresultapply}\label{thm:main-result-apply}
Given a weighted graph $G$ and $\ep\in (0,1)$, a random spanning tree $T$ of $G$ can be sampled from a distribution with total variation distance at most $\ep$ from the uniform distribution in time $m^{1+o(1)}\ep^{-o(1)}$ time.
\end{restatable}

Our techniques are based on random walks and are inspired by \cite{A90,B89,KM09,MST15}. Despite this, our runtime guarantees combine the best aspects of all of the former approaches. In particular, our $m^{1+o(1)}\ep^{-o(1)}$-time algorithm has no dependence on the edge weights, like the algorithms from the first two categories, but has subquadratic running time on sparse graphs, like the algorithms in the third category.

\subsection{Other contributions}

We use random walks to generate random spanning trees. The behavior of random walks can be understood through the lens of electrical networks. We prove several new results about electrical flows (for example Lemmas \ref{lem:well-sep-lev-score}, \ref{lem:ball-split}, and \ref{lem:second-order-deg}) and find new uses for many prior results (for example \cite{SRS17,I06,HHNRR08}). We highlight one of our new results here.

One particularly important quantity for understanding random walks is the effective resistance between two vertices:

\begin{definition}[Effective resistance]
The \emph{energy} of a flow $f\in \mathbb{R}^{E(G)}$ in an undirected graph $G$ with weights $\{c_e\}_{e\in E(G)}$ is

$$\sum_{e\in E(G)} \frac{f_e^2}{c_e}$$

For two vertices $s,t$ in a graph $G$, the $G$-\emph{effective resistance} between $s$ and $t$, denoted $\texttt{Reff}_G(s,t)$, is the minimum energy of any $s-t$ flow that sends one unit of flow from $s$ to $t$.
\end{definition}

We study the robustness of the $s-t$ effective resistance to random changes in the graph $G$. Specifically, we consider random graphs $H\sim G[F]$ obtained by conditioning on the intersection of a random spanning tree in $G$ with the set $F\subseteq E(G)$, which amounts to sampling a tree $T$ from $G$, contracting all edges in $E(T)\cap F$, and deleting all edges in $F\setminus E(T)$.

Surprisingly, the $s-t$ effective resistance in $H$ is the same as the $s-t$ effective resistance in $G$ in expectation as long as $G\setminus F$ is connected. No concentration holds in general, though. Despite this, we show that for every $F$, there exists a small set $F'\subseteq F$ for which $\texttt{Reff}_{H\setminus F'}(s,t)$ is not much smaller than its original value:

\begin{lemma}\label{lem:special-slow-fix}
Let $G$ be a graph, $F\subseteq E(G)$, $\ep\in (0,1)$, and $s,t\in V(G)$. Sample a graph $H$ by sampling a random spanning tree $T$ in $G$, contracting all edges in $E(T)\cap F$, and deleting all edges in $F\setminus E(T)$. Then, with high probability, there is a set $F'\subseteq F$ that depends on $H$ that satisfies both of the following guarantees:

\begin{compactitem}
\item (Effective resistance) $\texttt{Reff}_{H\setminus F'}(s,t)\ge (1 - \ep)\texttt{Reff}_G(s,t)$
\item (Size) $|F'|\le O((\log n)/\ep^2)$
\end{compactitem}
\end{lemma}

Even better, we show that $F'$ can be computed in almost-linear time. See Section \ref{sec:fast-fix} for details. Our algorithm uses a combination of matrix sketching \cite{AMS96,I06} and localization \cite{SRS17} that may be of independent interest.

\section{Algorithm Overview}\label{sec:fake-overview}

Our algorithm, like those of \cite{KM09} and \cite{MST15}, is based on the following beautiful result of Aldous \cite{A90} and Broder \cite{B89}:

\begin{restatable}[Aldous-Broder]{theorem}{thmaldousbroder}\label{thm:aldous-broder}
Pick an arbitrary vertex $u_0$ and run a random walk starting at $u_0$ in a weighted graph $G$. Let $T$ be the set of edges used to visit each vertex besides $u_0$ for the first time. Then $T$ is a weighted uniformly random spanning tree of $G$.
\end{restatable}

The runtime of Aldous-Broder is the amount of time it takes to visit each vertex for the first time, otherwise known as the \emph{cover time} of $G$. On one hand, the cover time can be as high as $\Theta(mn)$. On the other hand, Aldous-Broder has the convenient property that only a small number of vertex visits need to be stored. In particular, only $n-1$ visits to vertices add an edge to the sampled tree (the first visits to each vertex besides the starting vertex). This observation motivates the idea of \emph{shortcutting} the random walk.

\subsection{The shortcutting meta-algorithm}\label{subsec:meta-algo}

To motivate our algorithm, we classify all existing algorithms based on Aldous-Broder \cite{B89,A90,KM09,MST15} at a very high level. We start by describing Aldous-Broder in a way that is more readily generalizable:

\introstartspace

\begin{compactitem}
    \item[\textbf{Aldous-Broder}]
    \item For each $v\in V(G)$, let $S_v^{(0)} = \{v\}$.
    \item Pick an arbitrary vertex $u_0$ and set $u\gets u_0$.
    \item Until all vertices in $G$ have been visited
    \begin{compactitem}
        \item Sample the first edge that the random walk starting at $u$ uses to exit $S_u^{(0)}$
        \item Replace $u$ with the non-$S_u^{(0)}$ endpoint of this edge.
    \end{compactitem}
    \item Return all edges used to visit each vertex besides $u_0$ for the first time.
\end{compactitem}

\introendspace

We now generalize the above algorithm to allow it to use some number $\sigma_0$ of \emph{shortcutters} $\{v\}\subseteq S_v^{(i)}\subseteq V(G)$ for each vertex $v\in V(G)$. If the $\{S_v^{(i)}\}_{i=1}^{\sigma_0}$s are chosen carefully, then instead of running the random walk until it exits $S_v^{(i)}$, one can sample the exiting edge much faster using Laplacian solvers.

Ideally, we could shortcut the random walk directly to the next unvisited vertex in order to minimize the number of wasted visits. Unfortunately, we do not know how to do such shortcutting efficiently. Instead, we use multiple shortcutters per vertex. More shortcutters per vertex means a better approximation to the set of previously visited vertices, which leads to fewer unnecessary random walk steps and a better runtime.

\introstartspace

\begin{compactitem}
    \item[\textbf{Simple shortcutting meta-algorithm}]
    \item For each $v\in V(G)$, let $S_v^{(0)} = \{v\}$ \highlight{and \textbf{pick} shortcutters $\{S_v^{(i)}\}_{i=1}^{\sigma_0}$}
    \item Pick an arbitrary vertex $u_0$ and set $u\gets u_0$.
    \item Until all vertices in $G$ have been visited
    \begin{compactitem}
        \item \highlight{\parbox{0.8\textwidth}{Let $i^*\in \{0,1,\hdots,\sigma_0\}$ be the maximum value of $i$ for which all vertices in $S_u^{(i)}$ have been visited}}
        \item \textbf{Sample} the first edge that the random walk starting at $u$ uses to exit $S_u^{(i^*)}$
        \item Replace $u$ with the non-$S_u^{(i^*)}$ endpoint of this edge.
    \end{compactitem}
    \item Return all edges used to visit each vertex besides $u_0$ for the first time.
\end{compactitem}

\introendspace

To implement the above meta-algorithm, one must make two important choices, each of which is bolded above. Both of these choices only affect the runtime of the meta-algorithm; not its correctness:

\introstartspace

\begin{compactitem}
\item A set of shortcutters for each vertex $v\in V(G)$
\item A method for sampling the first exit edge from $S_u^{(i)}$, which we call a \emph{shortcutting method}
\end{compactitem}

\introendspace

The meta-algorithm could also choose an arbitrary starting location $u_0$, but this choice is not important to any shortcutting-based algorithm.

We now argue that the meta-algorithm correctly samples a uniformly random spanning tree, no matter the choice of the $S_v^{(i)}$s or shortcutting method. First of all, $i = 0$ is always a valid choice for $i^*$, so $i^*$ exists and the algorithm is well-defined. Since all vertices in $S_v^{(i^*)}$ have been previously visited, using the shortcutter $S_v^{(i^*)}$ does not skip any first visits. Therefore, by Theorem \ref{thm:aldous-broder}, the edges returned form a uniformly random spanning tree.

Next, we summarize all algorithms based on this meta-algorithm \cite{A90,B89,KM09,MST15} with $\sigma_0$ --- the number of shortcutters per vertex --- and the choice of shortcutting method. We also list bounds on the runtimes of these algorithms on unweighted graphs to offer context.

\begin{table}[H]
\begin{center}
\begin{tabular}{ c | c | c | c }
    \hline
    Algorithm & Shortcutting method & $\sigma_0$ & Runtime \\ \hline
    \cite{A90,B89} & N$\backslash$A & 0 & $O(mn)$\\ \hline
    \cite{KM09} & Offline & 1 & $O(m\sqrt{n})$\\ \hline
    \cite{MST15} & Offline & 2 & $O(m^{4/3})$\\ \hline
    This paper & Online & $\Theta(\log \log n)$ & $m^{1+O(1/(\log \log n))} = m^{1+o(1)}$\\
    \hline
\end{tabular}
\end{center}
\label{tab:runtimes}
\caption{Shortcutting methods, number of shortcutters, and runtimes for prior algorithms}
\end{table}

While we have not yet discussed what ``Offline'' and ``Online'' shortcutting are, we highlight that our shortcutting method is different from that of \cite{KM09} and \cite{MST15}. This is one of the key reasons why we are able to obtain a faster algorithm and are able to effectively use more shortcutters per vertex. 

\subsection{Shortcutting methods}

The starting point for our improvement is an \emph{online} shortcutting method. This method is based on the following observation:

\begin{keyidea}[Online shortcutting]
For a vertex $u$ and a shortcutter $S_u$ associated with $u$, the probability that a random walk exits $S_u$ through an edge $e$ can be $\ep$-additively approximated for all $e\in \partial S_u$ simultaneously using \emph{one} $\ep$-approximate Laplacian system solve on a graph with $|E(S_u)\cup \partial S_u|$ edges.

When combined with a trick due to Propp \cite{P10}, one can exactly sample an escape edge in expected $\tilde{O}(|E(S_u)\cup \partial S_u|)$ time with no preprocessing.
\end{keyidea}

We call this method \emph{online} due to its lack of preprocessing and the shortcutting technique used in \cite{KM09} and \cite{MST15} \emph{offline} due to its fast query time with high preprocessing time. We summarize the runtime properties of these shortcutting methods for a shortcutter $S_u$ here:

\begin{table}[H]
\begin{center}
\begin{tabular}{ c | c | c }
    \hline
    Shortcutting method & Preprocessing & Query\\ \hline
    Online & None & $\tilde{O}(|E(S_u)\cup \partial S_u|)$\\ \hline
    Offline & $\tilde{O}(|E(S_u)\cup \partial S_u||\partial S_u|)$ & $\tilde{O}(1)$\\ \hline
\end{tabular}
\end{center}
\label{tab:shortcutting-methods}
\caption{Shortcutting methods and their runtimes}
\end{table}

The upside to the online method is that its runtime does not depend quadratically on the size of $S_u$'s boundary. This is a critical barrier to improving the technique of \cite{KM09} and \cite{MST15} because it is impossible to obtain balanced cuts with arbitrarily small size that separate a graph into low-radius parts in most metrics. While the high query time for the online method may seem prohibitive initially, it is fine as long as online shortcutting does substantially less work than the random walk would have done to reach the boundary of $S_u$. In the following path example, it takes the random walk $\Theta(k^2)$ time for the random walk to escape $S_u$ starting at $u$, but online shortcutting only takes $\tilde{O}(k)$ time to find the escape edge:

\begin{figure}[H]
\begin{center}
\includegraphics[width=0.8\textwidth]{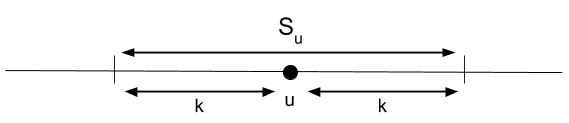}
\end{center}
\caption{How online shortcutting saves over the random walk.}
\label{fig:path-shortcutting}
\end{figure}

\subsection{Properties of shortcutters}

Now, we describe the machinery that allows us to bound the total amount of work. We start with a bound that captures the idea that most random walk steps happen far away from an unvisited vertex. This bound is a weighted generalization of Lemma A.4 given in \cite{MST15}. We prove it in Appendix \ref{sec:cover-time}:

\begin{restatable}[Key result for bounding the number of shortcutter uses]{lemma}{lemwalkbound}\label{lem:walk-bound}
Consider a random walk starting at an arbitrary vertex in a graph $I$, an edge $\{u,v\} = f\in E(I)$, and a set of vertices $S\subseteq V(I)$. The expected number of times the random walk traverses $f$ from $u\rightarrow v$ when $u$ is within $I$-effective resistance distance $R$ of some unvisited vertex in $S$ is at most $\tilde{O}(c_f R)$, where $c_f$ is the conductance of the edge $f$.
\end{restatable}

The effective resistance metric\footnote{The effective resistance $\texttt{Reff}_G$ satisfies the triangle inequality and therefore forms a metric space on the vertices of $G$} appears in the above lemma due to the relationship between random walks and electrical networks. Unlike \cite{KM09} and \cite{MST15}, we apply this lemma in the original graph \emph{and} in other graphs obtained by ``eliminating'' vertices from the original graph. Specifically, we apply Lemma \ref{lem:walk-bound} to \emph{Schur complements} of the input graph $G$. In all sections before Section \ref{sec:fix-preliminaries} --- including this one --- one does not need to understand how the Schur complement is constructed or its linear algebraic properties. We just use the following combinatorial, folklore fact about Schur complements:

\begin{theorem}\label{thm:schur-walk}
For a graph $G$ and a set of vertices $X\subseteq V(G)$, there is a graph $H = \texttt{Schur}(G,X)$ with $V(H) = X$ such that given a vertex $v_0\in X$, the distribution over random walk transcripts\footnote{List of all visited vertices} in $H$ starting at $v_0$ is the same as the distribution over random walk transcripts in $G$ starting at $v_0$ with all vertices outside of $X$ removed.
\end{theorem}

For a shortcutter $S_u$, consider the graph $H = \texttt{Schur}(G,(V(G)\setminus S_u)\cup \{u\})$. Each use of the shortcutter $S_u$ can be charged to crossing of at least one edge incident with $u$ in the random walk on $H$ by Theorem \ref{thm:schur-walk}. Therefore, to bound the number of times a shortcutter $S_u$ is used over the course of the algorithm, it suffices to bound the total conductance of edges between $u$ and $V(H)\setminus \{u\}$ in $H$. This motivates one of the key properties of shortcutters, which we call \emph{conductivity} in later sections:

\begin{keyidea}[Shortcutter Schur complement conductance bound]\label{keyidea:conductance}
Let $S_C$ denote a shortcutter for which $S_u^{(i)} = S_C$ for all $u\in C$.Then ``on average,'' the total conductance of the edges between $C$ and $V(G)\setminus S_C$ in the graph $\texttt{Schur}(G,C\cup (V(G)\setminus S_C))$, which we call the \emph{Schur complement conductance} of $S_C$, is at most $\frac{m^{o(1)}}{\R{i}}$. $\alpha$ is the ratio between the maximum and minimum effective resistance distance between any two vertices in $G$ and $r_{min}$ is the minimum effective resistance distance between any two vertices in $G$.
\end{keyidea}

By ``on average,'' we mean that the shortcutters $S_C$ are organized into $m^{o(1)}$ sets and within each set $\mc C$, the total Schur complement conductance of the shortcutters $S_C\in \mc C$ is at most $\frac{m^{o(1)}|\mc C|}{\R{i}}$. For the rest of this section, we think of \emph{each} shortcutter as having Schur complement conductance at most $\frac{m^{o(1)}}{\R{i}}$ in order to simplify the description. For a more formal description of the organization of our shortcutters, see Section \ref{sec:real-overview}.

To bound the number of times $S_C$ is used, Lemma \ref{lem:walk-bound} requires two things:

\introstartspace

\begin{compactenum}[(1)]
\item [\textbf{Sufficient properties for bounding shortcutter uses}]
\item A bound on the Schur complement conductance of $S_C$ \label{property:conductance}
\item A bound on the effective resistance distance to the nearest unvisited vertex outside of $S_C$ \label{property:distance}
\end{compactenum}

\introendspace

The Schur complement conductance is at most $\frac{m^{o(1)}}{\R{i}}$ by conductivity. Therefore, we just need to bound the distance to the nearest unvisited vertex. If there was an unvisited vertex within effective resistance distance $\R{i+1} m^{o(1)}$ of $C$, Lemma \ref{lem:walk-bound} would imply that $S_C$ is only used $$\left(\R{i+1}m^{o(1)}\right)\left(\frac{m^{o(1)}}{\R{i}}\right) = m^{o(1)}\alpha^{1/(\sigma_0+1)}$$ times over the course of the shortcutted random walk. To bound the total work done, the following fact suffices:

\begin{keyidea}[Shortcutter overlap]\label{keyidea:overlap}
Each vertex in $G$ is in at most $m^{o(1)}$ shortcutters.
\end{keyidea}

The above idea implies that the total size of all shortcutters is $O(m^{1 + o(1)})$. To use a shortcutter, we apply the online shortcutting method, which takes time proportional to the shortcutter's size (see Table \ref{tab:shortcutting-methods}). If each shortcutter is used at most $m^{o(1)}\alpha^{1/(\sigma_0+1)}$ times as described above, the total work due to all shortcutter uses is $(m^{1+o(1)})(m^{o(1)}\alpha^{1/(\sigma_0+1)})\le m^{1+o(1)}\alpha^{o(1)}$, as desired.

Therefore, if we can obtain shortcutters with bounds on (\ref{property:conductance}) and (\ref{property:distance}) that also respect Key Idea \ref{keyidea:overlap}, we would have an almost-linear time algorithm for sampling random spanning trees on weighted graphs.

\subsection{Obtaining shortcutters with Property (\ref{property:conductance}) and small overlap}

We have not yet discussed how to actually obtain shortcutters with the desired conductance property. We discuss this in detail in Section \ref{sec:compute-shortcutters}, but we give a summary here for interested readers. We construct $\{S_v^{(i)}\}_{v\in V(G)}$ for each $i$ independently by

\introstartspace

\begin{compactitem}
\item constructing a small number of families of sets that are each well-separated in the effective resistance metric, have distance separation roughly $\R{i}$, and together cover the graph. These are the \emph{cores} and the construction of these cores is similar to constructions of sparse covers of metric spaces (for example \cite{AP90}).
\item making the shortcutter around each core $C$ be the set of vertices $S_C\subseteq V(G)$ for which a random walk starting at $x\in S_C$ is more likely to hit $C$ before any other core in its family. The sparsity of the cover ensures that the shortcutters satisfy Key Idea \ref{keyidea:overlap}, while their well-separatedness ensures Property (\ref{property:conductance}).
\end{compactitem}

\subsection{Obtaining Property (\ref{property:distance}) using partial sampling and carving}

We now show that when $S_u^{(i)}$ is used, there is an unvisited vertex with effective resistance distance at most $m^{o(1)} \R{i+1}$ from $u$. Review the shortcutting meta-algorithm. When $S_u^{(i^*)}$ is used, there is some vertex $v\in S_u^{(i^*+1)}$ that the random walk has not yet visited. $v$, however, may not be close to $u$. This motivates the following property of a shortcutter $S_u^{(i)}$, which we call \emph{being carved with respect to (a set of vertices) $S$}:

\begin{keyidea}[Carving]\label{keyidea:carving}
A shortcutter $S_u^{(i)}$ is \emph{carved with respect to $S\subseteq V(G)$} if $S_u^{(i)}\cap S$ only consists of vertices that are within effective resistance distance $m^{o(1)}\R{i}$ of $u$.
\end{keyidea}

If $S_u^{(i^*+1)}$ is carved with respect to $V(G)$, then the unvisited vertex $v$ is within distance $m^{o(1)}\R{i+1}$ of $u$. As a result, there is an unvisited vertex within distance $m^{o(1)}\R{i+1}$ of $u$, as desired.

\begin{figure}[H]
\begin{center}
\includegraphics[width=0.8\textwidth]{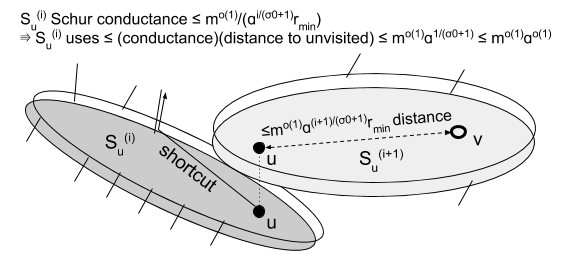}
\end{center}
\caption{Bounding the number of times $S_u^{(i)}$ is used.}
\label{fig:shortcutter-conductivity-fake}
\end{figure}

It is difficult to directly build shortcutters to make them carved with respect to some set $S$. Instead, we explore how we could remove vertices from shortcutters so that \emph{all} shortcutters for \emph{all} vertices are carved with respect to $V(G)$. To carve $V(G)$ out of all shortcutters $S_u^{(i)}$, one could just remove all vertices in $S_u^{(i)}$ that are farther than resistance distance $m^{o(1)}\R{i}$ away from $u$. Unfortunately, this could remove almost all of the shortcutter in general.

Instead, we compute a \emph{partial sample} of a random spanning tree in order to make it so that each shortcutter does not have to be carved with respect to as many vertices. Specifically, we modify the simple shortcutting meta-algorithm as follows:

\introstartspace

\begin{compactitem}
    \item[\textbf{Full shortcutting meta-algorithm (one round of partial sampling)}]
    \item \highlight{\textbf{Choose} a set $S\subseteq V(G)$ for partial sampling}
    \item For each $v\in V(G)$, let $S_v^{(0)} = \{v\}$ and \textbf{pick} shortcutters $\{S_v^{(i)}\}_{i=1}^{\sigma_0}$
    \item Pick an arbitrary vertex $u_0$ and set $u\gets u_0$.
    \item Until all vertices in \highlight{$S$} have been visited
    \begin{compactitem}
        \item Let $i^*\in \{0,1,\hdots,\sigma_0\}$ be the maximum value of $i$ for which all vertices in \highlight{$S\cap S_u^{(i)}$} have been visited
        \item \textbf{Sample} the first edge that the random walk starting at $u$ uses to exit $S_u^{(i^*)}$
        \item Replace $u$ with the non-$S_u^{(i^*)}$ endpoint of this edge.
    \end{compactitem}
    \item \parbox{0.8\textwidth}{\highlight{Let $T'$ be} all edges used to visit each vertex besides $u_0$ \highlight{in $S$} for the first time \highlight{that are in the induced subgraph $F := E(G[S])$}} 
    \item \highlight{\parbox{0.8\textwidth}{Condition on the partial sample, which amounts to contracting all edges in $E(T')\cap F$ and deleting all edges of $F\setminus E(T')$ in $G$}}
\end{compactitem}

\introendspace

\cite{MST15} also exploited partial sampling in this way. This algorithm correctly samples the intersection of a random spanning tree of $G$ with $E(G[S])$ because it does not skip any of the first visits to vertices in $S$ and only vertices in $S$ need to be visited in order to determine the edges in $G[S]$ that are in the sample. While this algorithm no longer samples the entirety of a tree, we only need shortcutters to be carved with respect to $S$ rather than all of $V(G)$ in order to show that the total work is $m^{1+o(1)}\alpha^{o(1)}$.

Our algorithm and \cite{MST15} exploit the full meta-algorithm in multiple rounds. During each round, we pick a set $S$ to condition on, run the meta-algorithm, and repeat until $G$ is a single vertex. At this point, we have sampled a complete spanning tree of $G$.

We want to choose $S$ to be small enough so that every shortcutter can be carved with respect to $S$ without increasing the Schur complement conductance of those shortcutters too much. As long as all shortcutters are carved with respect to $S$, the meta-algorithm takes $m^{1+o(1)}\alpha^{o(1)}$ time. However, we also want $S$ to be large enough to make substantial progress.

When $\sigma_0 = 1$, let $S$ be the set of vertices $u$ assigned to the largest shortcutters. By Key Idea \ref{keyidea:overlap}, there cannot be too many large shortcutters, which means that $S$ is the union of a small number of clusters with small effective resistance diameter ($\sqrt{\alpha}r_{min}$). Deleting each cluster from each shortcutter $S_C$ with a far-away core $C$ makes $S_C$ carved with respect to $S$. Furthermore, because the cluster was well-separated from $C$, its deletion did not increase the Schur complement conductance of $S_C$ much. For more details on this analysis, see Section \ref{subsec:cond-warmup}.

\subsection{Bounding the number of rounds of partial sampling when $\sigma_0 = 1$}\label{subsec:one-level-progress-disc}

In the previous section for the $\sigma_0 = 1$ case, we saw that conditioning on the induced subgraph of the vertices assigned to the largest shortcutters was a good idea for carving. We now show that computing a partial sample for the induced subgraph of these vertices allows us to make substantial progress. Ideally, one could show that conditioning on the induced subgraph of vertices assigned to the largest shortcutters decreases the size of the graph by a constant fraction. Unfortunately, we do not know how to establish this in general.

To get around this, we do not rebuild shortcutters from scratch after each partial sampling round. Instead, we show that it is possible to make shortcutters $S_u^{(i)'}$ that are contained within the shortcutter $S_u^{(i)}$ from the previous round. It is quite tricky to do this directly, as conditioning on a partial sample can change the metric structure of the graph $G$. In particular, the conductance of a shortcutter could dramatically increase after conditioning.

To cope with this, we show a concentration inequality that promises the existence of a small set of edges with high probability that, when deleted, restore the conductance of all shortcutters back to their value before conditioning. This result follows from a nontrivial generalization (Lemma \ref{lem:fast-fix}) of Lemma \ref{lem:special-slow-fix} that we reduce to in Section \ref{sec:fix-reduction}, prove in Section \ref{sec:slow-fix}, and find a fast algorithm for in Section \ref{sec:fast-fix}.

Given that $S_u^{(i)'}$ is contained in $S_u^{(i)}$ for all $u\in V(G)$ and all $i\in [\sigma_0]$, conditioning on the vertices with the near-largest shortcutters decreases the maximum size of a remaining shortcutter by a large factor. Therefore, after $O(\log n)$ rounds of conditioning on the induced subgraph of the vertices assigned to the largest shortcutters, no shortcutters are left. At this point, the algorithm is done. Each round takes $m^{1+o(1)}\sqrt{\alpha}$ time in the $\sigma_0 = 1$ case for a total of $\tilde{O}(m^{1+o(1)}\sqrt{\alpha})$ runtime.

\subsection{Carving and progress when $\sigma_0 > 1$}

The bottleneck in the algorithm for the $\sigma_0 = 1$ case is the sampling step. As discussed in Section \ref{subsec:meta-algo}, using more shortcutters allows us to approximate the set of previously visited vertices better, leading to a better runtime. In particular, the runtime-bounding argument presented earlier, given a carving and conductance bound, shows that using $\sigma_0$ shortcutters yields an $m^{1+o(1)}\alpha^{1/(\sigma_0+1)}$-runtime algorithm for sampling a spanning tree.

Unfortunately, carving shortcutters is more complicated when $\sigma_0 > 1$. We need to pick a relatively large set of vertices that can be carved out of \emph{all} shortcutters for \emph{all} vertices simultaneously. To do this, one could start by trying to generalize the strategy in the $\sigma_0 = 1$ case through repetition. Specifically, one could try the following strategy for picking a set $S$ for use in one round of the meta-algorithm with partial sampling:

\introstartspace

\begin{compactitem}
\item [\textbf{First attempt at conditioning when $\sigma_0 > 1$}]
\item $S_{\sigma_0+1}\gets V(G)$
\item For $i = \sigma_0,\sigma_0-1,\hdots,1$
    \begin{compactitem}
    \item $S_i\gets $ the vertices $u\in S_{i+1}$ with near-maximum size $S_u^{(i)}$ shortcutters; that is within a factor of $m^{-1/\sigma_1}$ of the maximum.
    \end{compactitem}
\item Let $S\gets S_1$ 
\end{compactitem}

\introendspace

This strategy has some benefits. If $S_1$ could be carved out of all shortcutters, the maximum size of $S_v^{(1)}$ shortcutters for vertices $v\in S_2$ would decrease by a factor of $m^{-1/\sigma_1}$.

Before moving on, we elaborate on how conditioning on the induced subgraph of vertices assigned to a shortcutter renders it unnecessary in the future. Start by refining all cores of all $S_v^{(i)}$ shortcutters to obtain $\sigma_0$ partitions $\{\mc P_i\}_{i=1}^{\sigma_0}$ of $V(G)$, with one for each $i\in [\sigma_0]$. Standard ball-growing (for example \cite{LR99}) ensures that the total conductance of all boundary edges of parts in $\mc P_i$ is at most $m^{1+o(1)}/(\R{i})$. Conditioning on the induced subgraph of a part $P$ deletes or contracts all edges in the induced subgraph of $P$, only leaving $P$'s boundary. Since $P$'s boundary is small, the random walk never needs to use $P$'s shortcutter again because the total number of steps across $\mc P_i$ boundary edges is at most $$\frac{m^{1+o(1)}}{\R{i}}\R{i+1}\le m^{1+o(1)}\alpha^{o(1)}$$ where the $\R{i+1}$ bound follows from carving. Therefore, conditioning on a part replaces it with its boundary, thus rendering its shortcutter unnecessary.

Now, we go back to analyzing our first attempt at a $\sigma_0 > 1$ algorithm for selecting $S$. If we could always carve all shortcutters with respect to $S_1$, conditioning $\sigma_1$ times on various $S_1$s would make all shortcutters for $\mc P_1$ parts intersecting $S_2$ irrelevant, thus making it possible to condition on $S_2$ directly. More generally, if carving were not an issue, every $\sigma_1$ rounds of conditioning on $S_i$ would pave the way for one round of conditioning on $S_{i+1}$. Combining this reasoning for all $i$ implies that we have sampled the entire tree after $\sigma_1^{\sigma_0}$ applications of the meta-algorithm.

Unfortunately, our first attempt does not produce carvable sets $S$ in general because there could be a very large shortcutter with core just outside of some $S_i$ that happens to intersect many of the vertices in $S_i$. To cope with this, we incorporate a ball-growing type approach that switches to conditioning on this very large but nearby shortcutter if one exists. Once this procedure stops, one can carve the parts assigned to the selected shortcutters out of all other shortcutters because the selected shortcutters are larger than all other shortcutters that the selected parts intersect. For more details, see Section \ref{sec:choose-parts}.

\subsection{Coping with the fixing lemma in the shortcutting method}

In Section \ref{subsec:one-level-progress-disc}, we established that we could obtain containment of shortcutters in past shortcutters if we deleted a small set of ``fixing edges'' from the graph. However, we cannot actually delete these edges from the graph, as we must do partial sampling in graphs resulting directly from conditioning in order to correctly sample a uniformly random spanning tree.

Instead of deleting these fixing edges from the graph, we just remove their endpoints from the shortcutters and use \emph{offline} shortcutting \cite{MST15,KM09} to make it so that shortcutting to the endpoints of these fixing edges only takes constant time rather than time proportional to the size of the shortcutter. Since there are a small number of fixing edges, the preprocessing time for offline shortcutting is small. Each shortcut to the endpoints of a removed edge takes $\tilde{O}(1)$ time and can be charged to crossing a boundary edge of some core. Constructing the cores using standard ball-growing makes the total conductance of these boundary edges small, so Lemma \ref{lem:walk-bound} can be used to show that the number of such shortcutting steps is small.

While this completes the high-level description of our algorithm, it does not describe all of the contributions of this paper. Along the way, we prove many new results about effective resistance metrics (like Lemma \ref{lem:well-sep-lev-score}) that may be of independent interest.

\newpage

\setcounter{tocdepth}{2}
\tableofcontents

\newpage

\section{Preliminaries}\label{sec:preliminaries}

\subsubsection{Graphs}

For a (directed or undirected) graph $G$, let $V(G)$ and $E(G)$ denote its vertex and edge set respectively. $n$ and $m$ refer to the number of vertices and edges respectively of the input graph to our random spanning tree generation algorithm. For a set of edges $F\subseteq E(G)$, $G\setminus F$ denotes the graph obtained by deleting the edges in $F$ from $G$. For a set of vertices $S\subseteq V(G)$, let $G/S$ be the graph obtained by \emph{identifying} all vertices in $S$ to a single vertex; that is

$$V(G/S) := (V(G)\setminus S)\cup \{s\}$$

and each endpoint of an edge $e\in E(G)$ that is also in $S$ is replaced with $s$. For a set of edges $F\subseteq E(G)$, let $V(F)$ denote the set of endpoints of edges in $F$. For an edge $f\in E(G)$, let $G\setminus f$ and $G/f$ denote the graph obtained by deleting and contracting $f$ respectively.

For two sets of vertices $S,S'\in V(G)$, let $E_G(S,S')$ denote the set of edges with one endpoint in $S$ and the other endpoint in $S'$. Let $G[S] := E_G(S) := E_G(S,S)$. When the graph $G$ is clear from context, we omit it from the subscript. For a graph $G$ with two sets $X,Y\subseteq V(G)$, let $G/(X,Y)$ denote the graph obtained by identifying all vertices in $X$ to one vertex $x$ and all vertices in $Y$ to one vertex $y$.

For a set of vertices $S\in V(G)$, let $\partial_G S := E_G(S,V(G)\setminus S)$ denote the boundary edges of $S$. For a singleton set $S = \{w\}$, $\partial_G S$ is abbreviated $\partial_G w$.

In this paper, graphs are sometimes weighted with \emph{conductances} $\{c_e\}_{e\in E(G)}$. For a set $F\subseteq E(G)$, let $c^G(F) := \sum_{e\in F} c_e^G$. Let $r_e^G = 1/c_e^G$. Let $\beta^G := (\max_{e\in E(G)} r_e^G)/(\min_{e\in E(G)} r_e^G)$. When the context is clear, the graph $G$ is omitted from all superscripts in the aforementioned definitions.

\subsubsection{Laplacian matrices, electrical flows, and effective resistances}

For an undirected graph $G$ with two vertices $s,t\in V(G)$, let $b_{st}\in \mathbb{R}^{V(G)}$ denote the vector with $b_s = 1$, $b_t = -1$, and $b_v = 0$ for $v\ne s,t$. Direct all of the edges of $G$ arbitrarily. Suppose that $e = \{a,b\}$ and is directed from $a$ to $b$. Define $b_e := b_{ab}$. Define the \emph{Laplacian matrix} of a weighted graph $G$ as

$$\sum_{e\in E(G)} c_e b_e b_e^T$$

This definition is invariant of the orientations of the edges. $L_G$ has nontrivial kernel, but still has a Moore-Penrose pseudoinverse $L_G^+$. The vector $L_G^+ b_{st}$ is a vector of \emph{potentials} for the \emph{electrical flow} $C_G B_G L_G^+ b_{st}$, where $B_G$ is the $|E(G)|\times |V(G)|$ matrix with rows equal to the vectors $b_e$ for $e\in G$ and $C_G$ is the $|E(G)|\times |E(G)|$ diagonal matrix of edge conductances.

The \emph{effective resistance} between two vertices $s,t\in V(G)$ is the energy of the electrical flow from $s$ to $t$, which equivalently is

$$\texttt{Reff}_G(s,t) := b_{st}^T L_G^+ b_{st}$$

For an edge $e = \{a,b\}\in E(G)$, let $\texttt{Reff}_G(e) := \texttt{Reff}_G(a,b)$. We use the following folklore fact about effective resistances extensively without reference:

\begin{remark}
The vertex set of $V(G)$ is a metric space with respect to the metric $\texttt{Reff}_G$. In particular, $\texttt{Reff}_G$ satisfies the triangle inequality; i.e.

$$\texttt{Reff}_G(s,t)\le \texttt{Reff}_G(s,w) + \texttt{Reff}_G(w,t)$$

for any three vertices $s,t,w\in V(G)$.
\end{remark}

For a set $S\in V(G)$, define its \emph{effective resistance diameter} to be

$$\max_{u,v\in S} \texttt{Reff}_G(u,v)$$

Often, for clarity, we call this the $G$-effective resistance diameter of $S$. Let $r_{min} := \min_{u,v\in V(G)} \texttt{Reff}_G(u,v)$, $r_{max} := \max_{u,v\in V(G)} \texttt{Reff}_G(u,v)$, and $\alpha = \frac{r_{max}}{r_{min}}$. Notice that $\beta\le \alpha\le m^2 \beta$. Therefore, to obtain an $m^{1 + o(1)}\beta^{o(1)}$-time algorithm, it suffices to obtain an $m^{1 + o(1)}\alpha^{o(1)}$ time algorithm.

\subsubsection{Laplacian solvers}

In this paper, we make extensive use of efficient approximate Laplacian solvers \cite{ST14,KMP14,KOSZ13,CKMPPRX14,PS14,LPS15,KS16}:

\begin{theorem}[\cite{CKMPPRX14}]
There is an $O(m\sqrt{\log n} \log(n/\ep))$ time algorithm, that, given a demand vector $d\in \mathbb{R}^{V(G)}$ for some graph $G$, computes a vector $p\in \mathbb{R}^{V(G)}$ such that

$$||p - L_G^+ d||_{\infty}\le \ep$$

with high probability.
\end{theorem}

\subsubsection{Random walks}

For a weighted graph $G$ and some vertex $a\in V(G)$, let $\Pr_a[E]$ denote the probability of an event $E$ over random walks starting at $a$. For a set of vertices $S\in V(G)$, let $t_S$ be the random variable denoting the hitting time to the set $S$. When $S$ is a singleton $\{b\}$, we abbreviate $t_S$ as $t_b$.

We use the following fact about random walks extensively:

\begin{theorem}[Proposition 2.2 of \cite{MR16}]\label{thm:edge-visits}
Let $G$ be a graph with conductances $\{c_e\}_e$. Consider two vertices $s,t\in V(G)$. For a vertex $u$, let $p_u = b_{st}^T L_G^+ b_{ut}$. Consider an edge $e = \{u,v\}$. Then

$$p_u c_e = \textbf{E}_s[\text{ number of times $e$ is crossed from $u\rightarrow v$ before } t_t]$$
\end{theorem}

\subsubsection{Low-dimensional embedding of the effective resistance metric}

Throughout this paper, we use the fact that the effective resistance metric can be embedded into low-dimensional Euclidean space:

\begin{theorem}[\cite{SS08}]\label{thm:jl}
With high probability, one can compute an embedding $D:V(G)\rightarrow \mathbb{R}^d$ of the vertices of a graph $G$ with $d\le (\log n)/\ep^2$ for which

$$||D(u) - D(v)||_2^2\in [(1-\ep)\texttt{Reff}(u,v),(1+\ep)\texttt{Reff}(u,v)]$$

in near-linear time. Furthermore, for each vertex $u\in V(G)$, $D(u)$ takes $O((\log n)/\ep^2)$ time to compute.
\end{theorem}

Throughout this paper, we use $\ep = 1/2$. In most of this paper, we use this result to approximately compute effective resistances. In the appendix and in Section \ref{sec:choose-parts}, we make greater use of this through approximate nearest neighbors. Specifically, we apply \emph{locality-sensitive hashing} for $\ell_2^2$:

\begin{definition}[Locality-sensitive hashing and approximate nearest neighbors\cite{AI06}]

A family $\mc H$ of functions with domain $\mathbb{R}^d$ is called $(R,cR,p_1,p_2)$-\emph{sensitive} if for any $p,q\in \mathbb{R}^d$,

\begin{itemize}
\item If $||p - q||_2^2\le R$, then $\Pr_{h\sim \mc H}[h(q) = h(p)] \ge p_1$.
\item If $||p - q||_2^2\ge cR$, then $\Pr_{h\sim \mc H}[h(q) = h(p)] \le p_2$.
\end{itemize}
\end{definition}

\begin{theorem}[\cite{AI06}]\label{thm:lsh}
For any $R > 0$ and dimension $d$, a $(R,O(c)R,1/n^{1/c},1/n^5)$-sensitive family of hash functions with query time $O(dn^{1/c})$ for $\mathbb{R}^d$ can be computed in almost-linear time.
\end{theorem}

Locality sensitive hashing can be used to find approximate nearest neighbors in Euclidean space. In particular, by Theorem \ref{thm:jl}, it can be used to find approximate nearest neighbors in effective resistance metrics of graphs with $c\gets \gammaann := \log n$:

\begin{theorem}[Fact 2.7 in \cite{AI06}]\label{thm:ann}
Given a graph $G$ and a set of vertices $S\subseteq V(G)$, there is a data structure $D$ computed by an algorithm $D\gets \PreprocANN(G,S)$ with query algorithm $v'\gets \ANN_D(v)$. $\ANN_D$ takes any vertex $v\in V(G)$ as input and uses the data structure $D$ to return a vertex $v'\in S$ with the following properties with probability at least $1 - 1/n^{10}$:

\begin{itemize}
\item (Closeness) $\texttt{Reff}_G(v,v')\le \min_{u\in S} \gammaann\texttt{Reff}_G(v,u)$
\item (Preprocessing runtime) $\PreprocANN$ takes $\tilde{O}(m)$ time.
\item (Query runtime) $\ANN$ takes $\tilde{O}(1)$ time.
\end{itemize}
\end{theorem}

\subsubsection{Basic facts about random spanning trees}

Let $T\sim G$ denote the distribution over spanning trees of $G$ with each tree selected with probability proportional to $\prod_{e\in E(T)} c_e$. The following shows that conditioning on a partial sample is equivalent to modifying the input graph:

\begin{theorem}[\cite{MST15}]\label{thm:conditioning}
Consider a graph $G$ and a set of edges $F\subseteq E(G)$. Fix a spanning tree $T_0$ of $G$ and let $F_0 := E(T_0)\cap F$. Obtain a graph $H$ of $G$ by contracting all edges in $F_0$ and deleting all edges in $F\setminus F_0$. Then

$$\Pr_{T\sim G}[T = T_0 | E(T)\cap F = F_0] = \Pr_{T'\sim H}[T' = T_0/F_0]$$ 
\end{theorem}

For any set $F\subseteq E(G)$, let $H\sim G[F]$ denote the distribution over minors $H$ of $G$ obtained by sampling a tree $T\sim G$, contracting all edges in $F\cap E(T)$, and deleting all edges in $F\setminus E(T)$. We also use the following folklore fact extensively:

\begin{theorem}\label{thm:edge-marginal}
Consider a graph $G$ and an edge $e\in E(G)$. Then

$$\Pr_{T\sim G}[e\in E(T)] = c_e^G \texttt{Reff}_G(e)$$
\end{theorem}

\subsubsection{Schur complements}

\begin{definition}[Schur complements]\label{def:schur}
The \emph{Schur complement} of a graph $I$ with respect to a subset of its vertices $S\subseteq V(I)$, denoted $\texttt{Schur}(I,S)$, is the weighted graph $J$ with $V(J) = S$ with Laplacian matrix

$$L_J = L_I[S,S] - L_I[S,S^c] L_I[S^c,S^c]^{-1} L_I[S^c,S]$$

where $M[S_0,S_1]$ denotes the submatrix of a matrix $M$ with rows and columns indexed by $S_0$ and $S_1$ respectively.
\end{definition}

In the above definition, it is not immediately clear that $L_J$ is the Laplacian matrix of a graph, but it turns out to be one. Furthermore, the following associativity property holds:

\begin{remark}\label{rmk:schur}
For any two disjoint sets of vertices $S_0,S_1\in V(I)$ for some graph $I$,

$$\texttt{Schur}(\texttt{Schur}(I,S_0\cup S_1),S_0) = \texttt{Schur}(I,S_0)$$
\end{remark}

Also, Schur complements commute with edge deletions and contractions in the kept set $S$:

\begin{remark}[Lemma 4.1 of \cite{CDN89}]\label{rmk:com-schur}
Let $S$ be a set of vertices in a graph $G$ and $f\in E_G(S)$. Then,

$$\texttt{Schur}(G\setminus f,S) = \texttt{Schur}(G,S)\setminus f$$

and

$$\texttt{Schur}(G/f,S) = \texttt{Schur}(G,S)/f$$
\end{remark}

Schur complements also have the following combinatorial property, which is the only property we use of Schur complements before Section \ref{sec:fix-preliminaries}:

\begin{theorem}\label{thm:schur-walk-equiv}
Consider a graph $I$ and some set of vertices $S\subseteq V(I)$. Let $J = \texttt{Schur}(I,S)$. Pick a vertex $v\in S$ and generate two lists of vertices as follows:

\begin{itemize}
\item Do a random walk in $J$ starting at $v$ and write down the list of visited vertices.
\item Do a random walk in $I$ starting at $v$ and write down the list of visited vertices that are also in $S$.
\end{itemize}

These two distributions over lists are identical.
\end{theorem}

\newpage

\section{Structure of the Paper and the Proof of Theorem \ref{thm:main-result-aspect}}\label{sec:real-overview}

Now, we formally introduce the concepts that were alluded to in Section \ref{sec:fake-overview}. In the process, we outline the structure of the paper and reduce the main result (Theorem \ref{thm:main-result-aspect}) given the four main components of our algorithm: building shortcutters, selecting vertices to condition on, sampling, and computing a set of fixing edges.

Throughout this section, we use two key parameters: $\sigma_0$ and $\sigma_1$. These parameters should be thought of as distance and shortcutter size-related parameters respectively. While there are other constants (like the $\mu$s, which are all $m^{o(1)}$), these constants are purely determined by proofs in the main sections. Only $\sigma_0$ and $\sigma_1$ are traded off in order to bound the runtime of the main algorithm $\ExactTree$. For more details on parameter values, see Appendix \ref{sec:constants}.

\subsection{Our shortcutting data structure}

Recall that in Section \ref{sec:fake-overview}, we stated that no vertices were in more than $m^{o(1)}$ different shortcutters. Here, we organize the shortcutters into a small number of families of disjoint shortcutters, which we call \emph{clans}, in order to achieve this property.

\linestart

\begin{definition}[Organization of shortcutters]

Consider a graph $H$ obtained as a minor of $G$. A \emph{cluster} is a set of vertices. In our algorithm, there are three kinds of clusters: parts, cores and shortcutters. We define parts in Definition \ref{def:covering-hordes-overlays}. A \emph{core} is an arbitrary cluster. A \emph{shortcutter} is a cluster $S_C$ that contains a core $C$ of vertices that are ``assigned'' to it. A \emph{clan} is a set of (vertex-)disjoint shortcutters. A \emph{horde} is set of clans.

\end{definition}

All hordes in our algorithm satisfy the following invariant:

\begin{invariant}\label{inv:numclans}
A horde $\mc H$ consists of at most $\ellmax\le m^{o(1)}$ clans.
\end{invariant}

\begin{definition}[Covering hordes and overlay partitions]\label{def:covering-hordes-overlays}

A horde $\mc H$ is said to \emph{cover} $H$ if each vertex in $H$ is in the core of some shortcutter in some clan of $\mc H$.

Given a collection of covering hordes $\{\mc H_i\}_{i=1}^{\sigma_0}$, the \emph{overlays} $\mc P_i(\{\mc H_i\}_{i=1}^{\sigma_0})$ are formed by refining all cores of shortcutters from all clans in $\cup_{j\ge i} \mc H_j$. More precisely, let $\chi_i$ denote the equivalence relation formed by letting $u\sim_{\chi_i} v$ if and only if for all clans $\mc C\in \cup_{j\ge i} \mc H_j$, $u$ and $v$ are either (a) both in the same core of $\mc C$ or (b) both not in any core of $\mc C$. Let $\mc P_i(\{\mc H_i\}_{i=1}^{\sigma_0})$ denote the equivalence classes of $\chi_i$.

Since all $\mc H_i$s are covering, each $\mc P_i(\{\mc H_i\}_{i=1}^{\sigma_0})$ is a partition of $V(H)$. A \emph{part} $P$ is some cluster in $\mc P_i(\{\mc H_i\}_{i=1}^{\sigma_0})$ for some $i\in [\sigma_0]$. Each part $P\in \mc P_i(\{\mc H_i\}_{i=1}^{\sigma_0})$ is assigned to a single core $C_P$ of a shortcutter $S_P$ in some clan of $\mc H_i$.

Let $\partial \mc P_i(\{\mc H_i\}_{i=1}^{\sigma_0})$ denote the set of boundary edges of parts in $\mc P_i(\{\mc H_i\}_{i=1}^{\sigma_0})$.
\end{definition}

\lineend

\begin{figure}
\includegraphics[width=1.0\textwidth]{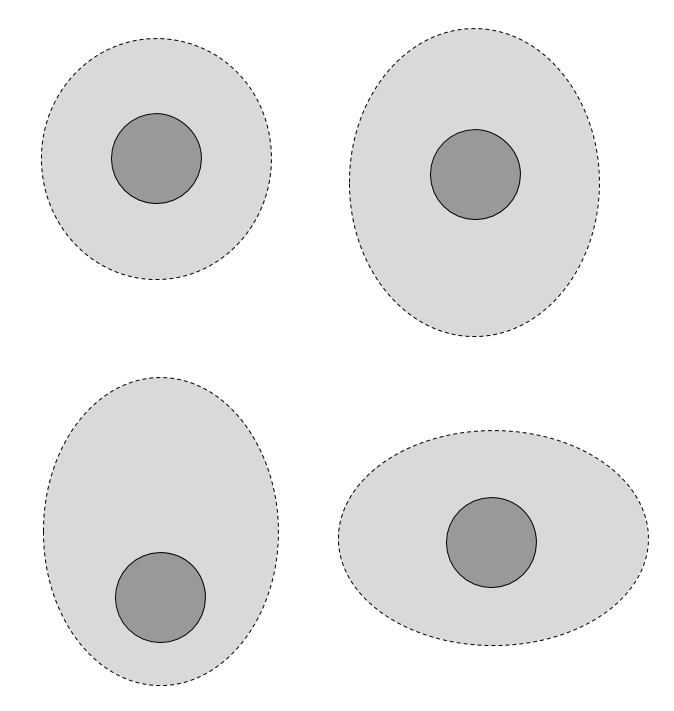}
\caption{Visualizing shortcutters and cores within a clan. Each shortcutter (light gray) contains a core (dark gray) of vertices for whom the shortcutter is usable.}
\label{fig:shortcutters-cores}
\end{figure}

Organizing shortcutters into clans allows us to define properties that hold for shortcutters in a clan ``on average.'' Now, we define various properties that cores, shortcutters, clans, and hordes should have. After defining these properties, we summarize their relevance to bounding the shortcutted random walk simulation time in Table \ref{tab:rw-runtime}.

For each of these definitions, fix a distance scale $R$. We start by insisting that each core consist of closeby vertices.

\begin{definition}[$R$-clans]
Call a clan an \emph{$R$-clan} if each shortcutter's core has $H$-effective resistance diameter at most $R$.
\end{definition}

$R$ may be referred to even if a clan is not an $R$-clan (i.e. the clan may not have bounded diameter cores).

Each clan contains shortcutters that are relatively similar to one another. This way, our analysis of the shortcutting scheme can focus on the clans within a horde independently. Specifically, a clan $\mc C$ is said to be \emph{bucketed} if the maximum size of a shortcutter in $\mc C$ is at most $4m/|\mc C|$.

Inverting this definition suggests a more convenient definition of the size of a clan.

\begin{definition}[Effective size and bucketing]

The \emph{effective size} of a clan $\mc C$, denoted $s_{\mc C}$, is the following:

$$s_{\mc C} := \frac{m}{\max_{S_C\in \mc C} |E(S_C)\cup \partial S_C|}$$

We say that a clan $\mc C$ is \emph{bucketed} if

$$|\mc C|\le 4 s_{\mc C}$$

\end{definition}

Clans also contain shortcutters with the property that using a shortcutter bypasses many random walk steps. Specifically, the \emph{conductance} of a shortcutter is relevant for assessing how many times it is used, as discussed in Section \ref{sec:fake-overview}. For an arbitrary graph $H'$, let $c^{H'}(S_C)$, the conductance of $S_C$ with respect to $H'$, be

$$c^{H'}(S_C) := \sum_{e\in E(C,V(H')\setminus S_C)} c_e^{\texttt{Schur}(H',C\cup (V(H')\setminus S_C))}$$

We define the conductance with respect to $H'$, not $H$, because we need to delete edges from $H$ in order to maintain the condition that $c^{H'}(S_C)$ is low after conditioning. $H'$ will be a graph obtained by deleting some edges $\texttt{deleted}(\mc C)$ from $H$:

\begin{definition}[Deletion set and the deletion set condition]\label{def:del-set-con}
For a clan $\mc C\in \mc E_i$, maintain a set $\texttt{deleted}(\mc C)$ of edges. This set must satisfy the \emph{deletion set condition}, which states that no deleted edge is incident with a nonempty part. Specifically, for any $P\in \mc P_i(\mc E)$ for which $E(P)\ne \emptyset$,

$$\texttt{deleted}(\mc C)\cap (\partial_H P) = \emptyset$$
\end{definition}

The deletion set condition ensures that precomputed random walk steps to the endpoint of a deleted edge cross a boundary edge of some part in $\mc P_i(\mc E)$. We exploit this in Section \ref{sec:random-walk}.

The following condition is used to bound the precomputation work during the shortcutting algorithm:

\begin{definition}[Modifiedness]
We say that a clan $\mc C$ is \emph{$\tau$-modified} if the number of deleted edges is not too high on average:

$$|\texttt{deleted}(\mc C)|\le \tau m^{1/\sigma_1} s_{\mc C}$$
\end{definition}

For a clan $\mc C$, let $H_{\mc C} := H\setminus \texttt{deleted}(\mc C)$. For a shortcutter $S_C\in \mc C$, let $c^{\mc C}(S_C) = c^{H_{\mc C}}(S_C)$.

\begin{definition}[Conductivity]
A clan $\mc C$ is \emph{$\zeta$-conductive} if

$$\sum_{S_C\in \mc C} c^{\mc C}(S_C)\le \frac{\zeta m^{1/\sigma_1} s_{\mc C}}{R}$$
\end{definition}

\begin{figure}
\includegraphics[width=1.0\textwidth]{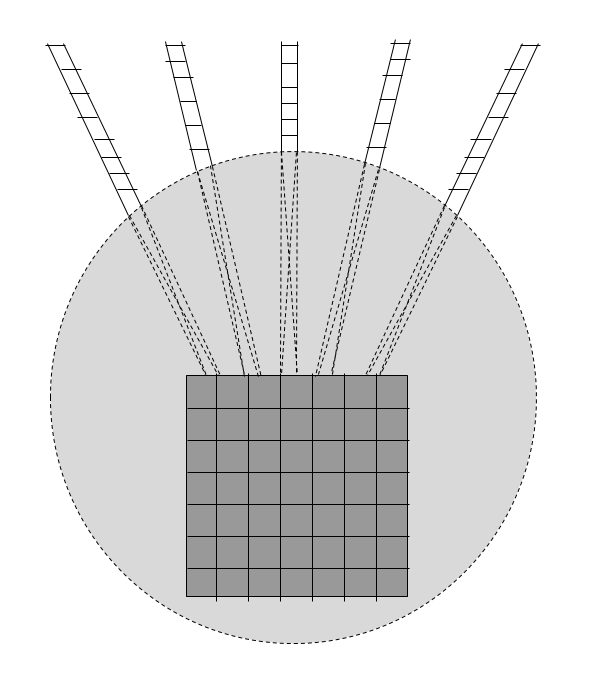}
\caption{The conductivity of the light gray shortcutter is determined by the conductance of the dashed edges, which are obtained by Schur complementing all vertices in the shortcutter with its core (dark gray) removed. The conductance of these edges is relevant for assessing the quality of a shortcutter because (1) doing a random walk on a Schur complement is equivalent to doing a random walk on the original graph and removing all eliminated vertices from the visit list and (2) Lemma \ref{lem:walk-bound}.}
\label{fig:shortcutter-largeness-real}
\end{figure}

The $\zeta$-conductive requirement is a way of saying that shortcutters within a clan are large on average. However, we also need a way of saying that they are not too large. If they are too large, the set of vertices that we are conditioning on may intersect too many shortcutters from another clan. View these vertices as being clustered into a small number of low effective resistance diameter balls and focus on each ball $C'$ one at a time. If $C'$ is close to the core of some shortcutter $S_C$, then $C'$ does not need to be removed from $S_C$ to make $S_C$ carved with respect to $C'$. Therefore, we only need to delete $C'$ from $S_C$ when $C'$ and $C$ are well-separated. This motivates the notion of \emph{ties}.

Consider an $R$-clan $\mc C$ and a shortcutter $S_C\in \mc C$. Let $C'$ be a cluster with $H_{\mc C}$-effective resistance diameter at most $\beta R$ for some $\beta\ge 1$. We say that $C'$ is \emph{tied} to $S_C$ if both of the following hold:

\begin{itemize}
\item (Intersection) $C'$ intersects $S_C$.
\item (Well-separatedness) $\min_{u\in C, v\in C'} \texttt{Reff}_{H_{\mc C}}(u,v) > \beta_0\beta R$, where $\beta_0 = 100$.
\end{itemize}

\begin{definition}[Well-spacedness]
An $R$-clan $\mc C$ is \emph{well-spaced} if no cluster $C'\subseteq V(H)$ is tied to more than one shortcutter $S_C\in \mc C$.
\end{definition}

The lack of ties for well-spaced clusters ensures that deleting $C'$ from all shortcutters in $\mc C$ does not increase the total conductance of shortcutters in $\mc C$ much.

All of the definitions that we have discussed leading up to this are used to show that conditioning once takes at most $O(m^{1 + o(1)}\alpha^{o(1)})$ time. Recall from Section \ref{sec:fake-overview} that sampling the intersection of a random tree with $E(S)$ for some set of vertices $S$ is supposed to allow us to get rid of some shortcutters because the boundary of $S$ is small.

\begin{definition}[Boundedness]
Say that a clan associated with distance scale $R$ is \emph{$\kappa$-bounded} if

$$\sum_{S_C\in \mc C} c^H(\partial C) \le \frac{\kappa m}{R}$$
\end{definition}

We now extend our definitions of clans to hordes. A horde $\mc H$ is an \emph{$R$-horde} if each clan in $\mc H$ is an $R$-clan. A horde $\mc H$ is \emph{bucketed} if each clan in it is bucketed. A horde is \emph{$\tau$-modified} if each clan in it is $\tau$-modified. A horde $\mc H$ is \emph{$\zeta$-conductive} if each clan in it is $\zeta$-conductive. A horde is \emph{well-spaced} if each of its clans are well-spaced. A horde is \emph{$\kappa$-bounded} if each of its clans is $\kappa$-bounded. A horde satisfies the deletion set condition if each of its clans satisfies it.

We now give definitions that are specific to hordes and to collections of hordes. Each vertex needs to have a shortcutter at every distance scale. Since a horde is associated with one distance scale $R$, each vertex should have a shortcutter in each horde. Now, we define a special collection of hordes called an \emph{empire} with which sampling can be performed:

\begin{definition}[Empires]
An \emph{empire} $\mc E$ is a set of covering hordes $\{\mc E_i\}_{i=1}^{\sigma_0}$, with $\mc E_i$ being an $\murad\R{i}$-horde. Define bucketedness, $\tau$-modifiedness, $\zeta$-conductivity, well-spacing, $\kappa$-boundedness, and the deletion set condition for empires as well if these conditions hold for all constituent hordes.
\end{definition}

Now, we show how these properties fit together to bound the runtime of our implementation of the full shortcutting meta-algorithm described in Section \ref{sec:fake-overview}. When the random walk is at a vertex $u$, the meta-algorithm first finds the maximum $i\in [\sigma_0]$ for which the intersection of a shortcutter $S_{P_i}$ with the set $S$ is covered, where $P_i\in \mc P_i(\mc E)$ is the unique part containing $u$. If $E(P_i) = \emptyset$, it does a standard random walk step. Otherwise, it samples whether the random walk hits an endpoint of an edge in $\texttt{deleted}(\mc C)$ before exiting $S_{P_i}$, where $\mc C$ is the clan containining $S_{P_i}$. If so, it uses offline shortcutting to shortcut to $\texttt{deleted}(\mc C)$. Otherwise, it uses online shortcutting to shortcut to the boundary of $S_{P_i}$.

The above discussion cites three kinds of random walk-related work and one kind of precomputation work. To bound the random walk-related work, we exploit Lemma \ref{lem:walk-bound}. Lemma \ref{lem:walk-bound} requires two things: a bound on conductance and a bound on the distance to an unvisited vertex. Work done using a part $P_i$ is charged to the clan containing $S_{P_i}$ as follows:

\begin{table}[H]
\begin{center}
\begin{tabular}{ c | c | c }
    \hline
    Type of work & Edge to charge to & Reason for charging\\ \hline
    Walk across $\partial P_i$ & $P_i$ boundary edge & Normal random walk step\\ \hline
    Shortcut to $\texttt{deleted}(\mc C)$ & $P_i$ boundary edge & Deletion set condition\\ \hline
    Shortcut to $\partial S_{P_i}$ & An edge in $\texttt{Schur}(H_{\mc C}, C_{P_i}\cup V(H)\setminus S_{P_i})$ & Theorem \ref{thm:schur-walk}\\ \hline
\end{tabular}

\vspace{0.2in}

\begin{tabular}{ c | c | c }
    \hline
    Type of work & \parbox{.3\textwidth}{\centering Conductance bound for work charged to a clan $\mc C$} & Distance to unvisited vertex\\ \hline
    Walk across $\partial P_i$ & $\kappa m/(\R{i})$ (boundedness) & $\R{i+1}$ (carving, $R$-clan)\\ \hline
    Shortcut to $\texttt{deleted}(\mc C)$ & $\kappa m/(\R{i})$ (boundedness) & $\R{i+1}$ (carving, $R$-clan)\\ \hline
    Shortcut to $\partial S_{P_i}$ & $\zeta m^{1/\sigma_1} s_{\mc C}/(\R{i})$ (conductivity) & $\R{i+1}$ (carving, $R$-clan)\\ \hline
\end{tabular}

\vspace{0.2in}

\begin{tabular}{ c | c | c | c }
    \hline
    Type of work & Total steps & Work per step & Total work\\ \hline
    Walk across $\partial P_i$ & $\kappa m \alpha^{1/(\sigma_0+1)}$ (Lemma \ref{lem:walk-bound}) & $O(1)$ & $\le m^{1 + o(1)}\alpha^{o(1)}$\\ \hline
    Shortcut to $\texttt{deleted}(\mc C)$ & $\kappa m \alpha^{1/(\sigma_0+1)}$ (Lemma \ref{lem:walk-bound}) & $\tilde{O}(1)$ (offline) & $\le m^{1 + o(1)}\alpha^{o(1)}$\\ \hline
    Shortcut to $\partial S_{P_i}$ & $\zeta m^{1/\sigma_1} s_{\mc C}\alpha^{1/(\sigma_0+1)}$ (Lemma \ref{lem:walk-bound}) & $\tilde{O}(\frac{m}{s_{\mc C}})$ (online) & $\le m^{1+o(1)}\alpha^{o(1)}$\\ \hline
    Precomputation & $\tau m^{1/\sigma_1} s_{\mc C}$ (modifiedness) & $\tilde{O}(\frac{m}{s_{\mc C}})$ & $\le m^{1 + o(1)}$\\ \hline
\end{tabular}
\end{center}
\label{tab:rw-runtime}
\caption{Accounting for work during each sampling round}
\end{table}

This table does not discuss the well-spacedness or bucketing conditions. Well-spacedness is used to bound the conductivity increase due to carving in the proof of Lemma \ref{lem:holes}, while the bucketing condition is used to bound the number of edges added to $\texttt{deleted}(\mc C)$ in the proof of Lemma \ref{lem:clan-fix}.

\subsection{Creating and maintaining shortcutters}

Our algorithm maintains an empire $\mc E$. Before each conditioning phase, it recomputes shortcutters in order to endow them with properties that are lost after one round of conditioning:

\begin{restatable}{lemma}{lemmaintainshortcutters}\label{lem:maintain-shortcutters}
There is an almost-linear time algorithm $\CreateEmpire(\{\mc H_i\}_{i=1}^{\sigma_0})$ that, when given a set of covering hordes $\{\mc H_i\}_{i=1}^{\sigma_0}$ with $\mc H_i$ associated with distance scale $\R{i}$ in a graph $H$, returns an empire $\mc E = \{\mc H_i'\}_{i=1}^{\sigma_0}$ with the following properties:

\begin{itemize}
\item (Bucketing) $\mc E$ is bucketed.
\item (Conductivity) If each horde $\mc H_i$ is $\zeta$-conductive, then $\mc E$ is $(8\log n)\zeta + (16\log n)\muapp$-conductive.
\item (Well-spacedness) $\mc E$ is well-spaced.
\item (Boundedness) If each horde $\mc H_i$ is $\kappa$-bounded, then $\mc E$ is $\kappa + \kappa_0$ bounded, for $\kappa_0\le m^{o(1)}$.
\item (Modifiedness and deletion set condition) If each horde $\mc H_i$ is $\tau$-modified, then $\mc E$ is $\tau$-modified as well. Furthermore, if the deletion set condition is satisfied in each clan of each $\mc H_i$, it continues to be satisfied in $\mc E$.
\item (Clan growth) The number of clans in $\mc E$ is at most $\muapp\log n$ times as high as the number of clans in all of the $\mc H_i$s.
\item (Containment) For any $i\in [\sigma_0]$, consider any part $P\in \mc P_i(\mc E)$. There is a unique part $Q\in \mc P_i(\{\mc H_j\}_j)$ for which $P\subseteq Q$. Furthermore, $C_P\subseteq C_Q$ and $S_P\subseteq S_Q$.
\end{itemize}
\end{restatable}

Our spanning tree generation algorithm starts by calling $\CreateEmpire$ on the set of hordes consisting of one clan, each of which just contains the one shortcutter $V(G)$. These hordes are clearly covering and have $\zeta = 0$, $\kappa = 0$, and $\tau = 0$. $\CreateEmpire$ is useful to call on the remnants of empires after conditioning later on in order to achieve the containment property. Containment is essential to our notion of progress, as discussed in Section \ref{sec:fake-overview}.

\subsection{Selecting parts to condition on}

Given an empire $\mc E$ with respect to a graph $H$, we can choose a set of vertices $S$ to condition on. The set $S$ is small enough that, when carved out of shortcutters in $\mc E$, does not increase their conductivity too much. The upside of carving is that each vertex in $S$ is close to the core of any shortcutter that it is in. We now define this precisely:

\linestart

\begin{definition}[Active parts and carving]

A part $P$ is called \emph{active} if it has nonempty interior, i.e. $E(P) \ne \emptyset$. A shortcutter is called \emph{active} if any part assigned to it is active.

A shortcutter $S_C$ in an $R$-clan has been \emph{carved} with respect to $S\subseteq V(G)$ if each vertex $v\in S\cap S_C$ is within $H$-effective resistance distance $\mucarve R$ of all vertices in $C$. An $R$-clan $\mc C$ in an empire $\mc E$ has been \emph{carved} with respect to $S$ if all of its active shortcutters have been carved with respect to $S$. An $R$-horde $\mc H$ in an empire $\mc E$ has been \emph{carved} with respect to $S$ if each clan in it has been carved with respect to $S$. An empire $\mc E$ has been \emph{carved} with respect to $S$ if each of its hordes has been carved with respect to $S$.

\end{definition}

\lineend

The routine $\ConditioningVerts$ both (a) selects parts $\mc K$ for conditioning on and (b) removes vertices from the shortcutters of the input empire $\mc E$ in order to ensure that it is carved with respect to $\cup_{P\in \mc K} P$. The $\ConditioningVerts$ subroutine maintains internal state and is the only method that exploits the ``Containment'' guarantee of Lemma \ref{lem:maintain-shortcutters}. The ``Progress'' input condition in the following definition captures the fact that partial sampling eliminates edges in the induced subgraph of the previously chosen parts:

\begin{definition}[$\ConditioningVerts$ input conditions]\label{def:cond-verts-input}
Given an empire $\mc E$ in a graph $H$, the algorithm $\ConditioningVerts(\mc E)$ returns a set of parts $\mc K$ to condition on and removes vertices from the shortcutters in the empire $\mc E$ to obtain $\mc E'$. Let $\mc E_{prev}$ be the argument supplied to the previous call to $\ConditioningVerts$, let $\mc K_{prev} := \ConditioningVerts(\mc E_{prev})$, and let $\mc E_{prev}'$ be the empire $\mc E_{prev}$ after being modified by $\ConditioningVerts$. Let $H_{prev}$ be the graph in which $\mc E_{prev}$ lies. The following conditions are the \emph{input conditions} for $\ConditioningVerts$:

\begin{itemize}
\item (Parameters) $\mc E$ is a bucketed, $\zeta$-conductive, well-spaced, $\tau$-modified, and $\kappa$-bounded empire that satisfies the deletion set condition.
\item (Containment) For any $i\in [\sigma_0]$, consider any part $P\in \mc P_i(\mc E)$. There is a unique part $Q\in \mc P_i(\mc E_{prev})$ for which $P\subseteq Q$. Furthermore, $C_P\subseteq C_Q$ and $S_P\subseteq S_Q$.
\item (Progress) For each $P\in \mc K_{prev}$, $E_H(P) = \emptyset$.
\end{itemize}
\end{definition}

\begin{restatable}{lemma}{lemholes}\label{lem:holes}
Given an empire $\mc E = \{\mc E_i\}_{i=1}^{\sigma_0}$ in a graph $H$ that satisfies the input conditions given in Definition \ref{def:cond-verts-input}, $\ConditioningVerts(\mc E)$ returns a set of parts $\mc K$ to condition on and removes vertices from the shortcutters in the empire $\mc E$ to obtain $\mc E'$. Let $S = \cup_{P\in \mc K} P \subseteq V(H)$. Then the following guarantees are satisfied:

\begin{itemize}
\item (Conductivity) $\mc E'$ is a bucketed, $\tau$-modified, $\zeta + 10(\log m)\muapp(\ellmax + \tau)$-conductive, well-spaced, $\kappa$-bounded empire that satisfies the deletion set condition.
\item (Carving) $\mc E'$ is carved with respect to $S$.
\end{itemize}
\end{restatable}

\subsection{Making enough progress during each round of conditioning}

In the previous section, we showed that $S$ is small enough to ensure that carving $S$ out of all shortcutters in $\mc E$ does not increase the conductivity of $\mc E$ too much. We now show that $S$ is large enough to make a lot of progress. Specifically, we show the following:

\begin{restatable}{lemma}{lemprogress}\label{lem:progress}
Consider a sequence of calls $\mc K^j\gets \ConditioningVerts(\mc E^j)$ that modifies $\mc E^j$ to obtain $(\mc E^j)'$. Suppose that $H^j$ is the graph in which $\mc E^j$ is defined. Suppose that for each $j > 0$, $\mc E\gets \mc E^j$, $\mc E_{prev}\gets \mc E^{j-1}$, $\mc K_{prev}\gets \mc K^{j-1}$ satisfies the input conditions in Definition \ref{def:cond-verts-input}. Let

$$j_{final} = (2\sigma_1)^{2\sigma_0}$$

Then $E(H^{j_{final}}) = \emptyset$.
\end{restatable}

This means that only $(2\sigma_1)^{2\sigma_0}\le o(\log n)$ rounds of conditioning are necessary to sample a random spanning tree.

\subsection{Conditioning on the intersection of a random tree with the selected vertices}

Now that each shortcutter $S_C$ only intersects vertices to condition on that are close to $C$, we can make the idea for using online shortcutting in Section \ref{sec:fake-overview} a reality:

\begin{restatable}{lemma}{lempartialsample}\label{lem:partial-sample}
Let $\mc K \subseteq \cup_{i=1}^{\sigma_0} \mc P_i(\mc E)$ be a set of parts. Let $F = \cup_{P\in \mc K} E(P)$ and $S = \cup_{P\in \mc K} P$. Suppose that the empire $\mc E$ is $\zeta$-conductive, $\kappa$-bounded, $\tau$-modified, satisfies the deletion set condition, and has been carved with respect to $S$. Then, there is an algorithm $\PartialSample(\mc E,\mc K)$ that returns the intersection of a random spanning tree $T$ in $H$ with $F$ in $\tilde{O}(((\zeta + \kappa)\mucarve + \tau)\ellmax m^{1+1/\sigma_1}\alpha^{1/(\sigma_0+1)})$ time.
\end{restatable}

\subsection{Fixing shortcutters}

After computing $T\cap F\gets\PartialSample$, contracting all edges in $F\cap T$ in $H$, and deleting all edges in $F\setminus T$ from $H$, $\mc E$ is no longer an empire with respect to $H$. In particular, the well-spacedness, $\zeta$-conductivity, and core diameter conditions break down. Well-spacedness and diameter can be fixed by applying $\CreateEmpire$. However, the $\zeta$-conductivity constraint accumulates over an old value. We could recompute the empire from scratch, but that forgoes the containment property that is so important to establishing progress. We deal with this issue by adding edges to $\texttt{deleted}(\mc C)$ for each clan $\mc C$ in $\mc E$:

\begin{restatable}{lemma}{lemclanfix}\label{lem:clan-fix}
Let $H$ be a graph, $\mc E$ be an empire in $H$ and $\mc K$ be a set of parts. Let $S = \cup_{P\in \mc K} P$ and let $F = \cup_{P\in \mc K} E(P)$. Let $H'\sim H[F]$. Suppose that the following input conditions hold $\mc E$:

\begin{itemize}
\item (Bucketing) The empire $\mc E$ is bucketed.
\item (Carving) $\mc E$ is carved with respect to $S$.
\end{itemize}

With high probability over $H'$, $\FixShortcutters(\mc E, H',\mc K)$ adds edges to the deletion set of each clan of $\mc E$ to obtain a set of covering hordes $\{\mc H_i'\}$ with the following properties:

\begin{itemize}
\item (Boundedness) For each $i$, if $\mc E_i$ is $\kappa$-bounded, then $\mc H_i'$ is $\ell\kappa$-bounded, where $\ell = \sum_{i=1}^{\sigma_0} |\mc E_i|$.
\item (Modifiedness and deletion set condition) For each $i$, if $\mc E_i$ is $\tau$-modified and satisfies the deletion set condition, then $\mc H_i'$ is $\mumod (\tau + \zeta)$-modified and also satisfies the deletion set condition.
\item (Conductivity) For each $i$, if $\mc E_i$ is $\zeta$-conductive with respect to $H$, then $\mc H_i'$ is at most $7\zeta$-conductive with respect to $H'$.
\end{itemize}

Futhermore, it does so in $m^{1 + o(1)}$ time.
\end{restatable}

\subsection{An $m^{1+o(1)}\alpha^{o(1)}$ time algorithm for exact random spanning tree generation}

We now tie the results from the previous sections together to prove Theorem \ref{thm:main-result-aspect}. We prove this result using the algorithm $\ExactTree$, which simply chains the algorithms from the previous sections in order:

\begin{algorithm}[H]
\SetAlgoLined
\DontPrintSemicolon
\caption{$\ExactTree(G)$}

    $H\gets G$\;

    \tcp{the set of hordes which contain one clan consisting of one shortcutter (the entire graph)}
    $\mc E\gets \{\{\{V(G)\}\}\}_{i=1}^{\sigma_0}$\;

    $T\gets \emptyset$\;

    \While{$E(H)\ne\emptyset$}{

        $\mc E\gets \CreateEmpire(\mc E)$\;

        $\mc K\gets \ConditioningVerts(\mc E)$\;

        $T\gets T\cup \PartialSample(\mc E, \mc K)$\;

        Contract all edges in $H$ added to $T$ and delete all other edges internal to parts of $\mc K$\; \label{line:condition}

        $\FixShortcutters(\mc E, H, \mc K)$\;

    }

    \Return $T$\;
\end{algorithm}

Most of the effort in proving the above result boils down to checking that all of the input conditions are satisfied for each of the subroutines that $\ExactTree$ calls.

\begin{proof}[Proof of Theorem \ref{thm:main-result-aspect}]

\underline{Invariant \ref{inv:numclans}.} Each of $\CreateEmpire$, $\ConditioningVerts$, $\PartialSample$, and $\FixShortcutters$ increases the number of clans by at most a factor of $(\log m)\muapp$. By Lemma \ref{lem:progress}, only $(2\sigma_1)^{2\sigma_0}$ iterations take place. Since there is only one clan initially, the number of clans at the end is at most $$((\log m)\muapp)^{(2\sigma_1)^{2\sigma_0}} = \ellmax\le m^{o(1)}$$ as desired.

\underline{$\kappa\le \kappamax$.} Each of the subroutines called in the while loop increases $\kappa$ by at most a factor of $\ellmax$ and additively by at most $\kappa_0\le m^{o(1)}$. Therefore, 

\begin{align*}
\kappa &\le (\ellmax)^{(2\sigma_1)^{2\sigma_0}}\\
&\le ((\log m)\muapp)^{(2\sigma_1)^{4\sigma_0}}\\
&= \kappamax\\
&\le m^{o(1)}\\
\end{align*}

as desired.

\underline{$\tau\le \taumax$ and $\zeta\le \zetamax$.} Each subroutine call increases $\max(\tau,\zeta)$ by a factor of at most $10(\log m)\muapp$ and additively by at most $10(\log m)\muapp\ellmax\mumod$. Therefore,

\begin{align*}
\max(\tau,\zeta) &\le (10(\log m)\muapp\ellmax\mumod)^{(2\sigma_1)^{2\sigma_0}}\\
&\le ((\log m)\muapp)^{(2\sigma_1)^{8\sigma_0}}\\
&= \max(\taumax,\zetamax)\\
&\le m^{o(1)}\\
\end{align*}

as desired.

\textbf{Well-definedness.} Start with $\CreateEmpire$. At the beginning of the algorithm, $\zeta = 0$, $\kappa = 0$, and all of the deletion sets are empty, so the deletion set condition is satisfied. $\mc E$ is not an empire when it is supplied to $\CreateEmpire$, but is a set of covering hordes because either (a) this is the first iteration and the cores are all $V(G)$ or (b) Lemma \ref{lem:clan-fix} states that the hordes $\mc H_i$ are covering. Therefore, $\CreateEmpire$'s input conditions given in Lemma \ref{lem:maintain-shortcutters} are always respected.

Next, consider $\ConditioningVerts$. The ``Parameters'' condition is the ``Parameters'' guarantee from Lemma \ref{lem:maintain-shortcutters}. The ``Containment'' condition follows from the ``Containment'' guarantee of Lemma \ref{lem:maintain-shortcutters}, along with the fact that $\FixShortcutters$ only adds to the deletion sets of the clans and $\PartialSample$ does not change $\mc E$. Line \ref{line:condition} of $\ExactTree$ contracts or deletes each edge internal to each part in $\mc K$. Therefore, the ``Progress'' condition is satisfied afterwards.

The desired parameter bounds for $\PartialSample$ are given in the ``Boundedness and covering'' guarantee of Lemma \ref{lem:holes}. The carving condition of Lemma \ref{lem:partial-sample} is the ``Carving'' guarantee of Lemma \ref{lem:holes}.

Finally, deal with $\FixShortcutters$. The input conditions for Lemma \ref{lem:clan-fix} are given directly as the ``Carving'' guarantee of Lemma \ref{lem:holes}, the ``Bucketing'' guarantee of Lemma \ref{lem:maintain-shortcutters}, and the fact that removing vertices from shortcutters preserves the bucketing guarantee.

\textbf{Correctness.} By Theorem \ref{thm:conditioning}, sampling a random tree in some $H$ is equivalent to partial sampling with $F = \cup_{P\in \mc K} E(P)$ and sampling a tree in the graph obtained by contracting the chosen edges in $F$ and deleting all others. By Lemma \ref{lem:partial-sample}, $\PartialSample$ returns a valid sample from a uniformly random spanning tree of $H$ intersected with $F$. Therefore, once $E(H) = \emptyset$, $T$ has been completely sampled and is valid.

\textbf{Runtime.} By Lemma \ref{lem:progress}, the while loop runs at most $(2\sigma_1)^{2\sigma_0}\le m^{o(1)}$ times. $\CreateEmpire$, $\ConditioningVerts$, $\PartialSample$, and $\FixShortcutters$ each take $m^{1+o(1)}\alpha^{o(1) + 1/(\sigma_0+1)}$ time by Lemmas \ref{lem:maintain-shortcutters}, \ref{lem:holes}, \ref{lem:partial-sample}, and \ref{lem:clan-fix} respectively and our bounds on $\ellmax$, $\taumax$, $\kappamax$, and $\zetamax$. Contracting and deleting edges only takes $O(m)$ time. Therefore, the entire algorithm only takes $O(m^{1+o(1)}\alpha^{o(1)+1/(\sigma_0+1)})$ time. Since $\sigma_0$ is superconstant, this runtime is $m^{1+o(1)}\alpha^{o(1)}$, as desired.
\end{proof}

\subsection{An $m^{1+o(1)}\ep^{-o(1)}$-time algorithm for generating a random spanning tree from a distribution with total variation distance $\ep$ from uniform}

In Section \ref{sec:apx-tree}, we give a simple reduction that proves Theorem \ref{thm:main-result-apply} given just Theorem \ref{thm:main-result-aspect}. The reduction samples the intersection of a random tree with a part of the graph with polynomial aspect ratio and smallest resistances. Conditioning on this part of the graph removes the edges with smallest resistance from the graph. A ball-growing-type technique and Theorem \ref{thm:main-result-aspect} ensures that each round of conditioning eliminates a number of edges from the graph proportional to the amount of work done.

\begin{figure}
\includegraphics[width=1.0\textwidth]{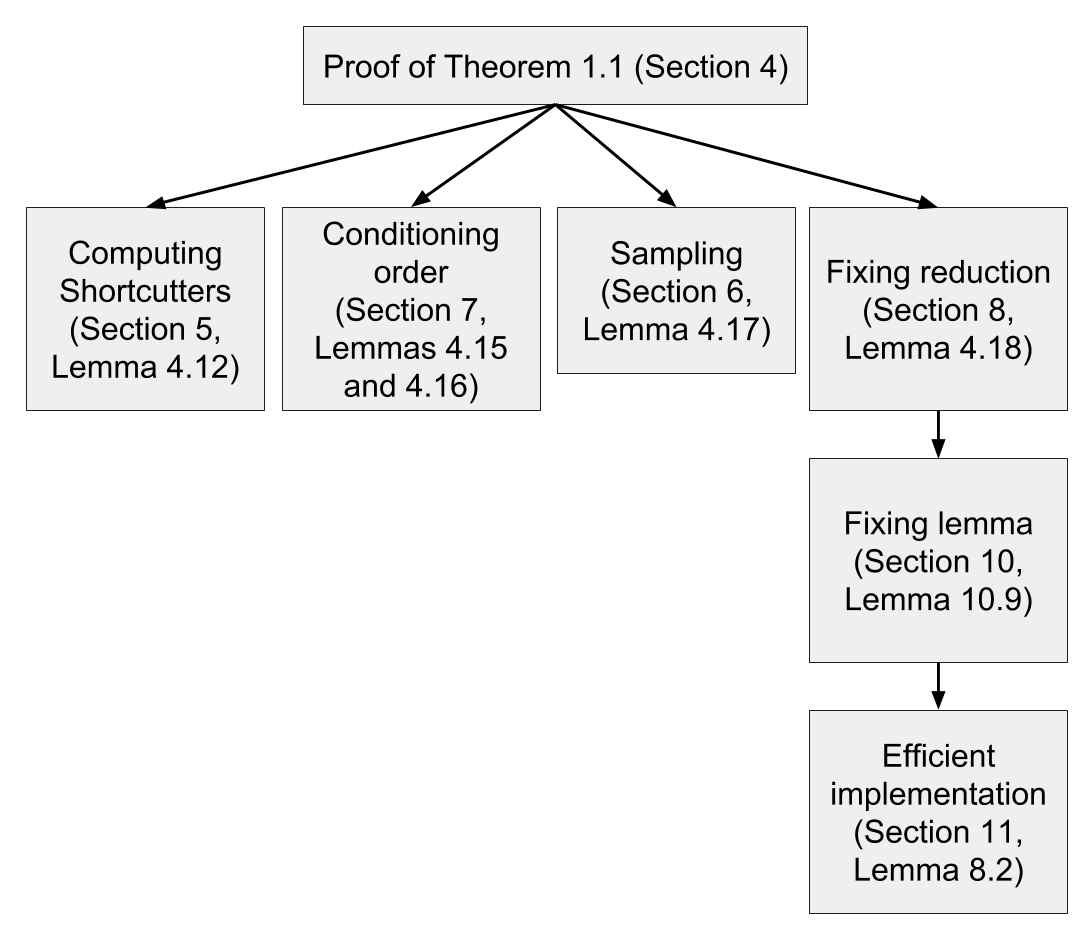}
\caption{Organization of the paper. Each section only depends on its children in the tree.}
\label{fig:paper-outline}
\end{figure}

\newpage

\section{Computing Shortcutters}\label{sec:compute-shortcutters}

In this section, we prove Lemma \ref{lem:maintain-shortcutters} by implementing $\CreateEmpire$. $\CreateEmpire$ starts by building cores in $\CoveringCommunity$. It then builds shortcutters around those cores in $\Voronoi$.

\subsection{Building cores}

In order to ensure that $\Voronoi$ can actually build good shortcutters around the cores, $\CoveringCommunity$ outputs cores that are organized into well-separated families. We start by giving some relevant definitions.

\begin{definition}[$\CoveringCommunity$-related definitions: families and communities]
Consider a graph $I$. A \emph{family} is a set of clusters $\mc F$. A \emph{community} is a set of families. An \emph{$R$-family} is a family of clusters with $I$-effective resistance diameter at most $R$. An \emph{$R$-community} is a community consisting of $R_{\mc F}$-families for possibly different values $\frac{R}{\murad}\le R_{\mc F}\le R$. An \emph{$(R,\gamma)$-well-separated family} is an $R$-family $\mc F$ with the additional property that the $I$-effective resistance distance between any two vertices in different clusters in $\mc F$ is at least $\gamma R$. A \emph{$\gamma$-well-separated community} is a community that consists of $(R_{\mc F},\gamma)$-well-separated families $\mc F$.
\end{definition}

Notice that in this definition, $\gamma$ is constant across all families but $R_{\mc F}$ is not. Well-separatedness is important for obtaining $\zeta$-conductive shortcutters. However, Lemma \ref{lem:maintain-shortcutters} also demands that cores have small total boundary. The boundary size is judged based on a cluster $C$ coming from remains of a former empire. This motivates the definition of $X$-constraint.

\begin{definition}[$\CoveringCommunity$-related definitions: $X$-constraint and boundedness]
Say that a community is \emph{$X$-constrained} if all vertices in clusters within families of the community are in $X$. Say that an $X$-constrained $R$-community $\mc D$ is \emph{$\kappa$-bounded} if for any family $\mc F\in \mc D$,

$$\sum_{C\in \mc F} c^I(\partial C) \le \frac{\kappa|E_I(X)\cup \partial_I X|}{R}$$
\end{definition}

We construct cores using the following result, which is proven in Section \ref{sec:covering-community}:

\begin{restatable}{lemma}{lemcoveringcommunity}\label{lem:covering-community}
The algorithm $\CoveringCommunity_D(X,I,R)$, when given a cluster $X$, a graph $I$, a radius $R$, and a Johnson-Lindenstrauss embedding $D$ of the vertices of $V(I)$, returns an $\murad R$-community $\mc D$ with the following properties:

\begin{itemize}
\item (Input constraint) $\mc D$ is $X$-constrained.
\item (Covering) Each vertex in $X$ is in some cluster of some family in $\mc D$.
\item (Boundedness) Each family $\mc F\in \mc D$ satisfies

$$\sum_{C\in \mc F} c^I(\partial_I C)\le \frac{\kappa_0|E_I(X)\cup \partial_I X|}{R} + c^I(\partial_I X)$$

\item (Well-separatedness) $\mc D$ is $\gammads$-well-separated.
\item (Number of families) $\mc D$ has at most $\muapp$ families.
\end{itemize}

Furthermore, $\CoveringCommunity$ takes almost-linear time in $|E(X)\cup\partial X|$.
\end{restatable}

\begin{figure}
\includegraphics[width=1.0\textwidth]{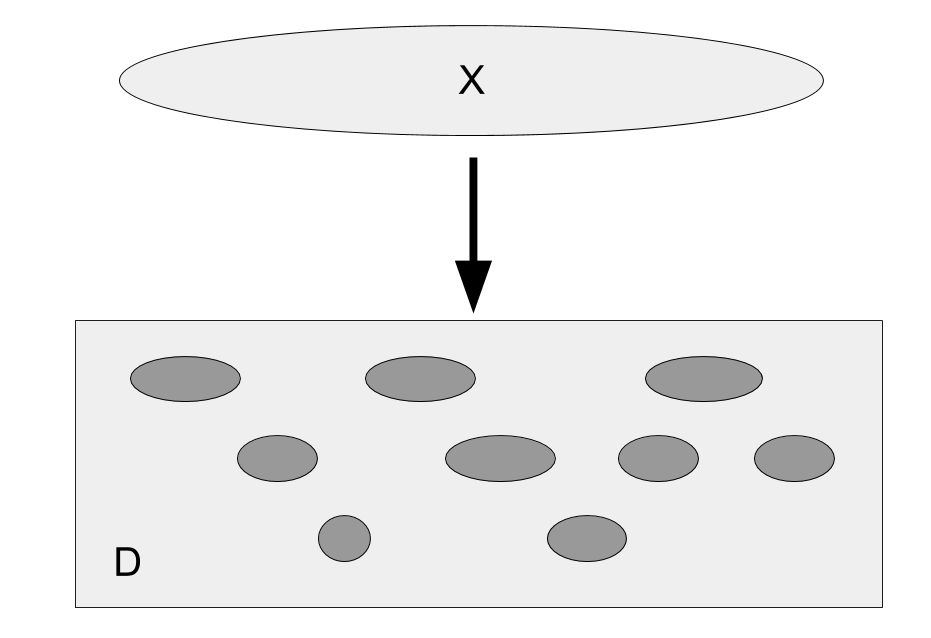}
\caption{$\CoveringCommunity$ splits the set $X$ into a small number of well-separated families of clusters in the effective resistance metric. $\CoveringCommunity$ is applied to each core of each clan independently. This splits each clan into $m^{o(1)}$ clans.}
\label{fig:building-cores-summary}
\end{figure}

$D$ is only given as input to $\CoveringCommunity$ for runtime purposes.

We now briefly discuss how each of these properties relates to properties of cores in Lemma \ref{lem:maintain-shortcutters}. We apply $\CoveringCommunity(X,I,R)$ to each core $C$ in the $R$-clain $\mc C$ with $I\gets H_{\mc C}$. The input, covering, and number of families constraints of Lemma \ref{lem:covering-community} relate to the parts of the containment constraint of Lemma \ref{lem:maintain-shortcutters} as it relates to $C$. The boundedness constraint relates to the boundedness constraint of Lemma \ref{lem:covering-community}. The well-separatedness constraint is used later on to obtain a $\zeta$-conductive clan of shortcutters.

$\CoveringCommunity$ is very similar to the sparse cover constructions (for example \cite{AP90}). When growing each cluster, though, (1) consider all nearby vertices for adjacent growth, not just adjacent balls and (2) ball-grow the resulting cluster afterwards to ensure that it has small boundary size. We give this construction in the appendix.

\subsection{Building the shortcutters}

Now we exploit the well-separatedness of the cores to build good shortcutters. We start by showing a result that allows us to translate well-separatedness into a shortcutter conductance upper bound:

\begin{restatable}{lemma}{lemwellseplevscore}\label{lem:well-sep-lev-score}
Consider a $\gammads = 2^{(\log n)^{2/3}}$ well-separated $R$-family of clusters $\mc F$ in a graph $G$. Let $H := \texttt{Schur}(G,\cup_{C\in \mc F} C)$. Then

$$\sum_{e\in E(C,C'), C\ne C'\in \mc F} \frac{\texttt{Reff}_H(e)}{r_e^H} \le \muapp |\mc F|$$
\end{restatable}

We prove this lemma in the appendix. We now give a brief description of the proof. If each cluster is a single vertex, the result is equivalent to Foster's Theorem (Remark \ref{rmk:lev-sum}). One can prove Remark \ref{rmk:lev-sum} combinatorially by running the Aldous-Broder algorithm on $H$ and writing down the sequence of relevant cluster visits; i.e. visits that end up adding a new edge to the tree. This sequence has the property that no two clusters can alternate more than once; otherwise at least one would be covered. Such sequences have been investigated before; they are \emph{Davenport-Schinzel sequences}. In particular, Davenport-Schinzel sequences are known to have linear length in the number of letters. This means that only a linear number of edges can be added.

This reasoning generalizes to the case where clusters are not single vertices. If the random walk alternates between two clusters more than $\log_{\gammads} n$ times, it covers one of them with probability at least $1 - 1/n$. Therefore, the sequence of visits is $\log_{\gammads} n$-Davenport-Schinzel. Picking $\gammads$ to be a high enough subpolynomial value gives a linear bound on the length of the sequence. There are some minor complications caused by the fact that a cluster can appear multiple times in a row in the sequence, so it is not truly Davenport-Schinzel. We describe how to cope with this issue in the proof of Lemma \ref{lem:well-sep-lev-score}.

If we could compute the Schur complement $H$, we could directly run Aldous-Broder on this Schur complement and sample the intersection of a random tree with all of intracluster edges in $\mc F$. This is expensive, though. Furthermore, the resulting graph is no longer sparse. Instead, we build shortcutters around each cluster in order to make it so that using a shortcutter effectively simulates one step of the Schur complement random walk in $\mc F$. Intuitively, we design a shortcutter around each core that has the property that using it takes the random walk from a cluster $C$ to a vertex from which it is more likely to hit some cluster besides $C$ before returning to $C$. This intuition motivates the following definition:

\begin{definition}[Potential level sets]
Consider a family of clusters $\mc F$ in a graph $G$ and a cluster $C\in \mc F$. Let $S_{\mc F}(p,C)$ denote the cluster of vertices $v\in V(G)$ for which

$$\Pr_v[t_C < t_{\mc F\setminus \{C\}}] \ge 1 - p$$
\end{definition}

Due to the interpretation of probabilities as electrical potentials, this set can be computed using one Laplacian solve on $G$ with each cluster in $\mc F$ identified to a vertex:

\begin{remark}\label{rmk:level-set-lsolve}
For any family $\mc F$ in a graph $G$, a cluster $C\in \mc F$, and a number $p\in (0,1)$. Then a set $S'\subseteq V(G)$ with $S_{\mc F}(p - 1/(m\alpha)^4,C)\subseteq S'\subseteq S_{\mc F}(p + 1/(m\alpha)^4,C)$ can be computed in near-linear time in $|E(G\setminus (\cup_{C'\in \mc F} C'))|$.
\end{remark}

The upper bound on $S_C$ in the following lemma is used to show that $\mc C$ is well-spaced; i.e. that the shortcutters are far away from one another. The lower bound is used to show that $\mc C$ is $\zeta$-conductive for some $\zeta\le m^{o(1)}\alpha^{o(1)}$; i.e. that using those shortcutters saves a lot of work. We show the following in Section \ref{sec:voronoi}:

\begin{restatable}{lemma}{lemvoronoi}\label{lem:voronoi}
The algorithm $\Voronoi(I,\mc F)$ takes a family $\mc F$ in the graph $I$ and outputs a clan $\mc C$ in near-linear time in $|E(Z)\cup \partial Z|$ with the property that for each $C\in \mc F$, there is a shortcutter $S_C\in \mc C$ with the property that $S_{\mc F}(1/(8\log n),C)\subseteq S_C\subseteq S_{\mc F}(1/8,C)$, where $Z = V(I)\setminus (\cup_{C\in \mc F} C)$.
\end{restatable}

\begin{figure}
\includegraphics[width=1.0\textwidth]{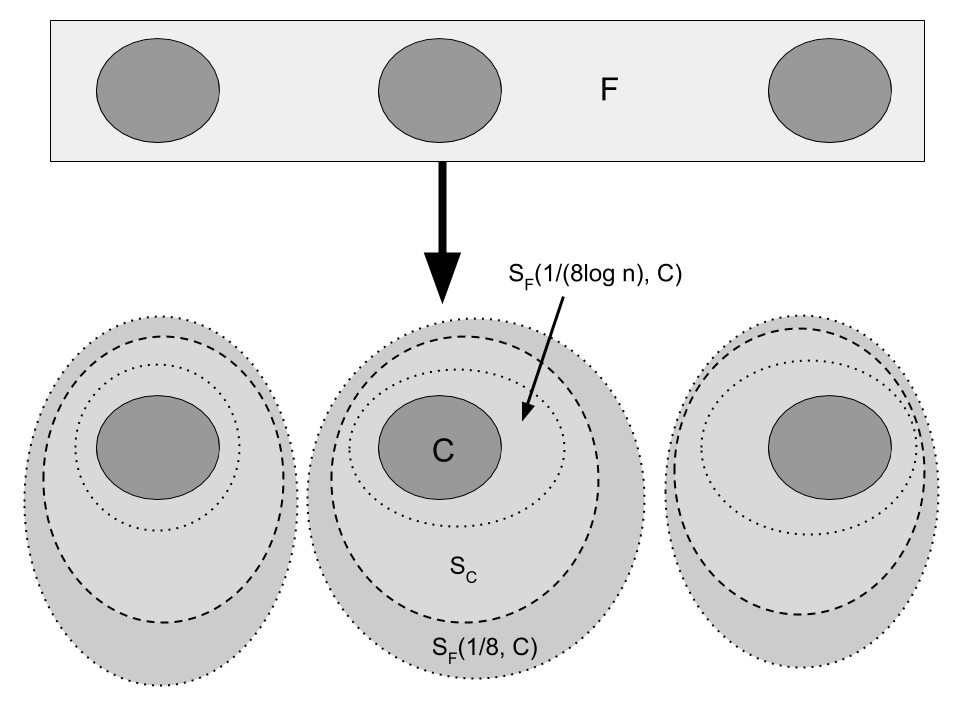}
\caption{$\Voronoi$ produces shortcutters $S_C$ that contain $S_{\mc F}(1/(8\log n),C)$ and are contained within $S_{\mc F}(1/8,C)$. This containment guarantee, along with the well-separatedness of $\mc F$, is translated into a bound on conductivity using Proposition \ref{prop:separated-large} and Lemma \ref{lem:well-sep-lev-score}.}
\label{fig:building-shortcutters-summary}
\end{figure}

One could satisfy the above lemma just by returning a clan of $S_{\mc F}(p,C)$s for some constant $p\in [1/(8\log n),1/8]$. $S_{\mc F}(p,C)$ can be computed using one approximate Laplacian system solve. We are not aware of a way to compute all shortcutters for $\mc F$ efficiently. Instead, $\Voronoi$ partitions clusters into two megaclusters in $\log n$ different ways so that no two clusters are in the same part of every partition. Then, it computes $S_{\mc F}(1/(8\log n),P)$ where $P$ is one side of the partition and intersects all of the partitions to obtain shortcutters. This only requires $O(\log n)$ Laplacian solves. The Laplacian solves are all in the graph with clusters identified to different vertices. This graph has size $|E(Z)\cup \partial Z|$.

As we saw in the Algorithm Overview, these shortcutters are modified many times over the course of the algorithm. The intuitive description of these shortcutters quickly breaks down after any modification. We use Lemma \ref{lem:well-sep-lev-score}, along with the following proposition, to establish that the clan output by $\Voronoi$ is $\zeta$-conductive for some reasonable $\zeta$:

\begin{restatable}{proposition}{propseparatedlarge}\label{prop:separated-large}
Consider a $\gamma$-well-separated $R$-family $\mc F$ in $Y\subseteq X\subseteq V(H_{\mc C})$ and let $I := \texttt{Schur}(H_{\mc C},\cup_{C\in \mc F} C)$. Suppose that

$$c^{\texttt{Schur}(H_{\mc C},Y\cup (V(H)\setminus X))}(E(Y,V(H)\setminus X)) \le \xi$$

For any $C\in \mc F$, let

$$\Delta_{\mc F}(C) := \sum_{e\in E_I(C,C'), C'\ne C\in \mc F} \frac{\texttt{Reff}_I(e)}{r_e^I}$$

Let $\mc F' := \mc F\cup \{V(H)\setminus X\}$ and consider any clusters $S_C$ with $S_{\mc F'}(p,C)\subseteq S_C$ for all $C\in \mc F$. Then

$$\sum_{C\in \mc F} c^{\mc C}(S_C)\le \left(\sum_{C\in \mc F}\frac{\Delta_{\mc F}(C)}{p(\gamma-4) R}\right) + \frac{\xi}{p}$$
\end{restatable}

A relatively simple argument about electrical potentials and their relationship with effective resistances shows the following, which is used to establish well-spacedness:

\begin{restatable}{proposition}{propwellspaced}\label{prop:well-spaced}
Consider a family $\mc F$ in a graph $H$. Let $C\in \mc F$ be a cluster with $H$-effective resistance diameter $R$. Consider some $S_C$ for which $C\subseteq S_C\subseteq S_{\mc F}(p,C)$ for any $p\in (0,1/2)$. Consider a cluster $C'$ that is tied to $C$. Then $C'\subseteq S_{\mc F}(p + 3/10,C)$.
\end{restatable}

We prove all of these statements in the appendix.

\subsection{Tying the parts together}

Now, we combine the statements of the previous two subsections into a proof of Lemma \ref{lem:maintain-shortcutters}. We first implement the algorithm $\CreateEmpire$:

\begin{algorithm}[H]
\DontPrintSemicolon
\caption{$\CreateEmpire(\{\mc H_i\}_{i=1}^{\sigma_0})$}

    $D\gets $ factor 2 error Johnson-Lindenstrauss embedding of $H$ into $\mathbb{R}^{C_1\log n}$ dimensions\;

    \ForEach{$i\in \{1,2,\hdots,\sigma_0\}$}{

        $\mc E_i\gets \emptyset$\;

        \ForEach{clan $\mc C\in \mc H_i$}{

            $T\gets \muapp$\;

            $\{\mc C_{1k}\}_{k=1}^{\log m},\{\mc C_{2k}\}_{k=1}^{\log m},\hdots,\{\mc C_{Tk}\}_{k=1}^{\log m}$ all are initialized to $\emptyset$\;

            \ForEach{shortcutter $S_C\in \mc C$}{

                $\mc F_1,\mc F_2,\hdots,\mc F_T\gets \emptyset$\;

                For each $j\in [T]$, $\mc F_j\gets $ $j$th family of the community $\CoveringCommunity_D(C,H_{\mc C},\R{i})$\;\label{line:cov-com}

                \tcp{bucketing}

                \ForEach{$j\in [T]$}{

                    $\mc F_{j1}\gets \mc F_j$\;

                    \For{$k=1,2,\hdots,\log m$}{

                        $\mc C_{jk}\gets \mc C_{jk}\cup $ the shortcutters in $\Voronoi(H_{\mc C},\mc F_{jk}\cup \{V(H)\setminus S_C\})$ with size at most $2^k$\; \label{line:clan-ub}

                        $\texttt{deleted}(\mc C_{jk})\gets \texttt{deleted}(\mc C)$\;

                        $\mc F_{j(k+1)}\gets $ the subset of $\mc F_{jk}$ with shortcutters in $\Voronoi(H_{\mc C},\mc F_{jk}\cup \{V(H)\setminus S_C\})$ with size greater than $2^k$\; \label{line:clan-lb}

                    }

                }

            }

            $\mc E_i\gets \mc E_i\cup_{j\in [T],k\in [\log m]} \{\mc C_{jk}\}$\;

        }

    }

    \ForEach{$i\in \{1,2,\hdots,\sigma_0\}$}{

        \ForEach{$P\in \mc P_i(\mc E)$}{

            \tcp{will justify well-definedness in ``Containment'' analysis}

            Let $Q$ be the unique part in $\mc P_i(\{\mc H_i\}_{i=1}^{\sigma_0})$ containing $P$\;\label{line:q-def}

            Let $C_P$ be an arbitrary core in an $\mc F_j$ obtained from the $\CoveringCommunity(C_Q,H_{\mc C},\R{i})$ call on Line \ref{line:cov-com}\;\label{line:cp-def}

            Let $S_P$ be the shortcutter in $\mc E$ assigned to $C_P$\;

        }

    }

    \Return $\{\mc E_i\}_{i=1}^{\sigma_0}$\;

\end{algorithm}

We now prove Lemma \ref{lem:maintain-shortcutters} given all of the propositions and lemmas in this section.

\lemmaintainshortcutters*

\begin{proof}[Proof of Lemma \ref{lem:maintain-shortcutters}]

\textbf{Radii of clans.} $\CoveringCommunity$ outputs an $\murad\R{i}$-community, so each cluster has $H$ effective resistance diameter at most $\murad\R{i}$, as desired.

\textbf{Bucketing.} Each clan $\mc C_{jk}\in \mc E_i$ arises from one clan $\mc C$ of $\mc H_i$ and $\mc F_{jk}$ for that clan $\mc C$. By Line \ref{line:clan-ub},

$$\max_{S_{C'}\in \mc C_{jk}} |E(S_{C'})\cup \partial S_{C'}|\le 2^k$$

Therefore,

$$s_{\mc C_{jk}} \ge \frac{m}{2^k}$$

Next, we bound $|\mc C_{jk}|$. The shortcutters in $\Voronoi(H_{\mc C_{jk}},\mc F_{jk})$ are disjoint because the shortcutters $\mc S_{\mc F}(1/2 - \ep,C')$ for $C'\in \mc F$ are disjoint for any $\ep\in (0,1/2)$. Line \ref{line:clan-lb} ensures that the clusters $\mc F_{jk}$ are cores of shortcutters with size at least $2^{k-1}$. The shortcutters in $\mc C_{jk}$ are disjoint because $\Voronoi$ produces disjoint subclusters of $S_C$ and the $S_C$s are disjoint by definition of $\mc C$. This means that

$$|\mc C_{jk}|\le m/2^{k-1}$$

Combining inequalities shows that

$$|\mc C_{jk}|\le 2s_{\mc C_{jk}}$$

which implies bucketing.

\textbf{Conductivity.} Suppose that each $\mc H_i$ is $\zeta$-conductive and consider a clan $\mc C\in \mc H_i$. Consider a shortcutter $S_C\in \mc C$, $j\in [T]$, and $k\in [\sigma_1]$. Let $X\gets S_C$, $Y\gets C$, $\mc F\gets \mc F_{jk}$, $\gamma\gets \gammads$, and $p\gets 8\log n$. This is valid because $\mc F_{jk}$ is a $\gammads$-well-separated $\R{i}$-family by Lemma \ref{lem:covering-community}. Therefore, Proposition \ref{prop:separated-large} applies and shows that

$$\sum_{C'\in \mc F_{jk}} c^{\mc C_{jk}}(S_{C'})\le \left(\sum_{C'\in \mc F_{jk}}\frac{2\Delta_{\mc F_{jk}}(C')}{p\gamma\R{i}}\right) + \frac{c^{\mc C}(S_C)}{p}$$

where the shortcutters $S_{C'}$ are in $\mc C_{jk}$. By Lemma \ref{lem:well-sep-lev-score},

$$\left(\sum_{C'\in \mc F_{jk}}\frac{2\Delta_{\mc F_{jk}}(C')}{p\gamma\R{i}}\right) + \frac{c^{\mc C}(S_C)}{p} \le \frac{2\muapp|\mc F_{jk}|}{p\gamma\R{i}} + \frac{c^{\mc C}(S_C)}{p}$$

Summing over all $S_C\in \mc C$ shows that

$$\sum_{S_{C'}\in \mc C_{jk}} c^{\mc C_{jk}}(S_{C'}) \le \left(\sum_{S_C\in \mc C} \frac{2\muapp|\mc F_{jk}|}{p\gamma\R{i}}\right) + \left(\sum_{S_C\in \mc C}\frac{c^{\mc C}(S_C)}{p}\right)$$

where $\mc F_{jk}$ is defined for $S_C$ in the above right hand side summand. Furthermore,

\begin{align*}
\left(\sum_{S_C\in \mc C} \frac{2\muapp|\mc F_{jk}|}{p\gamma\R{i}}\right) + \left(\sum_{S_C\in \mc C}\frac{c^{\mc C}(S_C)}{p}\right)&\le \frac{2\muapp|\mc C_{jk}|}{p\gamma\R{i}} + \frac{\zeta m^{1/\sigma_1}}{p\R{i}} s_{\mc C}\\
&\le \left(16(\log n)\muapp + (8\log n)\zeta m^{1/\sigma_1}\right) s_{\mc C}/(\R{i})
\end{align*}

by the bucketing of $\mc C_{jk}$. This is the desired conductivity statement.

\textbf{Well-spacedness.} Consider any cluster $C''$ that is tied to $S_{C'}\in \mc C_{jk}$. Suppose that $S_{C'}$ was generated from the shortcutter $S_C$. By Lemma \ref{lem:voronoi}, $S_{C'}\subseteq S_{\mc F_{jk}\cup \{V(H)\setminus S_C\}}(1/8,C')$, where $\mc F_{jk}$ corresponds to $S_C$. By Proposition \ref{prop:well-spaced} applied in the graph $H\gets H_{\mc C}$, $C''\subseteq S_{\mc F_{jk}\cup \{V(H)\setminus S_C\}}(3/8,C')$. Since each vertex has potential greater than 1/2 for only one cluster in $\mc F_{jk}\cup \{V(H)\setminus S_C\}$, $S_{\mc F_{jk}\cup \{V(H)\setminus S_C\}}(3/8,C')$ does not intersect any other $S_{\mc F_{jk}\cup \{V(H)\setminus S_C\}}(3/8,C''')$ for $C'''\in \mc F_{jk}\setminus \{C'\}$. Therefore, $C''$ cannot intersect any other shortcutter in $\mc C_{jk}$ for a core in $\mc F_{jk}$. Furthermore, $S_{\mc F_{jk}\cup \{V(H)\setminus S_C\}}(3/8,C')$ does not intersect $S_{\mc F_{jk}\cup \{V(H)\setminus S_C\}}(3/8,V(H)\setminus S_C)$. Therefore, $C''$ cannot contain any vertices outside of $S_C$. As a result, it cannot intersect any shortcutters of $\mc C_{jk}$ that are not for cores in $\mc F_{jk}$. Combining these two statements shows that $C''$ can only be tied to one shortcutter $S_{C'}\in \mc C_{jk}$, which is well-spacedness.

\textbf{Boundedness.} Boundedness follows from the following:

\begin{align*}
\sum_{S_{C'}\in \mc C_{jk}} c^H(\partial_H C') &= \sum_{S_C\in \mc C}\sum_{C'\in \mc F_{jk}} c^H(\partial_H C')\\
&\le \left(\sum_{S_C\in \mc C} c^H(\partial_H C) + \frac{\kappa_0|E_H(C)\cup \partial_H C|}{\R{i}}\right)\\
&\le \left(\sum_{S_C\in \mc C} c^H(\partial_H C)\right) + \frac{\kappa_0 m}{\R{i}}\\
&\le \frac{(\kappa + \kappa_0)m}{\R{i}}\\
\end{align*}

where the first inequality follows from the ``Boundedness'' guarantee of Lemma \ref{lem:covering-community}, the second inequality follows from the fact that $S_C$s in $\mc C$ are disjoint, and the third inequality follows from the $\kappa$-boundedness of $\mc C$.

\textbf{Clan growth.} Each clan $\mc C$ is replaced by clans of the form $\mc C_{jk}$. There are only $T\log m\le \muapp\log n$ clans of this form.

\textbf{Containment.} We start by showing that the part $Q$ defined on Line \ref{line:q-def} is well-defined. It suffices to show that $\mc P_i(\mc E)$ is a refinement of $\mc P_i(\{\mc H_j\}_{j=1}^{\sigma_0})$ for all $i\in \{1,2,\hdots,\sigma_0\}$. We show this by induction on decreasing $i$. Inductively assume that $\mc P_{i+1}(\mc E)$ is a refinement of $\mc P_{i+1}(\{\mc H_j\}_{j=1}^{\sigma_0})$ and consider a part $P\in \mc P_i(\mc E)$.

Let $C$ be a core in some clan of $\mc H_i$ that intersects $P$. By the ``Covering'' guarantee of Lemma \ref{lem:covering-community}, some core $C'$ in some family output during the call to $\CoveringCommunity$ on $C$ in Line \ref{line:cov-com} intersects $P$. By the ``Input constraint'' guarantee of Lemma \ref{lem:covering-community}, $C'\subseteq C$. By definition of $\mc P_i(\mc E)$, $P\subseteq C'$. Therefore, $P$ is contained in all cores $C$ in some clan of $\mc H_i$ that intersect $P$. Furthermore, $P$ is contained in a unique part of $\mc P_{i+1}(\mc E)$ by definition. This part is contained in a unique part of $\mc P_{i+1}(\{\mc H_j\}_{j=1}^{\sigma_0})$ by the inductive assumption. Since $\mc P_i(\mc E)$ is the refinement of all cores in $\mc E_i$ and $\mc P_{i+1}(\mc E)$ and $P$ is contained in all cores of $\mc H_i$ and parts of $\mc P_{i+1}(\mc E)$ that it intersects, $P$ is contained in a part $Q$ in the refinement $\mc P_i(\{\mc H_j\}_{j=1}^{\sigma_0})$. This completes the inductive step and shows that $Q$ is well-defined.

Let $C_Q$ be the core assigned to $\mc Q$ in $\{\mc H_j\}_{j=1}^{\sigma_0}$. By the ``Covering'' guarantee of Lemma \ref{lem:covering-community} applied to the $C_Q$ call, there is a core $C'\subseteq C_Q$ (by the ``Input constraint'') with $P\subseteq C'$. Therefore, there exists a choice of $C_P$ on Line \ref{line:cp-def} and $C_P\subseteq C_Q$.

Now, we show that $S_P\subseteq S_Q$, which is the same as showing that $S_{C_P}\subseteq S_{C_Q}$. Since $C_P$ is in some $\mc F_k$ created on Line \ref{line:cov-com} for $C_Q$, the $\Voronoi$ call that creates $S_{C_P}$ has the set $V(H)\setminus S_{C_Q}$ in its input family. This means that $S_{C_P}$ cannot intersect $V(H)\setminus S_{C_Q}$, which means that $S_{C_P}\subseteq S_{C_Q}$, as desired.

\textbf{Modifiedness and the deletion set condition.} This follows from the fact that $\texttt{deleted}(\mc C_{jk}) = \texttt{deleted}(\mc C)$ and the ``Containment'' guarantee for cores.

\textbf{Runtime.} For each core $C$ of a shortcutter in $\mc C$, the runtime of $\CoveringCommunity$ is at most $m^{o(1)}|E(C)\cup \partial C|$ by the runtime condition of Lemma \ref{lem:covering-community}. The runtime of the $j$ and $k$ loops is at most $m^{o(1)}|E(S_C)\cup \partial S_C|$, by Lemma \ref{lem:voronoi} and the fact that $\Voronoi$ takes $V(G)\setminus  S_C$ as one input cluster. Since $C\subseteq S_C$ and shortcutters in $\mc C$ are disjoint, the total work for shortcutters in $\mc C$ is $m^{1+o(1)}$. Since there are $m^{o(1)}$ clans in the input hordes, the total runtime of $\CreateEmpire$ is $m^{1+o(1)}$.

\end{proof}

\newpage

\section{Conditioning on the selected parts}\label{sec:random-walk}

Now, we prove Lemma \ref{lem:partial-sample}. Let $S = \cup_{P\in \mc K} P$. Aldous-Broder only needs to cover $S$ in order to sample the intersection of a random tree with $\cup_{P\in \mc K} E(P)$. Since $\mc E$ has been carved with respect to $S$, each active shortcutter $S_C \in \mc H_i\in \mc E$ has one of two statuses at any time:

\begin{itemize}
\item (Ready) $S_C$ has been covered and therefore can be used or $S_C$ is not active.
\item (Boundary) $S_C$ has not been covered and $C$ is within distance $\mucarve\R{i}$ of some unvisited vertex in $S$.
\end{itemize}

Now, suppose that a vertex $v$ is in a part $P_i$ in $\mc P_i(\mc E)$. Look at the maximum level $i$ for which $P_i$'s shortcutter has a ``Ready'' status. Since $P_{i+1}$'s shortcutter has a ``Boundary'' status, $v$ is within distance $\mucarve\R{i+1}$ of some unvisited vertex in $S$.

Since $P_i$'s shortcutter is ``Ready,'' it can be used without skipping any first visits to vertices in $S$. We now clarify what ``use'' means. Ideally, we could just use online shortcutting to shortcut from $P_i$ to $\partial S_{P_i}$. This could be too expensive for two reasons:

\begin{itemize}
\item $E(P_i) = \emptyset$, in which case $S_{P_i}$ is not active and may be larger than the promised maximum active shortcutter size.
\item $S_{P_i}$'s conductance is low in $H_{\mc C_i}$ where $\mc C_i$ is the clan containing $S_{P_i}$, not in $H$.
\end{itemize}

In the first case, no shortcutting is necessary, as $\partial S_{P_i}$ is part of a set of boundary edges with small enough total weight (conductance). In the second case, we precompute the probability that a random walk hits an edge in $\texttt{deleted}(\mc C_i)$ before $\partial S_{P_i}$. This is our use of offline shortcutting. We show that the precomputation work is small thanks to $\tau$-modifiedness. If the random walk does hit one of these edges first, shortcutting takes $\tilde{O}(1)$ work and the work can be charged to a traversal over an edge in $\partial P_i$. Otherwise, the shortcut step can be charged to a step over a small conductance set given by $\zeta$-conductivity.

Our random walk analysis relies heavily on the following lemma, which we prove in the appendix. One can think of this lemma as a softer version of the subgraph cover time bound used by \cite{KM09}. Unlike \cite{KM09} and \cite{MST15}, we use the following lemma on many different graphs $I$ corresponding to various Schur complements:

\lemwalkbound*

We prove this in the appendix.

\subsection{Our shortcutting method}

We now introduce our shortcutting primitive. It relies on two facts about electrical flows, the first of which was used by both Kelner-Madry and Madry-Straszak-Tarnawski:

\begin{theorem}[\cite{KM09,MST15}]\label{thm:offline-prob}
Consider two vertices $u,v$ in a graph $I$. Let $\textbf{p}\in \mathbb{R}^{V(I)}$ denote the potentials for a $u-v$ electrical flow with $p_u$ and $p_v$ normalized to 1 and 0 respectively. Then for any vertex $w\in V(I)$,

$$\Pr_w[t_u < t_v] = p_w$$

where $t_u$ is the hitting time to $u$.
\end{theorem}

\begin{theorem}[Special case of Theorem \ref{thm:edge-visits}]\label{thm:online-prob}
Consider two vertices $u,v\in V(I)$. For each edge $e\in \partial v$, let $f_e$ denote the unit $u-v$ electrical flow on $e$. Then for each $e\in \partial v$,

$$\Pr_u[e \text{ traversed to visit $v$ for the first time } ] = f_e$$
\end{theorem}

The first result motivates offline shortcutting, while the second motivates online shortcutting. Efficient Laplacian solvers are approximate rather than exact. A trick due to Propp \cite{P10} allows us to get around this issue in expected near-linear time.

Our algorithm maintains a data structure $D$ of shortcutting probabilities that is internal to it and an accuracy parameter $\ep_D$, which initially is $1/|X_C|$. Let $\mc C$ be the clan containing $S_C$. For each vertex $v\in C$ of a shortcutter $S_C$ and the set $X_C = V(\texttt{deleted}(\mc C))\cap S_C$, $D$ stores an $\ep_D$-approximations $q_{vx}$ for all $x\in X_C$ and $q_v$ to $\Pr_v[t_x < t_{(X_C\setminus \{x\})\cup \partial S_C}]$ for all $x\in X_C$ and $\Pr_v[t_{\partial S_C} < t_{X_C}]$ respectively. It occasionally recomputes $D$ when it needs higher accuracy in accordance with Propp's trick. Also in accordance with Propp's trick, $D$ represents these probabilities as subintervals of $[0,1]$, with $r_{vx} = q_{vx} + \sum_{y \text{ before } x} q_{vy}$ for an arbitrary ordering of $X_C$.

\begin{algorithm}[H]
\DontPrintSemicolon

    \tcp{offline part}

    $p\gets$ uniformly random number in $[0,1]$\;

    $x_1,x_2,\hdots,x_k\gets $ arbitrary ordering of $X_C$\;

    \While{$p$ is within $\ep_D$ distance of some $r_{vx}$}{
        
        $\ep_D\gets \ep_D/2$\;

        recompute $D$ for $S_C$\;
    }

    \If{$p$ is in a $q_{vx}$ interval}{

        \Return{an arbitrary edge incident with $x$}\;

    }\Else{

        \tcp{$p$ is in the $q_v$ interval, so switch to online}

        $I\gets H[S_C\cup\partial S_C]\setminus X_C$, with $\partial S_C$ identified to a vertex $s$\;

        Compute $\ep$-additive approximations to $u-s$ electrical flows for decreasing $\ep$ repeatedly using Propp's trick to sample an escape edge from $I$\;

        \Return{the sampled escape edge}

    }

\caption{$\Shortcut(S_C,v)$}
\end{algorithm}

The runtime analysis of this algorithm takes place over three parts: preprocessing, runtime to hit $X_C$, and the runtime to hit $\partial S_C$:

\begin{lemma}\label{lem:shortcut}
$\Shortcut(S_C,v)$ takes as input a shortcutter $S_C$ in a clan $\mc C$ and a vertex $v\in C$. It samples the first vertex $w\in X_C\cup \partial S_C$ that the random walk starting at $v$ in $H$ hits with the correct probability. The edge is correct if that vertex is outside of $S_C$. $\Shortcut$ satisfies the following runtime guarantees:

\begin{itemize}
\item (Preprocessing) The total work to update $D$ over an arbitrary number of uses of $\Shortcut$ on $S_C$ is at most $\tilde{O}(|X_C||E(S_C)\cup \partial S_C|)$ in expectation.
\item (Runtime to hit $X_C$) If $\Shortcut(S_C,v)$ returns some vertex in $X_C$, it took $\tilde{O}(1)$ time to do so, excluding time to update $D$.
\item (Runtime to hit $\partial S_C$) If $\Shortcut(S_C,v)$ returns some vertex in $\partial S_C$, it took $\tilde{O}(|E(S_C)\cup \partial S_C|)$ time to do so in expectation.
\end{itemize}
\end{lemma}

\begin{proof}

\textbf{Correctness.} Consider some $x\in X_C$. $p$ is in the $q_{vx}$ interval with probability exactly $\Pr_v[t_x < t_{(X_C\setminus \{x\})\cup \partial S_C}]$ by Theorem \ref{thm:offline-prob}, so any $x\in X_C$ is sampled with the right probability. $\partial S_C$ is sampled with probability exactly $\Pr_v[t_{\partial S_C} < t_{X_C}]$. Moreover, notice that for any non-$S_C$ endpoint $w$ of an edge in $\partial S_C$,

$$\Pr_v[t_w < t_{(\partial S_C\setminus \{w\})\cup X_C} \text{ in $H$ }| t_{\partial S_C} < t_{X_C} \text{ in $H$ }] = \Pr_v[t_w < t_{\partial S_C\setminus \{w\}} \text{ in $H\setminus X_C$}]$$

Since $I = H[S_C\cup \partial S_C]\setminus X_C$, Theorem \ref{thm:online-prob} implies that we are sampling escape edges with the right probabilities for endpoints $w$ of edges in $\partial S_C$.

\textbf{Preprocessing.} For a particular value of $\ep_D$, the probability of recomputation is at most $2|X_C|\ep_D$, as this is a bound on the probability that $p$ is within distance $\ep_D$ of an $r_{vx}$. If this happens, Theorem \ref{thm:offline-prob} implies that the $q_{ux}$s for all $u\in C$ and one $x$ can be computed using one Laplacian solve. Doing this for all vertices in $X_C$ takes $\tilde{O}(|X_C||E(S_C)\cup \partial S_C|\log(1/\ep_D))$ time. The expected time is at most

\begin{align*}
\sum_{\text{values of $\ep_D$}}^{\infty} 2|X_C|\ep_D \tilde{O}(|X_C||E(S_C)\cup \partial S_C|\log(1/\ep_D)) &= \sum_{i=0}^{\infty} \frac{i}{2^i} \tilde{O}(|X_C||E(S_C)\cup \partial S_C|)\\
&= \tilde{O}(|X_C||E(S_C)\cup \partial S_C|)\\
\end{align*}

to update $D$, as desired.

\textbf{Runtime to hit $X_C$.} If $x\in X_C$ is sampled, the else block does not execute. Everything in the if block takes $\tilde{O}(1)$ time besides updating $D$, as desired.

\textbf{Runtime to hit $\partial S_C$.} $i$ Laplacian solves on $I$ with error $2^{-i}$ are done with probability at most $2^{-i}$ to compute all of the exit probabilities for $v$ out of $S_C$. Therefore, the expected work is $\tilde{O}(|E(S_C)\cup\partial S_C|)$, as desired.
\end{proof}

\subsection{Our implementation of the shortcutting meta-algorithm}

Now, we implement $\PartialSample$. This algorithm is exactly the same as the shortcutting meta-algorithm given at the beginning of Section \ref{sec:fake-overview}:

\begin{algorithm}[H]
\DontPrintSemicolon

    $S\gets \cup_{P\in \mc K} P$\;

    $v\gets $ arbitrary vertex in $H$\;

    $F\gets\emptyset$\;

    \While{there is a vertex in $S$ that has not been visited}{

        $i\gets$ maximum $i$ for which the shortcutter $S_{P_i}$ for the part $P_i\in \mc P_i(\mc E)$ containing $v$ has status ``Ready''\;

        \If{$E(P_i) = \emptyset$}{

            $\{v,w\}\gets$ random edge incident with $v$ with probability proportional to conductance\;

            $e\gets \{v,w\}$\;

        }\Else{

            $e\gets \Shortcut(S_{P_i},v)$\;

        }

        $w\gets $ the $X_{P_i}\cup \partial S_{P_i}$ endpoint of $e$\;

        \If{both endpoints of $e$ are in $S$ and $w$ not previously visited}{

            $F\gets F\cup \{e\}$\;

        }

        $v\gets w$\;

    }

    \Return{$F$}

\caption{$\PartialSample(\mc E,\mc K)$}
\end{algorithm}

The runtime analysis does the following:

\begin{itemize}
\item Preprocessing takes a small amount of time because of $\tau$-modifiedness.
\item Random walk steps and shortcut steps to $X_{P_i}$ take $\tilde{O}(1)$ time. We can afford to charge these steps to traversals over boundary edges of $\mc P_i(\mc E)$ thanks to the deletion set condition.
\item Shortcut steps to $\partial S_{P_i}$ can be charged to a step over an edge in the Schur complement obtained by eliminating all vertices internal to $S_{P_i}$ besides $X_{P_i}$. $\zeta$-conductivity can be used to bound the number of these steps.
\end{itemize}

Picking the maximum level ``Ready'' shortcutter allows us to charge these steps to covering the $i+1$th horde, since the above ``Boundary'' shortcutter contains a closeby uncovered vertex.

Correctness relies on the following fact, which we restate from the overview:

\thmaldousbroder*

\begin{figure}
\includegraphics[width=1.0\textwidth]{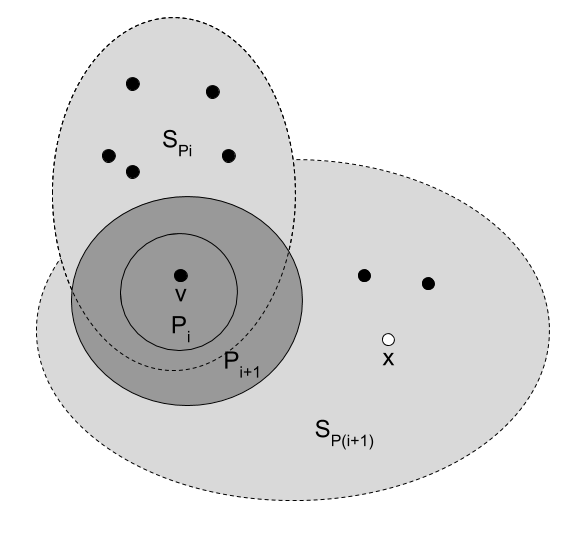}
\caption{Runtime analysis of the shortcutted random walk in $\PartialSample$. When $S_{P_i}$ is used to shortcut a random walk starting at $v$, there is guaranteed to be an unvisited vertex $x$ in $S_{P_{i+1}}$. By carving, $x$ is within distance $\mucarve \alpha^{(i+1)/(\sigma_0+1)}r_{min}$ of $v$. Therefore, $R = \mucarve \alpha^{(i+1)/(\sigma_0+1)}r_{min}$ can be used when applying Lemma \ref{lem:walk-bound}. $\alpha^{(i+1)/(\sigma_0+1)}r_{min}$ partially cancels with an $\alpha^{i/(\sigma_0+1)}r_{min}$ in the denominator of both the boundedness and conductivity of $\mc E_i$.}
\label{fig:rw-analysis}
\end{figure}

\lempartialsample*

\begin{proof}[Proof of Lemma \ref{lem:partial-sample}]
We start with correctness. We need to sample the intersection of a random tree with $Z = \cup_{P\in \mc K} E(P)$. By Theorem \ref{thm:aldous-broder}, it suffices to show that $\PartialSample$ finds all of the edges used to visit $S$ for the first time, since all edges in $Z$ have both of their endpoints in $S$.

Since $S_{P_i}$ has a status of ``Ready,'' it is either (1) covered or (2) not necessary to use. If (2) is the case, then $E(P_i) = \emptyset$, which means that a true random walk step out of $v$ is performed. If (1) is the case, then using the shortcutter $S_{P_i}$ will not skip any first visits to vertices in $S$. Furthermore, $X_{P_i}\cap S$ has been visited, so there is no need to keep track of the correct edge used to visit these vertices. By Lemma \ref{lem:shortcut}, $\Shortcut$ returns the correct visit edge if the walk exits through $\partial S_{P_i}$. Therefore, $\PartialSample$ does not miss any first visits to $S$, which means that it returns the intersection of a random tree with $Z$.

Now, we bound the runtime of $\PartialSample$. We break this analysis up into a number of different types of steps. First, though, we make an observation that is relevant for all types. $v$ is within distance $\alpha^{(i+1)/(\sigma_0+1)}\alpha^{o(1)}r_{min}$ of some unvisited vertex in $S$ because either (1) $S_{P_{i+1}}$ has status ``Boundary'' or (2) $i = \sigma_0$, in which case the diameter of the graph is at most $\alpha r_{min} = \alpha^{(i+1)/(\sigma_0+1)} r_{min}$.

\textbf{If-statement steps.} These occur when $E(P_i) = \emptyset$, including the case in which $i$ does not exist. In this case, one random walk step is performed incident with $v$. This takes $O(\log n)$ time to execute and occurs over an edge of $\partial \mc P_i(\mc E)$, where $\partial \mc P_i(\mc E) := \cup_{P'\in \mc P_i(\mc E)} \partial P'$.

Apply Lemma \ref{lem:walk-bound} to $I\gets H$, all edges in $\partial \mc P_i(\mc E)$, $S\gets S$, and $R\gets \mucarve\R{i+1}$. By the $\kappa$-boundedness of $\mc E$, the total number of steps across edges in $\partial \mc P_i(\mc E)$ within distance $R$ of an unvisited vertex of $S$ is at most

$$\tilde{O}(c^H(\partial \mc P_i(\mc E)) R)\le \frac{\kappamax m}{\R{i}} \mucarve \R{i+1} \le \mucarve \kappamax m \alpha^{1/(\sigma_0+1)}$$

per clan. Summing over all clans yields a bound of $\ellmax \mucarve \kappamax m \alpha^{1/(\sigma_0+1)}$, as desired.

\textbf{Else-statement $X_{P_i}$ steps.} Notice that since $E(P_i)$ is nonempty, the deletion set condition applies, which ensures that no edge in $\texttt{deleted}(\mc C)$ for any $\mc C$ is on the boundary of $\mc P_i$. The random walk, in going from $P_i$ to $X_{P_i}$, must cross an edge of $\partial P_i$, as the deletion set condition implies that $X_{P_i}$ is disjoint from $\partial P_i$. As discussed in the If-statement steps bound, only $\ellmax \mucarve \kappamax m \alpha^{1/(\sigma_0+1)}$ steps occur. By the $X_{P_i}$ step condition of Lemma \ref{lem:shortcut}, each of these steps takes $\tilde{O}(1)$ time to execute.

\textbf{Else-statement $\partial S_{P_i}$ steps.} Lemma \ref{lem:shortcut} says that $\Shortcut$ returns an edge incident with the first vertex in $X_{P_i}\cup \partial S_{P_i}$ that the random walk visits. In particular, a shortcut step directly to $\partial S_{P_i}$ can be charged to random walk steps across a $C_{P_i}-\partial S_{P_i}$ edge in $I \gets \texttt{Schur}(H,X_{P_i}\cup C_{P_i}\cup (V(H)\setminus S_{P_i}))$. Let $F\gets E_I(C_{P_i},V(H)\setminus S_{P_i})$, $S\gets V(I)\cap S$, and $R\gets \mucarve\alpha^{(i+1)/(\sigma_0+1)}\alpha^{o(1)}r_{min}$. By Lemma \ref{lem:walk-bound}, the total number of steps that occur across $E_I(C_{P_i},V(H)\setminus S_{P_i})$ that are within distance $R$ of $S\cap V(I)$ is at most

$$c^I(E_I(C_{P_i},V(H)\setminus S_{P_i})) R$$

in expectation. We start by arguing that each $\partial S_{P_i}$ shortcut step can be charged to one of these random walk steps. Recall that $S_{P_i}$ is only used when $S\cap S_{P_i}$ is covered. All vertices that were eliminated to obtain $I$ were in $S_{P_i}$, so all shortcut steps occur within distance $R$ of some vertex in $S\cap V(I)$. Each shortcut step can be charged to at least one step across an edge of $E_I(C_{P_i},V(H)\setminus S_{P_i})$, as discussed earlier. Therefore, the number of shortcut steps is at most $c^I(E_I(C_{P_i},V(H)\setminus S_{P_i})) R$.

Now, we bound $c^I(E_I(C_{P_i},V(H)\setminus S_{P_i}))$. Let $I' = \texttt{Schur}(H\setminus \texttt{deleted}(\mc C_{P_i}),X_{P_i}\cup C_{P_i}\cup (V(H)\setminus S_{P_i}))$, where $\mc C_{P_i}$ is the clan that contains the shortcutter $S_{P_i}$. Each edge of $\texttt{deleted}(\mc C_{P_i})$ with endpoints in $S_{P_i}$ has both of its (identified) endpoints in $X_{P_i}$, by definition of $X_{P_i}$. Therefore, deleting these edges does not affect the conductance of the relevant set:

$$c^I(E_I(C_{P_i},V(H)\setminus S_{P_i})) = c^{I'}(E_{I'}(C_{P_i},V(H)\setminus S_{P_i}))$$

Eliminating $X_{P_i}$ also can only increase the conductance of this set. Precisely, let $I'' = \texttt{Schur}(H\setminus \texttt{deleted}(\mc C_{P_i}),C_{P_i}\cup (V(H)\setminus S_{P_i}))$. Then

$$c^{I'}(E_{I'}(C_{P_i},V(H)\setminus S_{P_i}))\le c^{I''}(E_{I''}(C_{P_i},V(H)\setminus S_{P_i}))$$

This quantity has already been defined. Recall that $H_{\mc C_{P_i}} = H\setminus \texttt{deleted}(\mc C_{P_i})$. As a result,

$$c^{I''}(E_{I''}(C_{P_i},V(H)\setminus S_{P_i})) = c^{\mc C_{P_i}}(S_{P_i})$$

$\zeta$-conductivity can be used to bound this quantity. In particular, summing these bounds over $\mc C_{P_i}$ shows that the total number of times shortcutters in $\mc C_{P_i}$ can be used to travel to their true boundaries is at most

\begin{align*}
\tilde{O}\left(\sum_{S_C\in \mc C_{P_i}} c^{\mc C_{P_i}}(S_{P_i}) R\right) &\le \tilde{O}\left(\frac{\zeta m^{1/\sigma_1} s_{\mc C_{P_i}}}{\R{i}} R\right)\\
&\le \tilde{O}(\mucarve \zetamax s_{\mc C_{P_i}} m^{1/\sigma_1} \alpha^{1/(\sigma_0+1)})\\
\end{align*}

By Lemma \ref{lem:shortcut}, using any shortcutter $S_C\in \mc C_{P_i}$ takes at most $\max_{S_C\in \mc C_{P_i}} |E(S_C)\cup \partial S_C|$ time in expectation. By definition of $s_{\mc C_{P_i}}$, the total work done using shortcutters in $\mc C_{P_i}$ is at most

$$(\max_{S_C\in \mc C_{P_i}} |E(S_C)\cup \partial S_C|)\tilde{O}(\mucarve \zetamax s_{\mc C_{P_i}} m^{1/\sigma_1} \alpha^{1/(\sigma_0+1)})\le \mucarve\zetamax m^{1+1/\sigma_1}\alpha^{1/(\sigma_0+1)}$$

in expectation. Since $\mc E$ contains $\ellmax$ clans, the total amount of work due to online shortcut steps is at most $\ellmax \mucarve\zetamax m^{1+1/\sigma_1}\alpha^{1/(\sigma_0+1)}$.

\textbf{Else-statement preprocessing work.} Each edge in $\texttt{deleted}(\mc C)$ for any clan $\mc C$ is incident with at most two shortcutters in $\mc C$ since shortcutters in a clan are disjoint. By Lemma \ref{lem:shortcut}, an expected $O(1)$ Laplacian solves on a cluster with size at most $\max_{S_C\in \mc C} |E(S_C)\cup \partial S_C|$ happen for each edge in $\texttt{deleted}(\mc C)$. Therefore, the expected total amount of work is at most

$$(\max_{S_C\in \mc C} |E(S_C)\cup \partial S_C|)|\texttt{deleted}(\mc C)|$$

Since $\mc C$ is $\tau$-modified,

$$(\max_{S_C\in \mc C} |E(S_C)\cup \partial S_C|)|\texttt{deleted}(\mc C)|\le (\max_{S_C\in \mc C} |E(S_C)\cup \partial S_C|)(\taumax m^{1/\sigma_1} s_{\mc C})\le \taumax m^{1+1/\sigma_1}$$

The desired bound follows from the fact that $\mc E$ only has $\ellmax$ clans.

\textbf{Completing the proof.} All work that $\PartialSample$ does falls into one of these four categories for some $i$. Since there are only $\sigma_0$ possible values of $i$, the total runtime is at most $\tilde{O}(((\zetamax + \kappamax)\mucarve + \taumax)\ellmax m^{1+1/\sigma_1}\alpha^{1/(\sigma_0+1)})$, as desired.

\end{proof}

\newpage

\section{Choosing vertices to condition on}\label{sec:choose-parts}

In this section, we implement $\ConditioningVerts$. This amounts to choosing which parts to condition on when there are multiple levels of shortcutters.

Before doing this, we implement a simple version, $\SimpleConditioningVerts$, which works when $\sigma_0 = 1$. This method illustrates the utility of conditioning on parts whose shortcutters are largest. When $\sigma_0 > 1$, we can no longer chose parts whose shortcutters in all hordes are largest. Instead, we find parts whose shortcutters are ``locally largest'' in all hordes; i.e. they are well-separated from cores of larger shortcutters. In this case, the parts chosen for conditioning can be carved out anyways. This choice makes enough progress because one can show that the sizes of shortcutters for parts chosen decrease by a factor of $m^{1/\sigma_1}$. After conditioning $\sigma_1$ times, one can condition on parts in a higher-level horde. Since there are $\sigma_0$ hordes, the number of rounds required is $\sigma_1^{\sigma_0}$.

\subsection{Warmup: A version of $\ConditioningVerts$ for $\sigma_0 = 1$ (one shortcutter per vertex)}\label{subsec:cond-warmup}

In this section, we implement $\SimpleConditioningVerts$, which satisfies Lemmas \ref{lem:holes} and \ref{lem:progress} when $\sigma_0 = 1$. $\SimpleConditioningVerts$ is not used in any way to prove Theorem \ref{thm:main-result-aspect}, but is included to motivate some of the ideas behind $\ConditioningVerts$. Recall that using no shortcutters (Aldous-Broder) takes $\tilde{O}(m\alpha)$ time. Replacing $\ConditioningVerts$ with $\SimpleConditioningVerts$ results in a $\tilde{O}(m^{1+o(1)}\alpha^{1/2+o(1)})$-time algorithm.

The $\SimpleConditioningVerts(\mc E)$ routine just takes the input empire $\mc E$, outputs the parts whose shortcutters have size within an $m^{1/\sigma_1}$-factor of the maximum, and ``carves'' the selected parts out of every shortcutter:

\begin{figure}
\begin{center}
\includegraphics[width=0.8\textwidth]{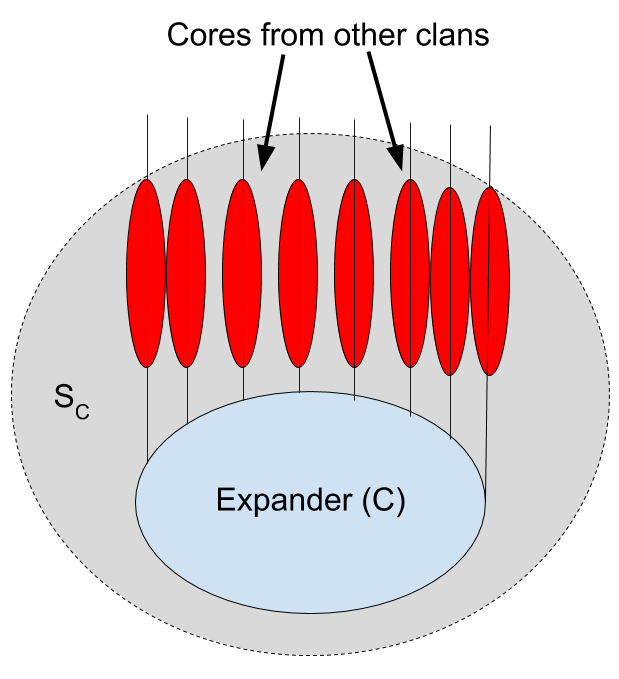}
\end{center}
\caption{The situation that $\SimpleConditioningVerts$ is trying to avoid. Suppose that the shortcutter $S_C$ makes up a large fraction of the graph and that $\SimpleConditioningVerts$ were to select the red cores to condition on first. Applying the algorithm in Lemma \ref{lem:holes} could cause the conductance of $S_C$ to increase untenably. Intuitively, this is because the shortcutter $S_C$ will take the random walk from $C$ to each of the deleted cores at least once. This is bad, though, as the deleted cores could correspond to very small shortcutters. As a result, there could be $\Theta(n)$ of them. $S_C$ contains $\Theta(m)$ edges, so the work just due to this shortcutter is $\Theta(mn)$, which is quadratic! If one conditions on parts in $C$ before the red cores, this issue does not come up because the large size of $S_C$ ensures that there are not many of them.}
\label{fig:quadratic-blowup}
\end{figure}

\begin{algorithm}[H]
\DontPrintSemicolon

    \KwData{an empire $\mc E$ consisting of one $\murad r_{min}\sqrt{\alpha}$-horde $\mc H$}

    \KwResult{a set of parts $\mc K$ to condition on}

    $\mc X\gets $ set of parts $P\in \mc P_1(\mc E)$ that have nonempty $E(P)$ (are active)\;

    $\mc K\gets $ set of parts $P\in \mc X$ whose shortcutters $S_P$ have size at least $m^{-1/\sigma_1} \max_{Q\in \mc X} |E(S_Q)\cup \partial S_Q|$\;\label{line:simple-size}

    \tcp{carving}
    \ForEach{active shortcutter $S_C$ in some clan of $\mc E$}{

        $\mc Z\gets $ parts in $\mc K$ with distance greater than $\mucarve r_{min}\sqrt{\alpha}$ from all vertices in $C$\;

        Remove all parts in $\mc Z$ from $S_C$\;\label{line:simple-deletion}

    }

    \Return{$\mc K$}

\caption{$\SimpleConditioningVerts(\mc E)$, never executed}
\end{algorithm}

We now analyze this algorithm. We need to show that

\begin{itemize}
\item conductivity does not increase much (Lemma \ref{lem:holes})
\item enough progress is made (Lemma \ref{lem:progress})
\end{itemize}

We start with Lemma \ref{lem:holes}. The key idea is that conditioning on the parts with largest shortcutters ensures that there are not many of them. Specifically, each of the parts $P$ assigned to a shortcutter $S_P$ is contained in $S_P$'s core $C_P$, which has low effective resistance diameter ($O(\sqrt{\alpha}r_{min})$). Since the shortcutters come from a small number of clans, each of which consists of disjoint shortcutters, the number of distinct shortcutters for chosen parts is at most $m |\mc E_1|/(\text{minimum size of shortcutter for a selected part})$. Since $\mc E_1$ only consists of $m^{o(1)}$ clans and all shortcutters for selected parts have size within an $m^{1/\sigma_1}$-factor of the maximum across all of $\mc E_1$, the number of distinct shortcutters for chosen parts is at most $m^{1+o(1)}m^{1/\sigma_1}s_{\mc C}$ for any clan $\mc C\in \mc H_1$. In particular, all parts chosen for conditioning are contained in a small number of low-radius clusters (the cores). The following key proposition finishes the proof:

\begin{proposition}\label{prop:delete-low-rad}
Let $H$ be a graph, $\mc E$ be an empire in this graph, and $\mc L$ be a set of clusters with $H$-effective resistance diameter at most $\psi \R{i}$ for some $i\in [\sigma_0]$. Consider a well-spaced clan $\mc C\in \mc E_i$. Let $S := \cup_{C\in \mc L} C$ be the set of vertices in clusters of $\mc L$.

Obtain a new clan $\mc C'$ by deleting all vertices $v\in S$ from shortcutters $S_C\in \mc C$ for which $v$ is not within $H$-effective resistance distance $\gammadel\psi\muapp \R{i}$ of any vertex in $C$, where $\gammadel = 1000$.

Then

$$\sum_{S_C\in \mc C'} c^{\mc C'}(S_C)\le \frac{\muapp(|\mc L| + |\texttt{deleted}(\mc C)|)}{\R{i}} + \sum_{S_C\in \mc C} c^{\mc C}(S_C)$$
\end{proposition}

The proof of this proposition relies on the following lemmas, both of which are proven in the appendix:

\begin{restatable}{lemma}{lemballsplit}\label{lem:ball-split}
Consider a graph $H$ and a set of clusters $\mc D$, each with effective resistance diameter at most $R$. Let $F$ be a set of edges in $H$. Then there is a set of clusters $\mc D'$ with the following properties:

\begin{itemize}
\item (Covering) Each vertex in a cluster of $\mc D$ is in a cluster of $\mc D'$.
\item (Diameter) The effective resistance diameter of each cluster in $\mc D'$ is at most $\muapp R$ in the graph $H\setminus F$.
\item (Number of clusters) $|\mc D'|\le \muapp(|\mc D| + |F|)$.
\end{itemize}
\end{restatable}

\begin{restatable}{lemma}{lemsmallholes}\label{lem:small-holes}
Consider a graph $H$ and two clusters $C$ and $S_C$, with $C\subseteq S_C$. Let $C'$ be disjoint from $C$. Additionally, suppose that

\begin{itemize}
\item The effective resistance diameters of $C$ and $C'$ in $H$ are both at most $R$.
\item The effective resistance distance between any pair of points in $C$ and $C'$ in $H$ is at least $\beta_1 R$.
\item $c^H(S_C) \le \frac{\tau}{R}$. 
\end{itemize}

Then $c^H(S_C\setminus C') \le \frac{\tau + 1/(\beta_1-4)}{R}$.
\end{restatable}

\begin{proof}[Proof of Proposition \ref{prop:delete-low-rad}]
By Lemma \ref{lem:ball-split} applied to the clusters $\mc L$ with deleted edge set $\texttt{deleted}(\mc C)$, $S$ is the union of a set of clusters $\mc L'$ with $H_{\mc C}$-effective resistance diameter $\muapp \psi \R{i}$. Furthermore, $|\mc L'|\le \muapp (|\mc L| + |\texttt{deleted}(\mc C)|)$.

Now, consider a cluster $C\in \mc L'$. To delete all vertices in $S$ from shortcutters with far-away cores, it suffices to delete each of the clusters $C$ from any shortcutter $S_{C'}$ for which the minimum $H_{\mc C}$ distance between vertices in $C'$ and $C$ is at least $(\gammadel - 2) \muapp\psi \R{i} \ge \beta_0 \muapp\psi \R{i}$ by the triangle inequality. In this case, $C$ is tied to $S_{C'}$. Since $\mc C$ is a well-spaced clan, $C$ cannot be tied to any other shortcutter in $\mc C$. Therefore, deleting $C$ from all shortcutters for whom the cores are at least $H_{\mc C}$-distance $(\gammadel - 2) m^{o(1)}\psi \R{i}$ only modifies one shortcutter $S_{C'}$. Furthermore, the conductance of this shortcutter only increases additively by $1/((\gammadel - 2) \muapp\psi \R{i})\ge 1/\R{i}$ by Lemma \ref{lem:small-holes}.

Removing vertices from shortcutters in $\mc C$ does not destroy its well-spacedness. Therefore, we can apply this reasoning for each of the clusters in $\mc L'$; incurring a $1/\R{i}$ conductance increase per cluster deletion. After doing this, we obtain a clan $\mc C'$ for which

\begin{align*}
\sum_{S_C\in \mc C'} c^{\mc C'}(S_C)&\le \frac{|\mc L'|}{\R{i}} + \sum_{S_C\in \mc C} c^{\mc C}(S_C)\\
&\le \frac{\muapp(|\mc L| + |\texttt{deleted}(\mc C)|)}{\R{i}} + \sum_{S_C\in \mc C} c^{\mc C}(S_C)
\end{align*}

as desired.
\end{proof}

Now, we use Proposition \ref{prop:delete-low-rad} to prove Lemma \ref{lem:holes} for $\sigma_0 = 1$:

\begin{proof}[Proof of Lemma \ref{lem:holes}, with $\ConditioningVerts$ replaced by $\SimpleConditioningVerts$]

Let $R = \mucarve\sqrt{\alpha}r_{min}$.

\underline{Number of cores containing parts in $\mc K$.} Recall that each part $P\in \mc K$ is assigned to a shortcutter $S_P$ in the horde $\mc E_1$ of $\mc E$. Let $C_P$ denote this core of $S_P$ and let $\mc K' = \{C_P : \forall P\in \mc K\}$. We now bound $|\mc K'|$. Start by bounding the size of the intersection of $\mc K'$ with the coreset of an arbitrary clan $\mc C$. $\mc C$ consists of disjoint shortcutters, so they must have total size at most $m$. As a result, $\mc K'$ has size at most

\begin{align*}
|\mc K'|&\le |\mc E_1| \frac{m}{\min_{C\in \mc K'} |E(S_C)\cup \partial S_C|}\\
&\le \ellmax \frac{m^{1+1/\sigma_1}}{\max_{C\in \mc K'} |E(S_C)\cup \partial S_C|}\\
&= \ellmax \frac{m^{1+1/\sigma_1}}{\max_{Q\in \mc X} |E(S_Q)\cup \partial S_Q|}\\
&\le \ellmax m^{1/\sigma_1} s_{\mc C}\\
\end{align*}

for any clan $\mc C\in \mc E_1$. The first inequality follows from Line \ref{line:simple-size} of $\SimpleConditioningVerts$. This is the desired bound.

\textbf{Conductivity.} Apply Proposition \ref{prop:delete-low-rad} with $i\gets 1$ and $\mc L\gets \mc K'$ on each of the clans in $\mc E_1$. $\SimpleConditioningVerts$ does strictly fewer vertex removals than the procedure described in Proposition \ref{prop:delete-low-rad}. As a result, Proposition \ref{prop:delete-low-rad} implies that the conductance of the shortcutters in a clan $\mc C$ additively increases by at most $\frac{\muapp(|\mc K'| + |\texttt{deleted}(\mc C)|)}{\sqrt{\alpha}r_{min}}$. By the ``Number of cores containing parts in $\mc K$'' and $\tau$-modifiedness of $\mc C$, the conductance increase is at most $\frac{\muapp(\ellmax m^{1/\sigma_1} + \tau m^{1/\sigma_1})s_{\mc C}}{\sqrt{\alpha}r_{min}}$. Therefore, the conductivity of $\mc C$ in $\mc E'$ (at the end of $\SimpleConditioningVerts$) is at most $\zeta_1$ higher than it was in $\mc E$, as desired.

\textbf{Carving.} Line \ref{line:simple-deletion} of $\SimpleConditioningVerts$ ensures that the shortcutter $S_C$ has been carved with respect to $\mc K$. Therefore, all active shortcutters in $\mc E'$ are carved with respect to $\mc K$, which means that $\mc E'$ is carved with respect to $\mc K$.
\end{proof}

Now, we show that $\SimpleConditioningVerts$ conditions on a large enough set to make substantial progress. Progress is measured by the maximum size of an active shortcutter:

\begin{proof}[Proof of Lemma \ref{lem:progress}, with $\ConditioningVerts$ replaced by $\SimpleConditioningVerts$]

We show that the maximum size of an active shortcutter decreases by a factor of $m^{1/\sigma_1}$ in between applications of $\SimpleConditioningVerts$. If we do this, then no shortcutter is active after $\sigma_1$ iterations. Since each part is assigned to a shortcutter, each part $P$ must have $E(P) = \emptyset$ after $\sigma_1$ iterations.

Now, consider any part $P\in \mc P_1(\mc E)$ with an active shortcutter. By the ``Containment'' input condition, $P$ is contained in a unique $Q\in \mc P_1(\mc E_{prev})$. $S_P$ being active implies that $S_Q$ was active. By the ``Progress'' condition, $Q\notin \mc K_{prev}$. Therefore, by the conditioning choice that chose $\mc K_{prev}$,

$$|E(S_Q)\cup\partial S_Q|\le m^{-1/\sigma_1}\max_{\text{ previously active parts X }} |E(S_X)\cup\partial S_X|$$

By the ``Containment'' condition, $S_P$ is smaller than $S_Q$, so

$$\max_{\text{ currently active parts P }} |E(S_P)\cup\partial S_P|\le m^{-1/\sigma_1} \max_{\text{ previously active parts X }} |E(S_X)\cup\partial S_X|$$

In particular, the maximum active shortcutter size decreased by a factor of $m^{-1/\sigma_1}$, as desired.
\end{proof}

Now, consider what happens when all shortcutters are inactive. By the above proof, this happens after $\sigma_1$ conditioning rounds. Since $\mc E$ is $m^{o(1)}$-bounded at this point and all parts are empty (by inactivity), the total conductance of all edges left over in the graph is at most $m^{1+o(1)}/(\sqrt{\alpha}r_{min})$. Running Aldous-Broder therefore takes $O((m^{1+o(1)}/(\sqrt{\alpha}r_{min}))\alpha r_{min}) = O(m^{1+o(1)}\sqrt{\alpha})$ time. Therefore, the entire algorithm for sampling a random spanning tree from the original graph $G$ takes $O(m^{1+o(1)}\sqrt{\alpha})$ time, as desired.

\subsection{Generalizing to $\sigma_0 > 1$}\label{subsec:multiple-levels-intuition}

We start by reviewing how the main algorithm $\ExactTree$ with $\sigma_0 = 1$ uses $\SimpleConditioningVerts$. $\ExactTree$ starts by making arbitrary shortcutters using $\CreateEmpire$. Afterwards, it calls $\SimpleConditioningVerts$ followed by the routines $\PartialSample$, conditioning, $\FixShortcutters$, and $\CreateEmpire$ again. $\SimpleConditioningVerts$ chooses all of the parts with near-largest shortcutters for conditioning. Applying the routines $\PartialSample$, conditioning, $\FixShortcutters$, and $\CreateEmpire$ makes the induced subgraphs of the parts with the largest shortcutters empty (inactive), as described by the ``Progress'' condition. By the ``Containment'' condition, all remaining shortcutters are smaller than they were previously, so the empire supplied to the second call of $\SimpleConditioningVerts$ has no shortcutters with size greater than $m^{1 - 1/\sigma_1}$. Specifically, the quantity

$$s_1 := \lfloor \log_{m^{1/\sigma_1}}(\max_{\text{active parts $P\in \mc P_1(\mc E)$}} |E(S_P)\cup \partial S_P|) \rfloor$$

strictly decreases during each iteration of the while loop in the $\ExactTree$ algorithm. $s_1 \le \sigma_1$ initially, so $s_1 = 0$ after at most $\sigma_1$ while loop iterations. At this point, $\mc E$ has no active parts. As a result, the graph at this point just consists of boundary edges for parts in $\mc P_1(\mc E)$, at which point just running Aldous-Broder without shortcutting is efficient. Specifically, conditioning on all parts in $\mc E_1$ paved the way for $\ExactTree$ to be able to efficiently condition on the entire graph.

Now, we generalize this reasoning to the case in which $\sigma_0 > 1$. Specifically, we design a scheme that is built around the idea of conditioning on parts in $\mc E_i$ in order to make conditioning on $\mc E_{i+1}$ efficient:

\begin{keyidea}\label{keyidea:cond-hier}
$\ConditioningVerts$ maintains state across multiple calls. Specifically, for all $i\in [\sigma_0]$, it maintains a choice of parts $\mc Q_i\subseteq \mc P_i(\mc E)$ that it would like to be able to condition on and a set of parts $\mc R_i\subseteq P_i(\mc E)$ with $\mc Q_i\subseteq \mc R_i$ that are ``relevant'' to being able to condition on $\mc Q_{i+1}$. Specifically, if all of the parts in $\mc R_i$ are inactive, then $\ConditioningVerts$ can condition on $\mc Q_{i+1}$. Conditioning on at most $\sigma_1$ different choices of $\mc Q_i$ will make all of the parts in $\mc R_i$ irrelevant.
\end{keyidea}

We encourage the reader to delay trying to internalize how exactly the $\mc R_i$s are constructed. We discuss this in detail in Section \ref{subsec:two-level}. $\mc Q_i$ can almost be thought of as the set of parts in $\mc R_i$ with near-maximum shortcutters. We discuss the reasons for saying ``almost'' in Section \ref{subsec:two-level}. In particular, the $\mc Q_i$ parts have ``locally maximum'' shortcutters. We discuss this in more detail in Section \ref{subsec:cond-hierarchies}.

To get a feel for what one could expect $\mc R_i$ and $\mc Q_i$ to be, it helps to think about the $\sigma_0 = 1$ case. In this case, $\mc R_1 = \mc P_1(\mc E)$ during each call of $\SimpleConditioningVerts$. $\mc Q_1$ is the set of active parts in $\mc R_1$ with near-maximum shortcutters. After $\sigma_1$ iterations of the while loop of $\ExactTree$, all parts in $\mc R_1$ are inactive, which allows us to condition on $\mc Q_2$, which is defined to be just one part containing the entire graph. After doing this, all parts in $\mc R_2$ --- which is also defined to be just one part with the entire graph --- are inactive. In particular, the graph contains no more edges and we have sampled a complete random spanning tree.

\subsubsection{Making progress and intuition for the proof of Lemma \ref{lem:progress}}

Now that we have some idea for what the $\mc R_i$s and $\mc Q_i$s could be, we can talk about our notion of progress. We start by generalizing $s_1$ to a quantity $s_i$. Roughly speaking, $m^{s_i/\sigma_1}$ is the maximum size of a relevant shortcutter in a clan of $\mc E_i$:

$$s_i := \lfloor \log_{m^{1/\sigma_1}}(\max_{\text{active parts $P\in \mc R_i$}} |E(S_P)\cup \partial S_P|) \rfloor$$

This is not how we actually define $s_i$, but it is a good way to think about it. In particular, the function $\DSize$ is similar to $\max_{\text{active parts $P\in \mc R_i$}} |E(S_P)\cup \partial S_P|$ and we encourage a confused reader to mentally replace $\DSize$ with $\max_{\text{active parts $P\in \mc R_i$}} |E(S_P)\cup \partial S_P|$ with the exception of one place, which we point out in Section \ref{subsec:cond-digraphs}. Now, we discuss progress:

\begin{keyidea}\label{keyidea:progress}
Our notion of progress is that each iteration of the while loop in $\ExactTree$ lexicographically decreases the word $s_{\sigma_0+1}s_{\sigma_0}s_{\sigma_0-1}s_{\sigma_0-2}\hdots s_2 s_1$.
\end{keyidea}

Now, we understand how $\ConditioningVerts$ could be implemented to guarantee such a lexicographic decrease. Each call to $\ConditioningVerts$ returns $\mc Q_k$ for some $k\in [\sigma_0]$ with $s_k > 0$. By the ``Progress'' input condition in Definition \ref{def:cond-verts-input}, all parts in $\mc Q_k$ become inactive before the next call to $\ConditioningVerts$. Roughly speaking, as described earlier, $\mc Q_k$ contains all of the parts with largest shortcutters in $\mc R_k$. As long as $\mc R_k$ shrinks (which we show that it does in Section \ref{subsec:two-level}), $s_k$ strictly decreases. All $s_i$ for $i > k$ do not increase by the ``Containment'' input condition. Therefore, the word $s_{\sigma_0+1}s_{\sigma_0}s_{\sigma_0-1}\hdots s_2 s_1$ lexicographically decreased between conditioning rounds.

We now briefly discuss why $s_i$s with $i < k$ could increase. To make progress again, $\mc Q_k$ needs to be changed. Recall in Key Idea \ref{keyidea:cond-hier} that the only purpose of $\mc Q_i$ for $i < k$ is to make it possible to condition on parts in $\mc Q_k$. Therefore, once $\mc Q_k$ changes, the $\mc Q_i$s and $\mc R_i$s for $i < k$ become useless and need to be chosen anew. In particular, $s_i$ for $i < k$ could increase.

Eventually, $s_{\sigma_0+1}s_{\sigma_0}s_{\sigma_0-1}\hdots s_2 s_1 = 00\hdots 00$. At this point, because $\mc R_{\sigma_0+1}$ consists of one part which is the entire graph, the entire graph is empty and consists of just one vertex. In particular, we have completely sampled a random spanning tree. Since there are only $\sigma_1^{\sigma_0+1}$ possible words, this happens after $\sigma_1^{\sigma_0+1}$ conditioning rounds, suggesting a proof of Lemma \ref{lem:progress}.

\subsubsection{Conditioning on $\mc Q_k$}

In the previous subsubsection, we needed to find some $\mc Q_k$ with $s_k > 0$ that could actually be conditioned on. Specifically, we need to find some $\mc Q_k$ that can be carved out of \emph{all} shortcutters in \emph{all} hordes of $\mc E$ with increasing the conductivity of any clan too much. 

At the beginning of each call to $\ConditioningVerts$, the $\mc Q_{k'}$ selected for conditioning from the previous iteration consists of inactive parts by the ``Progress'' input condition. Let $i^*$ be the maximum $i$ for which $\mc Q_{i^*} = \emptyset$. This exists because $k'$ is a valid choice for $i^*$. If $s_{i^*} > 0$, then $\mc R_{i^*}$ still contains active parts, which means that we shoud reselect $\mc Q_{i^*}$ in order to continue rendering parts in $\mc R_{i^*}$ inactive. This is done using the routine $\TwoLevelConditioningVerts$. After reselecting $\mc Q_{i^*}$, $\ConditioningVerts$ reselects all $\mc Q_i$s for $i < i^*$ using $\TwoLevelConditioningVerts$. Then, $\ConditioningVerts$ returns the parts $\mc Q_{k^*+1}$, where $k^*$ is the maximum value for which $s_{k^*} = 0$.

While the previous subsubsection illustrated that this algorithm ($\ConditioningVerts$) makes enough progress (satisfies Lemma \ref{lem:progress}), we still need to demonstrate that $\mc E$ can be carved with respect to $\mc Q_{k^*+1}$ without increasing its conductivity too much (by more than $m^{o(1)}$ additively). To do this, it suffices to design the $\mc R_i$s in a way that respects the following property:

\begin{keyidea}\label{keyidea:relevance}
For each part $P\in \mc P_i(\mc E)\setminus \mc R_i$, $S_P$ does not intersect any part $P'\in \cup_{j\le i+1} \mc Q_j$. In particular, $S_P$ is carved with respect to $\cup_{j\le i+1} \mc Q_j$.
\end{keyidea}

Now, we see how this idea enables $\ConditioningVerts$ to choose $\mc Q_{k^*+1}$. Each part $P\in \mc P_{k^*}(\mc E)$ is either in $\mc R_{k^*}$ or not in $\mc R_{k^*}$. If $P\in \mc R_{k^*}$, then $P$ is inactive because $s_{k^*} = 0$. Therefore, its shortcutter $S_P$ does not need to be carved with respect to $\mc Q_{k^*+1}$ because $S_P$ does not need to be used. If $P\not\in \mc R_{k^*}$, then $S_P$ is carved with respect to $\mc Q_{k^*+1}$ by Key Idea \ref{keyidea:relevance}. For parts $Q\in \mc P_j(\mc E)$ for $j \ge k^*+1$, either $Q\notin \mc R_j(\mc E)$ in which case Key Idea \ref{keyidea:relevance} implies that $S_Q$ is carved, or $Q\in \mc R_j(\mc E)$ and $\mc Q_j$ can be carved out of $S_Q$ because $\mc Q_j$'s shortcutters are bigger than $S_Q$. In particular, applying Proposition \ref{prop:delete-low-rad} here implies that carving does not increase the conductivity of $S_Q$'s clan too much (see the $\MakeNonedgesPermanent$ routine). Therefore, all shortcutters are carved with respect to $\mc Q_{k^*+1}$. For more details on this argument, see Proposition \ref{prop:make-nonedges-permanent}.

\subsubsection{Choosing conditioning hierarchies that respect the key ideas}

While we have discussed most of the key properties of conditioning hierarchies that allow $\ConditioningVerts$ to establish Lemmas \ref{lem:holes} and \ref{lem:progress} (as summarized in the Key Ideas), it has not been made clear how to actually obtain any of them. We now discuss some insufficient approaches for defining the $\mc Q_i$s and $\mc R_i$s. Specifically, we attempt to design $\TwoLevelConditioningVerts$, which is called on level $i^*$ and below if $s_{i^*} > 0$ in order to get closer to making all of $\mc R_{i^*}$ inactive. For simplicity, we focus our discussion on the first call to $\ConditioningVerts$, in which case $\mc Q_i = \mc R_i = \emptyset$ for all $i\le \sigma_0$ and $\mc Q_{\sigma_0+1}$ and $\mc R_{\sigma_0+1}$ consist of just one part that contains the entire graph.

One natural attempt at designing $\TwoLevelConditioningVerts$ is to do the following:

\begin{itemize}
\item For each $i \gets \sigma_0,\sigma_0-1,\hdots,1$,
    \begin{itemize}
    \item Let $\mc R_i$ be the set of parts with parents in $\mc Q_{i+1}$.
    \item Let $\mc Q_i$ be the parts $P\in \mc R_i$ with near-maximum-size shortcutters.
    \end{itemize}
\end{itemize}

This approach is a simple generalization of the algorithm $\SimpleConditioningVerts$. Unfortunately, it does not respect Key Idea \ref{keyidea:relevance} because parts $P\in \mc P_i(\mc E)\setminus \mc R_i$ could have shortcutters $S_P$ that intersect some part in $\mc Q_{i+1}$. More concretely, this approach does not work because there may be a part $P\in \mc P_{\sigma_0-2}(\mc E)\setminus R_{\sigma_0-2}$ for which (a) $S_P$ intersects $\mc Q_{i+1}$ and (b) $S_P$ is very large compared to the shortcutters of the parts in $\mc Q_i$. In this case, carving would increase the conductance of $S_P$ too much (see Figure \ref{fig:quadratic-blowup}).

We deal with the above issue by allowing $\mc Q_i$ to not necessarily be contained in $\mc Q_{i+1}$ and picking $\mc Q_i$ before $\mc R_i$ rather than the other way around. In particular, if there is a $P\in \mc P_i(\mc E)$ with a parent in $\mc R_{i+1}$ for which (a) $S_P$ intersects $\mc Q_{i+1}$ and (b) $S_P$ is very large compared to the shortcutters for the default choice of $\mc Q_i$ described in the first attempt, switch $\mc Q_i$ to be all parts satisfying (a) that have near-maximum shortcutter size (similar to (b)). Of course, the new choice for $\mc Q_i$ may be bad for the exact same reasons as the original choice. Luckily, though, this switching procedure can only happen $\sigma_1$ times because each switch of $\mc Q_i$ increases the sizes of the shortcutters considered by an $m^{1/\sigma_1}$ factor.

The above approach also has the unfortunate property that chosen parts at level $j$ for $j < i$ could make choices at higher levels unusable. To deal with this problem, we would like all parts at lower levels to look as if they are part of $\mc Q_i$. One way of doing this is to require chosen parts at level $\mc Q_j$ to be close to $\mc Q_i$ (see the ``Vertical closeness'' condition in the definition of conditioning hierarchies). This suggests modifying the algorithm mentioned in the previous paragraph to discuss the effective resistance metric:

\begin{itemize}
\item For each $i = \sigma_0,\sigma_0-1,\hdots,1$,
    \begin{itemize}
    \item $\mc Q_i\gets$ the set of parts of $\mc P_i(\mc E)$ with parents in $\mc Q_{i+1}$
    \item While there is a part $P\in \mc P_i(\mc E)$ with parent in $\mc R_i$ for which (a) $S_P$ intersects $\mc Q_{i+1}$ (b) $S_P$ is much larger than shortcutters for parts in $\mc Q_i$ and (c) $P$ is not too much farther away (say no more than 7 times farther away) from $\mc Q_{i+1}$ than some part in the current value of $\mc Q_i$
        \begin{itemize}
        \item Replace $\mc Q_i$ with all parts $P$ with parent in $\mc R_i$ that satisfy (a) and (c) with near-maximum size.
        \end{itemize}
    \item $\mc R_i\gets$ all parts with parents in $\mc R_i$ that satisfy (a) and (c) for $\mc Q_i$
    \end{itemize}
\end{itemize}

The irrelevant parts are far away, so $\mc Q_{i+1}$ can be carved out of their shortcutters. While this algorithm works for call to $\ConditioningVerts$, it happens to violate Key Idea \ref{keyidea:progress}. Luckily, we can fix this issue by observing that the above strategy is part of a more general family of strategies based on a digraph $\CG$ at each level $i$. This digraph is a digraph on parts in $\mc P_i(\mc E)$ with parents in $\mc R_{i+1}$. An edge is present in this digraph from $P\rightarrow Q$ if and only if (1) $P$ intersects $S_Q$ and (2) $Q$ is not much farther from $\mc Q_{i+1}$ than $P$. (1) and (2) are similar to (a) and (c) respectively.

As a result, the second attempt given above can be viewed as doing a BFS in this digraph and making a histogram of sizes of shortcutters for parts at these distances. The above strategy is equivalent to letting $\mc Q_i$ be the set of parts corresponding to the closest local maximum to the source $\mc Q_{i+1}$ in this histogram. See Figure \ref{fig:two-level-diagram} for a visual on the construction of this histogram. This visual suggests a large family of strategies based on finding local maxima in histograms. In $\TwoLevelConditioningVerts$, we give a simple histogram strategy which picks local maxima each time. This allows us to satisfy Key Idea \ref{keyidea:relevance}. However, this strategy also has the property that $\mc R_i$ shrinks across multiple conditioning rounds. This ensures that Key Idea \ref{keyidea:progress} is also respected.

There are other complications that we have not adequately addressed in this summary. For example, it is a priori unclear if conditioning digraphs from one call to $\ConditioningVerts$ relate in any way to conditioning digraphs in previous calls. Luckily, later digraphs are subgraphs of earlier digraphs, thanks to the ``Containment'' input condition, after modification by the routine $\MakeNonedgesPermanent$ that does not increase conductivity much.

\begin{figure}
\includegraphics[width=1.0\textwidth]{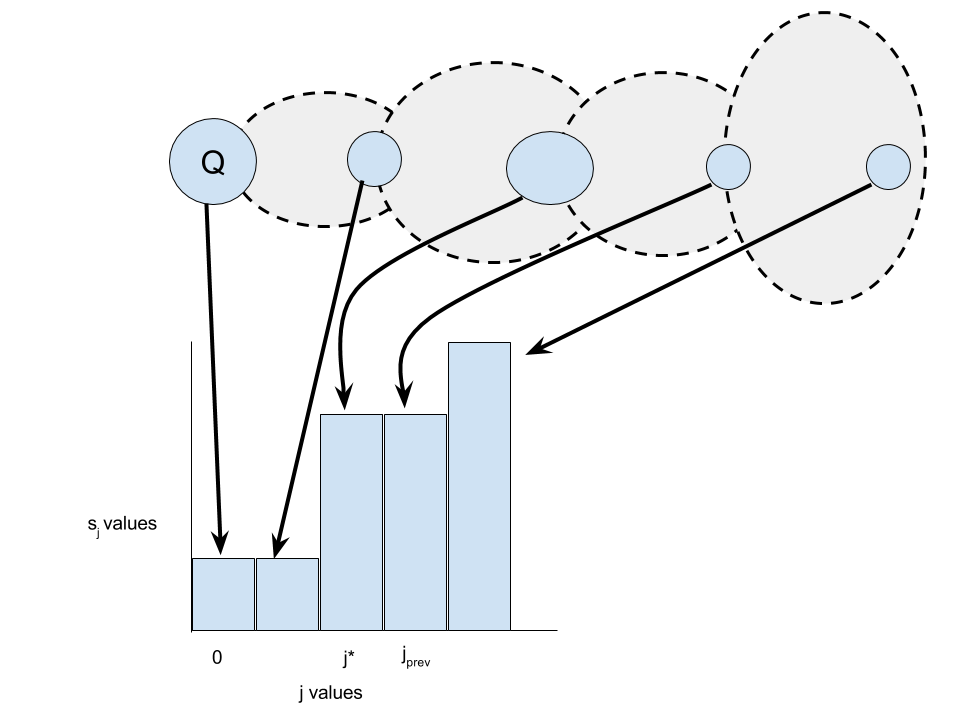}
\caption{Constructing $s_j$s from parts and their shortcutters in $\TwoLevelConditioningVerts$. Our choice of $j^*$ allows us to get rid of all parts that are closer to $\mc Q$ than $j^*$ that have near-maximum size, thus bringing us one step closer to being able to condition on $Q$ (one out of at most $\sigma_1$ steps.)}
\label{fig:two-level-diagram}
\end{figure}

\subsection{Conditioning hierarchies: the data structure for representing parts to condition on}\label{subsec:cond-hierarchies}

In this section, we formally define the data structure used to keep track of sets to condition on: the \emph{conditioning hierarchy}. We again encourage the reader to delay trying to understand exactly how the $\mc Q_i$s and $\mc R_i$s are constructed until Section \ref{subsec:two-level}, where $\TwoLevelConditioningVerts$ is introduced.

\begin{definition}[Conditioning hierarchies]
Consider the empire $\mc E$, a family of sets of \emph{chosen parts} $\{\mc Q_i\}_{i=1}^{\sigma_0}$, and a family of sets of \emph{relevant parts} $\{\mc R_i\}_{i=1}^{\sigma_0}$. The pair $\mc{CH} = (\{\mc Q_i\}_{i=1}^{\sigma_0}, \{\mc R_i\}_{i=1}^{\sigma_0})$ is called a \emph{conditioning hierarchy} if the following properties are true for all $i\in [\sigma_0]$:

\begin{itemize}
\item (Horizontal containment) $\mc Q_i\subseteq \mc R_i\subseteq \mc P_i(\mc E)$
\item (Vertical containment) For every part $P\in \mc R_i$, there is some $Q\in \mc R_{i+1}$ for which $P\subseteq Q$.
\item (Vertical closeness) Each vertex in a part of $\mc Q_i$ is within distance $(\mucarve/(100\gammatemp\gammaann^2))\R{i+1}$ of some vertex in a part of $\mc Q_{i+1}$, where $\gammatemp := 4\gammadel$.
\end{itemize}

\end{definition}

The ``Vertical closeness'' property is in some sense a weaker version of the ``Vertical containment'' property that applies to the chosen parts. While the chosen parts from different hordes are not truly nested, they are close to chosen parts from higher hordes. Recall from the end of the previous section that we do not want the $\mc Q_i$s to be truly nested, as there may be a much larger shortcutter that intersects $\mc Q_i$ but whose core is outside of $\mc Q_i$. 

As the conditioning algorithm proceeds, the ``Containment'' input condition allows us to argue that conditioning hierarchies are in some sense contained within one another:

\begin{definition}[Fitting of conditioning hierarchies] 
Consider two conditioning hierarchies $\mc{CH} = (\{\mc Q_k\}_{k=1}^{\sigma_0},\{\mc R_k\}_{k=1}^{\sigma_0})$ and $\mc{CH}' = (\{\mc Q_k'\}_{k=1}^{\sigma_0},\{\mc R_k'\}_{k=1}^{\sigma_0})$ for two (possibly different) empires $\mc E$ and $\mc E'$.

$\mc{CH}$ $i$-\emph{fits} within $\mc{CH}'$ if both of the following conditions hold for all $j\ge i$:

\begin{itemize}
\item ($\mc Q$-fitting) For all $P\in \mc Q_j$, there is a $Q\in \mc Q_j'$ for which $P\subseteq Q$.
\item ($\mc R$-fitting) For all $P\in \mc R_j$, there is a $Q\in \mc R_j'$ for which $P\subseteq Q$.
\end{itemize}

$\mc{CH}$ \emph{completely fits} in $\mc{CH}'$ if $\mc{CH}$ 1-fits in $\mc{CH}'$.
\end{definition} 

Next, we define a property of chosen parts that will allow us to carve without substantially increasing conductivity when combined with the ``Vertical closeness'' property of conditioning hierarchies:

\begin{definition}[Descendant size]
Given an empire $\mc E$ and a part $P\in \mc P_i(\mc E)$, the \emph{descendant size} of $P$, denoted $\DSize(\mc E, P)$, is the maximum number of edges incident with any shortcutter for a descendant part:

$$\DSize(\mc E, P) := \max_{Q\in \cup_{j\le i} \mc P_j(\mc E): Q\subseteq P} |E(S_Q)\cup \partial S_Q|$$
\end{definition}

When understanding these concepts for the first time, it is helpful to think of defining $\DSize(\mc E, P)$ as $|E(S_P)\cup \partial S_P|$ instead of the maximum over all descendant parts. $\DSize$ is used in place of $|E(S_P)\cup \partial S_P|$ for reasons discussed in Section \ref{subsec:cond-digraphs}.

\begin{definition}[Locally maximum conditioning hierarchies]
Consider a conditioning hierarchy $\mc{CH} = (\{\mc Q_k\}_{k=1}^{\sigma_0},\{\mc R_k\}_{k=1}^{\sigma_0})$ for an empire $\mc E$. $\mc{CH}$ is said to be \emph{locally maximum} if for all $i\in [\sigma_0]$ with $\mc Q_i\ne \emptyset$,

$$\min_{P\in \mc Q_j} \DSize(\mc E, P)\ge m^{-1/\sigma_1} \max_{Q\in \mc R_j} \DSize(\mc E, Q)$$
\end{definition}

$\ConditioningVerts$ keeps a conditioning hierarchy as continued state, along with a size and distance parameter associated with each $i\in [\sigma_0]$. The size parameter tracks the largest size of a shortcutter for a part in $\mc R_i$. We discuss the distance parameter in Section \ref{subsec:cond-digraphs}. The conditioning hierarchy is initialized with $\mc Q_i = \emptyset$ for all $i\in [\sigma_0]$ and $\mc R_i = \emptyset$ except for $i = \sigma_0 + 1$, for which $\mc Q_{\sigma_0+1} = \{V(G)\}$ and $\mc R_{\sigma_0+1} = \{V(G)\}$. On a particular call to $\ConditioningVerts$, it

\begin{itemize}
\item Picks the highest $i$ for which both (1) $\mc Q_i = \emptyset$ and (2) all parts in $\mc R_i$ have inactive shortcutters and calls this value $i^*$
\item If a part in $\mc R_{i^*}$ has an active shortcutter, $\ConditioningVerts$ uses a ball-growing strategy (see $\TwoLevelConditioningVerts$) to reselect the chosen and relevant parts for levels $i^*$ and below. After doing this, $\mc Q_1$ is selected for conditioning.
\item Otherwise, $\mc Q_{i^*+1}$ can be conditioned on (see Proposition \ref{prop:carve-hierarchy})
\end{itemize}

This process results in a conditioning hierarchy that $i^*+1$-fits in the previous conditioning hierarchy. As a result, all $s_j$ for $j > i^*$ do not increase.

\begin{figure}
\includegraphics[width=1.0\textwidth]{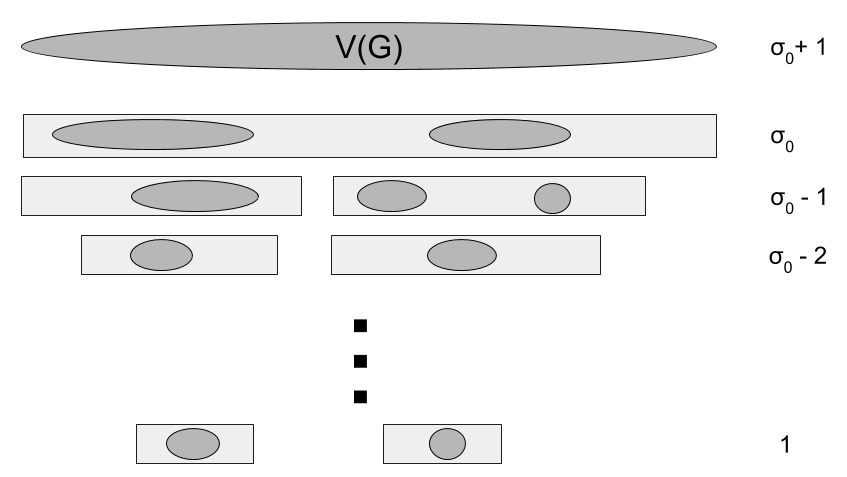}
\caption{Typical $\mc Q_i$s (dark gray ovals) and $\mc R_i$s (light gray rectangle). Notice that the $\mc R_i$s are laminar and that they contain the $\mc Q_i$ parts from the level above.}
\label{fig:multiple-levels}
\end{figure}

\subsection{Conditioning digraphs: the metric used for choosing parts to condition on}\label{subsec:cond-digraphs}

At the end of Section \ref{subsec:multiple-levels-intuition}, we discussed the fact that picking parts to condition on based on a plot of shortcutter $\DSize$s versus ``distance'' would be a good strategy. We now define $\CG$, from which that distance is defined.

The digraph $\CG(\mc E,i+1,\mc Q_{i+1},\mc R_{i+1})$ is used to obtain $\mc Q_i$ and $\mc R_i$ from $\mc Q_{i+1}$ from $\mc R_{i+1}$.

\begin{definition}[Conditioning Digraph]
Consider an empire $\mc E$ defined in a graph $H$, an index $i\in [\sigma_0]$, and a set of chosen parts $\mc Q\subseteq \mc P_i(\mc E)$, and a set of relevant parts $\mc R\subseteq \mc P_i(\mc E)$ with $\mc Q\subseteq \mc R$. For each part $P\in \mc R$, pick an arbitrary representative vertex $v_P\in P$. The digraph $I = \CG(\mc E,i,\mc Q,\mc R)$, called a \emph{temporary conditioning digraph}, is the digraph with vertex set $V(I) = \mc R$ and a directed edge $(P,Q)\in E(I)$ if and only if all of the following conditions hold, where $D\gets \PreprocANN(H,\{v_P\}_{P\in \mc Q})$:

\begin{itemize}
\item (Permanent condition) There exists some descendant $Q'\in \cup_{j\le i} \mc P_j(\mc E)$ of $Q$ for which $P$ intersects $S_{Q'}$.
\item (Temporary condition) $\texttt{JLReff}_H(v_Q,\ANN_D(v_P)) \le 8\gammatemp \gammaann (\texttt{JLReff}_H(v_P,\ANN_D(v_P)) + \R{i})$ where $\texttt{JLReff}_H(a,b)$ for $a,b\in V(H)$ denotes a multiplicative 2-approximation to $\texttt{Reff}_H(a,b)$ given using Theorem \ref{thm:jl}.
\end{itemize}

The digraph $J = \PCG(\mc E,i,\mc R)$, called a \emph{permanent conditioning digraph}, has vertex set $V(J) = \mc R$ and a directed edge $(P,Q)\in E(I)$ if and only if the ``Permanent Condition'' holds for the $(P,Q)$ pair.
\end{definition}

The ``Temporary condition'' implies the following relevant metric properties thanks to the ``Closeness'' guarantee of Theorem \ref{thm:ann}. We only use $\ANN$ in the definition of the temporary condition for runtime purposes:

\begin{corollary}[Temporary condition implications]\label{cor:ann}
Consider two different parts $P,Q\in \mc R$ and let $P_{\min} := \arg \min_{P'\in \mc Q} \texttt{Reff}_H(v_P,v_{P'})$. If the pair of parts $(P,Q)$ satisfies the ``Temporary condition,'' then

$$\texttt{Reff}_H(v_Q,v_{P_{\min}}) \le 34 \gammatemp \gammaann^2 (\texttt{Reff}_H(v_P,v_{P_{\min}}) + \R{i})$$

If the pair does not satisfy the ``Temporary condition,'' then

$$\texttt{Reff}_H(v_Q,v_{P_{\min}}) \ge \gammatemp (\texttt{Reff}_H(v_P,v_{P_{\min}}) + \R{i})$$
\end{corollary}

\begin{proof}
\textbf{Upper bound.} Since the ``Temporary condition'' is satisfied,

$$\texttt{JLReff}_H(v_Q,\ANN_D(v_P)) \le 8\gammatemp \gammaann (\texttt{JLReff}_H(v_P,\ANN_D(v_P)) + \R{i})$$

Since $\texttt{JLReff}_H$ is a 2-approximation to $\texttt{Reff}_H$,

$$\texttt{Reff}_H(v_Q,\ANN_D(v_P)) \le 32\gammatemp \gammaann (\texttt{Reff}_H(v_P,\ANN_D(v_P)) + \R{i})$$

By Theorem \ref{thm:ann},

$$\texttt{Reff}_H(v_Q,\ANN_D(v_P)) \le 32\gammatemp \gammaann^2 (\texttt{Reff}_H(v_P,v_{P_{\min}}) + \R{i})$$

By the triangle inequality and Theorem \ref{thm:ann},

\begin{align*}
\texttt{Reff}_H(v_Q,v_{P_{\min}}) &\le \texttt{Reff}_H(v_Q,\ANN_D(v_P)) + \texttt{Reff}_H(\ANN_D(v_P),v_P) + \texttt{Reff}_H(v_P,v_{P_{\min}})\\
&\le 34\gammatemp \gammaann^2 (\texttt{Reff}_H(v_P,v_{P_{\min}}) + \R{i})\\
\end{align*}

as desired.

\textbf{Lower bound.} Since the ``Temporary condition'' is not satisfied,

$$\texttt{JLReff}_H(v_Q,\ANN_D(v_P)) \ge 8 \gammatemp \gammaann (\texttt{JLReff}_H(v_P,\ANN_D(v_P)) + \R{i})$$

By the approximation of $\texttt{JLReff}_H$ to $\texttt{Reff}_H$,

$$\texttt{Reff}_H(v_Q,\ANN_D(v_P)) \ge 2 \gammatemp \gammaann (\texttt{Reff}_H(v_P,\ANN_D(v_P)) + \R{i})$$

By definition of $P_{\min}$,

$$\texttt{Reff}_H(v_Q,\ANN_D(v_P)) \ge 2 \gammatemp \gammaann (\texttt{Reff}_H(v_P,v_{P_{\min}}) + \R{i})$$

By the triangle inequality and Theorem \ref{thm:ann},

\begin{align*}
\texttt{Reff}_H(v_Q,v_{P_{\min}}) &\ge \texttt{Reff}_H(v_Q,\ANN_D(v_P)) - \texttt{Reff}_H(\ANN_D(v_P),v_P) - \texttt{Reff}_H(v_P,v_{P_{\min}})\\
&\ge \gammatemp \gammaann (\texttt{Reff}_H(v_P,v_{P_{\min}}) + \R{i})\\
&\ge \gammatemp (\texttt{Reff}_H(v_P,v_{P_{\min}}) + \R{i})\\
\end{align*}

as desired.

\end{proof}

Notice that the ``Permanent condition'' is slightly different from the condition (a) discused in the last subsubsection of Section \ref{subsec:multiple-levels-intuition} in that we are interested in \emph{any} intersection of descendant shortcutters of $Q$ with $P$, not just $S_Q$ with $P$. This is done for the same reason that we use $\DSize$ in place of the actual size of $S_Q$. Specifically, we use these definitions because they guarantee that any edge in $\PCG$ for level $i$ is also an edge in $\PCG$ for higher levels:

\begin{proposition}[Vertical monotonicity]\label{prop:pcg-vertical-monotonicity}
Consider any empire $\mc E$, two indices $j < i\in [\sigma_0]$, and two sets of parts $\mc R_j\subseteq \mc P_j(\mc E)$ and $\mc R_i\subseteq \mc P_i(\mc E)$ with every part in $\mc R_j$ having an ancestor in $\mc R_i$. Then for any $(P,Q)\in \PCG(\mc E,j,\mc R_j)$,

$(P',Q')\in \PCG(\mc E,i,\mc R_i)$, where $P'$ and $Q'$ are the ancestors of $P$ and $Q$ respectively in $\mc P_i(\mc E)$.
\end{proposition}

\begin{proof}
The permanence condition only gets more restrictive in lower levels. Therefore, if $(P,Q)\in \PCG(\mc E,j,\mc R_j)$,

the edge $(P',Q')\in \PCG(\mc E,i,\mc R_i)$, as desired.
\end{proof}

$\TwoLevelConditioningVerts$ uses $I \gets \CG(\mc E,i+1,\mc Q_{i+1},\mc R_{i+1})$ to define a directed metric on the set $\mc R_{i+1}'\subseteq \mc P_i(\mc E)$ of parts with ancestors in $\mc R_{i+1}$. Specifically, define a (distance) function $d_{dir}$ on $\mc R_{i+1}'$ by letting $d_{dir}(Q)$ be the distance from $\mc Q_{i+1}$ to the parent of $Q$ in $\mc R_{i+1}$ in the digraph $I$.

$\TwoLevelConditioningVerts$ could use $d_{dir}$ by picking some distance threshold $j^*+1$ from $\mc Q_{i+1}$, letting $\mc R_i$ be the set of all $P\in \mc R_{i+1}'$ with $d_{dir}(P)\le j^*+1$, and defining $\mc Q_i$ to be the set of parts $P\in \mc R_i$ with $d_{dir}(P)\le j^*$ with near-largest shortcutter size. If it did this, $\mc R_{i+1}'\setminus \mc R_i$ would consist of parts that at least distance $\R{i+1}$ from any of the chosen parts $\mc Q_i$. Even better, the $\mc Q_i$s are at least $\gammatemp$ times closer to $\mc Q_{i+1}$ in the $H$-effective resistance metric than any part in $\mc R_{i+1}'\setminus \mc R_i$. Therefore, the conditioning hierarchy being locally maximum at level $i+1$ shows that deleting all parts of $\mc Q_i$ from shortcutters for parts in $\mc Q_{i+1}$ that they have edges to does not increase the conductivity of $\mc E$ much. The upside of doing this is that all shortcutters for parts in $\mc R_{i+1}'\setminus \mc R_i$ have shortcutters that are carved with respect to $\cup_{j\le i+1} \mc Q_j$, thus respecting Key Idea \ref{keyidea:relevance}. This approach does not quite work across multiple conditioning rounds, but is easily fixed in Section \ref{subsec:two-level} by using $j^* + (2\sigma_1)^i$ instead of $j^*+1$.

The previous paragraph suggests that conditioning digraphs can be used to pick good parts to condition on once. To pick good parts to condition on multiple times, we need to show that (a) distances only grow in conditioning digraphs and (b) parts chosen at lower levels cannot intersect shortcutters for parts in higher levels that were declared irrelevant. We discuss concern (a) in Section \ref{subsec:maintained-state}, while we discuss concern (b) in Section \ref{subsec:two-level}.

\subsubsection{Fast computation with conditioning digraphs}

We now show that computations involving conditioning digraphs are cheap. These graphs may be very dense, but luckily they are not dense compared to the input graph $H$:

\begin{proposition}\label{prop:runtime-cg}
The conditioning digraphs has two nice computational properties:

\begin{itemize}
\item (Size) $\CG(\mc E,i,\mc Q,\mc R)$ and

$\PCG(\mc E,i,\mc R)$ have at most $m^{1 + o(1)}$ edges.
\item (Computation time) The graphs $\CG(\mc E,i,\mc Q,\mc R)$ and

$\PCG(\mc E,i,\mc R)$ can be computed in $m^{1 + o(1)}$ time.
\end{itemize}

\end{proposition}

\begin{proof}

\textbf{Size.} It suffices to bound the number of possible intersections between a cluster $P$ and a shortcutter $S_Q$ (property (1) for edge existence). Each intersection must contain a vertex and each vertex is only in one shortcutter in each clan. Since there are only $m^{o(1)}$ clans in $\mc E$, only $(m^{o(1)})^2 \le m^{o(1)}$ pairwise intersections can involve this vertex. Therefore, there are at most $nm^{o(1)}\le m^{1 + o(1)}$ candidate pairs that satisfy (1).

\textbf{Computation time.} For each candidate edge described in the ``Size'' part, it takes two approximate nearest neighbor calls to check whether then ``Temporary condition'' is satisfied. By Theorems \ref{thm:ann} and \ref{thm:jl}, these queries take $\tilde{O}(1)$ time each, with $\tilde{O}(m)$ preprocessing time up front. Therefore, it only takes $m^{1 + o(1)}$ time to compute $\CG$.
\end{proof}

Since this graph has almost-linear size, $\CG(\mc E,i,\mc Q,\mc R)$

and $\PCG(\mc E,i,\mc R)$-distances from all parts in $\mc Q$ to each part in $\mc R$ can be computed in total time $m^{1 + o(1)}$.

\subsection{Maintained state across multiple conditioning rounds (and deleting far-away conditioned vertices) using $\MakeNonedgesPermanent$}\label{subsec:maintained-state}

In this section, we deal with issue (a) raised in the previous section by turning nonedges that do not satisfy the ``Temporary condition'' into ones that do not satisfy the ``Permanent condition.'' We start by showing that the permanent conditioning digraph at each level only gets smaller:

\begin{proposition}[Permanent horizontal monotonicity]\label{prop:pcg-horizontal-monotonicity}
Consider two empires $\mc E_{prev}'$ and $\mc E$ that satisfy the ``Containment'' input condition in Definition \ref{def:cond-verts-input}. Consider some $i\in [\sigma_0]$ and relevant sets $\mc R_{prev}'\subseteq \mc P_i(\mc E_{prev}')$ and $\mc R\subseteq \mc P_i(\mc E)$. Suppose that each part $P\in \mc R$ is contained in some (unique by ``Containment'') part $Q\in \mc R_{prev}'$. Let $I_{prev}' := \PCG(\mc E_{prev}',i,\mc R_{prev}')$ and $I := \PCG(\mc E,i,\mc R)$.

Consider any parts $P,P'\in V(I)$ and let $Q,Q'\in V(I_{prev}')$ be the unique parts that contain $P$ and $P'$ respectively. If $(P,P')\in E(I)$, then $(Q,Q')\in E(I_{prev}')$.
\end{proposition}

\begin{proof}
Since $(P,P')\in E(I)$, there exists a descendant $W'\in \mc P_j(\mc E)$ of $P'$ for which $P$ intersects $S_{W'}$ for some $j \le i$. Let $U'\in \mc P_j(\mc E_{prev}')$ be the unique part for which $W'\subseteq U'$. This exists by the ``Containment'' input condition in Definition \ref{def:cond-verts-input}. Also by ``Containment,'' $S_{W'}\subseteq S_{U'}$. Therefore, $S_{U'}$ intersects $Q$ since $P\subseteq Q$.

Now, we just need to show that $U'$ is a descendant of $Q'$; i.e. that $U'\subseteq Q'$. $W'\subseteq P'$ since $W'$ is a descendant of $P'$. $P'\subseteq Q'$ by definition of $Q'$. $W'\subseteq U'$ by definition of $U'$, which means that $U'\cap Q' \ne \emptyset$. Since the overlay partitions of any empire form a laminar family of sets, this means that $U'\subseteq Q'$ or that $Q'\subseteq U'$. Since $U'$ is in a lower horde than $Q'$, $U'\subseteq Q'$. Therefore, $Q'$ has a descendant whose shortcutter intersects $Q$, so $(Q,Q')\in E(I_{prev}')$, as desired.
\end{proof}

Proposition \ref{prop:pcg-horizontal-monotonicity} only reasons about containment of permanent conditioning digraphs between applications of $\ConditioningVerts$. In particular, it does not say anything a priori about containment of temporary conditioning digraphs. In general, temporary conditioning digraphs may not be contained within prior conditioning digraphs because conditioning can substantially change the effective resistance metric of a graph. We show that the shortcutters in $\mc E$ can be modified to make temporary conditioning digraphs into permanent ones. As long as the $\mc Q$s defining each temporary conditioning digraph are part of a locally maximum conditioning hierarchy, this modification does not substantially increase the conductivity of $\mc E$.

\begin{algorithm}[H]
\SetAlgoLined
\DontPrintSemicolon
\caption{$\MakeNonedgesPermanent(\mc E, i, \mc Q, \mc R)$}

    \ForEach{clan $\mc C$ in the horde $\mc E_i$ of $\mc E$}{

        split $\mc C$ into clans $\mc C_0, \mc C_1, \hdots, \mc C_{\log m}$, where $\mc C_j$ consists of all shortcutters in $\mc C$ with between $2^j$ and $2^{j+1}$ incident edges\; \label{line:split-clans-mnp}

    }

    $I\gets \CG(\mc E,i,\mc Q,\mc R)$\;

    \ForEach{pair $(P,Q)$ with $P,Q\in \mc R$ and $P$ intersecting $S_{Q'}$ for some descendant $Q'$ of $Q$}{

        \If{$(P,Q)\notin E(I)$}{

            \ForEach{descendant $Q'$ of $Q$}{

                $S_{Q'}\gets S_{Q'}\setminus P$\; \label{line:remove-intersections}

            }

        }

    }

\end{algorithm}

\begin{proposition}\label{prop:make-nonedges-permanent}
$\MakeNonedgesPermanent(\mc E,i,\mc Q,\mc R)$ takes an empire $\mc E$, an index $i$, a set of parts $\mc R\subseteq \mc P_i(\mc E)$, and a set of parts $\mc Q\subseteq \mc R$.

Suppose that that the locally maximum condition holds; i.e. that

$$\min_{P\in \mc Q} \DSize(\mc E, P) \ge m^{-1/\sigma_1} \max_{P\in \mc R} \DSize(\mc E, P)$$

Then $\MakeNonedgesPermanent$ deletes vertices from shortcutters in $\mc E$ and splits clans to obtain an empire $\mc E'$ with the following guarantees:

\begin{itemize}
\item (Permanence) $\PCG(\mc E',i,\mc R)$ is a subdigraph of

$\CG(\mc E,i,\mc Q,\mc R)$.
\item (Conductivity and number of clans) If $\mc E$ is $\zeta$-conductive, then $\mc E'$ is $\zeta + (\log n)\muapp(\ellmax + \tau)$-conductive. Furthermore, $|\mc E_i'|\le O(\log n) |\mc E_i|$ for all $i\in [\sigma_0]$.
\end{itemize}

Furthermore, $\MakeNonedgesPermanent$ takes almost-linear time.
\end{proposition}

\begin{figure}

\includegraphics[width=1.0\textwidth]{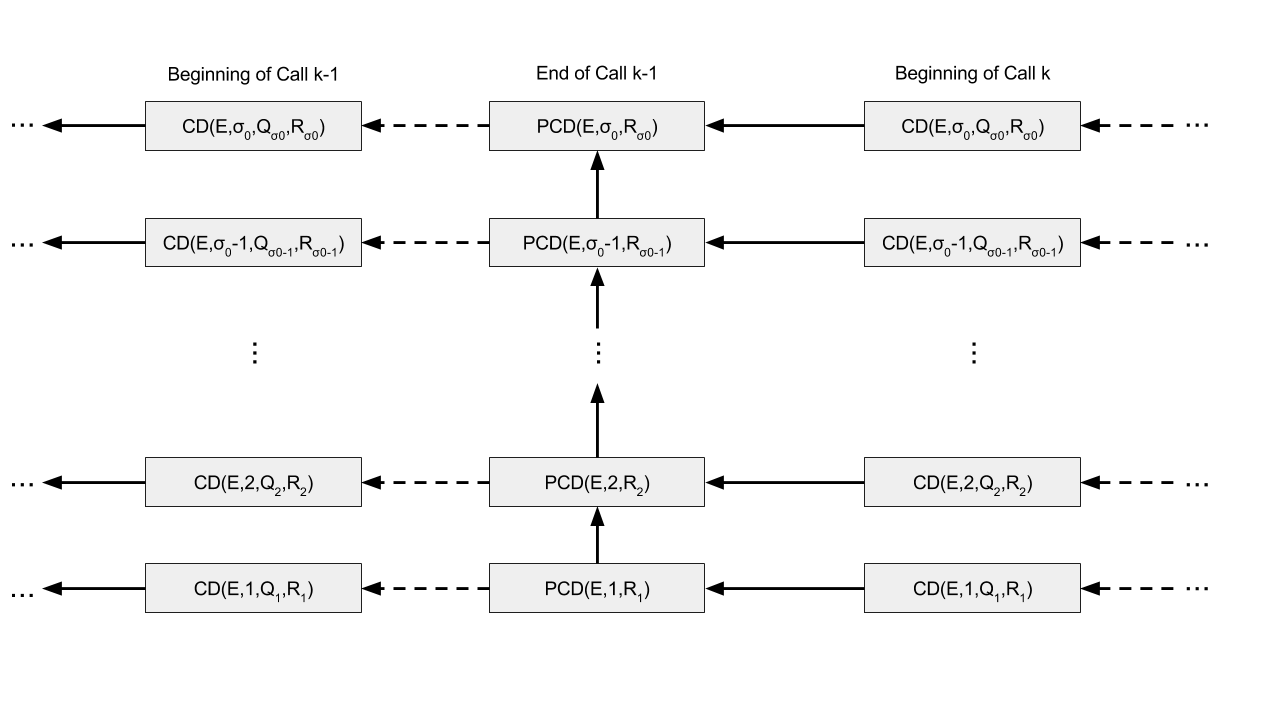}

\caption{Containment relationships between conditioning digraphs. $CD$ and $PCD$ abbreviate $\CG$ and $\PCG$ respectively. A conditioning digraph $I_0$ is \emph{contained} in another conditioning digraph $I_1$ if (a) each vertex (part) $P\in V(I_0)$ is entirely contained in a unique part $Q\in V(I_1)$ and (b) for any edge $(P,P')\in E(I_0)$, $(Q,Q')\in E(I_1)$, where $Q$ and $Q'$ are the unique containing parts for $P$ and $P'$. Arrows point towards the containing graph. The dashed horizontal edges are present due to Lemma \ref{lem:tcg-horizontal-monotonicity} and calls to $\MakeNonedgesPermanent$. The solid horizontal edges are present due to Proposition \ref{prop:pcg-horizontal-monotonicity}. The vertical edges are present thanks to Proposition \ref{prop:pcg-vertical-monotonicity}.}
\label{fig:containment-relationships}
\end{figure}

\begin{proof}
\textbf{Permanence.} Line \ref{line:remove-intersections} removes any intersection between two parts with no edge from one to the other, thus ensuring that any nonedge in $\CG(\mc E,i,\mc Q,\mc R)$, which may satisfy the ``Permanent condition,'' does not satisfy it in $\mc E'$.

\textbf{Conductivity and number of clans.} Removing shortcutters from a clan cannot decrease its effective size, so $s_{\mc C_k}\le s_{\mc C}$ for all $k\in \{0,1,\hdots,\log m\}$. Therefore, the number of clans grew by a factor of at most $O(\log n)$.

Now, we bound the increase in conductivity. It suffices to consider a clan $\mc C_k\in \mc E_j$ for some $j\le i$ that contains a shortcutter for a descendant of a part $P_*\in \mc R$. Otherwise, no shortcutter in $\mc C_k$ is modified by Line \ref{line:remove-intersections}, which means that its conductivity does not increase.

We start by characterizing the removals done by Line \ref{line:remove-intersections} in a simpler way. Let $\mc N_z$ be a set of clusters with radius $2^z\R{i}$ that contain all vertices with $H$-effective resistance distance at most $2^z\R{i}$ from $\mc Q$. Then all vertex removals done by Line \ref{line:remove-intersections} are also done by, for all $z\in \{0,1,\hdots,\log m\}$, removing the clusters in $\mc N_z$ from shortcutters for cores in $\mc C_k$ with $H$-distance greater than $\gammatemp 2^z\R{i}$ from $\mc N_z$. This holds by the ``Temporary condition'' not being satisfied for nonedges of $\CG(\mc E,i,\mc Q,\mc R)$ and Corollary \ref{cor:ann}.

Now, we bound the effect of deleting $\mc N_z$ for one $z\in \{0,1,\hdots,\log m\}$ from shortcutters for $\gammatemp$-separated cores in $\mc C_k$. We start by bounding $|\mc N_z|$ in terms of $s_{\mc C_k}$. Start by noticing to bound $|\mc N_z|$, it suffices to bound the number of $O(\R{i})$ $H$-effective resistance diameter clusters needed to cover all parts in $\mc Q$.
By definition of $\DSize(\mc E,P')$, there is a shortcutter with core in $\mc E_j$ for some $j \le i$ whose size is $\rho = \DSize(\mc E,P')$ that intersects $P'$. Each clan has at most $m/\rho$ such shortcutters. Furthermore, this core has $H$-effective resistance diameter $\R{j}\le \R{i}$. ``Assigning'' a part $P'\in \mc Q$ to the core of the maximum size descendant shortcutter therefore shows that,

\begin{align*}
|\mc N_z| &\le (\sum_{i'=1}^i |\mc E_{i'}|)\frac{m}{\min_{Q\in \mc Q} \DSize(\mc E,Q)}\\
&\le \frac{\ellmax m}{\min_{Q\in \mc Q} \DSize(\mc E,Q)}\\
\end{align*}

By the ``Locally maximum'' input condition,

$$\frac{\ellmax m}{\min_{Q\in \mc Q} \DSize(\mc E,Q)} \le \frac{\ellmax m^{1+1/\sigma_1}}{\max_{Q\in \mc R} \DSize(\mc E,Q)}$$

Since $\mc C_k$ contains a shortcutter for a descendant of a part in $\mc R$,

$$\frac{\ellmax m^{1+1/\sigma_1}}{\max_{Q\in \mc R} \DSize(\mc E,Q)} \le \frac{\ellmax m^{1+1/\sigma_1}}{\min_{C\in \mc C_k} |E(S_C)\cup \partial S_C|}$$

By definition of $\mc C_k$ on Line \ref{line:split-clans-mnp},

$$\frac{\ellmax m^{1+1/\sigma_1}}{\min_{C\in \mc C_k} |E(S_C)\cup \partial S_C|} \le \frac{2\ellmax m^{1+1/\sigma_1}}{\max_{C\in \mc C_k} |E(S_C)\cup \partial S_C|}$$

By definition of $s_{\mc C_k}$,

$$|\mc N_z|\le \frac{2\ellmax m^{1+1/\sigma_1}}{\max_{C\in \mc C_k} |E(S_C)\cup \partial S_C|} \le 2\ellmax m^{1/\sigma_1}s_{\mc C_k}$$

Therefore, by Proposition \ref{prop:delete-low-rad} applied to $\mc L\gets \mc N_z$ and $\mc C\gets \mc C_k$, Line \ref{line:remove-intersections} only additively increases the conductivity of $\mc C_k$ by $\muapp(2\ellmax + \tau)$. Therefore, summing over each $z\in \{0,1,2,\hdots,\log m\}$ shows that the total increase in conductivity of $\mc C_k$ due to Line \ref{line:remove-intersections} is at most $(\log m)\muapp(2\ellmax + \tau)$, as desired.

\textbf{Runtime.} Looping over intersecting pairs takes $m^{1+o(1)}$ time. Checking the presence of an edge takes constant time. Therefore, this procedure takes $m^{1+o(1)}$ time.
\end{proof}

Now, we are ready to prove the main result of this section. Specifically, applying $\MakeNonedgesPermanent$ to each temporary conditioning digraph ensures that temporary conditioning digraphs from later applications of $\ConditioningVerts$ are contained in earlier ones:

\begin{lemma}[Temporary horizontal monotonicity]\label{lem:tcg-horizontal-monotonicity}
Consider some $i\in [\sigma_0]$ and two conditioning hierarchies $\mc{CH}_{prev} = (\{\mc Q_i^{prev}\}_{i=1}^{\sigma_0},\{\mc R_i^{prev}\}_{i=1}^{\sigma_0})$ and $\mc{CH} = (\{\mc Q_i\}_{i=1}^{\sigma_0},\{\mc R_i\}_{i=1}^{\sigma_0})$ for the empires $\mc E_{prev}$ and $\mc E$ in graphs $H_{prev}$ and $H$ respectively.

Define an empire $\mc E_{prev}'$ in $H_{prev}$ that satisfies the following conditions:

\begin{itemize}
\item ($\mc E_{prev}'$ description) $\mc E_{prev}'$ is obtained by applying $\MakeNonedgesPermanent(\mc E_{prev},i,\mc Q_i^{prev},\mc R_i^{prev})$ for each $i$ and letting $\mc E_{prev}' \gets \mc E_{prev}$ afterwards.
\item (Containment) $\mc E_{prev}'$ and $\mc E$ satisfy the ``Containment'' input condition in Definition \ref{def:cond-verts-input}.
\item (Locally maximum) $\mc{CH}_{prev}$ is locally maximum.
\end{itemize}

For each $i$, let $I := \CG(\mc E,i,\mc Q_i,\mc R_i)$ and

$I_{prev} := \CG(\mc E_{prev},i,\mc Q_i^{prev},\mc R_i^{prev})$. Consider any parts $P,P'\in V(I)$ and let $Q,Q'\in V(I_{prev})$ be the unique parts that contain $P$ and $P'$ respectively. If $(P,P')\in E(I)$, then $(Q,Q')\in E(I_{prev})$.
\end{lemma}

\begin{proof}
Since edges in $\CG$ are more constrained then those in

$\PCG$, $I$ is a subdigraph of $\PCG(\mc E,i,\mc R_i)$. Therefore,

$$(P,P')\in E(\PCG(\mc E,i,\mc R_i))$$

By the ``Containment'' condition of this lemma, Proposition \ref{prop:pcg-horizontal-monotonicity} applies and shows that

$$(Q,Q')\in E(\PCG(\mc E_{prev}',i,\mc R_i^{prev}))$$

By the ``$E_{prev}'$ description'' and ``Locally maximum'' conditions, Proposition \ref{prop:make-nonedges-permanent} applies. By the ``Permanence'' guarantee of Proposition \ref{prop:make-nonedges-permanent}, $\PCG(\mc E_{prev}',i,\mc R_i^{prev})$ is a subdigraph of $I_{prev}$, which means that

$$(Q,Q')\in E(I_{prev})$$

as desired.
\end{proof}

\subsection{$\MakeNonedgesPermanent$ implies carving as a side effect}\label{subsec:carve-shortcutters}

In this section, we define \emph{carvable} conditioning hierarchies and show that the empire $\mc E$ associated with a conditioning hierarchy is carved with respect to the bottom of the hierarchy. In later sections, we will show that $\ConditioningVerts$ maintains a carvable conditioning hierarchy, which will make it easy to show Lemma \ref{lem:holes}.

\begin{definition}[Carvable conditioning hierarchies]
Consider an empire $\mc E$ and two families of sets $\{\mc Q_i\}_{i=1}^{\sigma_0}$ and $\{\mc R_i\}_{i=1}^{\sigma_0}$. We say that the pair $(\{\mc Q_i\}_{i=1}^{\sigma_0}, \{\mc R_i\}_{i=1}^{\sigma_0})$ is \emph{carvable} if satisfies the following properties for all $i\in [\sigma_0]$:

\begin{itemize}
\item (Horizontal containment) $\mc Q_i\subseteq \mc R_i\subseteq \mc P_i(\mc E)$
\item (Vertical containment) For every part $P\in \mc R_i$, there is some $Q\in \mc R_{i+1}$ for which $P\subseteq Q$.
\item (Strong vertical closeness) The $\mc R_{i+1}$ parent of each part in $\mc Q_i$ is within distance $(2\sigma_1)(4\sigma_1)^i$ of some part in $\mc Q_{i+1}$ in $\CG(\mc E,i+1,\mc Q_{i+1},\mc R_{i+1})$.
\item (Locally maximum) If $\mc Q_i\ne \emptyset$, $\min_{P\in \mc Q_i} \DSize(\mc E, P) \ge m^{-1/\sigma_1}\max_{Q\in \mc R_i} \DSize(\mc E, Q)$.
\item (Irrelevant parts are far away) Any part $P\in \mc P_i(\mc E)\setminus \mc R_i$ with parent $Q\in \mc R_{i+1}$ has the property that $Q$ is at least distance $(4\sigma_1)^i$ from any parent of a part in $\mc Q_i$ in the digraph $\CG(\mc E,i+1,\mc Q_{i+1},\mc R_{i+1})$.
\end{itemize}
\end{definition}

The above definition does not specifically call the $(\{\mc Q_i\}_{i=1}^{\sigma_0}, \{\mc R_i\}_{i=1}^{\sigma_0})$ pair a conditioning hierarchy. We now show that it is. This allows us to call such a pair of families of sets a \emph{carvable conditioning hierarchy} for $\mc E$:

\begin{remark}\label{rmk:carvable-ch}
A carvable pair of families of sets $(\{\mc Q_i\}_{i=1}^{\sigma_0},\{\mc R_i\}_{i=1}^{\sigma_0})$ for $\mc E$ is a locally maximum conditioning hierarchy.
\end{remark}

\begin{proof}

\textbf{Horizontal containment, vertical containment, and locally maximum.} All of these conditions are exactly the same as the definitions for conditioning hierarchies.

\textbf{Vertical closeness.} Consider the shortest path $P = P_0\rightarrow P_1\rightarrow \hdots\rightarrow P_t = Q$ in $\CG(\mc E,i+1,\mc Q_{i+1},\mc R_{i+1})$ from a part $P\in \mc Q_{i+1}$ to the parent of a part $Q\in \mc Q_i$. By Corollary \ref{cor:ann} and the triangle inequality, all vertices in $Q$ are at most distance $(34\gammatemp\gammaann^2)^t \R{i+1}$ from $\mc Q_{i+1}$ in the effective resistance metric of $H$, the graph associated with $\mc E$. By the ``Strong vertical closeness'' condition, $t\le (2\sigma_1)(4\sigma_1)^i$. Since $(34\gammatemp\gammaann^2)^{(2\sigma_1)(4\sigma_1)^i}\le (\mucarve/(100\gammatemp\gammaann^2))$, all vertices in $Q$ are at most $(\mucarve/(100\gammatemp\gammaann^2))\R{i+1}$-distance away from $P\in \mc Q_{i+1}$, as desired.
\end{proof}

Now, we show that applying $\MakeNonedgesPermanent$ for each $i\in [\sigma_0]$ does the required amount of carving:

\begin{proposition}\label{prop:carve-hierarchy}
Consider the empire $\mc E$ and a carvable conditioning hierarchy $\mc{CH} = (\{\mc Q_i\}_{i=1}^{\sigma_0}, \{\mc R_i\}_{i=1}^{\sigma_0})$ for $\mc E$. Obtain an empire $\mc E'$ by applying $\MakeNonedgesPermanent(\mc E,i,\mc Q_i,\mc R_i)$ for each $i\in [\sigma_0]$ and letting $\mc E'$ be $\mc E$ at the end. $\mc E'$ has the following property:

\begin{itemize}
\item (Carving) For each active part $P$, let $Q\in \mc P_{i^*}(\mc E)$ be the smallest ancestor of $P$ (possibly $P$) for which $Q\in \mc R_{i^*}$. Then $S_P$ is carved with respect to $\cup_{j\le i^*} \mc Q_j$.
\end{itemize}
\end{proposition}

\begin{proof}

\textbf{Well-definedness.} We just need to make sure that the input conditions for

$\MakeNonedgesPermanent$ are satisfied. $\mc Q_i\subseteq \mc R_i$ because $\mc{CH}$ is a conditioning hierarchy. Since $\mc{CH}$ is locally maximum,

$$\min_{P\in \mc Q_i} \DSize(\mc E,P) \ge m^{-1/\sigma_1} \max_{P\in \mc R_i} \DSize(\mc E,P)$$

for all $i\in [\sigma_0]$. This is the remaining input condition to $\MakeNonedgesPermanent$.

\textbf{Carving.} Let $i'\le i^*$ be the value for which $P\in \mc P_{i'}(\mc E)$. Break the reasoning up into two cases:

\underline{$i' < i^*$.} By the ``Irrelevant parts are far away'' condition of carvable conditioning hierarchies, the distance from any part in $\mc Q_{i^*}$ to $Q$ in $\PCG(\mc E',i^*,\mc R_{i^*})$ is at least $(4\sigma_1)^{i^*-1}$. By ``Strong vertical closeness'' for all $j < i^*$ and Proposition \ref{prop:pcg-vertical-monotonicity}, the distance from $\mc Q_{i^*}$ to any $\mc P_{i^*}(\mc E)$-ancestor of a part in $\mc Q_j$ is at most

$$(2\sigma_1)((4\sigma_1)^{i^*-2} + (4\sigma_1)^{i^*-3} + \hdots + 1) < (4\sigma_1)^{i^*-1} - 1$$

In particular, no $\mc P_{i^*}(\mc E)$-ancestor of a part in $\mc Q_j$ for any $j \le i^*$ has an edge to $Q$ in

$\PCG(\mc E',i^*,\mc R_{i^*})$. Therefore, since the ``Permanent condition'' is not satisfied for nonedges of this graph and $P\subseteq Q$, $S_P$ does not intersect any part in $\cup_{j\le i^*} \mc Q_j$, which is the desired carving property.

\underline{$i' = i^*$.} In this case, $P = Q\in \mc R_{i^*}$. Consider any $X\in \mc Q_j$ for some $j\le i^*$ that intersects $S_P$ and let $Y$ be its ancestor in $\mc P_{i^*}(\mc E)$. By ``Vertical containment,'' $Y\in \mc R_{i^*}$. Let $Y_{\min} = \arg\min_{Y'\in \mc Q_{i^*}} \texttt{Reff}_H(v_Y,v_{Y'})$. By ``Vertical closeness'' and the triangle inequality,

$$\texttt{Reff}_H(v_Y,v_{Y_{\min}})\le (\mucarve/(100 \gammatemp \gammaann^2))\R{i^*}$$

By construction of $\mc E'$, Lemma \ref{lem:tcg-horizontal-monotonicity} applies. Therefore, edges in

$\PCG(\mc E',i^*,\mc R_{i^*})$ also satisfy the ``Temporary condition,'' which means that

$$\texttt{Reff}_H(v_P,v_{Y_{\min}})\le (34 \gammatemp \gammaann^2)(\texttt{Reff}_H(v_Y,v_{Y_{\min}}) + \R{i^*})$$

by Corollary \ref{cor:ann}. By the triangle inequality and the previous two inequalities,

\begin{align*}
\texttt{Reff}_H(v_P,v_Y)&\le \texttt{Reff}_H(v_P,v_{Y_{\min}}) + \texttt{Reff}_H(v_{Y_{\min}},v_Y)\\
&\le (\mucarve/2) \R{i^*}\\
\end{align*}

Therefore, by $\R{i^*}$-boundedness of $\mc E_{i^*}$ and the triangle inequality, all vertices of $X$ are within distance $\mucarve \R{i^*}$ of $P$, which means that $S_P$ is carved with respect to $\cup_{j\le i^*} \mc Q_j$, as desired.
\end{proof}

\subsection{Deciding on parts to condition on from $\mc P_i(\mc E)$ given a choice from $\mc P_{i+1}(\mc E)$}\label{subsec:two-level}

In the previous section, we defined carvable conditioning hierarchies. Proposition \ref{prop:carve-hierarchy} shows that $\MakeNonedgesPermanent$ makes $\mc E$ carved with respect to the lowest $\mc Q_i$ with active shortcutters. In this section, we give an algorithm, $\TwoLevelConditioningVerts$, which adds one level to the bottom of a carvable conditioning hierarchy. $\TwoLevelConditioningVerts$ extends carvable conditioning hierarchies while ensuring that if we can condition on $\mc Q_i$, we make a substantial amount of progress. We now define conditioning hierarchies for which substantial progress can be made at any level:

\begin{definition}[Advanceable conditioning hierarchies]
Consider an empire $\mc E$, a conditioning hierarchy $\mc{CH} = (\{\mc Q_i\}_{i=1}^{\sigma_0},\{\mc R_i\}_{i=1}^{\sigma_0})$ for $\mc E$, and a tuple of \emph{distance indices} $(j_i)_{i=1}^{\sigma_0}$. $\mc{CH}$ is \emph{advanceable for} $(j_i)_{i=1}^{\sigma_0}$ if the following holds for all $i\in [\sigma_0]$:

\begin{itemize}
\item (All local maxima) Let $s_i := \lfloor \log_{m^{1/\sigma_1}}(\max_{P\in \mc R_i} \DSize(\mc E,P)) \rfloor$. Then $\mc Q_i$ contains all parts $Q\in \mc R_i$ whose $\mc P_{i+1}(\mc E)$ parents are within distance $j_i(4\sigma_1)^i$ of some part in $\mc Q_{i+1}$ in

$\CG(\mc E,i+1,\mc Q_{i+1},\mc R_{i+1})$ with $\DSize(\mc E,Q)\ge m^{s_i/\sigma_1}$.
\item (Relevant set bound) When $\mc Q_{i+1}\ne \emptyset$, $\mc R_i$ contains the set of parts whose $\mc P_{i+1}(\mc E)$ parents are within distance $(j_i+1)(4\sigma_1)^i$ of $\mc Q_{i+1}$ in $\CG(\mc E,i+1,\mc Q_{i+1},\mc R_{i+1})$.
\end{itemize}

\end{definition}

The ``All local maxima'' condition, when coupled with the ``Progress'' input condition for $\ConditioningVerts$, shows that conditioning on $\mc Q_i$ replaces the word of sizes $s_{\sigma_0+1}s_{\sigma_0}\hdots s_1$ with a lexicographically smaller word, as desired in Key Idea \ref{keyidea:progress}. The ``Relevant set bound'' is used to show that relevant sets only get smaller in the output of $\TwoLevelConditioningVerts$.

We now give an algorithm $\TwoLevelConditioningVerts$ that adds an additional level to a carvable conditioning hierarchy that is advanceable for some tuple of indices that is lexicographically smaller than the previous tuple. The intuition behind $\TwoLevelConditioningVerts$ is that the distance index $j^*$ returned has the property that parts just beyond that distance (within distance $(j^*+1)(4\sigma_1)^i$) do not have larger shortcutters. As a result, after ``Temporary condition'' nonedges are converted into ``Permanent condition'' nonedges, the shortcutters with parts that are farther than distance $(j^*+1)(4\sigma_1)^i$ respect Key Idea \ref{keyidea:relevance}.

\begin{algorithm}[H]
\SetAlgoLined
\DontPrintSemicolon
\caption{$\TwoLevelConditioningVerts(\mc E,i,\mc Q_{i+1},\mc R_{i+1},j_{prev})$}

    $I\gets \CG(\mc E,i+1,\mc Q_{i+1},\mc R_{i+1})$\;

    \ForEach{$j\in \{0,1,2,\hdots,2\sigma_1\}$}{

        $\mc Q_j'\gets $ the set of parts in $\mc P_i(\mc E)$ with $\mc P_{i+1}(\mc E)$ parents that are in $\mc R_{i+1}$ and have distance at most $j(4\sigma_1)^i$ from $\mc Q_{i+1}$ in $I$\;

        $s_j\gets \lfloor\log_{m^{1/\sigma_1}}(\max_{Q \in \mc Q_j'} \DSize(\mc E,Q)\rfloor$\; \label{line:si-def}

    }

    \tcp{selection rule: pick the maximum local maximum closer than the previous one}

    $j^*\gets \arg\max_{j: (1) s_{j+1} = s_j \text{ and } (2) j \le j_{prev}} j$\; \label{line:j-opt}

    \tcp{condition on all parts with shortcutters with size in the $s_{j^*}$ bucket}

    \Return{($j^*$, $s_{j^*}$, all parts $Q$ in $\mc Q_{j^*}'$ with $\DSize(\mc E,Q)\ge m^{s_{j^*}/\sigma_1}$, all parts $Q$ in $\mc Q_{j^*+1}'$)}

\end{algorithm}

\begin{figure}

\includegraphics[width=1.0\textwidth]{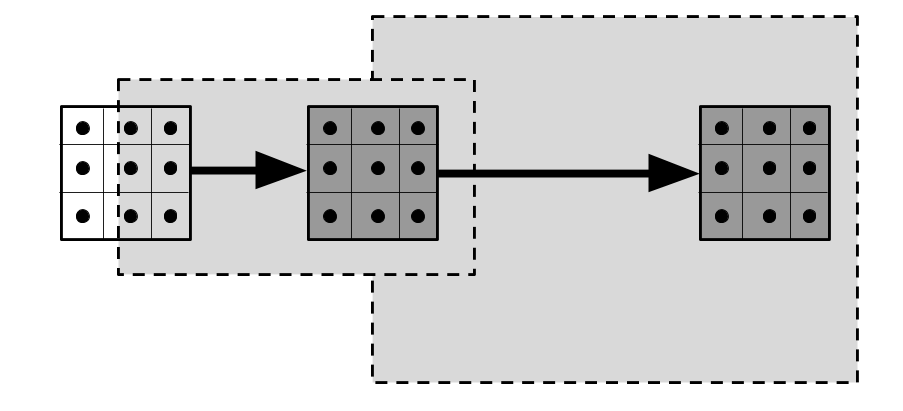}

\includegraphics[width=1.0\textwidth]{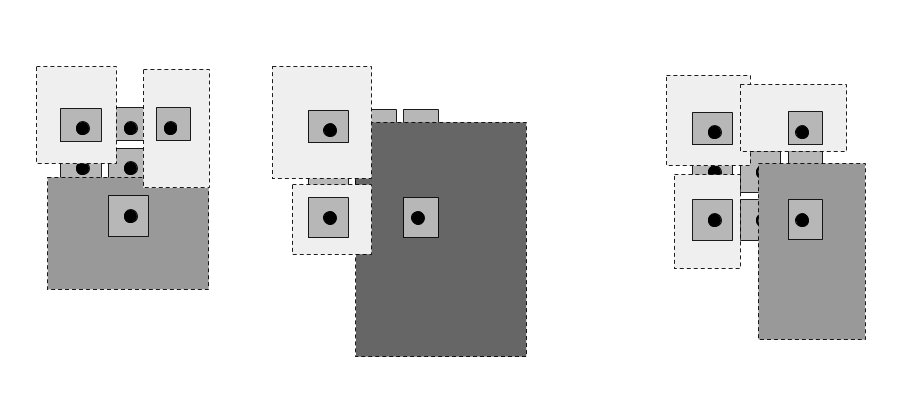}

\caption{Picking parts to condition on in $\TwoLevelConditioningVerts$. The top diagram depicts parts in $\mc P_{i+1}(\mc E)$ (thick outline) with their shortcutters. All of the thick-outlined parts are in $\mc R$ and the leftmost one is also in $\mc Q$. The thick arrows depict directed edges in the graph $\CG(\mc E, i+1, \mc Q, \mc R)$. The bottom diagram depicts parts in $\mc P_i(\mc E)$. The three thick-outlined parts from left to right contain the thin-outlined parts in $\mc Q_1',\mc Q_2'$, and $\mc Q_3'$ respectively. The bottom diagram depicts the shortcutters for some of those parts (the thin-outlined, light-gray dashed regions). The medium gray shortcutters determine $s_j$ for each $\mc Q_j'$. The darkest shortcutter motivates the choice of $j^* = 2$ in this example.}
\label{fig:two-level-conditioning-verts}
\end{figure}

We now state the properties of the output of $\TwoLevelConditioningVerts$. The ``All shortcutters with particular size'' property ensures that if all parts of $\mc Q''$ are conditioned on, the next $s_{j^*}$ will be smaller (thanks to the ``Progress'' input condition). The ``Larger shortcutters are far away'' condition ensures that all parts declared irrelevant can be carved and will remain carved regardless of what is chosen for conditioning in lower levels. The bound on $j_i'$ is used in the analysis of $\ConditioningVerts$ to show continued feasibility (i.e. that each $j_i\ge \sigma_1 + 2$).

\begin{proposition}\label{prop:two-level-holes}
$\TwoLevelConditioningVerts(\mc E,i,\mc Q_{i+1},\mc R_{i+1},j_{prev})$ takes in an empire $\mc E$, an index $i$, and $\mc Q_{i+1}$ and $\mc R_{i+1}$ for a carvable conditioning hierarchy $\mc{CH} = (\{\mc Q_k\}_{k=1}^{\sigma_0},\{\mc R_k\}_{k=1}^{\sigma_0})$ that is advanceable with respect to some distance indices $(j_k)_{k=1}^{\sigma_0}$ with $j_i = j_{prev}$. If

\begin{itemize}
\item $j_{prev}\ge \sigma_1$
\item $\mc Q_k = \emptyset$ for all $k\le i$
\end{itemize}

then $\TwoLevelConditioningVerts(\mc E,i,\mc Q_{i+1},\mc R_{i+1},j_{prev})$ returns $j_i'$, $s_{j_i'}$, $\mc Q_i$, and $\mc R_i$ that, when added to $\mc{CH}$, make it have the following properties:

\begin{itemize}
\item (Carvable) $\mc{CH}$ is a carvable conditioning hierarchy for $\mc E$.
\item (Advanceable) $\mc{CH}$ is advanceable for the distance indices $(j_k')_{k=1}^{\sigma_0}$, where $j_k' = j_k$ for all $k > i$ and

$$j_i - (s_{j_i} - s_{j_i'}) - 1 \le j_i' \le j_i$$
\item (Strong relevant set bound) $\mc R_i$ \textbf{is} the set of parts whose $\mc P_{i+1}(\mc E)$ parents are within distance $(j_i'+1)(4\sigma_1)^i$ of $\mc Q_{i+1}$ in $\CG(\mc E,i+1,\mc Q_{i+1},\mc R_{i+1})$.
\end{itemize}

\end{proposition}

\begin{proof}

\textbf{Well-definedness.} It suffices to show that the optimization problem on Line \ref{line:j-opt} is feasible. Since $\mc Q_j'\subseteq \mc Q_{j+1}'$ for all $j$, $s_j\le s_{j+1}$ for all $j$. Since $j_{prev} > \sigma_1 + 2$ and the $s_j$s are integers between 0 and $\sigma_1$ inclusive, there must exist a $j\le j_{prev}$ that satisfies condition (1) of the optimization problem on Line \ref{line:j-opt}. Therefore, the algorithm is well-defined.

\textbf{Carvable conditioning hierarchy.}

\underline{Horizontal containment.} Follows from the fact that $\mc Q_j'\subseteq \mc Q_{j+1}'$ for all $j$.

\underline{Vertical containment.} All of the parts in $\mc Q_{j^*+1}'$ have parents in $\mc R_{i+1}$.

\underline{Strong vertical closeness.} The algorithm only considers $j\le 2\sigma_1$.

\underline{Locally maximum.} This is the definition of $s_{j^*}$ and $\mc Q_{j^*}$.

\underline{Irrelevant parts are far away.} All parts in $\mc P_i(\mc E)\setminus \mc R_i$ have parents in $\CG(\mc E,i+1,\mc Q_{i+1},\mc R_{i+1})$ with distance greater than $(j^*+1)(4\sigma_1)^i$ from $\mc Q_{i+1}$, while all parts in $\mc Q_i$ have distance at most $j^*(4\sigma_1)^i$ from $\mc Q_{i+1}$. Therefore, by the triangle inequality, the distance from parents of $\mc Q_i$ to any parent of a part in $\mc P_i(\mc E)\setminus \mc R_i$ is at least $(4\sigma_1)^i$ as long as the parent is in $\mc R_{i+1}$. Otherwise, the desired bound follows from the fact that $\mc{CH}$ satisfied ``Irrelevant parts are far away'' on higher levels along with Proposition \ref{prop:pcg-vertical-monotonicity}.

\textbf{Advanceable.}

\underline{$j_i'$ bound.} In the output, $j_i'\gets j^*$, so we just need to bound $j^*$. $j^*\le j_{prev}$ by condition (2) for $j^*$. For all $j$, $s_{j+1}\ge s_j$. This combined with condition (1) shows that for all $j > j^*$, $s_{j+1}\ge s_j + 1$. Therefore, $s_{j_{prev}} - s_{j^*}\ge j_{prev} - j^* - 1$. Rearrangement gives the desired lower bound.

\underline{All local maxima.} By definition in the output, $\mc Q_{j_i'}$ contains all parts $P$ with $\DSize(\mc E, P)\ge m^{s_{j_i'}/\sigma_1}$. This is the desired property by Line \ref{line:si-def}.

\underline{Strong relevant set bound.} This is the definition of $\mc Q_{j_i'+1}'$.
\end{proof}

\subsection{$\ConditioningVerts$}\label{subsec:cond-verts}

Now, we implement $\ConditioningVerts$, which returns parts to condition on in order to satisfy Lemmas \ref{lem:holes} and \ref{lem:progress}. $\ConditioningVerts$ maintains a carvable conditioning hierarchy that is advanceable for distance indices $(j_i)_{i=1}^{\sigma_0}$. While maintaining such a hierarchy is relatively simple (just get rid of anything outside of the promised distances/sizes), many lower level $\mc Q_i$s will become empty. To make more progress, $\ConditioningVerts$ applies $\TwoLevelConditioningVerts$ to refill the lower $\mc Q_i$s with parts for smaller shortcutters than before.

\newpage

\begin{algorithm}[H]
\SetAlgoLined
\DontPrintSemicolon
\caption{$\ConditioningVerts(\mc E)$ initialization}

    \tcp{Initialization; only occurs on first $\ConditioningVerts$ call}

    $\mc Q_{\sigma_0+1}\gets \{V(G)\}$\;

    $\mc R_{\sigma_0+1}\gets \{V(G)\}$\;

    $s_{\sigma_0+1}\gets \sigma_1$\;

    $j_{\sigma_0+1}\gets 2\sigma_1$\;

    \For{$i=\{1,2,\hdots,\sigma_0\}$}{

        $\mc Q_i\gets \emptyset$\;

        $\mc R_i\gets \emptyset$\;

        $s_i\gets \sigma_1$\;

        $j_i\gets 2\sigma_1$\;

    }

    \tcp{End initialization}
\end{algorithm}

\begin{algorithm}[H]
\SetAlgoLined
\DontPrintSemicolon
\caption{$\ConditioningVerts(\mc E)$ each execution}
    \For{$i=\{\sigma_0,\sigma_0-1,\hdots,2,1\}$}{

        $\mc Q_i\gets $ refinement of $\mc Q_i$ by $\mc P_i(\mc E)$, with all parts with shortcutters of size either (a) below $m^{s_i/\sigma_1}$ or (b) parent distance to $\mc Q_{i+1}$ in $\CG(\mc E,i+1,\mc Q_{i+1},\mc R_{i+1})$ farther than $j_i (4\sigma_1)^i$ removed\; \label{line:remove-small}

        $\mc R_i\gets $ refinement of $\mc R_i$ by $\mc P_i(\mc E)$\;

    }

    \tcp{claim: always exists}

    $i^*\gets $ maximum $i$ for which $\mc Q_i = \emptyset$\;

    \uIf{$s_{i^*} = 0$}{

        \tcp{can condition on $i^*+1$ directly, as no $i^*$ shortcutters are active}

        $\Postprocess(\mc E)$\;

        \Return $\mc Q_{i^*+1}$\;

    }\Else{

        Set all $j_i$s to $2\sigma_1$ for $i < i^*$\; \label{line:low-clear}

        \For{$i=i^*,i^*-1,\hdots,1$}{

            $(j_i,s_i,\mc Q_i,\mc R_i)\gets \TwoLevelConditioningVerts(\mc E,i,\mc Q_{i+1},\mc R_{i+1},j_i)$\; \label{line:two-level-cond}

        }

        $\Postprocess(\mc E)$\;

        \Return $\mc Q_1$\;

    }

\end{algorithm}

\begin{algorithm}[H]
\SetAlgoLined
\DontPrintSemicolon
\caption{$\Postprocess(\mc E)$, which shares state with $\ConditioningVerts$}

    \For{$i=\{1,2,\hdots,\sigma_0\}$}{

        $\MakeNonedgesPermanent(\mc E, i, \mc Q_i, \mc R_i)$\;

    }
    
\end{algorithm}

\begin{figure}
\includegraphics[width=1.0\textwidth]{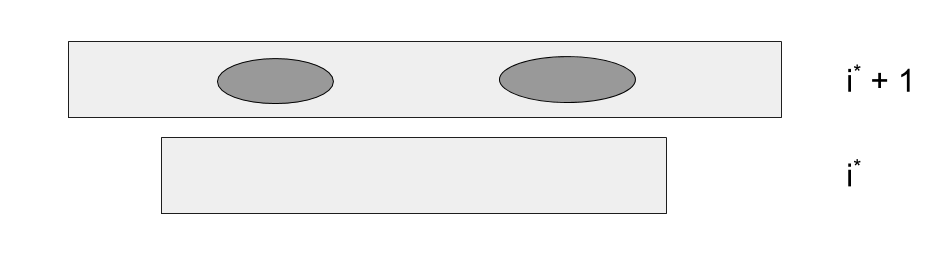}
\caption{In the above diagram, $\mc Q_{i^*}$ is empty. Since $s_{i^*} = 0$ and $\mc Q_{i^*}$ contains all of the parts in $\mc R_{i^*}$ with nearly-largest shortcutters, all parts in the relevant set have size 0 (inactive) shortcutters. Therefore, any shortcutter in level $i^*$ or below is either irrelevant or inactive. This means that carving is only required on levels $i^*+1$ and higher, which has already been achieved.}
\label{fig:two-level-irrelevance}
\end{figure}

We start by checking that this algorithm is well-defined. In doing so, we also describe the conditioning hierarchy throughout the algorithm:

\begin{proposition}\label{prop:cv-well-def}
$\ConditioningVerts$ is well-defined. More precisely, all of the input conditions for algorithms called are met.

Furthermore, $\mc{CH} = (\{\mc Q_i\}_{i=1}^{\sigma_0},\{\mc R_i\}_{i=1}^{\sigma_0})$ is a carvable conditioning hierarchy that is advanceable for the distance indices $(j_i)_{i=1}^{\sigma_0}$ between the end of the first for loop and the $\Postprocess$ call.
\end{proposition}

\begin{proof}
We inductively show that call $k$ to $\ConditioningVerts$ is well-defined. In particular, we show that $i^*$ always exists and that the input conditions to $\TwoLevelConditioningVerts$, $\MakeNonedgesPermanent$, and Proposition \ref{prop:carve-hierarchy} are satisfied. In this induction, we also show that the carvable and advanceable conditions are maintained.

\textbf{Base case.} During the first call to $\ConditioningVerts$, $i^* = \sigma_0$ because everything besides $Q_{\sigma_0+1}$ is initialized to the empty set. The input conditions to $\TwoLevelConditioningVerts$ for each $i$ are satisfied because $j_i = 2\sigma_1 \ge \sigma_1$ and the ``Carvable'' and ``Advanceable'' guarantees from the level $i+1$ application of $\TwoLevelConditioningVerts$. The ``Carvable'' guarantee ensures that the input conditions for both $\MakeNonedgesPermanent$ and Proposition \ref{prop:carve-hierarchy} are satisfied. This completes the proof that the first call to $\ConditioningVerts$ is well-defined and that $\mc{CH}$ is carvable and advanceable in the desired line range.

\textbf{Call $k$ for $k > 1$ $i^*$ existence.} Consider the set $\mc Q_{i_{prev}}$ returned during the $(k-1)$th call to $\ConditioningVerts$. By the ``Progress'' input condition to $\ConditioningVerts$, every part $P\in \mc Q_{i_{prev}}$ has the property that $E(P) = \emptyset$ at the beginning of the $k$th call to $\ConditioningVerts$. In particular, Line \ref{line:remove-small} eliminates $P$ from $\mc Q_{i_{prev}}$. In particular, $\mc Q_{i_{prev}} = \emptyset$ after Line \ref{line:remove-small}. In particular, the optimization problem defining $i^*$ is feasible, so $i^*$ exists.

\textbf{Call $k$ for $k > 1$ $\TwoLevelConditioningVerts$ $j_{prev}$ feasibility.} It suffices to show that $j_i\ge \sigma_1 + s_i$ after each execution of $\ConditioningVerts$, because $s_i\ge 0$ for all $i$. When $j_i$ is set to $2\sigma_1$, $j_i\ge \sigma_1 + s_i$ holds trivially. Only $\TwoLevelConditioningVerts$ changes $j_i$ to some other value. When it does this during the $l$th call to $\ConditioningVerts$ for $l < k$,

$$j_i^l \ge j_i^{l-1} - ((s_i^{l-1}-1) - s_i^l) - 1 \ge j_i^{l-1} - (s_i^{l-1} - s_i^l)$$

by Proposition \ref{prop:two-level-holes} inductively working in earlier $\ConditioningVerts$ calls. The superscript denotes the value after the $l$th call to $\ConditioningVerts$. The $(s_i^{l-1} - 1)$ bound holds because when $\mc Q_i$ becomes empty, $s_i$ decreases. By the inductive hypothesis,

$$j_i^{l-1}\ge \sigma_1 + s_i^{l-1}$$

Conbining these two inequalities shows that $j_i^l \ge \sigma_1 + s_i^l$. This completes the inductive step and shows that $j_i \ge \sigma_1 + s_i\ge \sigma_1$ before the $k$th call to $\ConditioningVerts$. Therefore, the $j_{prev}$ argument supplied to $\TwoLevelConditioningVerts$ is always at least $\sigma_1 + s_i^{l-1}\ge \sigma_1$, as desired.

\textbf{Call $k$ for $k > 1$ carvable.} Suppose that during call $k-1$ for $k > 1$ to $\ConditioningVerts$, $\mc{CH}$ is a carvable conditioning hierarchy just before $\Postprocess$. After $\Postprocess$ from call $k-1$, $\mc{CH}$ has not changed, as Propositions \ref{prop:make-nonedges-permanent} and \ref{prop:carve-hierarchy} imply that $\Postprocess$ does not modify $\mc{CH}$. Therefore, $\mc{CH}$ is still a locally maximum conditioning hierarchy at this point that satisfies the ``Irrelevant parts are far away'' condition. The conditions of Lemma \ref{lem:tcg-horizontal-monotonicity} are satisfied, so distances in each temporary conditioning digraph do not decrease between the end of call $k-1$ and the beginning of call $k$. Therefore, the ``Irrelevant parts are far away'' condition still applies. Line \ref{line:remove-small} maintains ``Horizontal containment,''``Vertical containment,'' and ``Locally maximum'' while restoring ``Strong vertical closeness.''

Therefore, $(\{\mc Q_i\}_{i=1}^{\sigma_0},\{\mc R_i\}_{i=1}^{\sigma_0})$ is a carvable conditioning hiearchy at the end of the first for loop. Since $\TwoLevelConditioningVerts$ maintains carvable conditioning hierarchies, $(\{\mc Q_i\}_{i=1}^{\sigma_0},\{\mc R_i\}_{i=1}^{\sigma_0})$ is a carvable conditioning hierarchy when $\Postprocess$ is called during the $k$th call to $\ConditioningVerts$. This completes the inductive step.

\textbf{Call $k$ for $k > 1$ advanceable.} Suppose that during call $k-1$ for $k > 1$ to $\ConditioningVerts$, $\mc{CH}$ is an advanceable conditioning hierarchy for indices $(j_i)_{i=1}^{\sigma_0}$. By Lemma \ref{lem:tcg-horizontal-monotonicity}, conditioning digraph distances only increase, which means that the ``Relevant set bound'' continues to at the beginning of $\ConditioningVerts$ call $k$. Line \ref{line:remove-small} restores the ``All local maxima'' property. These are both of the conditions for $\mc{CH}$ being advanceable, so $\mc{CH}$ is advanceable for $(j_i)_{i=1}^{\sigma_0}$ after the first for loop. $\TwoLevelConditioningVerts$ maintains advanceability. This completes the inductive step.

\textbf{Call $k$ for $k > 1$ $\TwoLevelConditioningVerts$.} We have already shown that $\mc{CH}$ is carvable and advanceable before the for loop containing Line \ref{line:two-level-cond}. Furthermore, we have already shown that $j_i\ge \sigma_1$. These are all of the conditions for $\TwoLevelConditioningVerts$.

\textbf{Call $k$ for $k > 1$ $\MakeNonedgesPermanent$.} $\mc{CH}$ is carvable when it is supplied to $\Postprocess$. This is a sufficient input condition for $\MakeNonedgesPermanent$.
\end{proof}

\subsection{Proof of Lemmas \ref{lem:holes} (conductivity not increased much) and \ref{lem:progress} (enough progress is made)}\label{subsec:holes-and-progress}

Now, we are finally ready to show that $\ConditioningVerts$ satisfies Lemmas \ref{lem:holes} and \ref{lem:progress}.

\lemholes*

\begin{proof}[Proof of Lemma \ref{lem:holes}]

\textbf{Conductivity and clan count.} By Proposition \ref{prop:cv-well-def}, the input to both $\MakeNonedgesPermanent$ is a carvable conditioning hierarchy. Therefore, the ``Conductivity'' guarantee of Proposition \ref{prop:make-nonedges-permanent} shows the desired conductivity increase. These propositions also show that the number of levels only increases by a factor of $O(\log n)$. No other parameters increase, as desired.

\textbf{Carving.} Proposition \ref{prop:carve-hierarchy} implies that all shortcutters in $\mc E'$ are carved with respect to $\mc Q_1$. When $s_{i^*} = 0$, all parts in $\mc R_{i^*}$ have inactive shortcutters, so they do not need to be carved.

Now, consider any part $P\in \mc P_k(\mc E)\setminus \mc R_k$ for some $k\le i^*$ with $E(P)\ne\emptyset$. Let $Q$ be the ancestor of $P$ in $\mc R_j$ for minimum $j$. Since $E(P)\ne\emptyset$ and $s_{i^*} = 0$, $j > i^*$, which means that $S_P$ is carved with respect to $\mc Q_{i^*+1}$ by Proposition \ref{prop:carve-hierarchy}. 

The shortcutters in $\cup_{\ell\ge i^*+1} \mc E_{\ell}$ are carved with respect to vertices in parts of $\mc Q_{i^*+1}$ by Proposition \ref{prop:carve-hierarchy}. Therefore, all shortcutters are carved with respect to $\mc Q_{i^*+1}$ if $s_{i^*} = 0$, as desired.
\end{proof}

The proof of Lemma \ref{lem:progress} shows that enough progress is made by showing that each conditioning round lexicographically decreases the word $s_{\sigma_0}s_{\sigma_0-1}\hdots s_2s_1$:

\lemprogress*

\begin{proof}[Proof of Lemma \ref{lem:progress}]

By the ``Relevant set bound,'' Lemma \ref{lem:tcg-horizontal-monotonicity}, and the ``Strong relevant set bound,'' $\mc R_i$ decreases (all parts replaced with subsets) during an interval of iterations in which $i^*\le i$. If $s_{i^*} > 0$, then $s_{i^*}$ strictly decreases by the ``All local maxima'' condition defining the $s_i$s. Furthermore, the $s_i$s for $i > i^*$ do not increase by the ``Containment'' input condition. In particular, the new word is lexicographically smaller than the old one.

Therefore, we just need to consider the case in which $s_{i^*} = 0$. In this case, the ``Progress'' input condition implies that $s_{i^*+1}$ decreases in the next call to $\ConditioningVerts$. By ``Containment,'' $s_i$ for $i > i^*+1$ does not increase. Therefore, the new word is again lexicographically smaller than the old one. These are all of the cases. The desired bound on the number of conditioning rounds follows from the fact that there are only $\sigma_1^{\sigma_0}$ possible words.
\end{proof}

\newpage

\section{Fixing reduction}\label{sec:fix-reduction}

In this section, we reduce Lemma \ref{lem:clan-fix} to a simpler result (Lemma \ref{lem:fast-fix}).

\begin{definition}[Schur complement conductance and degree]
Consider a graph $I$. Let $J := \texttt{Schur}(I,S\cup S')$, $J' := \texttt{Schur}(I/S,S\cup S')$,

$$c^I(S,S') := c^J(E_J(S,S'))$$

and

$$\Delta^I(S,S') := \sum_{e\in E_{J'}(S,S')} \texttt{Reff}_{J'}(e)c_e^{J'}$$
\end{definition}

\begin{lemma}\label{lem:fast-fix}
There is an algorithm $F'\gets \FastFix(I,J,D,S,S',F,\ep)$ that takes as input a graph $I$, a random graph $J$ that is a valid sample from the distribution $I[F]$, a set $D\subseteq E(I)$, $S,S'\subseteq V(G)$, $F\subseteq E(I)\setminus D$, and an accuracy parameter $\ep \in (0,1)$. With high probability on $J$, it outputs a set $F'\subseteq F$ with two properties:

\begin{itemize}
\item (Conductance) $c^{J\setminus D\setminus F'}(S,S') \le (1 + \ep) c^{I\setminus D}(S,S')$
\item (Size) $|F'|\le \mucon(\Delta^{I\setminus D}(S,S') + \Delta^{I\setminus D}(S',S) + |D|)\ep^{-3}$
\end{itemize}

Furthermore, $\FastFix$ takes $m^{1+o(1)}$ time.
\end{lemma}

When trying to understand the statement of Lemma \ref{lem:fast-fix} for the first time, it is helpful to think about the case when $D = \emptyset$. The $D = \emptyset$ case is very similar to the general case. It is also helpful to disregard the $J$ argument to $\FastFix$. This argument is only provided for runtime purposes, since sampling from a tree is only efficient in our case due to $\PartialSample$. Our reduction only uses $\ep = \Omega(1)$.

We prove this result in Section \ref{sec:slow-fix} and implement the almost-linear time algorithm for producing $F'$ in Section \ref{sec:fast-fix}. We illustrate many of the ideas behind the analysis of $\FastFix$ in Section \ref{sec:slow-warmup}.

\subsection{Reduction from Lemma \ref{lem:clan-fix} to Lemma \ref{lem:fast-fix}}

Let $\mc L_i$ be the set of clusters with $H$-effective resistance diameter at most $\R{i}\mucarve$ for which $S\subseteq \cup_{C\in \mc L_i} C$. These sets are guaranteed to exist by the ``Carved'' condition of Lemma \ref{lem:clan-fix}. Each of these sets of clusters maps to another set of clusters with $H_{\mc C}$-effective resistance diameter at most $\R{i}\mucarve$ for any clan $\mc C\in \mc E_i$. Call this set $\mc L_i^{\mc C}$. $\mc L_i^{\mc C}$ is not much larger than $\mc L_i$ by Lemma \ref{lem:ball-split}. We restate Lemma \ref{lem:ball-split} here to emphasize that one does not need to read Section \ref{sec:choose-parts} before reading this section. It is proven in the appendix:

\lemballsplit*

We can now focus on the clans $\mc C$ in $\mc E_i$ independently. Consider each shortcutter $S_C\in \mc C$. $S_C$ intersects some clusters in $\mc L_i^{\mc C}$. If clusters in $\mc L_i^{\mc C}$ contain a path from $C$ to $\partial S_C$ and that path is contracted after conditioning, the conductance of $S_C$ becomes infinite. We want to avoid this situation. To do this, we could try to apply Lemma \ref{lem:fast-fix} with $S\gets C$, $S'\gets \partial S_C$, $I\gets H$, $D\gets \texttt{deleted}(\mc C)$, $F\gets F$, and $\ep \gets 1/2$. While this ensures that the conductance of $S_C$ does not increase much after deleting $F'$, $\Delta^{H\setminus D}(S,S')$ could be very large. Our goal is for $F'$ to have average size $m^{o(1)}$, where the average is over the shortcutters in $\mc C$.

To achieve this, instead of directly tempering $S_C$'s conductance, it is useful to consider fixing the conductance between $C$ and all clusters in $\mc L_i^{\mc C}$ that intersect $S_C$. One slight issue with this is that some of these clusters may be very close to $C$, resulting in the initial $C-\mc L_i^{\mc C}$ conductance being much higher than $1/\R{i}$. This can be alleviated by fixing the conductance between $C$ and the subclusters in $\mc L_i^{\mc C}$ restricted to $S_C\setminus S_{\{C,V(H)\setminus S_C\}}(1/4,C)$. $S_{\{C,V(H)\setminus S_C\}}(1/4,C)$ serves as ``buffer space'' between $C$ and the restrictions of $\mc L_i^{\mc C}$, ensuring that the initial conductance is at most $\frac{m^{o(1)}}{\R{i}}$. The fact that $\mc L_i^{\mc C}$ consists of a relatively small set of clusters ensures that $\Delta^{H\setminus D}(C,\cup_{C'\in \mc L_i^{\mc C'}} C')$ is small. This ensures that there is a small set $F'$ that, when deleted, nearly reverses the increase in the $C-\mc L_i^{\mc C}$ conductance due to conditioning.

However, we need to reverse the increase in $S_C$'s conductance, not the $C-\mc L_i^{\mc C}$ conductance. $S_C$'s conductance, though, can essentially be upper bounded by the sum of the $C-\mc L_i^{\mc C}$ conductance and the conductance of edges in the Schur complement with respect to $C\cup \mc L_i^{\mc C}$ that go directly from the boundary of the buffer zone $S_{\{C,V(H)\setminus S_C\}}(1/4,C)$ and $\partial S_C$. The latter conductance is not affected by conditioning, because all edges of $F$ that are not in the buffer zone are covered by clusters in $\mc L_i^{\mc C}$. We have already argued that the former conductance is restored by deleting $F'$. Therefore, $S_C$'s conductance only increases by a small constant factor over what it used to be.

We formalize the above intuition by implementing the reduction in $\FixShortcutters$ using $\FastFix$.

\begin{algorithm}[H]
\DontPrintSemicolon
\caption{$\FixShortcutters(\mc E, H', \mc K)$}

    \ForEach{$i\in [\sigma_0]$}{

        \ForEach{clan $\mc C\in \mc E_i$}{

            $S'\gets \emptyset$\;

            $\mc K_{\mc C}\gets $ the parts in $\mc K$ that intersect some $S_C\in \mc C$\;\label{line:chosen-intersect}

            \ForEach{part $P\in \mc K_{\mc C}$}{

                $S'\gets S'\cup (P$ with all $S_{\{C,V(H)\setminus S_C\}}(1/4,C)$s for cores $C$ of shortcutters $S_C\in \mc C$ removed$)$\;

            }

            $\texttt{deleted}(\mc C)\gets \texttt{deleted}(\mc C)\cup \FastFix(H,H',\texttt{deleted}(\mc C),\cup_{S_C\in \mc C} C,S',F,1/4)$\;\label{line:aug-del}

        }

    }

    \Return $\mc E$\;
\end{algorithm}

\begin{figure}

\begin{center}
\includegraphics[width=0.6\textwidth]{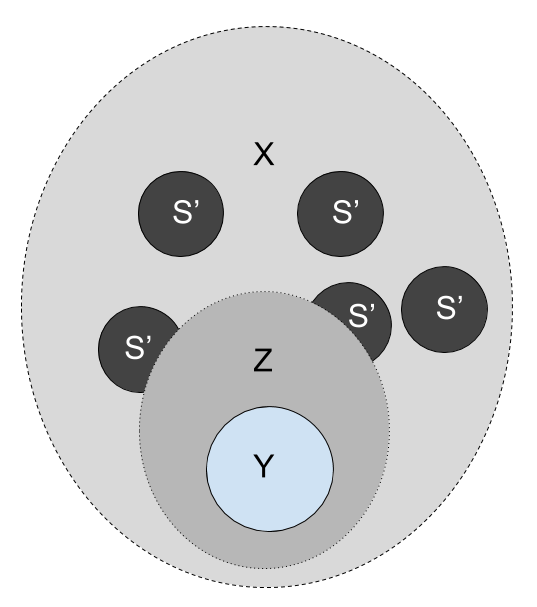}
\end{center}

\caption{The reduction for one shortcutter in some clan. All edges of $F$ are between two vertices of $S'$, two vertices of $Z$, or one vertex in $S'$ and one in $Z$. The edges between two vertices in $S'$ are irrelevant, as $Z$ is primarily responsible for certifying that $X$ has low conductance after conditioning. The edges between one vertex in $S'$ and one vertex in $Z$ do not affect the $(V(H)\setminus X)-Z$ direct conductance. The edges in $F$ with both edges in $Z$ can be in $F'$.}
\label{fig:fixing-reduction}
\end{figure}

Now, we prove that this algorithm has the desired effect. Before doing so, we state some useful facts about Schur complements and conductances that we prove in the appendix.

\begin{restatable}{lemma}{lemmonotoneschur}\label{lem:monotone-schur}
Consider any three disjoint sets of vertices $S_0,S_1,S_2\subseteq V(I)$ and $S_0'\subseteq S_0$. Let $J = \texttt{Schur}(I,S_0\cup S_1\cup S_2)$ and $J' = \texttt{Schur}(I,S_0'\cup S_1\cup S_2)$. Then

$$c^{J'}(E_{J'}(S_0',S_1)) \le c^J(E_J(S_0,S_1))$$
\end{restatable}

\begin{restatable}{lemma}{lemmonotonecond}\label{lem:monotone-cond}
Consider any two disjoint sets $S_0,S_1\subseteq V(I)$ with $S_0'\subseteq S_0$. Then $c^I(S_0',S_1)\le c^I(S_0,S_1)$.
\end{restatable}

\begin{restatable}{lemma}{lempcond}\label{lem:p-cond}
For any cluster $S_C$ in a graph $I$ and any $p\in (0,1)$,

$$c^I(C,V(I)\setminus S_{\{C,V(I)\setminus S_C\}}(p,C))\le \frac{c^I(C,V(I)\setminus S_C)}{p}$$
\end{restatable}

\begin{restatable}{lemma}{lemwellsep}\label{lem:well-sep}
Consider a graph $I$ with two clusters $C_1$ and $C_2$ with two properties:

\begin{itemize}
\item The $I$-effective resistance diameters of $C_1$ and $C_2$ are both at most $R$.
\item The minimum effective resistance between a vertex in $C_1$ and a vertex in $C_2$ is at least $\gamma R$ for $\gamma > 4$.
\end{itemize}

Let $J$ be the graph with $C_1$ and $C_2$ identified to $s$ and $t$ respectively. Then $\texttt{Reff}_J(s,t)\ge (\gamma-4)R$.
\end{restatable}

\begin{figure}[H]

\begin{center}
\includegraphics[width=0.6\textwidth]{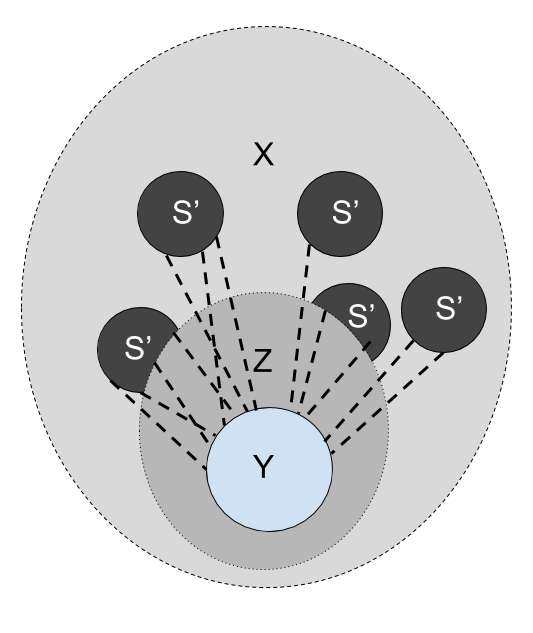}
\end{center}

\caption{The first step in bounding conductivity. To bound the $Y-S'$ conductance in $\texttt{Schur}(H_{\mc C},Y\cup S'\cup (V(H)\setminus X))$, it suffices to bound the $Y-(V(H)\setminus Z)$-conductance in $\texttt{Schur}(H_{\mc C},Y\cup (V(H)\setminus Z))$, which is bounded because $Z$ is 1/4th of the way from $Y$ to $\partial X$. More formally, we apply Lemma \ref{lem:p-cond} with $p = 1/4$.}

\label{fig:reduction-first-step}
\end{figure}

\begin{figure}[H]

\includegraphics[width=1.0\textwidth]{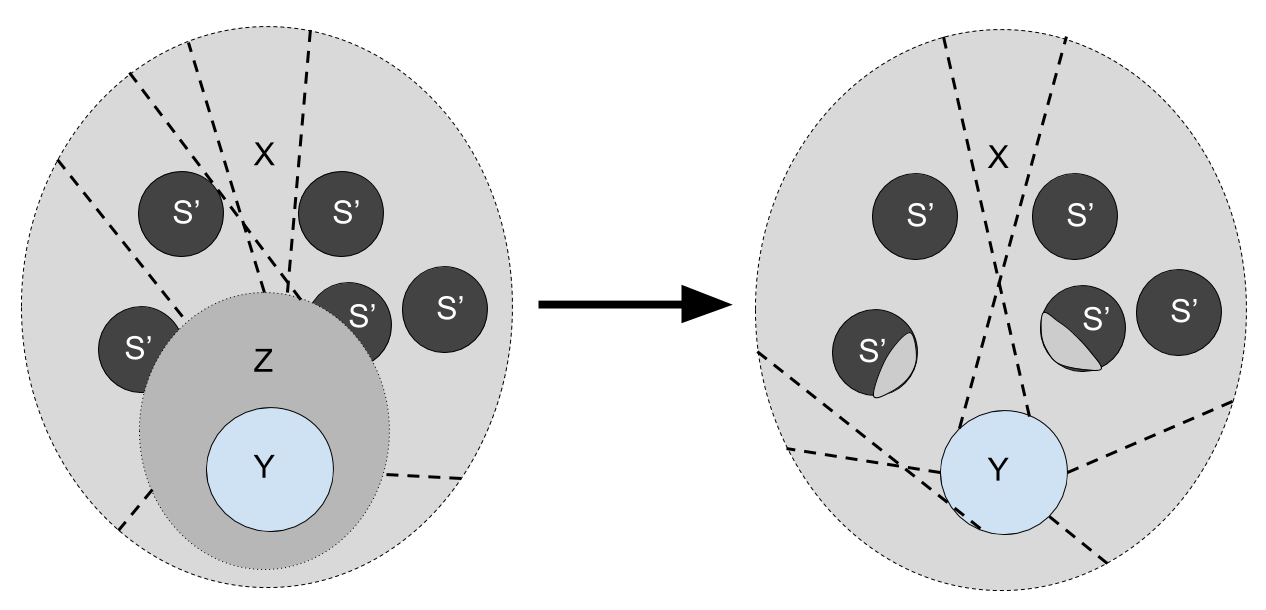}

\caption{The second step in bounding conductivity. To bound the direct $Y-V(H)\setminus X$ conductance, it suffices to bound the direct $Z-V(H)\setminus X$ conductance. This is bounded using Lemma \ref{lem:p-cond} on $X\setminus Z$ with $p = 3/4$.}

\label{fig:reduction-second-step}
\end{figure}

\begin{figure}[H]

\includegraphics[width=1.0\textwidth]{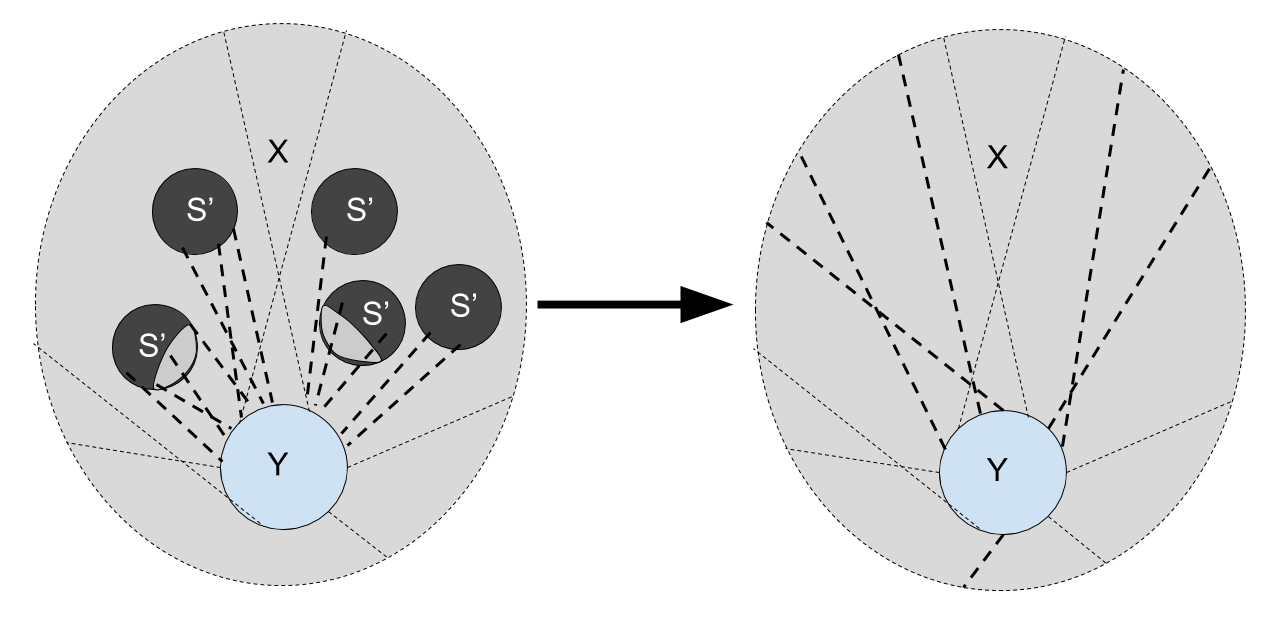}

\caption{The third step in bounding conductivity. The light dashed edges are the direct $Y-V(H)\setminus X$ edges, whose total conductance was bounded in the second step. The dark dashed edges obtained by eliminating $S'$ have smaller total conductance than the direct dark dashed edges to $S'$ by Lemma \ref{lem:monotone-schur}. The total conductance of these edges in $\texttt{Schur}(H_{\mc C}',Y\cup S'\cup V(H)\setminus X)$ is bounded by the first part and Lemma \ref{lem:fast-fix}. The end goal was to bound $Y-(V(H)\setminus X)$ conductance in the Schur complement $\texttt{Schur}(H_{\mc C}',Y\cup (V(H)\setminus X))$ obtained from the fixed graph $H_{\mc C}'$, so we are done.}

\label{fig:reduction-third-step}
\end{figure}

\lemclanfix*

\begin{proof}[Proof of Lemma \ref{lem:clan-fix} given Lemma \ref{lem:fast-fix}]

\textbf{Boundedness.} No part in $\mc K\cap (\cup_{k\le i} \mc P_k(\mc E))$ intersects the boundary of a core $C$ in some clan of $\mc E_i$, as $C$ is part of the refinement used to define $\mc P_k(\mc E)$. Therefore, each $H'$-boundary edge of the image $C'$ of $C$ in $H'$ is either a boundary edge of some part in $\mc K\cap (\cup_{k > i} \mc P_k(\mc E))$ or is a boundary edge of $C$. The total conductance of such edges is at most $\ell\kappa m/(\R{i})$, as desired.

\textbf{Conductivity.} Consider each clan $\mc C$ in isolation and consider the set $S'$ generated for the clan $\mc C$. Let $X = \cup_{S_{C'}\in \mc C} S_{C'}$ and $Y = \cup_{S_{C'}\in \mc C} C'$.

We start by bounding the $Y-S'$ conductance before conditioning. This is depicted in Figure \ref{fig:reduction-first-step}. By definition of $S'$ and the fact that the $S_C$s are vertex disjoint, $S'$ does not intersect $Z = S_{\{Y, V(H)\setminus X\}}(1/4,Y)$. Therefore, by Lemma \ref{lem:monotone-cond} with $S_0\gets S', S_0'\gets V(H)\setminus Z$, and $S_1\gets Y$,

$$c^{H_{\mc C}}(S',Y)\le c^{H_{\mc C}}(V(H)\setminus Z,Y)$$

By Lemma \ref{lem:p-cond} with $C\gets Y$, $S_C\gets X$ and $p\gets 1/4$ and the definition of $Z$,

$$c^{H_{\mc C}}(V(H)\setminus Z,Y) \le 4 c^{H_{\mc C}}(X,Y)$$

Since $\mc C$ was $\zeta$-conductive before conditioning,

$$c^{H_{\mc C}}(X,Y)\le \frac{\zeta m^{1/\sigma_1} s_{\mc C}}{\R{i}}$$

Combining these inequalities shows that

$$c^{H_{\mc C}}(S',Y)\le \frac{4\zeta m^{1/\sigma_1} s_{\mc C}}{\R{i}}$$

Let $H_{\mc C}'$ be the graph obtained by deleting $\texttt{deleted}(\mc C)$ from $H'$ after adding the $\FastFix$ edges. By the ``Conductance'' guarantee of Lemma \ref{lem:fast-fix},

$$c^{H_{\mc C}'}(S',Y)\le \frac{5\zeta m^{1/\sigma_1} s_{\mc C}}{\R{i}}$$

We have now finished bounding the $Y-S'$ conductance after conditioning. Now, we bound the $Y-(V(H)\setminus X)$ conductance after conditioning. This is depicted in Figure \ref{fig:reduction-second-step}. By definition, $Y\subseteq Z$. Notice that all edges in $F$ have endpoints in $S'\cup Z$. As a result, the direct $Z$ to $V(H)\setminus X$ conductance does not depend on $F$:

$$c^{\texttt{Schur}(H_{\mc C}',Z\cup S'\cup (V(H)\setminus X))}(E(Z,V(H)\cup X)) = c^{\texttt{Schur}(H_{\mc C},Z\cup S'\cup (V(H)\setminus X))}(E(Z,V(H)\cup X))$$

The graph subscript for $E$ is eliminated here for clarity, as the graph is the same as the Schur complement given in the superscript. Apply Lemma \ref{lem:monotone-schur} with $I\gets H_{\mc C}'$, $S_0\gets Z$, $S_0'\gets Y$, $S_1\gets V(H)\setminus X$, and $S_2\gets S'$ to conclude that

\begin{align*}
c^{\texttt{Schur}(H_{\mc C}',Y\cup S'\cup (V(H)\setminus X))}(E(Y,V(H)\cup X)) &\le c^{\texttt{Schur}(H_{\mc C}',Z\cup S'\cup (V(H)\setminus X))}(E(Z,V(H)\cup X))\\
&= c^{\texttt{Schur}(H_{\mc C},Z\cup S'\cup (V(H)\setminus X))}(E(Z,V(H)\cup X))\\
&\le \frac{4\zeta m^{1/\sigma_1} s_{\mc C}}{3\R{i}}\\
\end{align*}

where the last inequality follows from the fact that $X\setminus Z = S_{\{Y, V(H)\setminus X\}}(3/4,V(H)\setminus X)$, Lemma \ref{lem:p-cond}, and the $\zeta$-conductivity of $\mc C$ in the graph $H_{\mc C}$.

Now, we have bounds on the direct $Y-V(H)\setminus X$ and $Y-S'$ conductances in $H_{\mc C}'$. We now eliminate $S'$ to obtain the desired bound on the $Y-(V(H)\setminus X)$ conductance in $H_{\mc C}'$. This is depicted in Figure \ref{fig:reduction-third-step}. Start by applying Lemma \ref{lem:monotone-schur} with $I\gets H_{\mc C}'$, $S_0\gets (V(H)\setminus X)\cup S'$, $S_0'\gets V(H)\setminus X$, $S_1\gets Y$, and $S_2\gets \emptyset$ to show that

\begin{align*}
c^{\texttt{Schur}(H_{\mc C}',Y\cup (V(H)\setminus X))}(E(Y,V(H)\setminus X))&\le c^{\texttt{Schur}(H_{\mc C}',Y\cup S'\cup (V(H)\setminus X))}(E(Y,S'\cup (V(H)\setminus X)))\\
&\le c^{\texttt{Schur}(H_{\mc C}',Y\cup S'\cup (V(H)\setminus X))}(E(Y,S'))\\
&+ c^{\texttt{Schur}(H_{\mc C}',Y\cup S'\cup (V(H)\setminus X))}(E(Y,(V(H)\setminus X)))\\
&\le c^{H_{\mc C}'}(Y,S') + c^{\texttt{Schur}(H_{\mc C}',Y\cup S'\cup (V(H)\setminus X))}(E(Y,(V(H)\setminus X)))\\
&\le \frac{19\zeta m^{1/\sigma_1} s_{\mc C}}{3\R{i}}\\
\end{align*}

This is precisely saying that $\mc C$ is $\frac{19\zeta}{3}$-conductive after running $\FixShortcutters$, as desired.

\textbf{Modifiedness and deletion set condition.} We start with the deletion set condition. Consider any part $P\in \mc P_i(\mc E)$ for some $i\in [\sigma_0]$. By the laminarity of the overlay partition, either $P\subseteq Q$ for some $Q\in \mc K$ or $P$ does not intersect any part in $\mc K$. In the former case, conditioning on all parts in $\mc K$ makes $E_{H'}(P) = \emptyset$ and the deletion set condition does not need to apply for $P$. In the latter case, $F$ does not intersect $E_H(P)\cup \partial_H P$. Since $\FastFix$ adds a subset of $F$ to each $\texttt{deleted}(\mc C')$ for any clan $\mc C'$ in the empire $\mc E$, the deletion set condition remains satisfied for the part $P$, as desired.

Now, we move onto modifiedness. To apply Lemma \ref{lem:fast-fix}, we just need to bound $\Delta^{H_{\mc C}}(Y,S')$, $\Delta^{H_{\mc C}}(S',Y)$, and $|D|$. The reasoning for the first two bounds is very similar. By the modifiedness of the input, $|D|\le \tau m^{1/\sigma_1} s_{\mc C}$, so it suffices to bound the first two quantities.

We start with bounding $\Delta^{H_{\mc C}}(S',Y)$. Consider the graph $H_{\mc C}/S'$ with the vertices $S'$ identified to one vertex $s$. By Rayleigh monotonicity, cores of shortcutters in $\mc C$ have $H_{\mc C}/S'$ effective resistance diameter at most $\murad\R{i}\le \mucarve\muapp\R{i} := R$. Break the cores of shortcutters in $\mc C$ up into two sets $\mc C_{near}$ and $\mc C_{far}$, with these sets being the cores in $\mc C$ with and without respectively a vertex with $H_{\mc C}/S'$ effective resistance distance from $s$ at most $3R$. We bound the contribution to the degree of $\mc C_{near}$ and $\mc C_{far}$ independently.

First, we bound the $\mc C_{near}$ contribution. Every vertex in $\mc C_{near}$ clusters is within $H_{\mc C}/S'$ effective resistance distance $4R$ of $s$ by the triangle inequality. Recall that $Z$ separates the $\mc C_{near}$ clusters from $s$. Therefore, by Lemma \ref{lem:monotone-schur}, Lemma \ref{lem:p-cond} with $p = 1/4$, and $\zeta$-conductivity before conditioning,

$$c^{H_{\mc C}/S'}(s,\cup_{C\in \mc C_{near}} C)\le \frac{4\zeta m^{1/\sigma_1} s_{\mc C}}{\R{i}}$$

Let $K = \texttt{Schur}(H_{\mc C}/S',\{s\}\cup (\cup_{C\in \mc C} C))$. Then,

\begin{align*}
\sum_{e\in E_K(s,\cup_{C\in \mc C_{near}} C)} \texttt{Reff}_K(e) c_e^K &\le 4R c^K(E_K(s,\cup_{C\in \mc C_{near}} C))\\
&\le 16\mucarve \zeta m^{1/\sigma_1} s_{\mc C}\\
\end{align*}

Now, we bound the $\mc C_{far}$ contribution. To do this, we exploit Lemma \ref{lem:well-sep} on each cluster. In particular, for each cluster $C\in \mc C_{far}$ with a vertex at distance exactly $\gamma_C$ from $s$, Lemma \ref{lem:well-sep} implies that

$$c^K(s,C) \le \frac{1}{(\gamma_C-4)R}$$

By the triangle inequality, all vertices have $K$-effective resistance distance at most $(\gamma_C + 1)R$ from $s$. Therefore, they contribute at most $(\gamma_C+1)/(\gamma_C-4)\ge (5+1)/(5-4) = 6$ to the degree each. By bucketing, there are at most $4s_{\mc C}$ clusters in $\mc C$. Therefore,

$$\sum_{e\in E_K(s,\cup_{C\in \mc C_{far}} C)} \texttt{Reff}_K(e) c_e^K\le 24s_{\mc C}$$

We have now shown that

$$\Delta^{H_{\mc C}}(S',Y)\le (24 + 24\zeta \mucarve m^{1/\sigma_1}) s_{\mc C}$$

Next, we bound $\Delta^{H_{\mc C}}(Y,S')$. Since the shortcutters of $\mc C$ are carved with respect to $S'$ and all parts in $S'$ intersect some shortcutter of $\mc C$ (by Line \ref{line:chosen-intersect}), $S'$ can be written as a union of $|\mc C|$ clusters with $H$-effective resistance diameter at most $\mucarve\R{i}$ given by a $\mucarve\R{i}$ $H$-effective resistance ball around each core in $\mc C$. Let $\mc L$ be this set. By the bucketing of $\mc C$, $|\mc L|\le 4s_{\mc C}$. By Lemma \ref{lem:ball-split}, the clusters in $\mc L$ can be covered by a set $\mc L'$ with size $|\mc L'|\le \muapp(|\mc L| + |\texttt{deleted}(\mc C)|)\le \muapp (4 + \tau m^{1/\sigma_1})s_{\mc C}$ of clusters with $H_{\mc C}$ effective resistance diameter at most $\muapp\mucarve\R{i} \le R$.

Split $\mc L'$ into two sets $\mc L_{near}$ and $\mc L_{far}$, based on whether or not they contain a vertex within $H_{\mc C}/Y$ distance $3R$ of the identification $t$ of $Y$ in $H_{\mc C}/Y$. Applying the same argument as above with $s$ replaced with $t$ and the core set of $\mc C$ replaced with $\mc L'$ shows that

$$\Delta^{H_{\mc C}}(Y,S')\le 24 |\mc L_{far}| + 24 \muapp \mucarve (\zeta m^{1/\sigma_1} s_{\mc C})$$

As discussed after defining $\mc L'$,

$$|L_{far}|\le |\mc L'|\le \muapp (4 + \tau m^{1/\sigma_1})s_{\mc C}$$

Plugging this bound in shows that $\Delta^{H_{\mc C}/Y}(Y,S')\le 100\muapp (\mucarve \zeta + \tau) m^{1/\sigma_1} s_{\mc C}$. By Lemma \ref{lem:fast-fix}, $\mc C$ is $200\muapp\mucarve\mucon(\zeta + \tau) = \mumod(\zeta + \tau)$-modified after applying $\FixShortcutters$, as desired.

\textbf{Runtime.} $\FixShortcutters$ does a linear amount of work and one call to $\FastFix$ for each clan. Since there are at most $\ellmax \le m^{o(1)}$ clans, the total runtime is almost-linear.

\end{proof}

\newpage

\section{Conductance concentration inequality (fixing lemma) preliminaries}\label{sec:fix-preliminaries}

Now, we give preliminaries for Lemmas \ref{lem:slow-fix} and \ref{lem:fast-fix}.

\subsection{Linear algebraic descriptions of conductances and rank one updates}

For the rest of the paper, it will be helpful to think about the quantities $c^I(S,S')$ and $\Delta^I(S,S')$ linear-algebraically. We start with $c^I(S,S')$:

\begin{proposition}\label{prop:lin-alg-total-cond}
$$c^I(S,S') = \frac{1}{b_{ss'}^T L_{I/(S,S')}^+ b_{ss'}}$$
\end{proposition}

\begin{proof}
Since all edges accounted for in $c^I(S,S')$ go directly between $S$ and $S'$ in $\texttt{Schur}(I,S\cup S')$, $c^{I/(S,S')}(s,s') = c^I(S,S')$. $c^{I/(S,S')}(s,s')$ is the conductance of the one edge between $s$ and $s'$ in the graph $\texttt{Schur}(I/(S,S'),\{s,s'\})$. The conductance of this edge is the reciprocal of its resistance. By commutativity of Schur complements with identification, this edge's resistance is $b_{ss'}^T L_{I/(S,S')}^+ b_{ss'}$, as desired.
\end{proof}

Now, we interpret $\Delta^I(S,S')$ linear-algebraically. As discussed in Section \ref{sec:preliminaries}, the effective resistance between vertices $x$ and $y$ in a graph $I$ is just $b_{xy}^T L_I^+ b_{xy}$. Therefore, we just need to interpret condunctances between a vertex $s$ and a set of vertices $S'$. Proposition \ref{prop:lin-alg-total-cond} gives us a way of thinking about the total conductance of edges between $s$ and $S'$, so it suffices to describe the normalized conductances. It turns out that these conductances can be computed using one electrical flow:

\begin{proposition}\label{prop:lin-alg-norm-cond}
Consider a graph $I$ with a vertex $s$ and a set $S'\subseteq V(I)$. Let $J\gets \texttt{Schur}(I,\{s\}\cup S')$. Let $K = I/S'$ with $S'$ identified to a vertex $s'$. For each vertex $w\in S'$,

$$\frac{c_{sw}^J}{c^J(s,S')} = \sum_{e\in\partial_I w} (b_{ss'}^T L_K^+ b_e)c_e^K$$ 
\end{proposition}

\begin{proof}
By Theorem \ref{thm:edge-visits},

$$\sum_{e\in \partial_I w} (b_{ss'}^T L_K^+ b_e)c_e^K = \Pr_s [t_w < t_{S'\setminus \{w\}}]$$

The random walk in the above expression is done in the graph $I$. By Theorem \ref{thm:schur-walk-equiv},

$$\Pr_s [t_w < t_{Y\setminus \{w\}}] = \frac{c_{sw}^J}{c^J(s,S')}$$

Combining these equalities yields the desired result.
\end{proof}

We now characterize how conductances and effective resistances change after deleting single edges.

\begin{definition}[Leverage and nonleverage scores]
For any graph $I$ and $e\in E(I)$, let

$$\texttt{lev}_I(e) := c_e^I\texttt{Reff}_I(e)$$

and

$$\texttt{nonlev}_I(e) := 1 - c_e^I\texttt{Reff}_I(e)$$

\end{definition}

The following bounds follow immediately from the Sherman-Morrison rank one update formula (\cite{SM50}).

\begin{proposition}\label{prop:deletion-update}
Let $I$ be a graph. Consider two vertices $s,t\in V(I)$, an edge $f$, and a demand vector $d\in \mathbb{R}^{V(I)}$. Then

$$b_{st}^T L_{I\backslash f}^+ d = b_{st}^T L_I^+ d + \frac{(b_{st}^T L_I^+ b_f)(b_f^T L_I^+ d)}{r_f(1 - \texttt{lev}_I(f))}$$
\end{proposition}

\begin{proposition}\label{prop:contraction-update}
Let $I$ be a graph. Consider two vertices $s,t\in V(I)$, an edge $f$, and a demand vector $d\in \mathbb{R}^{V(I)}$. Then

$$b_{st}^T L_{I/f}^+ d = b_{st}^T L_I^+ d - \frac{(b_{st}^T L_I^+ b_f)(b_f^T L_I^+ d)}{r_f\texttt{lev}_I(f)}$$
\end{proposition}

\subsection{Splitting edges}

In Propositions \ref{prop:deletion-update} and \ref{prop:contraction-update}, the dependencies on $\texttt{nonlev}_I(f)$ and $\texttt{lev}_I(f)$ are inconvenient. One can mitigate this by \emph{splitting} edges in one of two ways. When an edge has low leverage score, it should be split \emph{in series} by replacing it with two edges that have half the resistance. When an edge has high leverage score, it should be split \emph{in parallel} by replacing it with two edges that have double the resistance. For an edge $e\in E(I)$, define $(J,\{e_1,e_2\})\gets \Split(I,e)$ to be the routine that does the following:

\begin{itemize}
\item If $\texttt{lev}_I(e)\le 1/2$, let $J$ be the graph with $e$ replaced by a path of two edges $e_1$ and $e_2$, each with conductance $2c_e$.
\item If $\texttt{lev}_I(e) > 1/2$, let $J$ be the graph with $e$ replaced by two parallel edges $e_1$ and $e_2$, each with conductance $c_e/2$.
\end{itemize}

Let $(I',F')\gets \Split(I,F)$ be the graph-set pair that results from splitting all edges in $F$. For any $F\subseteq E(I)$, the distribution over graphs $I'\sim I[[F]]$ is obtained by splitting all edges in $F$ and conditioning on an arbitrary copy for each edge. The purpose of doing this is as follows:

\begin{proposition}\label{prop:split-cond}
For an arbitrary subset $F\subseteq E(I)$, let $(I',F')\gets \Split(I,F)$. For all copies $f_i$, $i\in \{1,2\}$ of edges $f\in F$ in $I'$,

$$\texttt{lev}_{I'}(f_i)\in [1/4,3/4]$$

\end{proposition}

\begin{proof}
Splitting an edge $e\ne f\in F$ does not change the leverage score of $f$, so it suffices to show the desired proposition when $F = \{f\}$. When $\texttt{lev}_I(f)\ge \frac{1}{2}$, splitting it into two parallel copies does not change its effective resistance but doubles its resistance. Therefore, $\frac{1}{4}\le \texttt{lev}_{I'}(f_i)\le \frac{1}{2}$. When $\texttt{lev}_I(e)\le \frac{1}{2}$, subdividing $f$ to a copy $f_i$ makes

$$\texttt{lev}_{I'}(f_i) = \frac{1}{2}\texttt{lev}_I(f) + \frac{1}{2}$$

Since $0\le \texttt{lev}_I(f)\le \frac{1}{2}$, the desired inequality follows.
\end{proof}

Furthermore, doing this does not change the random spanning tree distribution:

\begin{proposition}\label{prop:split-sample}
Sampling a random spanning tree $T\sim J$ for $(J,\{e_1,e_2\})\gets \Split(I,e)$ is equivalent to doing the following:

\begin{itemize}
\item Sample a random tree $T\sim I$.
\item If $\texttt{lev}_I(e)\le 1/2$, do the following:
    \begin{itemize}
    \item If $e\in T$, return a tree $T'$ obtained by replacing $e$ with $e_1$ and $e_2$.
    \item If $e\notin T$, return a tree $T'$ obtained by adding $e_1$ with probability $1/2$ and $e_2$ otherwise.
    \end{itemize}
\item If $\texttt{lev}_I(e)\ge 1/2$, do the following:
    \begin{itemize}
    \item If $e\in T$, return a tree $T'$ obtained by replacing $e$ with $e_1$ with probability $1/2$ and $e_2$ otherwise.
    \item If $e\notin T$, return $T'\gets T$.
    \end{itemize}
\end{itemize}
\end{proposition}

\begin{proof}
Consider each of the two cases separately:

\textbf{Leverage score below 1/2.} In this case, we need to show that (a) splitting $e$ in series into edges $e_1$ and $e_2$ and sampling a tree from the resulting graph $J$ is equivalent to (b) sampling a tree from $I$, splitting $e$ in series if $e$ is in the tree, and adding one of $e_1,e_2$ to the tree otherwise with probability $1/2$ for each. It suffices to show that each spanning tree $T$ of $J$ has the same probability of being generated through Procedures (a) and (b). Suppose first that $e_1,e_2\in E(T)$. $T$ is generated with probability proportional to $c_{e_1}^Jc_{e_2}^J\prod_{f\ne e_1,e_2\in E(T)} c_f^J = 4(c_e^I)^2\prod_{f\ne e_1,e_2\in E(T)} c_f^I$ using Procedure (a) and with probability proportional to $c_e^I\prod_{f\ne e_1,e_2\in E(T)} c_f^I$ using Procedure (b). Now, suppose that one of $e_1,e_2$ is not in $T$ and that $e_1\in E(T)$ without loss of generality. Procedure (a) generates $T$ with probability proportional to $(c_{e_1}^J)\prod_{f\ne e_1,e_2\in E(T)} c_f^J = 2c_e^I\prod_{f\ne e_1,e_2\in E(T)} c_f^I$. Procedure (b) generates $T$ with probability proportional to $(1/2)\prod_{f\ne e_1,e_2\in E(T)} c_f^I$. In both cases, Procedure (a) generates trees with weight $4c_e^I$ times the weight that Procedure (b) uses. Since the weights used for the Procedure (a) and (b) are proportional, Procedures (a) and (b) generate trees from the same distribution.

\textbf{Leverage score above 1/2.} We need to show that (c) splitting $e$ in parallel into edges $e_1$ and $e_2$ and sampling from the resulting graph $J$ is equivalent to (d) sampling a tree from $I$ and replacing $e$ with $e_1$ or $e_2$ with probability $1/2$ for each if $e$ is in the tree. First, consider the case in which one of $e_1,e_2$ is in $T$. Without loss of generality, suppose that $e_1\in E(T)$. $T$ comes from Procedure (c) with probability proportional to $c_{e_1}^J\prod_{f\ne e_1,e_2\in E(T)} c_f^J = (c_e^I/2)\prod_{f\ne e_1,e_2\in E(T)} c_f^I$. $T$ comes from Procedure (d) with probability proportional to $(1/2) c_e^I\prod_{f\ne e_1,e_2\in E(T)} c_f^I$. Now, suppose that $e_1,e_2\notin E(T)$. In both procedures, $T$ is generated with probability proportional to $\prod_{f\in E(T)} c_f^I$. The weights for Procedures (c) and (d) are equal, so Procedures (c) and (d) generate trees from the same distribution, as desired.
\end{proof}

We use the above proposition to show that computing a graph from the distribution $H[F]$ (the second algorithm below) is equivalent to computing a graph using a more incremental strategy that is particularly amenable to analysis using rank 1 updates (the first algorithm below):

\begin{proposition}\label{prop:precomp-partial}
Consider a graph $H$ and a set $F\subseteq E(H)$. Consider any algorithm of the following form:

\begin{itemize}
\item Initialize $H_0\gets H$ and $k \gets 0$
\item While $F$ is not empty
    \begin{itemize}
    \item Pick a random edge $f_k\sim \mc D_k$, where $\mc D_k$ is an arbitrary distribution over $F$ that only depends on the contractions/deletions of the previous edges $f_0,f_1,\hdots,f_{k-1}$
    \item Let $H_{k+1}\sim H_k[[f_k]]$
    \item Remove $f_k$ from $F$ if a self-loop/leaf is created; otherwise replace $f_k$ in $F$ with the remaining copy
    \item Increment $k$
    \end{itemize}
\end{itemize}

Then $H_k$ is equivalent in distribution to $H_k'$ in the following algorithm for all $k\ge 0$, which only requires sampling the intersection of $T_0$ with $F$, not all of $T_0$:

\begin{itemize}
\item Sample $T_0\sim H$, $H_0'\gets H$, and set $k \gets 0$
\item While $F$ is not empty
    \begin{itemize}
    \item Pick a random edge $f_k\sim \mc D_k$ and let $(J,\{f^{(0)},f^{(1)}\})\gets \Split(H_k',f_k)$.
    \item Let $T_{k+1}$ be the random spanning tree of $J$ obtained by setting $T\gets T_k$ and $e\gets f_k$ in Proposition \ref{prop:split-sample}
    \item Let $H_{k+1}'$ be the subgraph of $J$ with $f^{(0)}$ contracted if $f^{(0)}\in E(T_{k+1})$ and $f^{(0)}$ deleted otherwise
    \item Increment $k$
    \end{itemize}
\end{itemize}

\end{proposition}

\begin{proof}
We prove this fact by induction on $k$. For $k = 0$, clearly $H_0 = H_0'$. $H_k$ is obtained from $H_{k-1}$ by sampling the intersection of a random spanning tree in the graph $J$ obtained by doing $(J,\{f^{(0)},f^{(1)}\})\gets \Split(H_{k-1},f_{k-1})$. By Proposition \ref{prop:split-sample}, this is equivalent to sampling a random tree in $H_{k-1}$ and splitting it randomly as described in Proposition \ref{prop:split-sample}.

By the inductive hypothesis, $H_{k-1}$ and $H_{k-1}'$ have the same distribution. Furthermore, for $k > 0$, notice that $f_{k-1}$ is obtained from the same distribution $\mc D_{k-1}$. These replacements, when coupled with the previous paragraph, yield the second algorithm. Therefore, $H_k$ and $H_k'$ are equivalent in distribution, completing the induction step.
\end{proof}

\begin{proposition}\label{prop:split-partial}
Consider a graph $H$ and a set $F\subseteq E(H)$. Sampling $H'\sim H[F]$ is equivalent to any strategy of the following form:

\begin{itemize}
\item Initialize $H_0\gets H$ and $k \gets 0$
\item While $F$ is not empty
    \begin{itemize}
    \item Pick a random edge $f_k\sim \mc D_k$, where $\mc D_k$ is an arbitrary distribution over $F$ that only depends on the contractions/deletions of the previous edges $f_0,f_1,\hdots,f_{k-1}$
    \item Let $H_{k+1}\sim H_k[[f_k]]$
    \item Remove $f_k$ from $F$ if a self-loop/leaf is created; otherwise replace $f_k$ in $F$ with the remaining copy
    \item Increment $k$
    \end{itemize}
\end{itemize}

\end{proposition}

\begin{proof}
By Proposition \ref{prop:precomp-partial}, this algorithm is the same as the second algorithm in Proposition \ref{prop:precomp-partial}. The final output of the second algorithm given in Proposition \ref{prop:precomp-partial} is equivalent to sampling from $H[F]$ in distribution. Therefore, the algorithm given in Proposition \ref{prop:split-partial} is equivalent to sampling from $H[F]$ in distribution, as desired.
\end{proof}

\subsection{Concentration inequalities}

The following is relevant for the slow version of Lemma \ref{lem:fast-fix}:

\begin{theorem}[Theorem 16 of \cite{CL06}]\label{thm:azuma}
Let $X$ be the martingale satisfying $|X_{i+1} - X_i|\le c_i$ for all $i$ specifying martingale increments. Then

$$\Pr[|X_n - \textbf{E}[X_n]|\ge \lambda]\le 2e^{-\frac{\lambda^2}{2\sum_{i=1}^n c_i^2}}$$

where $n$ is the total number of martingale increments.
\end{theorem}

To speed up the algorithm and prove Lemma \ref{lem:fast-fix}, we exploit the following additional concentration inequality:

\begin{theorem}[Theorem 18 in \cite{CL06}]\label{thm:martingale-2}
Let $X$ be a martingale satisfying:

\begin{itemize}
\item $\text{Var}(X_i|X_{i-1}) \le \sigma_i^2$ for all $i$.
\item $|X_i - X_{i-1}| \le M$ for all $i$.
\end{itemize}

Then

$$\Pr[|X_n - \textbf{E}[X_n]| \ge \lambda] \le e^{-\frac{\lambda^2}{2(\sum_{i=1}^n \sigma_i^2 + M\lambda / 3)}}$$
\end{theorem}

\subsection{\cite{HHNRR08} bound for source-sharing demand vectors}

We use the following bound to bound changes in normalized potentials of vertices due to conditioning:

\begin{theorem}[Lemma 3.1 of \cite{HHNRR08}, with $n$ replaced by $\tau$]\label{thm:harsha}
Let $G$ be a graph and consider three vertices $s_1,s_2,t\in V(G)$. Then for any positive integer $\tau$

$$\sum_{f\in E(G)} \frac{|b_{s_1t}^T L_G^+ b_f||b_f^T L_G^+ b_{s_2t}|}{r_f}\le (8\log \tau) |b_{s_1t}^T L_G^+ b_{s_2t}| + \frac{1}{\tau} (b_{s_1t}^T L_G^+ b_{s_1t} + b_{s_2t}^T L_G^+ b_{s_2t})$$
\end{theorem}

\subsection{Matrix sketching}

The $\FastFix$ algorithm needs to compute $\ell_p$ norms of vectors whose entries would naively require as many as $\Theta(m)$ Laplacian solves to compute. We show that these vector norms can be approximated with $O(\text{polylog}(n))$ Laplacian solves using matrix sketching:

\begin{theorem}[\cite{I06}, Theorem 3 stated for $\ell_p$ rather than just $\ell_1$]\label{thm:linear-sketch}
An efficiently computable, $\text{polylog}(d)$-space linear sketch exists for all $\ell_p$ norms exists for all $p\in (0,2]$. That is, given a $d\in \mathbb{Z}_{\ge 1}$, $\delta\in (0,1)$, $p\in (0,2]$, and $\epsilon\in (0,1)$, there is a matrix $C = \SketchMatrix(d,\delta,p,\epsilon)\in \mathbb{R}^{l \times d}$ and an algorithm $\RecoverNorm(s,d,\delta,p,\epsilon)$ with the following properties:

\begin{itemize}
\item (Approximation) For any vector $v\in \mathbb{R}^d$, with probability at least $1 - \delta$ over the randomness of $\SketchMatrix$, the value $r = \RecoverNorm(Cv,d,\delta,p,\epsilon)$ is as follows:

$$(1 - \epsilon) ||v||_p\le r\le (1 + \epsilon) ||v||_p$$

\item $l = c/\epsilon^2 \log(1/\delta)$ for some constant $c > 1$

\item (Runtime) $\SketchMatrix$ and $\RecoverNorm$ take $O(\ell d)$ and $\text{poly}(\ell)$ time respectively.
\end{itemize}
\end{theorem}

\subsection{Localization bound}

The $\FastFix$ algorithm needs to be able to reuse low-contribution edges for many iterations. We exploit the following bound to do this:

\begin{theorem}[Theorem 1 of \cite{SRS17}]\label{thm:localization}
Let $G$ be a graph and $F\subseteq E(G)$. Then for any vector $w\in \mathbb{R}^{E(G)}_{\ge 0}$,

$$\sum_{e_1,e_2} w_e w_f\frac{|b_e^T L_G^+ b_f|}{\sqrt{r_e}\sqrt{r_f}}\le O(\log^2 n) ||w||_2^2$$
\end{theorem}

\newpage

\section{Conductance concentration inequality}\label{sec:slow-fix}

In this section, we prove the following inefficient analogue of Lemma \ref{lem:fast-fix}:

\begin{lemma}\label{lem:slow-fix}
Given a graph $I$, a random graph $J$ that is a valid sample from the distribution $I[F]$, a set $D\subseteq E(I)$, $S,S'\subseteq V(G)$, $F\subseteq E(I)\setminus D$, and an accuracy parameter $\ep \in (0,1)$. With high probability on $J$, there is a set $F'\subseteq F$ with two properties:

\begin{itemize}
\item (Conductance) $c^{J\setminus D\setminus F'}(S,S') \le (1 + \ep) c^{I\setminus D}(S,S')$
\item (Size) $|F'|\le \tilde{O}(\ep^{-3}(\Delta^{I\setminus D}(S,S') + \Delta^{I\setminus D}(S',S) + |D|))$
\end{itemize}
\end{lemma}

\begin{figure}
\includegraphics[width=1.0\textwidth]{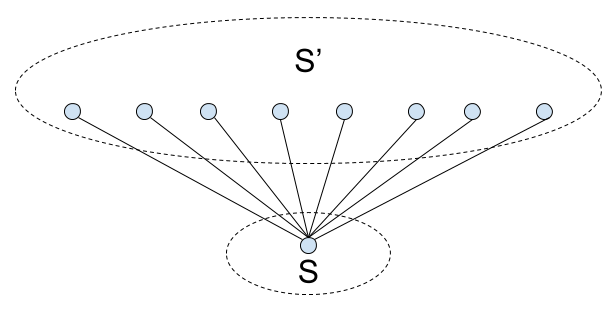}
\caption{Why Lemma \ref{lem:slow-fix} has a dependence on $\Delta^I(S,S')$. In this diagram, let $F$ be the set of all edges and $D = \emptyset$. Conditioning on a random spanning tree in $I$ contracts all of the edges in $F$. In order to make the $S-S'$ conductance finite, all of the edges in $F$ must be deleted. There are $\Delta^I(S,S')$ of these edges, no matter what the conductances of these edges are.}
\label{fig:degree-dependence}
\end{figure}

The proof conditions on one edge in $F$ at a time. In principle, doing this could drastically change the conductance between $S$ and $S'$. Luckily, it does not change much in expectation if one conditions on a random edge in $F$. While the conductance does not change much in expectation, it could change in absolute amount. In Proposition \ref{prop:stable-concentration}, we show that there always exists an edge an edge that, if conditioned on, does not change the $S-S'$ conductance much. Using these statements in concert with Azuma's Inequality shows that the $S-S'$ conductance does not change by more than a $(1 + \ep)$ factor with high probability if one conditions on all but a very small number of edges in $F$. Letting these edges be $F'$ proves the desired result.

To prove Lemma \ref{lem:slow-fix}, we will need to use Azuma's Inequality to control a number of other quantities besides the $S-S'$ conductance to show that one can always find the edges described in Proposition \ref{prop:stable-concentration}. As a result, we choose to discuss a much simpler result first, which may be of independent interest.

\subsection{Warmup: Controlling the effective resistance between two vertices}\label{sec:slow-warmup}

In this subsection, we prove Lemma \ref{lem:slow-fix} with $S = \{s\}$ and $S' = \{t\}$ each being one vertex and $D$ being empty. This is also Lemma \ref{lem:special-slow-fix}, which we restate here with notation that agrees with Lemma \ref{lem:slow-fix}:

\begin{lemma}[Restatement of Lemma \ref{lem:special-slow-fix}]\label{lem:restate-special-slow-fix}
Let $I$ be a graph, $F\subseteq E(I)$, $\ep\in (0,1)$, and $s,t\in V(I)$. Sample a graph $J\sim I[F]$. Then, with high probability, there is a set $F'\gets \SpecialFix(I,s,t,F,\ep)\subseteq F$ that satisfies both of the following guarantees:

\begin{itemize}
\item (Effective resistance) $b_{st}^T L_{J\setminus F'}^+ b_{st}\ge (1 - \ep)(b_{st}^T L_I^+ b_{st})$
\item (Size) $|F'|\le O((\log n)/\ep^2)$
\end{itemize}

\end{lemma}

While this special case is not directly relevant to the proof of Lemma \ref{lem:slow-fix}, it offers a simpler way of demonstrating most of the key ideas.

$\SpecialFix$ splits and conditions on the edge with lowest energy until there are $O((\log n)/\ep^2)$ edges left. At this point, the algorithm returns the remaining edges in $F$ to be in $F'$.

\begin{algorithm}[H]
\DontPrintSemicolon
\caption{$\SpecialFix(I,s,t,F,\ep)$, never executed}

    \KwData{the graph $I$, $s,t\in V(I)$, arbitrary $F\subseteq E(I)$, $\ep\in (0,1)$}

    \KwOut{the subset $F'\subseteq F$ of edges that should be deleted}

    $J\gets I$\;

    \While{$F$ has more than $9(\log n)/\ep^2$ non-leaf/loop edges}{

        $f\gets $ the non-leaf/loop edge in $E(J)$ that minimizes $(b_{st}^T L_J^+ b_f)^2c_f$\;

        $(J,\{f_1,f_2\})\gets \Split(J,f)$\;

        $J\sim J[f_1]$\;

        $F\gets (F\setminus \{f\})\cup \{f_2\}$\;

    }

    \Return{ the non-leaf/loop edges of $F$}
\end{algorithm}

\begin{figure}
\includegraphics[width=1.0\textwidth]{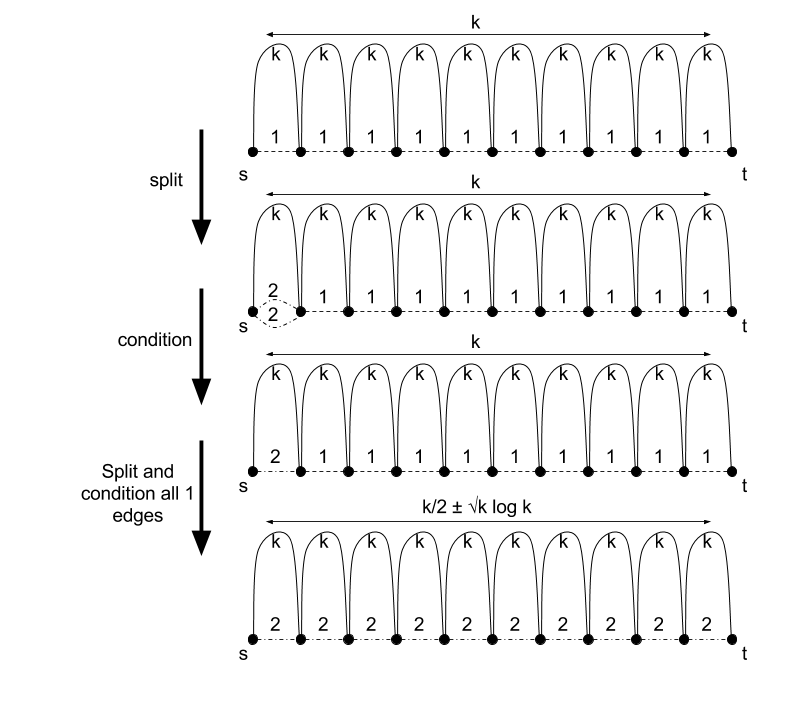}
\caption{The $\SpecialFix$ algorithm for $|F|/2$ iterations. In the above example, each of the dashed edges is in $F$. During each iteration, one edge is chosen, split, and conditioned on. During the first conditioning round in this example, the chosen edge is deleted. In the last diagram, all contracted edges yield self-loops with resistance $k$, which are irrelevant for the $s-t$ resistance. With high probability, roughly $k/2$ edges are left in $F$ after conditioning on each edge once. Furthermore, the $s-t$ resistance is extremely close to its original value by martingale concentration.}
\label{fig:simple-conditioning}
\end{figure}

\begin{proof}[Proof of Lemma \ref{lem:restate-special-slow-fix}]

Let $F^0 = F$ and for all $i \ge 0$, let $F^{i+1}$ be the set of non-leaf/loop edges left over after conditioning on $|F^i|/8$ edges in $F^i$. When a non-leaf/loop edge is split and conditioned on, it has a probability of at least $1/4$ of resulting in an additional loop or leaf. Furthermore, the number of non-leaf/loop edges in $F$ never increases. Therefore, since all conditionings are independent of each other, $|F^{i+1}|\le 63|F^i|/64$ with high probability as long as $|F^i|\ge \log n$ by Chernoff bounds. Therefore, there are only $\log n$ $F^i$s before the algorithm terminates.

Let $J^i$ be the graph $J$ with $F = F^i$. We now show that

$$b_{st}^T L_{J^{i+1}}^+ b_{st}\ge (1-\frac{10\sqrt{\log n}}{\sqrt{|F^i|}})(b_{st}^T L_{J^i}^+ b_{st})$$

with high probability. Let $J_j$ be the graph immediately before the $j$th iteration after $J^i$ and let $f_j$ be the edge conditioned on during this iteration to obtain $J_{j+1}$. Let $F_j$ be the value of $f$ before conditioning on $f_j$. By Propositions \ref{prop:deletion-update} and \ref{prop:contraction-update},

\begin{align*}
\textbf{E}_{J_{j+1}}[b_{st}^T L_{J_{j+1}}^+ b_{st}] &= \texttt{lev}_{J_j}(f_j)\left(b_{st}^T L_{J_j}^+ b_{st} - \frac{(b_{st}^T L_{J_j}^+ b_{f_j})^2}{r_{f_j}\texttt{lev}_{J_j}(f_j)}\right)\\
&+ \texttt{nonlev}_{J_j}(f_j)\left(b_{st}^T L_{J_j}^+ b_{st} + \frac{(b_{st}^T L_{J_j}^+ b_{f_j})^2}{r_{f_j}\texttt{nonlev}_{J_j}(f_j)}\right)\\
&= b_{st}^T L_{J_j}^+ b_{st}\\
\end{align*}

so $X_j := b_{st}^T L_{J_j}^+ b_{st}$ is a martingale. We now show that our choice of $f_j$ makes the martingale Lipschitz. Proposition \ref{prop:split-cond} implies that $\texttt{lev}_{J_j}(f_j)\ge 1/4$ and that $\texttt{nonlev}_{J_j}\ge 1/4$. By Propositions \ref{prop:deletion-update} and \ref{prop:contraction-update},

$$|b_{st}^T L_{J_{j+1}}^+ b_{st} - b_{st}^T L_{J_j}^+ b_{st}|\le 4 \frac{(b_{st}^T L_{J_j}^+ b_{f_j})^2}{r_{f_j}}$$

This is the energy contribution of $f_j$ to the overall $s-t$ energy. Therefore,

$$|b_{st}^T L_{J_{j+1}}^+ b_{st} - b_{st}^T L_{J_j}^+ b_{st}|\le \frac{4}{|F_j|}b_{st}^T L_{J_j}^+ b_{st}$$

Since $|F^i|/8$ conditionings occur during $J^i$ and $J^{i+1}$, $|F_j|\ge 7|F^i|/8$ for all $j$ and

$$b_{st}^T L_{J_{j+1}}^+ b_{st}\le \left(1 + \frac{32}{7|F^i|}\right)^{|F^i|/8}b_{st}^T L_{J_j}^+ b_{st}\le e b_{st}^T L_{J_j}^+ b_{st}$$

We now bootstrap this bound using Azuma's Inequality to obtain a much better bound on the change in the $s-t$ resistance. In particular,

$$|b_{st}^T L_{J_{j+1}}^+ b_{st} - b_{st}^T L_{J_j}^+ b_{st}|\le \frac{4e}{|F_j|}b_{st}^T L_{J^i}^+ b_{st}$$

Therefore, by Theorem \ref{thm:azuma},

$$\Pr[|b_{st}^T L_{J^{i+1}}^+ b_{st} - b_{st}^T L_{J^i}^+ b_{st}|\ge \frac{10\sqrt{\log n}}{\sqrt{|F^i|}} b_{st}^T L_{J^i} b_{st}]\le 1/\text{poly}(n)$$

Therefore, with high probability, the change is at most a $(1 - \ep)$ factor unless $|F^i|\le (\log n)/\ep^2$. Deleting these edges can only increase the $s-t$ resistance. Conditioning on $F'$ instead of deleting would have given a graph $J_{final}$ equivalent to sampling $J_{final}\sim I[F]$ by Proposition \ref{prop:split-sample}. This completes the proof.
\end{proof}

\subsection{Generalizing the warmup}\label{subsec:slow-intuition}

By Proposition \ref{prop:lin-alg-total-cond}, for any graph $H$ with disjoint sets of vertices $S,S'\subseteq V(H)$, $1/c^H(S,S') = b_{ss'}^T L_{H'}^+ b_{ss'}$, where $H'$ is a graph obtained from $H$ by identifying $S$ to $s$ and $S'$ to $s'$. Therefore, to show that $c^H(S,S')$ does not increase by more than a $(1 + \ep)$ factor, it suffices to show that $b_{ss'} L_{H'}^+ b_{ss'}$ does not decrease by more than a $(1 - \ep)$ factor, where $H$ is obtained from $I\setminus D$ by conditioning on a partial sample from a random spanning tree in $I$.

This quantity is similar to the quantity controlled during the warmup. There are two initial obstacles to directly using the approach in the warmup:

\begin{itemize}
\item The tree being sampled is sampled in a different graph ($H\cup D$) from the graph in which the quantity $b_{ss'}^T L_{H'}^+ b_{ss'}$ is defined. This discrepancy causes $b_{ss'}^T L_{H'}^+ b_{ss'}$ to not be a martingale under conditioning.

\item Conditioning could cause the quantity in the ``Size'' upper bound to change.
\end{itemize}

To get around these issues, we show that there exists an edge that changes $b_{ss'}^T L_{H'}^+ b_{ss'}$, $\Delta^H(S,S')$, and $\Delta^H(S',S)$ by a very small factor in expectation. Ideally, there would be an edge to condition on that changes $b_{ss'}^T L_{H'}^+ b_{ss'}$ and the ``Size'' bound in expectation by at most a $(1 + (\Delta^H(S,S') + \Delta^H(S',S) + |D|)/|F|^2)$ factor in expectation. If this is the case, then conditioning on all but $(\Delta^H(S,S') + \Delta^H(S',S) + |D|)\ep^{-1}$ edges changes the expectation by at most a $(1 + \ep)$ factor in expectation. When combined with martingale concentration, this shows that the size bound and the conductance change by at most a factor of $(1 + O(\ep))$ with high probability.

Unfortunately, such an increase is hard to obtain. For example, the expected change in $b_{ss'}^T L_{H'}^+ b_{ss'}$ due to conditioning on an edge $f$ can be written as a product of two terms:

\begin{itemize}
\item The change in $f$'s leverage score due to identifying $S$ to $s$, identifying $S'$ to $s'$, and deleting all edges in $D$.
\item The energy of the $b_{ss'}$ electrical flow in $H'$ on $f$.
\end{itemize}

The second quantity is at most $O(1/|F|)b_{ss'}^T L_{H'}^+ b_{ss'}$ for most edges in $F$ due to the fact that the average energy on an edge is small. The first quantity is trickier to bound and sometimes cannot be bounded at all. For example, if $F$ solely consists of edges with endpoints in $S$, the change in leverage score upon identifying $S$ to $s$ is for all edges in $F$ is high (constant). Luckily, though, the endpoints of these edges are very close to $s$ in the $L_{H'}^+ b_{ss'}$ potential embedding.

This phenomenon holds in general. In particular, Lemma \ref{lem:first-order-deg} can be used to show that the total leverage score decrease of all edges with an endpoint that has $s-s'$ normalized potential at least $p$ away from $s$ and $s'$ when $S$ and $S'$ are identified is at most $O((\Delta^H(S',S) + \Delta^H(S,S'))/p)$. The total leverage score increase is at most $|D|$, so the total leverage score change is at most the sum of these two quantities. Therefore, if a large fraction of the edges in $F$ have an endpoint with normalized $s-s'$ potential that is between $p$ and $1-p$, there is an edge $f\in F$, which, when conditioned on, causes an

$$1 + O\left(\left(\frac{|D| + (\Delta^H(S',S) + \Delta^H(S,S'))/p}{\ep}\right)\left(\frac{b_{ss'}^T L_{H'}^+ b_{ss'}}{|F|}\right)\right) \le 1 + \frac{(|D| + \Delta^H(S',S) + \Delta^H(S,S'))}{p|F|^2}$$

factor change in the ``Size'' bound and $b_{ss'}^T L_{H'}^+ b_{ss'}$ in expectation. This differs from the desired bound due to the factor of $1/p$.

This bound is good enough, though, to allow conditioning on all but roughly $\frac{(|D| + \Delta^H(S',S) + \Delta^H(S,S'))}{p\ep}$ edges before the ``Size'' bound or $b_{ss'}^T L_{H'}^+ b_{ss'}$ changes by more than an $(1 + \ep)$ factor. Once this many edges are left that have maximum normalized endpoint potential at most $O(p)$, defer conditioning on these edges until the very end of the algorithm.

Deferring conditioning on these edges works for the following reasons.

\begin{itemize}

\item The total normalized potential of such edges is at most $\frac{(|D| + \Delta^H(S',S) + \Delta^H(S,S'))}{p\ep} O(p) = O((|D| + \Delta^H(S',S) + \Delta^H(S,S'))\ep^{-1})$.

\item If we can show that the total normalized potential of deferred edges only increases by a $(1 + \ep)$ factor over the course of conditioning on the rest of $F$, the total normalized potential of all deferred edges is at most $O((|D| + \Delta^H(S',S) + \Delta^H(S,S'))\ep^{-1})$ at the end of the algorithm.

\item By Markov's inequality, at most $O((|D| + \Delta^H(S',S) + \Delta^H(S,S'))\ep^{-2})$ of the deferred edges have any endpoint with normalized potential in the interval $[\ep,1-\ep]$. Conditioning on the rest of the edges therefore decreases the $s-s'$ resistance in $H'$ by at most a factor of $(1-O(\ep))$ by Proposition \ref{prop:cond-pot-bound}. Letting $F'$ be the remaining $O((|D| + \Delta^H(S',S) + \Delta^H(S,S'))\ep^{-2})$ edges yields the desired result by Rayleigh monotonicity.

\end{itemize}

\subsection{Reducing the fixing lemma to the existence of stable oracles}

We prove Lemma \ref{lem:slow-fix} using the algorithm $\Fix$ with the oracle $\SlowOracle$. This algorithm picks edges to condition on in $F$ one by one in the graph $I$. It selects edges that do not change the following quantities very much:

\begin{itemize}
\item $\Delta^{I\setminus D}(S,S')$
\item $\Delta^{I\setminus D}(S',S)$
\item $c^{I\setminus D}(S,S')$
\item Various sums of potential drops
\end{itemize}

The key idea is that there always is an edge that does not change any of these quantities very much. Unlike in the special case when $S$ and $S'$ are single vertices, conditioning might decrease the effective resistance in expectation. We show that the same choice of edge that does not change these quantities also leads to a very small change in the expectation.

\subsubsection{Stability implies concentration}

We now formalize the concept of low-change edges using \emph{stable functions}:

\begin{definition}[Stable functions]
Call a function $g$ on graphs \emph{electrical} if $g$ is preserved under two operations:

\begin{itemize}
\item edge subdivision; that is replacing an edge $e$ with a path of length two consisting of edges with resistance $r_e/2$
\item edge splitting; that is replacing an edge $e$ with two parallel copies with resistance $2r_e$
\end{itemize}

A set of edges $X\subseteq F$ is called a \emph{$(\rho_L, \rho_E, \delta)$-stable} subset of $F$ for an electrical function $g$ if

\begin{itemize}
\item $|g(H/f) - g(H)|\le \frac{\rho_L}{|F|} g(H) + \delta$ for all edges $f\in X$
\item $|g(H\backslash f) - g(H)|\le \frac{\rho_L}{|F|} g(H) + \delta$ for all $f\in X$
\item $\frac{1}{|X|}\sum_{f\in X}|g(H) - \textbf{E}_{H'\sim H[f]}[g(H')]|\le \frac{\rho_E}{|F|^2} g(H) + \delta$
\end{itemize}

An set $X$ is called \emph{$(\rho_L, \rho_E)$-stable} for $g$ if it is $(\rho_L, \rho_E, 0)$-stable for $g$.
\end{definition}

\begin{definition}[Multiplicative functions]
A function $h: \mathbb{R}_{>0}^d\rightarrow \mathbb{R}_{>0}$ is called \emph{$\gamma$-multlipschitz} if the function $h'(x_1,\hdots,x_d) := \log h(\exp(x_1), \exp(x_2), \hdots, \exp(x_d))$ is $\gamma$-Lipschitz in the $\ell_1$ norm.
\end{definition}

We encapsulate our use of standard concentration inequalities into one proposition. This proposition is used to show that conditioning on a long sequence of edges that are stable for the functions $\{x_a\}_{a=1}^{\ell_1}$ and $\{y_b\}_{b=1}^{\ell_2}$ does not change the value of any of these functions by a large amount with high probability. This proposition allows the stability to depend on the $x_a$ functions, which one should think of as being related to $\Delta^{I\setminus D}(S,S')$.

\linestart

\begin{proposition}\label{prop:stable-concentration}
Consider two collections of electrical functions $\{x_a\}_{a=1}^{\ell_1}$ and $\{y_b\}_{b=1}^{\ell_2}$. Fix a graph $I_0 = I$ and a set $F_0 = F\subseteq E(I)$. Consider two sequences of random graphs $\{I_i\}_i$, $\{I_i'\}_i$, edges $\{f_i\}_i$, and edge sets $\{F_i\}_i$, $\{X_i\}_i$ with the following properties for all $i$:

\begin{itemize}
\item $I_{i+1} \sim I_i[[f_i]]$ for a uniformly random edge $f_i\in X_i$
\item $F_{i+1} := F_i\cap E(I_{i+1})$
\item $X_i\subseteq F_i$
\item $X_i$ is $(I_i',\rho_L,\rho_i,0)$-stable for all functions $x_a$ and $(I_i',F_i') := \texttt{Split}(I_i,F_i)$
\item $X_i$ is $(I_i',\rho_L,\rho_i,\delta)$-stable for all functions $y_b$
\item $|F_i|\ge \sigma$
\item $\rho_i := \rho_E(\{x_a(H_i)\}_a)$; that is $\rho_i$ is only a function of the values of the $x_a$s on $H_i$.
\item $\rho_E$ is a $\gamma$-multlipschitz function.
\end{itemize}

Let $\tau_E := \max_i (\rho_i/|X_i|)$. Then with high probability, all of the following bounds apply for all $i$:

\begin{itemize}
\item For all $w\in \{x_a\}_a$,

$$|w(I) - w(I_i)|\le \tilde{O}\left(\tau_E + \frac{\sqrt{\gamma\ell_1}\rho_L^{3/2}}{\sqrt{\sigma}}\right) w(H)$$

\item For all $w\in \{y_b\}_b$,

$$|w(I) - w(I_i)|\le \tilde{O}\left(\tau_E + \frac{\sqrt{\gamma\ell_1}\rho_L^{3/2}}{\sqrt{\sigma}}\right) w(H) + \tilde{O}(\rho_L\gamma\ell_1|F|\delta)$$

\end{itemize}
\end{proposition}

\lineend

Before proving this, notice that $f_i\in F_{i+1}$ if and only if the edge set of $I_i$ doesn't change, i.e. one of the following two conditions holds:

\begin{itemize}
\item $\texttt{lev}_{I_i}(f_i)\le 1/2$ and $f_i$ is contracted to form $I_{i+1}$
\item $\texttt{lev}_{I_i}(f_i) > 1/2$ and $f_i$ is deleted to form $I_{i+1}$
\end{itemize}

The proof of Proposition \ref{prop:stable-concentration} uses Azuma's Inequality.

\begin{proof}[Proof of Proposition \ref{prop:stable-concentration}]
Break up the sequence of edges to condition on into a small number of subsequences with guaranteed small changes in $x_a(I)$ and $y_b(I)$ values. Then, apply Theorem \ref{thm:azuma} to each subsequence. Specifically, let $i_0 = 0$ and $i_{j+1}$ be the maximum value for which $|F_{i_{j+1}}| \ge |F_{i_j}|(1 - 1/(16\gamma\ell_1\rho_L))$. Concentration bounds will be shown for each subsequence of $i$s between $i_j$ and $i_{j+1}$.

We start with bounds on $|i_{j+1} - i_j|$. Consider the random variable $|F_i|$. By the fourth property of stable edges, $\textbf{E}[|F_{i+1}| | F_i]\le |F_i| - \frac{1}{4}$. In particular for any $z > i_j$,

$$\textbf{E}[|F_z| | F_{i_j}, i_j]\le |F_{i_j}| - \frac{1}{4}(z - i_j)$$

$X_i = |F_i| - \textbf{E}[|F_i|]$ is a martingale with $c_i = 1$ for all $i$, so by Theorem \ref{thm:azuma}

$$|F_z|\le |F_{i_j}| - \frac{1}{4}(z - i_j) + \sqrt{(z - i_j)\log n}$$

with probability at least $1 - \frac{1}{\text{poly(n)}}$. As long as $z - i_j > 8\log n$, the above inequality implies that

$$|z - i_j| \le 8(|F_{i_j}| - |F_z|)$$

For $z = i_{j+1}$, $|i_{j+1} - i_j|\le \frac{1}{2\gamma\ell_1\rho_L}|F_{i_j}|$ with high probability.

By the first two properties of stable edges and the fact that $|F_i|$ is nondecreasing,

\begin{align*}
x_a(I_i)&\le x_a(I_{i_j})\left(1 + \frac{\rho_L}{|F_{i_{j+1}}|}\right)^{i_{j+1} - i_j}\\
&\le \left(1 + \frac{1}{\gamma\ell_1}\right) x_a(I_{i_j})\\
\end{align*}

and $x_a(I_i)\ge (1 - \frac{1}{\gamma\ell_1}) x_a(I_{i_j})$ for all $i\in [i_j,i_{j+1}]$ and all $a$ with high probability. By the multlipschitzness of $\rho_E$,

$$\rho_i\le (1 + 1/(\gamma \ell_1))^{\gamma \ell_1}\rho_{i_j}\le e \rho_{i_j}$$

for all $i\in [i_j,i_{j+1}]$ with high probability. Similarly,

$$y_b(I_i)\le (1 + 1/(\gamma\ell_1))(y_b(I_{i_j}) + |F|\delta)$$

for all $i\in [i_j,i_{j+1}]$ and $b$ with high probability. 

We now use these worst-case bounds on $x_a$s and $y_b$s within each interval to get a better bound using the expectation. By the third property of stable edges and the definition of $\tau_E$,

\begin{align*}
|\textbf{E}[w(I_{i_{j+1}})|I_{i_j},i_j] - w(I_{i_j})|&\le (i_{j+1} - i_j) \frac{\tau_E}{|F_{i_{j+1}}|} e (w(I_{i_j}) + |F|\delta) + |F|\delta\\
&\le 2e\frac{\tau_E}{\gamma\ell_1\rho_L} w(I_{i_j}) + O(|F|\delta)\\
\end{align*}

for functions $w\in \{y_b\}_b$. For functions $w\in \{x_a\}_a$, a similar bound holds:

\begin{align*}
|\textbf{E}[w(I_{i_{j+1}})|I_{i_j},i_j] - w(I_{i_j})|&\le (i_{j+1} - i_j) \frac{\tau_E}{|F_{i_{j+1}}|} e w(I_{i_j})\\
&\le 2e\frac{\tau_E}{\gamma\ell_1\rho_L} w(I_{i_j})\\
\end{align*}

Therefore, by Theorem \ref{thm:azuma},

$$|w(I_{i_{j+1}}) - w(I_{i_j})|\le 2e\frac{\tau_E}{\gamma\ell_1\rho_L} w(I_{i_j}) + 2e\sqrt{(\log n)\frac{|F_{i_j}|}{2\gamma\ell_1\rho_L}}\frac{\rho_L}{|F_i|} w(I_{i_j}) + \tilde{O}(|F|\delta)$$

with high probability for $w\in \{y_b\}_b$, with analogous statements for $x_a$ with $\delta = 0$ (see proposition guarantees). Since $|F_{i_{j+1}+1}|\le |F_{i_j}|(1 - \frac{1}{16 \rho_L\gamma\ell_1})$, $j\le 16 \rho_L\gamma\ell_1 (\log n)$. Since $|F_i|\ge \sigma$ for all $i$, the total change is

$$\tilde{O}\left(\tau_e + \frac{\rho_L^{3/2}\sqrt{\gamma\ell_1}}{\sqrt{\sigma}}\right)w(I) + \tilde{O}(\rho_L\gamma\ell_1|F|\delta)$$

as desired.
\end{proof}

\subsubsection{Stable oracles and their use in proving the fixing lemmas}

In this section, we exploit Proposition \ref{prop:stable-concentration} to reduce Lemmas \ref{lem:slow-fix} and \ref{lem:fast-fix} to the construction of a certain oracle. This oracle outputs a set of edges remains stable for the main objective $b_{ss'}^T L_H^+ b_{ss'}$, the ``Size'' bound, and an auxilliary objective that bounds the endpoint potentials of two sets of ``deferred'' edges.

Before defining the oracle, we define a proxy for the degree function, called the \emph{normalized degree}, that will be more convenient to control directly than $\Delta^H(X,Y)$:

\begin{definition}[Normalized degree]
For a graph $H$ with disjoint vertex sets $X,Y$, let

$$\delta_{X,Y}(H) := \frac{\Delta^H(X,Y)}{c^H(X,Y)}$$
\end{definition}

It is easier to control due to the following representation, which involves quantities whose change can be easily analyzed using rank 1 updates:

\begin{remark}
$\delta_{X,Y}(H)$ can be written in a different way:

$$\delta_{X,Y}(H) = \sum_{w\in Y} \sum_{e\in \partial_H w}(b_{xw}^T L_{H/X}^+ b_{xw})\frac{b_{xy}^T L_{H/(X\cup Y)}^+ b_e}{r_e}$$

where $x$ and $y$ are the identifications of $X$ and $Y$ respectively in $H/(X\cup Y)$ and $e$ is oriented towards $w$.
\end{remark}

\begin{proof}
By Proposition \ref{prop:lin-alg-norm-cond}, for each $w\in Y$,

$$\sum_{e\in \partial_H w} \frac{b_{xy}^T L_{H/(X\cup Y)}^+ b_e}{r_e} = \frac{c_{xw}^{\texttt{Schur}(H/X,X\cup Y)}}{\sum_{w'\in Y} c_{xw'}^{\texttt{Schur}(H/X,X\cup Y)}}$$

Since Schur complement conductances do not depend on edges with endpoints that were not eliminated,

$$\sum_{w'\in Y} c_{xw'}^{\texttt{Schur}(H/X,X\cup Y)} = c^H(X,Y)$$

Combining these equalities shows that

$$\sum_{w\in Y} (b_{xw}^T L_{H/X}^+ b_{xw})\sum_{e\in \partial_H w}\frac{b_{xy}^T L_{H/(X\cup Y)}^+ b_e}{r_e} = \frac{\Delta^H(X,Y)}{c^H(X,Y)}$$

as desired.
\end{proof}

Now, we define an oracle that generalizes the concept of picking the minimum energy edge that was so important to $\SpecialFix$:

\linestart

\begin{definition}[Stable oracles]\label{def:oracle}
An $(\rho,K(z))$-\emph{stable oracle}, is a function

$Z\gets \Oracle(I,S,S',D,A,B,W)$ that takes in a graph $I$, two disjoint sets of vertices $S,S'\subseteq V(I)$, a set of \emph{deleted} edges $D\subseteq E(I)$, two sets of \emph{deferred} edges $A,B\subseteq E(I)$, and a set of edges $W$ to condition on. $K$ is allowed to be a function of some parameter $z$. This oracle is given inputs that satisfy the following conditions:

\begin{itemize}
\item (Bounded leverage score difference) For all $e\in W$, $|\texttt{lev}_{I\setminus D}(e) - \texttt{lev}_{I/(S,S')}(e)|\le 1/16$
\item (Narrow potential neighborhood) There is a $p\le 1/4$ for which \textbf{one} of the following conditions holds:
    \begin{itemize}
    \item ($s$ narrow potential neighborhood) $\frac{b_{ss'}^T L_{(I\setminus D)/(S,S')}^+ ((b_{su} + b_{sv})/2)}{b_{ss'}^T L_{(I\setminus D)/(S,S')}^+ b_{ss'}}\in [p,2p]$ for all $\{u,v\}\in W$.
    \item ($s'$ narrow potential neighborhood) $\frac{b_{ss'}^T L_{(I\setminus D)/(S,S')}^+ ((b_{us'} + b_{vs'})/2)}{b_{ss'}^T L_{(I\setminus D)/(S,S')}^+ b_{ss'}}\in [p,2p]$ for all $\{u,v\}\in W$.
    \end{itemize}
\end{itemize}

It outputs a set $Z\subseteq E(I)$ of edges to condition on.

Let $I_0 = I$ and for each $i > 0$, obtain $I_i$ by picking a uniformly random edge $f_{i-1}\in Z$, letting $I_i\gets I_{i-1}[[f_{i-1}]]$, and removing $f_{i-1}$ from $Z$.

$\Oracle$ satisfies the following stability-related properties with high probability for all $i < K(|W|)$:

\begin{itemize}
\item (Size of $Z$) $|Z|\ge |W|/(\log^4 n)$
\item (Leverage score stability)
    \begin{itemize}
    \item (Upper leverage score stability) $|\texttt{lev}_{I_i\setminus D}(f_i) - \texttt{lev}_{I\setminus D}(f_i)|\le 1/16$
    \item (Lower leverage score stability) $|\texttt{lev}_{I_i/(S,S')}(f_i) - \texttt{lev}_{I/(S,S')}(f_i)|\le 1/16$
    \end{itemize}

\item (Midpoint potential stability)
    \begin{itemize}
    \item ($s$ midpoint potential stability) Let $f_i = \{u_i,v_i\}$. Then

    $$\frac{b_{ss'}^T L_{(I_i\setminus D)/(S,S')}^+ (b_{su_i} + b_{sv_i})/2}{b_{ss'}^T L_{(I_i\setminus D)/(S,S')}^+ b_{ss'}} \ge p/2$$

    \item ($s'$ midpoint potential stability) Let $f_i = \{u_i,v_i\}$. Then 

    $$\frac{b_{ss'}^T L_{(I_i\setminus D)/(S,S')}^+ (b_{u_is'} + b_{v_is'})/2}{b_{ss'}^T L_{(I_i\setminus D)/(S,S')}^+ b_{ss'}} \ge p/2$$

    \end{itemize}

\item ($S-S'$-normalized degree change stability)
    \begin{itemize}
    \item ($S-S'$ conductance term stability)

    $$\sum_{w\in S'} (b_{sw}^T L_{(I_i\setminus D)/S}^+ b_{sw}) \sum_{e\in \partial_{I_i}w} \frac{|b_{ss'}^T L_{(I_i\setminus D)/(S,S')}^+ b_{f_i}| |b_{f_i}^T L_{(I_i\setminus D)/(S,S')}^+ b_e|}{r_{f_i}r_e} \le \frac{\rho}{|W|}\delta_{S,S'}(I\setminus D)$$

    \item ($S'-S$ conductance term stability)

    $$\sum_{w\in S} (b_{s'w}^T L_{(I_i\setminus D)/S'}^+ b_{s'w}) \sum_{e\in \partial_{I_i}w} \frac{|b_{ss'}^T L_{(I_i\setminus D)/(S,S')}^+ b_{f_i}| |b_{f_i}^T L_{(I_i\setminus D)/(S,S')}^+ b_e|}{r_{f_i}r_e} \le \frac{\rho}{|W|}\delta_{S',S}(I\setminus D)$$

    \item ($S-S'$ energy term stability)

    $$\sum_{w\in S'} \frac{(b_{sw}^T L_{(I_i\setminus D)/S}^+ b_{f_i})^2}{r_{f_i}} \sum_{e\in \partial_{I_i}w} \frac{b_{ss'}^T L_{(I_i\setminus D)/(S,S')}^+ b_e}{r_e}\le \frac{\rho}{|W|}\delta_{S,S'}(I\setminus D)$$

    \item ($S'-S$ energy term stability)

    $$\sum_{w\in S} \frac{(b_{s'w}^T L_{(I_i\setminus D)/S'}^+ b_{f_i})^2}{r_{f_i}} \sum_{e\in \partial_{I_i}w} \frac{b_{ss'}^T L_{(I_i\setminus D)/(S,S')}^+ b_e}{r_e}\le \frac{\rho}{|W|}\delta_{S',S}(I\setminus D)$$
    \end{itemize}

\item (Deferred endpoint potential change stability)

\begin{align*}
&\left(\sum_{\{u,v\}\in A} \frac{|b_{ss'}^T L_{(I_i\setminus D)/(S,S')}^+ b_{f_i}| |b_{f_i}^T L_{(I_i\setminus D)/(S,S')}^+ (b_{su} + b_{sv})|}{r_{f_i}}\right)\\
&+ \left(\sum_{\{u,v\}\in B} \frac{|b_{ss'}^T L_{(I_i\setminus D)/(S,S')}^+ b_{f_i}| |b_{f_i}^T L_{(I_i\setminus D)/(S,S')}^+ (b_{us'} + b_{vs'})|}{r_{f_i}}\right)\\
&\le \frac{\rho}{|W|}\left(\sum_{\{u,v\}\in A} b_{ss'}^T L_{(I\setminus D)/(S,S')}^+ (b_{su} + b_{sv}) + \sum_{\{u,v\}\in B} b_{ss'}^T L_{(I\setminus D)/(S,S')}^+ (b_{us'} + b_{vs'})\right) + \frac{r_{min}}{n^4}
\end{align*}

\item (Main objective change stability)

$$\frac{(b_{ss'}^T L_{(I_i\setminus D)/(S,S')}^+ b_{f_i})^2}{r_{f_i}}\le \frac{\rho}{|W|} b_{ss'}^T L_{(I\setminus D)/(S,S')}^+ b_{ss'}$$

\end{itemize}

$|W|$ in each of the above guarantees refers to the original size of $W$.
\end{definition}

\lineend

We now use this oracle to prove the following result, which will later be used to show Lemmas \ref{lem:slow-fix} and \ref{lem:fast-fix}:

\linestart

\begin{lemma}\label{lem:oracle-fix}
Given a $(\rho,K)$-stable oracle $\Oracle$, there is an algorithm $\Fix(I,J,S,S',\ep,F,D)$ that takes in a graph $I$, a random graph $J$ that is a valid sample from the distribution $I[F]$, a set $D\subseteq E(I)$ of deleted edges, $S,S'\subseteq V(G)$, $F\subseteq E(I)\setminus D$, and an accuracy parameter $\ep\in (0,1)$. With high probability on $J$, there is a set $F'\gets \Fix(I,J,S,S',\ep,F,D) \subseteq F$ with two properties:

\begin{itemize}
\item (Conductance) $c^{J\setminus D\setminus F'}(S,S')\le (1 + \ep)c^{I\setminus D}(S,S')$
\item (Size) $|F'|\le \tilde{O}(\rho^3\ep^{-3}(\Delta^{I\setminus D}(S,S') + \Delta^{I\setminus D}(S',S) + |D|))$
\end{itemize}

Futhermore, $\Fix$ takes $\tilde{O}(m + (\max_{z\le |F|} \frac{z \rho^3}{K(z)} + \log^2 n)(m + \mc T(\Oracle)))$ time, where $\mc T(\Oracle)$ is the runtime of $\Oracle$.
\end{lemma}

\lineend

The runtime in the above guarantee corresponds to running $\Oracle$ $\frac{|F|\rho^3}{K}$ times. Later in this section, we implement a simple $(O(\text{polylog}(n)),1)$-stable oracle called $\SlowOracle$. This oracle is sufficient to prove Lemma \ref{lem:slow-fix}. However, $\mc T(\SlowOracle) = \Theta(m^2)$. Furthermore, $K = 1$. As a result, the runtime of the algorithm could be as high as $\Theta(m^3)$. To reduce this runtime, we exploit Theorem \ref{thm:localization} and sketching techniques to obtain $\FastOracle$. Sketching allows us to make $\mc T(\FastOracle) = m^{1 + o(1)}$ and Theorem \ref{thm:localization} allows us to reuse $Z$ for many iterations. Specifically, $\FastOracle$ is an $(m^{o(1)},z m^{-o(1)})$-stable oracle. Therefore, plugging in $\FastOracle$ to $\Fix$ proves Lemma \ref{lem:fast-fix}.

To prove Lemma \ref{lem:oracle-fix}, proceed as described in Section \ref{subsec:slow-intuition}. At any given point, bucket edges in $F$ by their maximum normalized endpoint $S-S'$ potential. Let $[p,2p]$ be the potential range with the most edges of $F$ in the corresponding bucket. If this bucket $W$ has more than $\rho^3((\Delta^{I\setminus D}(S,S') + \Delta^{I\setminus D}(S',S))/p + |D|)$, call $\Oracle$ to produce a subset $Z\subseteq W$ and split and condition on a uniformly random sample of $K$ elements from $Z$. Otherwise, ``defer'' conditioning on all edges of $Z$ by adding them to $A$ or $B$, depending on whether they are closer to $S$ or $S'$ respectively in the potential embedding. Once all remaining edges of $F$ are in either $A$ or $B$, delete the ones with an endpoint with normalized potential in the interval $[\ep,1-\ep]$ and condition on all others.

We now give pseudocode that implements the above description:

\begin{algorithm}[H]
\SetAlgoLined
\DontPrintSemicolon
\caption{$\Fix(I,J,S,S',\ep,F,D)$}

    \KwData{two disjoint sets of vertices $S$ and $S'$, the graph $I$, the random sample $J\sim I[F]$, $\epsilon\in (0,1)$, $F,D\subseteq E(G)$}

    \KwResult{the set $F'$}

    \tcp{edges whose endpoint potentials we want to control}
    $A,B\gets \emptyset$\;

    \While{$|F|$ is not empty}{

        \tcp{identifications $s,s'$}

        $I'\gets (I\setminus D)/(S,S')$\;

        \ForEach{$i \gets 0,1,\hdots$}{

            $X_i\gets \{\{u,v\}\in F: \frac{b_{ss'}^T L_{I'}^+ (b_{su} + b_{sv})/2}{b_{ss'}^T L_{I'}^+ b_{ss'}}\in (2^{-i-2}, 2^{-i-1}]\}$\;

            $Y_i\gets \{\{u,v\}\in F: \frac{b_{ss'}^T L_{I'}^+ (b_{us'} + b_{vs'})/2}{b_{ss'}^T L_{I'}^+ b_{ss'}}\in [2^{-i-2}, 2^{-i-1})\}$\;

        }

        $i_{\max}\gets 2\log n$\;

        $X_{\text{low}}\gets \cup_{i > i_{\max}} X_i$\;

        $Y_{\text{low}}\gets \cup_{i > i_{\max}} Y_i$\;

        $W \gets \arg\max_{W\in \{X_0,Y_0,\hdots,X_{i_{\max}},Y_{i_{\max}},X_{\text{low}},Y_{\text{low}}\}} |W|$\;

        $i\gets $ index of $W$ as $X_i$ or $Y_i$, with $i\gets i_{\max}+1$ if $W$ is low\;

        $p\gets 2^{-i-2}$\;

        \uIf{$|W|\le 10000(\log n)\xibuc^2\rho^3\epsilon^{-2}(\Delta^{I\setminus D}(S,S') + \Delta^{I\setminus D}(S',S) + |D|)2^i$}{

            \tcp{defer all edges in $W$}

            $F\gets F\setminus W$\;\label{line:remove-w}

            Add edges of $W$ to $A$ if $W = X_i$, otherwise add $W$ to $B$\;

        }\Else{

            Remove all edges $e$ from $W$ for which $\texttt{lev}_{I\setminus D}(e) - \texttt{lev}_{I/(S,S')}(e) > 1/32$ using Johnson-Lindenstrauss with $\ep = 1/64$\;\label{line:remove-high-levcng}

            $Z\gets \Oracle(I,S,S',D,A,B,W)$\;

            $Z'\gets$ a uniformly random subset of $Z$ with size $K(|W|)$\; \label{line:select-1}

            \tcp{can be implemented using $J$ in $O(|Z'|)$ time}
            $I\gets I[[Z']]$\;\label{line:split-condition}

            $F\gets F\cap E(I)$\; \label{line:select-2}

        }

    }

    $F'\gets $ set of edges in $A\cup B$ that have some endpoint with normalized $I'$ potential between $\epsilon/2$ and $1 - \epsilon/2$\; \label{line:contract-remaining}

    \Return $F'$\;
\end{algorithm}

Our analysis of this algorithm centers around applying Proposition \ref{prop:stable-concentration}. Specifically, Proposition \ref{prop:stable-concentration} is applied to show that the electrical functions

\begin{itemize}
\item $\delta_{S,S'}(H\setminus D)$
\item $\delta_{S',S}(H\setminus D)$
\item $b_{ss'}^T L_{(H\setminus D)/(S,S')}^+ b_{ss'}$
\item $\left(\sum_{u\in V(A)} b_{su}^T L_{(H\setminus D)/(S,S')}^+ b_{ss'}\right) + \left(\sum_{v\in V(B)} b_{vs'}^T L_{(H\setminus D)/(S,S')}^+ b_{ss'}\right)$
\end{itemize}

do not change much over the course of conditioning on many edges. We focus our discussion of intuition on the first two functions, since the others are easier. The first two functions are very similar, so we focus our discussion on the first one. Furthermore, bounding $\rho_E$ is in Proposition \ref{prop:stable-concentration} is more difficult than bounding $\rho_L$, so we focus on that.

To bound expected changes in $\delta_{S,S'}(H\setminus D)$, it helps to define a quantity related to the discrepancy between leverage scores in $H$ (the graph a tree is being sampled in) and $(H\setminus D)/(S,S')$ (the graph in which quantities of interest are defined):

\begin{definition}[Leverage score change]
Consider an edge $e$ in a graph $H$ obtained by identifying vertices and deleting edges in a graph $G$. Define

$$\texttt{levcng}_{G\rightarrow H}(e) = \frac{r_e - b_e^T L_G^+ b_e}{r_e - b_e^T L_H^+ b_e} - \frac{b_e^T L_G^+ b_e}{b_e^T L_H^+ b_e} = \frac{r_e(b_e^T L_H^+ b_e - b_e^T L_G^+ b_e)}{(r_e - b_e^T L_H^+ b_e) b_e^T L_H^+ b_e}$$
\end{definition}

It is also helpful to define the following maximum energy fraction, which is used to define a generalization of the fact that the effective resistance is the sum of energies on all edges:

\begin{definition}[Maximum energy fraction]
For a graph $H$, a set $Y\subseteq V(H)$, a vertex $x_0\notin Y$, and some edge $f\in E(X)$, let

$$\alpha_{x_0,Y}^H(f) = \max_{w\in Y} \frac{(b_{x_0w}^T L_H^+ b_f)^2}{(b_{x_0w}^T L_H^+ b_{x_0w})r_f}$$
\end{definition}

\begin{remark}\label{rmk:lev-cng}
When $\gamma\le \texttt{lev}_H(e)\le 1 - \gamma$ for some $\gamma\in [0,1/2]$,

$$|\texttt{lev}_G(e) - \texttt{lev}_H(e)|\le |\texttt{levcng}_{G\rightarrow H}(e)|\le 2|\texttt{lev}_G(e) - \texttt{lev}_H(e)|/\gamma$$
\end{remark}

Consider any edge $f\in E(H)\setminus D$. One can write the expected change in $\delta_{S,S'}(H\setminus D)$ after conditioning using rank-one updates:

\begin{align*}
&\textbf{E}_{H'\sim H[f]}[\delta_{S,S'}(H'\setminus D)] - \delta_{S,S'}(H\setminus D)\\
&= -\texttt{levcng}_{H\rightarrow (H\setminus D)/S}(f)\sum_{w\in S'} \sum_{e\in \partial_H w} \left(\frac{(b_{sw}^T L_{(H\setminus D)/S}^+ b_f)^2}{r_f}\right)\left(\frac{b_{ss'}^T L_{(H\setminus D)/(S,S')}^+ b_e}{r_e}\right)\\
&- \texttt{levcng}_{H\rightarrow (H\setminus D)/(S,S')}(f) \sum_{w\in S'} \sum_{e\in \partial_H w} \left(b_{sw}^T L_{(H\setminus D)/S}^+ b_{sw}\right)\left(\frac{(b_{ss'}^T L_{(H\setminus D)/(S,S')}^+ b_f)(b_f^T L_{(H\setminus D)/(S,S')}^+ b_e)}{r_fr_e}\right)\\
&+ \left(\frac{\texttt{lev}_H(f)}{\texttt{lev}_{(H\setminus D)/S}(f)\texttt{lev}_{(H\setminus D)/(S,S')}(f)} + \frac{1 - \texttt{lev}_H(f)}{(1 - \texttt{lev}_{(H\setminus D)/S}(f))(1 - \texttt{lev}_{(H\setminus D)/(S,S')}(f))}\right)\\
&\sum_{w\in S'} \sum_{e\in \partial_H w} \left(\frac{(b_{sw}^T L_{(H\setminus D)/S}^+ b_f)^2}{r_f}\right)\left(\frac{(b_{ss'}^T L_{(H\setminus D)/(S,S')}^+ b_f)(b_f^T L_{(H\setminus D)/(S,S')}^+ b_e)}{r_fr_e}\right)\\
\end{align*}

By the triangle inequality and the definition of $\alpha_{s,S'}^{(H\setminus D)/S}(f)$,

\begin{align*}
&|\textbf{E}_{H'\sim H[f]}[\delta_{S,S'}(H'\setminus D)] - \delta_{S,S'}(H\setminus D)|\\
&\le |\texttt{levcng}_{H\rightarrow (H\setminus D)/S}(f)|\sum_{w\in S'} \sum_{e\in \partial_H w} \left(\frac{(b_{sw}^T L_{(H\setminus D)/S}^+ b_f)^2}{r_f}\right)\left(\frac{b_{ss'}^T L_{(H\setminus D)/(S\cup S')}^+ b_e}{r_e}\right)\\
& |\texttt{levcng}_{H\rightarrow (H\setminus D)/(S,S')}(f)| \sum_{w\in S'} \sum_{e\in \partial_H w} \left(b_{sw}^T L_{(H\setminus D)/S}^+ b_{sw}\right)\left(\frac{|b_{ss'}^T L_{(H\setminus D)/(S,S')}^+ b_f| |b_f^T L_{(H\setminus D)/(S,S')}^+ b_e|}{r_fr_e}\right)\\
&+ \left(\frac{\texttt{lev}_H(f)}{\texttt{lev}_{(H\setminus D)/S}(f)\texttt{lev}_{(H\setminus D)/(S,S')}(f)} + \frac{1 - \texttt{lev}_H(f)}{(1 - \texttt{lev}_{(H\setminus D)/S}(f))(1 - \texttt{lev}_{(H\setminus D)/(S,S')}(f))}\right)\\
&\alpha_{s,S'}^{(H\setminus D)/S}(f)\sum_{w\in S'} \sum_{e\in \partial_H w} \left(b_{sw}^T L_{(H\setminus D)/S}^+ b_{sw}\right)\left(\frac{|b_{ss'}^T L_{(H\setminus D)/(S,S')}^+ b_f| |b_f^T L_{(H\setminus D)/(S,S')}^+ b_e|}{r_fr_e}\right)\\
\end{align*}

The stable oracle guarantees can be used on all of the quantities in the above sum. The ``$S-S'$ normalized degree change stability'' guarantees bound the double sums in each of the three above terms. The ``Leverage score stability'' guarantees bound the leverage score quantity in the second order term. Initially, the $\texttt{levcng}$ and $\alpha$ quantities may seem trickier to bound. Luckily, we can prove bounds on their average values in terms of $\Delta^{H\setminus D}(U_S,S') + \Delta^{H\setminus D}(U_{S'},S)$, where $U_S$ and $U_S'$ are the sets of vertices with normalized potential less than $1-p$ and greater than $p$ respectively, with $s$ and $s'$ assigned to 0 and 1 respectively. These $\Delta$ quantities are in turn bounded using the ``Midpoint potential stability'' guarantee, along with the following:

\begin{restatable}{proposition}{propdegpotbound}\label{prop:deg-pot-bound}
Consider two disjoint sets of vertices $X$ and $Y$ in a graph $G$. Let $G' = G/(X\cup Y)$, with $x$ and $y$ the identifications of $X$ and $Y$ respectively. Let $Z$ be the set of vertices $v$ with electrical potential $p_v \le \gamma$ for some $\gamma\in (0,1)$ with boundary conditions $p_x = 0$ and $p_y = 1$. Then

$$\Delta^G(Z,Y)\le \frac{1}{1-\gamma}\Delta^G(X,Y)$$

where $z$ is the identification of $Z$ in $G/Z$.
\end{restatable}

This bound on $\Delta$ is then used in conjunction with the following to bound $\texttt{levcng}$ and $\alpha$:

\begin{restatable}[Bounding first order terms, part 1]{lemma}{lemfirstorderdeg}\label{lem:first-order-deg}
Consider a graph $G$ and two sets $X,Y\subseteq V(G)$ with $X\cap Y = \emptyset$. Let $H = G/Y$ with $y$ the identification of $Y$ in $G$. Let $\Delta = \Delta^G(X,Y)$. Then

$$\sum_{f\in G[X]\cup \partial_G X: 1/4\le \texttt{lev}_G(f)\le 3/4} |\texttt{levcng}_{G\rightarrow H}(f)|\le 32\Delta$$
\end{restatable}

\begin{restatable}[Bounding first order terms, part 2]{lemma}{lemdelfirstorderdeg}\label{lem:del-first-order-deg}

Consider a graph $G$ and a set of edges $D\subseteq E(G)$. Then

$$\sum_{e\in E(G)\setminus D: 1/4\le \texttt{lev}_G(e)\le 3/4} |\texttt{levcng}_{G\rightarrow G\setminus D}(e)|\le 4|D|$$
\end{restatable}

\begin{figure}
\includegraphics[width=1.0\textwidth]{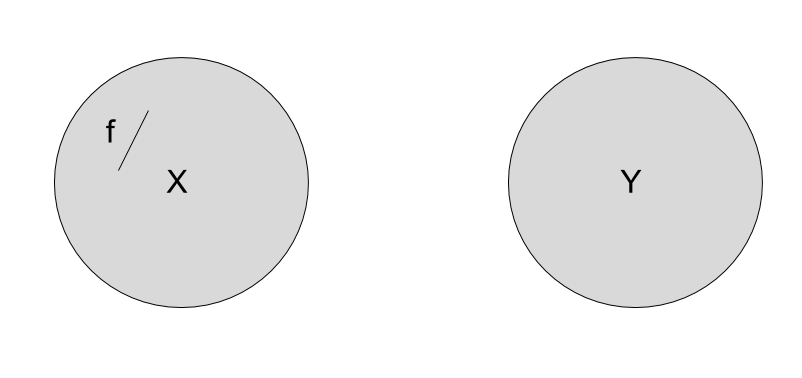}
\caption{Depiction of Lemma \ref{lem:first-order-deg}. When $Y$ is identified, the expected number of random spanning tree edges decreases.}
\label{fig:first-order-bound}
\end{figure}

\begin{restatable}[Bounding the second order term]{lemma}{lemsecondorderdeg}\label{lem:second-order-deg}
Consider a graph $G$ and two disjoint sets of vertices $X,Y\subseteq V(G)$. For any $s\in X$,

$$\sum_{e\in E_G(X)\cup \partial_G X} \alpha_{s,Y}^G(e) = \sum_{e\in E_G(X)\cup \partial_G X} \max_{t\in Y} \frac{(b_{st}^T L_G^+ b_e)^2}{(b_{st}^T L_G^+ b_{st})r_e} \le 24\xibuc^2\Delta^G(X,Y)$$

where $\xibuc = \log (m\alpha)$.
\end{restatable}

\begin{figure}
\includegraphics[width=1.0\textwidth]{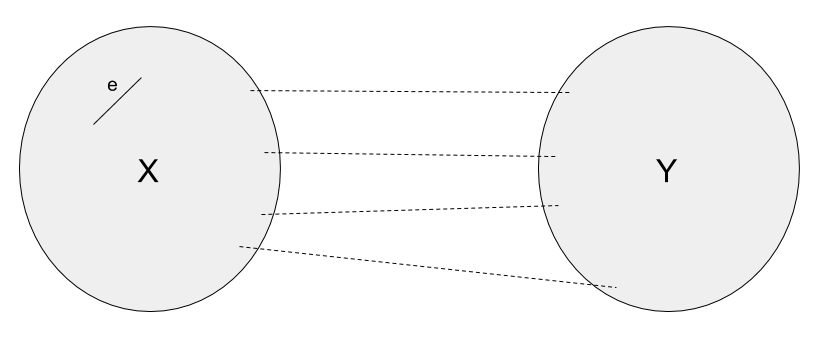}
\caption{Depiction of Lemma \ref{lem:second-order-deg}. When $\Delta^G(X,Y)$ is low, the vertices in $Y$ are similar to each other from the perspective of $s-Y$ demand vectors.}
\label{fig:second-order-bound}
\end{figure}

Appendix \ref{sec:stable-appendix} is dedicated towards proving these bounds. Notice that in Line \ref{line:select-1}, the algorithm chooses edges to condition on randomly. This allows us to exploit the above bounds in order to eliminate the $\texttt{levcng}$ and $\alpha$ dependencies. This completes the outline of the proof of the bound on $\rho_E$ for the normalized degree and illustrates how we prepare to apply Proposition \ref{prop:stable-concentration} in general.

We now state all of the stability results that a stability oracle implies in the algorithm $\Fix$. These are shown in Appendix \ref{sec:alg-stable-appendix}. These propositions analyze of Lines \ref{line:select-1}-\ref{line:select-2} by conditioning on a uniformly random sample one edge at a time.

\linestart

\begin{definition}[Stability propositions setup]\label{def:stability-input}

Let $Z_0 = Z\gets \Oracle(I,S,S',D,A,B,W)$ and $I_0 = I$, where $\Oracle$ is $(\rho,K)$-stable. Assume that the input to $\Oracle$ satisfies the conditions described in Definition \ref{def:oracle}.

Obtain $Z_k$ and $I_k$ for $k > 0$ by choosing a uniformly random edge $f_{k-1}\in Z_{k-1}$, letting $Z_k \gets Z_{k-1}\setminus \{f_{k-1}\}$, and letting $I_k\gets I_{k-1}[[f_{k-1}]]$. Let $\Delta_k := \Delta^{I_k\setminus D}(S,S') + \Delta^{I_k\setminus D}(S',S) + |D|$.
\end{definition}

\lineend

Notice that the set $Z'$ defined in the algorithm $\Fix$ could analogously be defined by letting $Z'\gets Z\setminus Z_K$. We now show that the following stability properties hold:

\begin{restatable}[Stability with respect to $\Delta$]{proposition}{propdeltastability}\label{prop:delta-stability}
For all $k\in \{0,1,\hdots,K(|W|)-1\}$, the set $Z_k$ is a $(\tilde{O}(\rho),\tilde{O}(\rho \Delta_k/p),0)$-stable subset of $W$ for the electrical functions

$$\delta_{S,S'}(H\setminus D)$$

and

$$\delta_{S',S}(H\setminus D)$$

of $H$ with high probability.
\end{restatable}

\begin{restatable}[Stability with respect to sums of deferred potentials]{proposition}{propdeferredstability}\label{prop:deferred-stability}
For all $k\in \{0,1,\hdots,K(|W|)-1\}$, the set $Z_k$ is a $(\tilde{O}(\rho),\tilde{O}(\rho \Delta_k/p),r_{min}/n^4)$-stable subset of $W$ for the electrical function

$$\left(\sum_{\{u,v\}\in A} b_{ss'}^T L_{(H\setminus D)/(S,S')}^+ (b_{su} + b_{sv})\right) + \left(\sum_{\{u,v\}\in B} b_{ss'}^T L_{(H\setminus D)/(S,S')}^+ (b_{us'} + b_{vs'})\right)$$

of $H$ with high probability.
\end{restatable}

\begin{restatable}[Stability with respect to the main objective]{proposition}{propmainstability}\label{prop:main-stability}

For all $k\in \{0,1,\hdots,K(|W|)-1\}$, the set $Z_k$ is a $(\tilde{O}(\rho),\tilde{O}(\rho \Delta_k/p),0)$-stable subset of $W$ for the electrical function

$$b_{ss'}^T L_{(H\setminus D)/(S,S')}^+ b_{ss'}$$

of $H$ with high probability.
\end{restatable}

Now, we use Proposition \ref{prop:stable-concentration} to show the following result:

\linestart

\begin{proposition}\label{prop:end-state}
Immediately before Line \ref{line:contract-remaining} of the algorithm $\Fix$, the graph $I$ and the sets $A$ and $B$ have the following properties with high probability:

\begin{itemize}
\item (Main objective) Let $I_0$ be the graph supplied as input to $\Fix$. Then $b_{ss'}^T L_{I\setminus D}^+ b_{ss'}\ge (1 - \ep) b_{ss'}^T L_{I_0\setminus D}^+ b_{ss'}$
\item (Normalized potentials of deferred edges are not too high on average)

\begin{align*}
&\frac{1}{b_{ss'}^T L_{I\setminus D}^+ b_{ss'}}\left(\left(\sum_{\{u,v\}\in A} b_{ss'}^T L_{I\setminus D}^+ (b_{su} + b_{sv})\right) + \left(\sum_{\{u,v\}\in B} b_{ss'}^T L_{I\setminus D}^+ (b_{us'} + b_{vs'})\right)\right)\\
&\le \tilde{O}(\rho^3\ep^{-2}(\Delta^{I_0\setminus D}(S,S') + \Delta^{I_0\setminus D}(S',S) + |D|))
\end{align*}

\end{itemize}
\end{proposition}

\lineend

The first condition of the above proposition states that the value that $\Fix$ needs to preserve ($b_{ss'}^T L_{I\setminus D}^+ b_{ss'} = \frac{1}{c^{I\setminus D}(S,S')}$) is in fact similar to its value at the beginning of the algorithm. This is not enough, however, because the deferred edges have not been conditioned on before Line \ref{line:contract-remaining}. The second condition of Proposition \ref{prop:end-state} ensures that contracting or deleting all but $\tilde{O}((\Delta^{I\setminus D}(S,S') + \Delta^{I\setminus D}(S',S) + |D|)\ep^{-1})$ of the deferred edges does not decrease the main objective by more than a factor of $1 - \ep$. Setting the remaining edges to be $F'$ shows the desired result.

We now show Proposition \ref{prop:end-state}:

\begin{proof}[Proof of Proposition \ref{prop:end-state}]

\underline{Number of ``If'' block visits.} Start by bounding the number of times the ``If'' block in the algorithm $\Fix$ is entered. Each time the ``If'' statement is entered, Line \ref{line:remove-w} removes $W$ from $F$. By construction, $|W| \ge |F|/(2i_{\max}) = |F|/(2\log n)$. Therefore, this removal can only take place $(\log |F|)(2\log n)\le 2\log^2 n$ times, so the ``If'' statement only executes $2\log^2 n$ times over the course of $\Fix$.

\underline{Verifying stable oracle input conditions and size of $W$.} After Line \ref{line:remove-high-levcng}, $W$ only contains edges for which $|\texttt{lev}_{I\setminus D}(e) - \texttt{lev}_{I/(S,S')}(e)|\le 1/32 + 2(1/64)\le 1/16$. Therefore, $W$ satisfies the ``Bounded leverage score difference'' condition. $W$ satisfies one of the ``Narrow potential neighborhood'' conditions by definition of the $X_i$s and $Y_i$s.

Now, we lower bound the size of $W$ after Line \ref{line:remove-high-levcng}. Let $U_S$ and $U_{S'}$ be the vertices with $(I\setminus D)/(S,S')$-normalized potentials greater than $p$ and less than $1 - p$ respectively, with $s\gets 0$ and $s'\gets 1$. By Proposition \ref{prop:deg-pot-bound} with $\gamma\gets 1 - p$,

$$\Delta^{I\setminus D}(U_S,S')\le \Delta^{I\setminus D}(S,S')/p$$

and

$$\Delta^{(I\setminus D)/S'}(U_{S'},S)\le \Delta^{I\setminus D}(U_{S'},S)\le \Delta^{I\setminus D}(S',S)/p$$

Every edge in $W$ has an endpoint in $U_S\cap U_{S'}$ by the bucketing definition. Applying Lemma \ref{lem:first-order-deg} twice and the first inequality of Remark \ref{rmk:lev-cng} shows that

$$\sum_{e\in W} \texttt{lev}_{I\setminus D}(e) - \texttt{lev}_{(I\setminus D)/(S,S')}(e) \le 32(\Delta^{I\setminus D}(S,S') + \Delta^{I\setminus D}(S',S))/p$$

By Lemma \ref{lem:del-first-order-deg} and Remark \ref{rmk:lev-cng},

$$\sum_{e\in W} \texttt{lev}_{(I\setminus D)/(S,S')}(e) - \texttt{lev}_{I/(S,S')}(e) \le |D|$$

Therefore, by Rayleigh monotonicity,

$$\sum_{e\in W} |\texttt{lev}_{I\setminus D}(e) - \texttt{lev}_{I/(S,S')}(e)| \le 32(\Delta^{I\setminus D}(S,S') + \Delta^{I\setminus D}(S',S) + |D|)/p$$

By the ``If'' condition,

$$|W|\ge 4096(\Delta^{I\setminus D}(S,S') + \Delta^{I\setminus D}(S',S) + |D|)/p$$

before Line \ref{line:remove-high-levcng}. In particular,

$$\sum_{e\in W} |\texttt{lev}_{I\setminus D}(e) - \texttt{lev}_{I/(S,S')}(e)| \le |W|/128$$

which means that Line \ref{line:remove-high-levcng} can only remove $(|W|/128)/(1/32) \le |W|/4$ edges from $W$. Therefore, Line \ref{line:remove-high-levcng} only decreases the size of $W$ by a factor of $3/4$. In particular, $|W| \ge \tilde{\Omega}(\rho^3 \ep^{-2}(\Delta^{I\setminus D}(S,S') + \Delta^{I\setminus D}(S',S) + |D|)/p)$ after Line \ref{line:remove-high-levcng}.

\underline{Concentration in ``Else''-only intervals.} Now, define the functions

$$x_0(H) := b_{ss'}^T L_{(H\setminus D)/(S,S')} b_{ss'} = 1/c^{H\setminus D}(S,S')$$

$$x_1(H) := \delta_{S,S'}(H\setminus D)$$

$$x_2(H) := \delta_{S',S}(H\setminus D)$$

and

$$y_0(H) := \left(\sum_{\{u,v\}\in A} b_{ss'}^T L_{H\setminus D}^+ (b_{su} + b_{sv})\right) + \left(\sum_{\{u,v\}\in B} b_{ss'}^T L_{H\setminus D}^+ (b_{us'} + b_{vs'})\right)$$

All of these functions are electrical functions for the graph $H$, as they are preserved under splitting edges in series and in parallel.

We now check that Proposition \ref{prop:stable-concentration} can be applied between ``If'' block executions to show that $x_0$, $x_1$, $x_2$, and $y_0$ do not change by more than a $(1 + \ep/(2\log^2 n))$-factor over each ``Else''-only interval. Line \ref{line:select-1} is equivalent picking uniformly random edges from $Z'$ without replacement $K$ times. Let $\Delta := \Delta^{I\setminus D}(S,S') + \Delta^{I\setminus D}(S',S) + |D|$. By Propositions \ref{prop:delta-stability},\ref{prop:deferred-stability}, and \ref{prop:main-stability}, each of the selected edges is sampled from a set that is a $(\tilde{O}(\rho),\tilde{O}(\rho (\Delta/p)),0)$-stable subset of $W$ for the $x_i$s and $(\tilde{O}(\rho),\tilde{O}(\rho (\Delta/p)),r_{min}/n^4)$-stable for $y_0$. Therefore, we may apply Proposition \ref{prop:stable-concentration} with

\begin{align*}
\rho_E(\{x_a(H)\}_a) &\gets \tilde{O}(\rho (\Delta^{H\setminus D}(S,S') + \Delta^{H\setminus D}(S',S) + |D|)/p))\\ 
&= \tilde{O}(\rho (\delta_{S,S'}(H\setminus D) + \delta_{S',S}(H\setminus D))/(p b_{ss'}^T L_{(H\setminus D)/(S,S')}^+ b_{ss'}))\\
&= \tilde{O}(\rho (x_1(H) + x_2(H))/(p x_0(H)))\\
\end{align*}

$$\rho_L \gets \tilde{O}(\rho)$$

$$\delta \gets r_{min}/n^4$$

and

$$\sigma \gets \tilde{\Omega}(\rho^3/\ep^2)$$

because $\rho_E$ is a $1$-multlipschitz function in its inputs. In order to apply Proposition \ref{prop:stable-concentration}, we need to bound $\tau_E$. By the ``If'' statement and the ``Size of $W$'' bound earlier in this proof, $\tau_E\le \ep/(100\log^2 n)$ and $\sqrt{\gamma \ell_1}\rho_L^{3/2}/\sqrt{\sigma}\le \ep/(100\log^2 n)$. Therefore, Proposition \ref{prop:stable-concentration} implies that each $x_i$ function changes by at most a factor of $(1 + \ep/(8\log^2 n))$ during each ``Else''-only interval. Furthermore, $y_0$ changes by at most an $(1 + \ep/(8\log^2 n))$ factor during each interval, along with an additive change of at most $\tilde{O}(\rho r_{min}|F|/n^4)\le r_{min}/n^3$. This is the desired change in each ``Else''-interval.

\textbf{Main objective.} Each ``Else'' interval causes $x_0(H) = b_{ss'}^T L_{H\setminus D}^+ b_{ss'}$ to change by a factor of at most $(1 + \ep/(8\log^2 n))$. By ``Number of If block visits,'' there are at most $2\log^2 n$ of these intervals. Therefore, $x_0(I)\ge (1 - \ep/(8\log^2 n))^{2\log^2 n}x_0(I_0)\ge (1 - \ep)x_0(I_0)$. In particular, $b_{ss'}^T L_{I\setminus D}^+ b_{ss'} \ge (1 - \ep) b_{ss'}^T L_{I_0\setminus D}^+ b_{ss'}$, as desired.

\textbf{Normalized potentials are not too high on average.} Each ``Else''interval causes the quantity $\Delta^{I\setminus D}(S,S') + \Delta^{I\setminus D}(S',S)$ to increase by a factor of at most $(1 + \ep/(8\log^2 n))$ by Proposition \ref{prop:stable-concentration} applied to all of the $x_i$s. Therefore, the total increase over the course of the entire algorithm is at most a factor of $(1 + \ep)\le 2$, since there are at most $2\log^2 n$ ``Else'' intervals. Therefore, $I$ in the bound of the ``If'' statement can be replaced with $I_0$ with at most a factor of 2 increase in the value.

Each ``If'' statement adds at most $\tilde{O}(\rho^3\ep^{-2}(\Delta^{I_0\setminus D}(S,S') + \Delta^{I_0\setminus D}(S',S) + |D|)2^i)$ edges to $A$ or $B$. Each of these edges contributes at most $2^{-i-2}(b_{ss'}^T L_{I\setminus D}^+ b_{ss'})\le 2^{-i-1}(b_{ss'}^T L_{I_0\setminus D}^+ b_{ss'})$ (by Proposition \ref{prop:stable-concentration}) to the sum, so each ``If'' block increases $y_0$ by at most 

\begin{align*}
&\tilde{O}(\rho^3\ep^{-2}(\Delta^{I_0\setminus D}(S,S') + \Delta^{I_0\setminus D}(S',S) + |D|)2^i) 2^{-i-1}(b_{ss'}^T L_{I_0\setminus D}^+ b_{ss'})\\
&\le \tilde{O}(\rho^3\ep^{-2}(\Delta^{I_0\setminus D}(S,S') + \Delta^{I_0\setminus D}(S',S) + |D|))(b_{ss'}^T L_{I_0\setminus D}^+ b_{ss'})\\
\end{align*}

additively. Each ``Else'' statement increases the value of $y_0$ by at most a factor of $(1 + \ep/(8\log^2 n))$ along with an additive increase of at most $r_{min}/n^3$. Therefore, the value of $y_0$ immediately before Line \ref{line:contract-remaining} is at most $\tilde{O}(\rho^3\ep^{-2}(\Delta^{I_0\setminus D}(S,S') + \Delta^{I_0\setminus D}(S',S) + |D|))(b_{ss'}^T L_{I_0\setminus D}^+ b_{ss'})$, as desired. Dividing both sides by $(b_{ss'}^T L_{I_0\setminus D}^+ b_{ss'})$ and using the concentration of $x_0$ gives the desired result.
\end{proof}

We now use this proposition, along with the following result:

\begin{restatable}{proposition}{propcondpotbound}\label{prop:cond-pot-bound}
Consider two disjoint sets of vertices $X$ and $Y$ in a graph $G$. Let $G' = G/(X,Y)$, with $x$ and $y$ the identifications of $X$ and $Y$ respectively. Let $A$ and $B$ be sets of edges for which both endpoints have normalized $L_{G'}^+ b_{xy}$ potential at most $\gamma$ and at least $1-\gamma$ respectively for some $\gamma\in (0,1/2)$. Arbitrarily contract and delete edges in $A$ and $B$ in $G$ to obtain the graph $H$. Then

$$c^H(X,Y)\le \frac{1}{(1-\gamma)^2}c^G(X,Y)$$
\end{restatable}

to show Lemma \ref{lem:oracle-fix} using the intuition given earlier:

\begin{proof}[Proof of Lemma \ref{lem:oracle-fix}]

\underline{Correctness and runtime of conditioning.} The conditioning process given in $\Fix$ can be written in the form described in Proposition \ref{prop:split-partial}. Therefore, by Proposition \ref{prop:split-partial}, its output is equivalent in distribution to sampling from $I[F]$. Furthermore, Line \ref{line:split-condition} can be implemented using the sample $J$ in constant time per edge by Proposition \ref{prop:precomp-partial}.

\textbf{Conductance.} By the ``Main objective'' condition, the $s-s'$ effective resistance in $I\setminus D$ is at least $(1 - \ep)$ times its original value right before Line \ref{line:contract-remaining}. Contracting or deleting $(A\cup B)\setminus F'$ in $I$ only decreases the resistance by a factor of at most $(1 - \ep/2)^2\ge (1 - \ep)$ by Proposition \ref{prop:cond-pot-bound}. By Rayleigh monotonicity, deleting $F'$ only increases the $s-s'$ resistance. The total decrease due to all of these changes is at most  a factor of $(1 - \ep)^2\ge 1 - O(\ep)$. Taking reciprocals and using the fact that $c^{J\setminus D\setminus F'}(S,S') = 1/(b_{ss'}^T L_{J\setminus D\setminus F'}^+ b_{ss'})$ yields the desired result.

\textbf{Size.} By Markov's Inequality and the ``Normalized potentials of deferred edges are not too high on average'' guarantee of Proposition \ref{prop:end-state}, only $\tilde{O}(\rho^3\ep^{-3}(\Delta^{I\setminus D}(S,S') + \Delta^{I\setminus D}(S',S) + |D|))$ edges in $A\cup B$ can have any endpoint with normalized potential in the interval $[\ep/2, 1-\ep/2]$. Therefore, $|F'|\le \tilde{O}(\rho^3\ep^{-3}(\Delta^{I\setminus D}(S,S') + \Delta^{I\setminus D}(S',S) + |D|))$, as desired.

\textbf{Runtime.} Each ``While'' loop iteration takes $\tilde{O}(m + \mc{T}(\Oracle))$ time, so it suffices to bound the number of ``While'' loop iterations. Each ``If'' block ``While'' loop iteration reduces the size of $|F|$ by a factor of $1 - 1/(i_{\max})\le 1 - 1/(2\log n)$, so only $O(\log^2 n)$ such iterations can occur. Each ``Else'' block ``While'' loop iteration decreases the size of $F$ in expectation by at least $K(|F|)/4$ by Proposition \ref{prop:split-cond}. By Chernoff bounds, for $K(|F|)\ge \text{polylog}(n)$, $F$ decreases in size by at least $K(|F|)/8$ with high probability. Therefore, each ``Else'' iteration reduces the size of $|F|$ by a factor of $(1 - \frac{K(|F|)}{|F|})$ with high probability. Therefore, only $\tilde{O}(|F|/K(|F|))$ ``Else'' iterations can occur. All ``While'' loop iterations either call ``If'' or ``Else,'' so we have finished the runtime analysis.
\end{proof}

\subsection{Warmup: A $(\text{polylog}(n),1)$-stable oracle}

We now give a stable oracle $\SlowOracle$ that suffices for proving Lemma \ref{lem:slow-fix}. $\SlowOracle$ is similar in concept to picking the minimum energy edge in $W$. However, it needs to do so for multiple functions simultaneously and for more complicated functions than effective resistances. To cope with these complexities, we exploit Theorem \ref{thm:harsha} in place of the fact that the sum of the energies on edges is the effective resistance.

\begin{algorithm}[H]
\SetAlgoLined
\DontPrintSemicolon
\caption{$\SlowOracle(I,S,S',D,A,B,W)$, never executed}

    Return all edges in $W$ that satisfy all of the inequalities in the ``$S-S'$-normalized degree change stability,'' ``Deferred endpoint potential change stability,'' and ``Main objective change stability'' guarantees of $\Oracle$ with $\rho = 400 \log (\beta n)$\;

\end{algorithm}

\begin{proposition}\label{prop:slow-oracle}
There is a $(400\log (\beta n),1)$-stable oracle $\SlowOracle$.
\end{proposition}

\begin{proof}
\textbf{Everything besides size of $Z$.} Let $Z\gets \SlowOracle(I,S,S',D,A,B,W)$. By construction, all edges in $Z$ (in particular $f_0$) satisfy the ``$S-S'$-normalized degree change stability,'' ``Deferred endpoint potential change stability,'' and ``Main objective change stability'' guarantees of $\Oracle$. The ``Midpoint potential stability'' and ``Leverage score stability'' guarantees follow from the fact that $K = 1$.

\textbf{Size of $Z$.} In each of the ``Conductance term stability'' quantities, the edge $e$ has $s'$ as an endpoint in the graph $(I_0\setminus D)/(S,S')$. Furthermore, in the ``Deferred endpoint potential change stability'' quantity, $b_{su} + b_{sv}$ has one source that is the same as $b_{ss'}$'s source. Similarly, $b_{us'} + b_{vs'}$ and $b_{ss'}$ have the same sink. Therefore, Theorem \ref{thm:harsha} applies with $\tau = \beta^{12} n^{12}$ and shows that

\begin{align*}
&\sum_{f\in W} \sum_{w\in S'} (b_{sw}^T L_{(I_0\setminus D)/S}^+ b_{sw}) \sum_{e\in \partial_{I_0}w} \frac{|b_{ss'}^T L_{(I_0\setminus D)/(S,S')}^+ b_f| |b_f^T L_{(I_0\setminus D)/(S,S')}^+ b_e|}{r_fr_e}\\
&= \sum_{w\in S'} (b_{sw}^T L_{(I_0\setminus D)/S}^+ b_{sw}) \sum_{e\in \partial_{I_0}w} \sum_{f\in W} \frac{|b_{ss'}^T L_{(I_0\setminus D)/(S,S')}^+ b_f| |b_f^T L_{(I_0\setminus D)/(S,S')}^+ b_e|}{r_fr_e}\\
&\le \sum_{w\in S'} (b_{sw}^T L_{(I_0\setminus D)/S}^+ b_{sw}) \sum_{e\in \partial_{I_0}w} \left((12\log (n\beta))\frac{b_{ss'}^T L_{(I_0\setminus D)/(S,S')}^+ b_e}{r_e} + \frac{n^2r_{max}}{r_e\beta^{12} n^{12}}\right)\\
&\le \sum_{w\in S'} (b_{sw}^T L_{(I_0\setminus D)/S}^+ b_{sw}) \sum_{e\in \partial_{I_0}w} \left((24\log (n\beta))\frac{b_{ss'}^T L_{(I_0\setminus D)/(S,S')}^+ b_e}{r_e}\right)\\
&= 24\log(n\beta) \delta_{S,S'}(I\setminus D)\\
\end{align*}

where the second-to-last inequality follows from the fact that effective resistances in any graph are within a factor of $n^2\beta$ of one another. Similarly,

$$\sum_{f\in W} \sum_{w\in S} (b_{s'w}^T L_{(I_0\setminus D)/S}^+ b_{s'w}) \sum_{e\in \partial_{I_0}w} \frac{|b_{ss'}^T L_{(I_0\setminus D)/(S,S')}^+ b_f| |b_f^T L_{(I_0\setminus D)/(S,S')}^+ b_e|}{r_fr_e}\le 24\log(n\beta)\delta_{S',S}(I\setminus D)$$

$$\sum_{f\in W} \sum_{w\in S'}\frac{(b_{sw}^T L_{(I_0\setminus D)/S}^+ b_f)^2}{r_f} \sum_{e\in \partial_{I_0}w} \frac{b_{ss'}^T L_{(I_0\setminus D)/(S,S')}^+ b_e}{r_e}\le \delta_{S,S'}(I\setminus D)$$

$$\sum_{f\in W} \sum_{w\in S} \frac{(b_{s'w}^T L_{(I_0\setminus D)/S'}^+ b_f)^2}{r_f} \sum_{e\in \partial_{I_0}w} \frac{b_{ss'}^T L_{(I_0\setminus D)/(S,S')}^+ b_e}{r_e}\le \delta_{S',S}(I\setminus D)$$

\begin{align*}
&\sum_{f\in W}\left(\sum_{\{u,v\}\in A} \frac{|b_{ss'}^T L_{(I_0\setminus D)/(S,S')}^+ b_f| |b_f^T L_{(I_0\setminus D)/(S,S')}^+ (b_{su} + b_{sv})|}{r_f}\right)\\
&+ \sum_{f\in W} \left(\sum_{\{u,v\}\in B} \frac{|b_{ss'}^T L_{(I_0\setminus D)/(S,S')}^+ b_f| |b_f^T L_{(I_0\setminus D)/(S,S')}^+ (b_{us'} + b_{vs'})|}{r_f}\right)\\
&\le 12\log(\beta n)\sum_{\{u,v\}\in A} b_{ss'}^T L_{(I\setminus D)/(S,S')}^+ (b_{su} + b_{sv}) + 12\log(\beta n)\sum_{\{u,v\}\in B} b_{ss'}^T L_{(I\setminus D)/(S,S')}^+ (b_{us'} + b_{vs'}) + \frac{r_{min}}{n^6}
\end{align*}

$$\sum_{f\in W}\frac{(b_{ss'}^T L_{(I_0\setminus D)/(S,S')}^+ b_f)^2}{r_f}\le b_{ss'}^T L_{(I\setminus D)/(S,S')}^+ b_{ss'}$$

By Markov's Inequality, only $|W|/16$ edges in $W$ can violate any one of the six conditions tested in $\SlowOracle$. Therefore, $|Z|\ge |W| - 6|W|/16\ge |W|/2$, as desired.

\end{proof}

\begin{proof}[Proof of Lemma \ref{lem:slow-fix}]
Follows directly from Lemma \ref{lem:oracle-fix}, with $\SlowOracle$ substituted in for $\Oracle$ by Proposition \ref{prop:slow-oracle}.
\end{proof}

\newpage

\section{Efficient construction for the conductance concentration inequality}\label{sec:fast-fix}

To accelerate $\Fix$, we need to construct an almost-linear time $(m^{o(1)},|W|m^{-o(1)})$-stable oracle. To do this, we need to do the following:

\begin{itemize}
\item Compute a large subset of $W$ consisting of edges that respect all of the conditions of stable oracles in almost-linear time.
\item Show that this large subset of $W$ continues to satisfy the stable conditions even after conditioning on a significant fraction of $W$.
\end{itemize}

The first objective boils down to computing approximations to moments of Laplacian inner products. This can be done using techniques from streaming algorithms; for example \cite{AMS96,I06}. Specifically, we use Theorem \ref{thm:linear-sketch}. The second objective boils down to showing that a set of stable edges remains stable for many iterations. To show this, we use Theorem \ref{thm:localization}.

\subsection{Exploiting localization}

\subsubsection{Concentration preliminaries}

In this subsection, it is helpful to have a few concentration inequalities that allow us to control the $\ell_{\infty}$ norm of certain vectors. We apply Theorem \ref{thm:martingale-2} to obtain two concentration inequalities that will be applied directly:

\begin{restatable}{proposition}{proprandommatixcontrol}\label{prop:random-matrix-control}
Let $\{M^{(k)}\}_k\in \mathbb{R}^{n\times n}$ be a sequence of symmetric, nonnegative random matrices, $\{Z^{(k)}\}_k\in \{0,1\}^n$, and $S^{(k)}\subseteq [n]$ with the following properties:

\begin{itemize}
\item For all $i\in [n]$, $\sum_{j=1,j\ne i}^n M_{ij}^{(0)} \le \sigma$ where $\sigma\le \sigma_0$. Furthermore, $S^{(0)} = \emptyset$.
\item The random variables $\{Z^{(k)}\}_k$ are defined by making $Z^{(k+1)}$ the indicator of a uniformly random choice $w^{(k+1)}\in [n]\setminus S^{(k)}$. Let $S^{(k+1)} := S^{(k)}\cup \{w^{(k+1)}\}$.
\item For all $i,j\in [n]$ and $k$, $M_{ij}^{(k+1)} \le M_{ij}^{(k)} + \gamma \sum_{l=1}^n M_{il}^{(k)} Z_l^{(k+1)} M_{lj}^{(k)}$.
\end{itemize}

With probability at least $1 - 1/n^8$,

$$\sum_{j\ne i, j\notin S^{(k)}} M_{ij}^{(k)} \le \sigma_1$$

for all $i\notin S^{(k)}$ and all $k\le n/2$.
\end{restatable}

We also need a bound on how $M$ affects a random vector $v$:

\begin{restatable}{proposition}{proprandomvectorcontrol}\label{prop:random-vector-control}
Let $\{M^{(k)}\}_k \in \mathbb{R}^{n\times n}$ be a sequence of symmetric, nonnegative random matrices, $\{v^{(k)}\}_k \in \mathbb{R}^n$ be a sequence of nonnegative random vectors, $\{Z^{(k)}\}_k \in \{0,1\}^n$, and $\{S^{(k)}\}_k \subseteq [n]$ with the following properties:

\begin{itemize}
\item For all $i\in [n]$ and all $k$, $\sum_{j\ne i, j\notin S^{(k)}} M_{ij}^{(k)} \le \sigma_1$.
\item The random variables $\{Z^{(k)}\}_k$ are defined by making $Z^{(k+1)}$ the indicator of a uniformly random choice $w^{(k+1)}\in [n]\setminus S^{(k)}$. Let $S^{(k+1)} := S^{(k)}\cup \{w^{(k+1)}\}$.
\item For all $i\in [n]$, $v_i^{(0)} \le \tau$.
\item For all $i\in [n]$ and $k$, $v_i^{(k+1)} \le v_i^{(k)} + \gamma \sum_{l=1}^n M_{il}^{(k)}Z_l^{(k+1)}(v_l^{(k)} + v_i^{(k)})$.
\end{itemize}

With probability at least $1 - 1/n^8$,

$$v_i^{(k)} \le \tau_1$$

for all $i\notin S^{(k)}$ and all $k\le n/2$.
\end{restatable}

We prove both of these propositions in Appendix \ref{sec:fast-concentration-appendix}.

\subsubsection{Flexible functions}

We now discuss how to exploit the concentration inequalities from the previous subsubsection to make stable sets of edges remain stable for an almost-linear number of iterations. The bounds below are motivated by applications of Sherman-Morrison.

\begin{definition}[Flexible families of functions]\label{def:flexible-families}
Let $G'$ be a graph and consider any minor $H$ of $G'$. Let $X\subseteq E(G')$ and consider a family of electrical functions $\{g_e\}_{e\in X}$ on minors of $G'$ along with a function $\phi$ that maps minors $H$ of $G'$ to graphs with a subset of the edges in $H$. This family is called \emph{flexible} if for any graph $H$ and any edges $e,f\in X\cap E(H)$, all of the following hold:

\begin{itemize}
\item ($\phi$ commutes with modifications) For any minor $H$ of $G'$ and any edge $f\in E(H)$, $\phi(H/f) = \phi(H)/f$ and $\phi(H\backslash f) = \phi(H)\backslash f$.
\item (Contractions) $|g_e(H/f) - g_e(H)| \le \frac{|b_e^T L_{\phi(H)}^+ b_f|}{\sqrt{r_e}\sqrt{r_f}} \frac{3}{(\texttt{lev}_{\phi(H)}(f))^2} (g_f(H) + g_e(H))$
\item (Deletions) $|g_e(H\backslash f) - g_e(H)| \le \frac{|b_e^T L_{\phi(H)}^+ b_f|}{\sqrt{r_e}\sqrt{r_f}} \frac{3}{(1 - \texttt{lev}_{\phi(H)}(f))^2} (g_f(H) + g_e(H))$
\end{itemize}

\end{definition}

\subsubsection{Crudely controlling flexible functions}

We now make use of Theorem \ref{thm:localization} to find a set of edges for which flexible functions do not increase much if one splits and conditions for a linear number of iterations.

In the $\Fix$ algorithm, edges were split in one of two ways before conditioning on them. This ensured that they had a leverage score that was bounded away from both 0 and 1. In order to decide which way to split an edge, the algorithm must approximately know its effective resistance. Naively, this requires recomputing approximate effective resistances during each iteration. One can avoid recomputation for a linear number of iterations, however, by showing the following:

\begin{proposition}\label{prop:very-stable}
Given any set $X\subseteq E(G')$ in some graph $G'$ along with a graph $G''$ on a subset of the edges of $G'$ with the following property:

\begin{itemize}
\item Each edge $e\in X$ has $\texttt{lev}_{G''}(e)\in [1/4,3/4]$
\end{itemize}

there is a set $Y\gets \VeryStable_{G'}(G'',X)$ with $Y\subseteq X$ and the following additional properties:

\begin{itemize}
\item (Size) $|Y| \ge |X|/\text{polylog}(n)$ with probability at least $\sigma/(16O(\log^2 n))$.

\item (Leverage score bound) Pick a subset $Y_0\subseteq Y$ and sample a random sequence of edges $f_0,f_1,\hdots,f_{|Y_0|/2}$ without replacement. For any integer $0\le i\le |Y_0|/2$, let $G_i''$ denote the graph obtained by arbitrarily deleting/contracting the edges $f_0,f_1,\hdots,f_{i-1}$ in $G''$. Then, with probability at least $1 - 1/n^6$,

$$|\texttt{lev}_{G_i''}(e) - \texttt{lev}_{G''}(e)|\le \frac{1}{8}$$

for all $e\in Y_0\setminus \{f_0,f_1,\hdots,f_{i-1}\}$.

\item (Flexible function bound) For any flexible family of functions $(\{g_e\}_{e\in X},\phi)$ with $G'' = \phi(G')$,

$$\max_{e\in Y_0\setminus \{f_0,f_1,\hdots,f_{i-1}\}} g_e(G_i') \le 2\max_{e\in Y_0} g_e(G')$$

with probability at least $1 - 1/n^6$.

\item (Runtime) The algorithm takes $\tilde{O}(m + \mc T(\ColumnApxPreproc) + |X|\mc T(\ColumnApx))$ time.
\end{itemize}

\end{proposition}

To apply this proposition, one just needs to run it $O((\log^3 n)/\sigma)$ times to obtain the desired set with high probability. We encourage the reader to ignore the approximation aspect of the $v_e$s this is only included for efficiency purposes later on:

\begin{algorithm}[H]
\SetAlgoLined
\DontPrintSemicolon

    \KwData{An ambient graph $G'$, a graph $G''$ that stable functions are ``defined'' in (is the image of $\phi$), and a set of edges $X\in E(G')$ for possible conditioning}

    \KwResult{The set $Y\subseteq X$}

    $Z\gets $ subset of edges $e\in X$, with each edge independently added to $Z$ with probability $\sigma/(8O(\log^2 n))$\;

    $\ColumnApxPreproc(G'',Z)$\;

    $Y\gets \emptyset$\;

    \ForEach{$e\in Z$}{

        \tcp{multiplicative 2-approximation to the quantity $\sum_{f\in Z, f\ne e} \frac{|b_e^T L_{G''}^+ b_f|}{\sqrt{r_e}\sqrt{r_f}}$}

        $v_e\gets \ColumnApx(e)$\;

        Add $e$ to $Y$ if $v_e \le \sigma$\;

    }

    \Return{$Y$}

\caption{$\VeryStable_{G'}(G'',X)$}
\end{algorithm}

\begin{proof}
\textbf{Size.} By Markov's Inequality and Theorem \ref{thm:localization} applied to the vector $w = 1_X$, the subset $X_0\subseteq X$ consisting of edges $e\in X$ with $\sum_{f\in X} \frac{|b_e^T L_{G''}^+ b_f|}{\sqrt{r_e}\sqrt{r_f}} \le 2O(\log^2 n)$ has size at least $|X|/2$. For any edge $e\in X_0$,

\begin{align*}
\textbf{E}_Z\left[\sum_{f\in Z, f\ne e} \frac{|b_e^T L_{G''}^+ b_f|}{\sqrt{r_e}\sqrt{r_f}}\mid e\in Z\right] &= \textbf{E}_Z\left[\sum_{f\in Z, f\ne e} \frac{|b_e^T L_{G''}^+ b_f|}{\sqrt{r_e}\sqrt{r_f}}\right]\\
&= \frac{\sigma}{8O(\log^2 n)} \sum_{f\in Z, f\ne e} \frac{|b_e^T L_{G''}^+ b_f|}{\sqrt{r_e}\sqrt{r_f}}\\
&\le \frac{\sigma}{4}\\
\end{align*}

where the first equality follows from the fact that edges in $X$ are added to $X_0$ independently. Therefore, for $e\in X_0$, $\textbf{E}[v_e|e\in Z]\le \sigma/2$ and

\begin{align*}
\Pr_Z[e\in Y] &= \Pr_Z[e\in Y, e\in Z]\\
&= (1 - \Pr_Z[e\notin Y\mid e\in Z])\Pr_Z[e\in Z]\\
&= \left(1 - \Pr_Z\left[v_e > \sigma\mid e\in Z\right]\right)\Pr_Z[e\in Z]\\
&\ge \frac{1}{2}\Pr_Z[e\in Z]\\
&= \frac{\sigma}{16O(\log^2 n)}
\end{align*}

where the inequality follows from Markov. Since $|Y\cap X_0|\le |X_0|$,

$$\Pr_Z\left[|Y\cap X_0| > \frac{\sigma|X_0|}{(32O(\log^2 n))}\right]|X_0| + \frac{\sigma|X_0|}{(32O(\log^2 n))} \ge \textbf{E}_Z[|Y\cap X_0|] \ge \frac{\sigma|X_0|}{(16O(\log^2 n))}$$

and

$$\Pr_Z\left[|Y| > \frac{\sigma|X|}{(64O(\log^2 n))}\right]\ge \frac{\sigma}{(32O(\log^2 n))}$$

thus completing the size bound.

\underline{\textbf{Leverage score bound} and continued electrical flow sum bound.} We now show inductively that for all $k\le |Y_0|/2$ and $e\in Y_0\setminus \{f_0,f_1,\hdots,f_{k-1}\}$, both

\begin{align}
\sum_{f\in Y_0\setminus \{e,f_0,f_1,\hdots,f_{k-1}\}} \frac{|b_e^T L_{G_k''}^+ b_f|}{\sqrt{r_e}\sqrt{r_f}} \le 4\sigma \label{eqn:local}
\end{align}

and

\begin{align}
|\texttt{lev}_{G_k''}(e) - \texttt{lev}_{G''}|\le \frac{1}{8} \label{eqn:lev}
\end{align}

where $G_k''$ is a graph obtained by contracting or deleting the edges $f_0,f_1,\hdots,f_{k-1}$ in $G''$.

We start by checking the base case. The base case for (\ref{eqn:local}) follows immediately from the approximation lower bound on $v_e$. The base case for (\ref{eqn:lev}) follows from the input condition to Proposition \ref{prop:very-stable}.

Now, we continue on to the inductive step. By the input condition,

$$\frac{1}{4}\le \texttt{lev}_{G''}(e) \le \frac{3}{4}$$

for all $e\in Y_0$. Let $M_{ef}^{(k)} \gets \frac{|b_e^T L_{G_k''}^+ b_f|}{\sqrt{r_e}\sqrt{r_f}}$, $\gamma \gets 8$, and $\sigma\gets 2\sigma$ (last one uses the lower bound for $v_e$). Notice that

$$M_{ef}^{(k+1)} \le M_{ef}^{(k)} + \sum_{g\in Y_0} \frac{1}{\min(\texttt{lev}_{G_k''}(g),\texttt{nonlev}_{G_k''}(g))}Z_g^{(k+1)} M_{eg}^{(k)}M_{gf}^{(k)}$$

by Sherman-Morrison and the triangle inequality. $Z^{(k+1)}$ is the indicator for the edge $f_{k+1}$. By the inductive assumption and the triangle inequality, $\min(\texttt{lev}_{G_k''}(g),\texttt{nonlev}_{G_k''}(g)) \ge 1/8 = 1/\gamma$. Therefore, Proposition \ref{prop:random-matrix-control} applies. It shows that for all $k\le |Y_0|/2$ and all $e\ne f_0,\hdots,f_{k-1}$,

$$\sum_{f\in Y_0\setminus\{e,f_0,\hdots,f_{k-1}\}} \frac{|b_e^T L_{G_k''}^+ b_f|}{\sqrt{r_e}\sqrt{r_f}} \le \sigma_1$$

with probability at least $1 - 1/n^8$. This is (\ref{eqn:local}). Now, we bound how much $e$'s effective resistance can change using martingale concentration. The above inequality implies that

$$\left|\textbf{E}_{G_{k+1}''}\left[\frac{b_e^T L_{G_{k+1}''}^+ b_e}{r_e} \mid G_k''\right] - \frac{b_e^T L_{G_k''}^+ b_e}{r_e}\right|\le 8 \textbf{E}_{G_{k+1}''}\left[\frac{(b_e^T L_{G_k''}^+ b_{f_k})^2}{r_e r_{f_k}}\right] \le \frac{16}{|Y_0|}\sigma_1^2$$

\begin{align*}
\textbf{Var}_{G_{k+1}''}\left[\frac{b_e^T L_{G_{k+1}''}^+ b_e}{r_e} \mid G_k''\right] &= \textbf{Var}_{G_{k+1}''}\left[\frac{b_e^T L_{G_{k+1}''}^+ b_e}{r_e} - \frac{b_e^T L_{G_k''}^+ b_e}{r_e} \mid G_k'\right]\\
&\le 64\textbf{E}_{G_{k+1}''}\left[\frac{(b_e^T L_{G_k''}^+ b_{f_k})^4}{r_e^2 r_{f_k}^2}\mid G_k''\right]\\
&\le \frac{128\sigma_1^4}{|Y_0|}
\end{align*}

$$\left|\frac{b_e^T L_{G_{k+1}''}^+ b_e}{r_e} - \frac{b_e^T L_{G_k''}^+ b_e}{r_e}\right| \le 8\sigma_1^2$$

if $e\ne f_k$. By Theorem \ref{thm:martingale-2}, $\Pr[|\texttt{lev}_{G_k''}(e) - \texttt{lev}_{G''}(e)| > 160(\log n)\sigma_1^2]\le 1/n^{10}$. Since $\texttt{lev}_{G''}(e)\in [1/4,3/4]$ and $160(\log n)\sigma_1^2 < 1/8$, $|\texttt{lev}_{G_k''}(e) - \texttt{lev}_{G''}(e)|\le 1/8$ for all $k$ and all $e\ne f_0,f_1,\hdots,f_{k-1}$ with probability at least $1 - n^2/n^{10} = 1 - 1/n^8$. This verifies the inductive hypothesis and proves (\ref{eqn:lev}).

\textbf{Flexible function bound.} Set $M_{ef}^{(k)} \gets \frac{|b_e^T L_{G_k''}^+ b_f|}{\sqrt{r_e}\sqrt{r_f}}$, $v_e^{(k)} \gets g_e(G_k')$, $\gamma\gets 300$, $\sigma_1\gets \sigma_1$ (using (\ref{eqn:local})), and $\tau\gets \max_{e\in Y_0} g_e(G')$. Notice that

$$M_{ef}^{(k)} = \frac{|b_e^T L_{\phi(G_k')}^+ b_f|}{\sqrt{r_e}\sqrt{r_f}}$$

as well by the ``$\phi$ commutes with modifications'' property of $\phi$. By definition of flexibility and (\ref{eqn:lev}),

\begin{align*}
v_e^{(k+1)} &\le v_e^{(k)} + \sum_{f\in Y_0} \frac{3}{\min(\texttt{lev}_{G_k''}(f), 1 - \texttt{lev}_{G_k''}(f))}M_{ef}^{(k)} Z_f^{(k)} (v_f^{(k)} + v_e^{(k)})\\
&\le v_e^{(k)} + \sum_{f\in Y_0} \gamma M_{ef}^{(k)} Z_f^{(k)} (v_f^{(k)} + v_e^{(k)})\\
\end{align*}

In particular, Proposition \ref{prop:random-vector-control} applies and shows that

$$v_e^{(k)}\le \tau_1$$

for all $e\in Y_0\setminus \{f_0,f_1,\hdots,f_{k-1}\}$, as desired.

\textbf{Runtime.} $\ColumnApx$ is called at most $|X|$ times and $\ColumnApxPreproc$ is called once.
\end{proof}

\subsection{A $(m^{o(1)},|W|m^{-o(1)})$-stable oracle that runs in almost-linear time ($\FastOracle$) given fast approximations to certain quantities}

Now, we implement $\FastOracle$ modulo some subroutines that efficiently return approximations to certain quantities. We start by intuitively discussing how to make each of the quantities that $\FastOracle$ needs to control not change for multiple iterations.

The ``Bounded leverage score difference'' condition ensures that when an edge is split, it can be split in the same direction for all graphs whose edge resistances are between those in $I\setminus D$ and $I/(S,S')$. Specifically, in all graphs between $I\setminus D$ and $I/(S,S')$, splitting an edge $e$ in one particular direction ensures that its leverage score is bounded away from 0 and 1.

\subsubsection{Leverage score stability intuition}

This bound follows immediately from returning a subset of $\VeryStable(I\setminus D,\VeryStable(I/(S,S'),W))$. The ``Leverage score bound'' of Proposition \ref{prop:very-stable} yields the desired result.

\subsubsection{Midpoint potential stability intuition}\label{subsubsec:midpoint}

By Theorem \ref{thm:harsha}, for any $\{u,v\}\in W$,

$$\sum_{f\in W} \frac{|b_{ss'}^T L_{(I\setminus D)/(S,S')}^+ b_f| |b_f^T L_{(I\setminus D)/(S,S')}^+ (b_{su} + b_{sv})/2|}{r_f}\le b_{ss'}^T L_{(I\setminus D)/(S,S')}^+ (b_{su} + b_{sv})/2 \le O(\log (n/p)) p$$

In particular, using the birthday paradox (as in $\VeryStable$), one can find a $1/\text{polylog(n/p)}$ fraction $W'\subseteq W$ for which

$$\frac{|b_{ss'}^T L_{(I\setminus D)/(S,S')}^+ b_f| |b_f^T L_{(I\setminus D)/(S,S')}^+ (b_{su} + b_{sv})/2|}{r_f} \le \frac{p}{\log n}$$

for all pairs of distinct edges $e = \{u,v\}, f\in W'$. Passing $W'$ through $\VeryStable$ makes it so that the following flexible functions of $H$ do not change by more than a factor of 2 over the course of many edge contractions or deletions:

$$g_f^{(0)}(H) := \frac{|b_{ss'}^T L_{(H\setminus D)/(S,S')}^+ b_f|}{\sqrt{r_f}}$$

$$\phi^{(0)}(H) := (H\setminus D)/(S,S')$$

$$g_f^{(1),X,s}(H) := \sum_{\{u,v\}\in X\setminus f} \frac{|b_f^T L_{(H\setminus D)/(S,S')}^+ (b_{su} + b_{sv})/2|}{\sqrt{r_f}}$$

$$g_f^{(1),X,s'}(H) := \sum_{\{u,v\}\in X\setminus f} \frac{|b_f^T L_{(H\setminus D)/(S,S')}^+ (b_{us'} + b_{vs'})/2|}{\sqrt{r_f}}$$

$$\phi^{(1)}(H) := (H\setminus D)/(S,S')$$

for some set of edges $X\subseteq W$.

We have to bucket edges by similar value of these functions in order to fully exploit this (since the bound at the end is in terms of the maximum). Therefore, we can find a relatively large subset of $W$ for which random sampling will on average causes a $(1 - \tilde{O}(1/|W|))$-factor decrease in the function $ b_{ss'}^T L_{(H\setminus D)/(S,S')}^+ (b_{su} + b_{sv})/2$ and in the worst case only decreases it by a $(1 - 1/(\log n))$ factor. Therefore, by Theorem \ref{thm:martingale-2}, $ b_{ss'}^T L_{(I\setminus D)/(S,S')}^+ (b_{su} + b_{sv})/2 $ only decreases by a small constant factor over the course of conditioning on this relatively large subset of $W$.

\subsubsection{Main objective change stability intuition}

$g_f^{(0)}(H)$ is the square root of the quantity that needs to be bounded here. Therefore, we are done by the discussion of the previous subsubsection.

\subsubsection{Deferred endpoint potential change stability intuition}

Exploit bucketing and Theorem \ref{thm:harsha}, as discussed in Subsubsection \ref{subsubsec:midpoint}. In particular, we will control the function

$$g_f^{(2)}(H) := \sum_{\{u,v\}\in A} \frac{|b_f^T L_{(H\setminus D)/(S,S')}^+ (b_{su} + b_{sv})|}{\sqrt{r_f}} + \sum_{\{u,v\}\in B} \frac{|b_f^T L_{(H\setminus D)/(S,S')}^+ (b_{us'} + b_{vs'})|}{\sqrt{r_f}}$$

$$\phi^{(2)}(H) := (H\setminus D)/(S,S')$$

\subsubsection{$S-S'$ normalized degree change stability intuition}

The only complication over the previous subsubsection is that the coefficients $b_{sw}^T L_{(I_i\setminus D)/S}^+ b_{sw}$ and $\frac{b_{ss'}^T L_{(I_i\setminus D)/(S,S')}^+ b_e}{r_e}$ can increase. We show using Theorem \ref{thm:harsha} that stopping after $|W|m^{-o(1)}$ samples does not result in an overly large increase in the coefficients. This is the only part of this section that requires $\rho = m^{o(1)}$.

Specifically, the following proposition will help us. Think of $v^{(k)} = b_{sw}^T L_{(I_k\setminus D)/S}^+ b_{sw}$ and $v^{(k)} = b_{ss'}^T L_{(I_k\setminus D)/(S,S')}^+ b_e$ in two different applications of Proposition \ref{prop:random-sequence-control}. It is proven in Appendix \ref{sec:fast-concentration-appendix}:

\linestart

\begin{restatable}{proposition}{proprandomsequencecontrol}\label{prop:random-sequence-control}
Consider a random sequence $\{v^{(k)}\}_{k\ge 0}$ generated as follows. Given $v^{(k)}$,

\begin{itemize}
\item Pick $\{u_i^{(k)}\}_{i=1}^{\ell_k}$ and $\{w_i^{(k)}\}_{i=1}^{\ell_k}$, with $\sum_{i=1}^{\ell_k} u_i^{(k)}w_i^{(k)} \le \eta v^{(k)}$
\item Let $Z^{(k+1)}\in \{0,1\}^{\ell_k}$ denote the indicator of a uniformly random choice over $[\ell_k]$
\item Pick $v^{(k+1)}\le v^{(k)} + \gamma \sum_{i=1}^{\ell_k} u_i^{(k)} Z_i^{(k+1)} w_i^{(k)}$
\end{itemize}

Let $m_0 = \min_k \ell_k$ and $M_0 = \max_k \ell_k$. Then with probability at least $1 - 2\tau$,

$$v^{(k')} \le (2\gamma \eta)^{\rho} v^{(0)}$$

for all $k'\le m_0 \min(\frac{1}{(\log (M_0^2/\tau))\eta^2\gamma^2},\frac{1}{200\eta \gamma^2 \log (M_0^2/\tau)}(\tau/M_0^2)^{1/\rho})$ 
\end{restatable}

\lineend

We use this proposition to show that controlling the following functions with constant coefficients suffices:

$$g_f^{(3),s}(H) = \sum_{w\in S'} b_{sw}^T L_{(I\setminus D)/S}^+ b_{sw} \sum_{e\in \partial_H w} \frac{|b_f^T L_{(H\setminus D)/(S,S')}^+ b_e|}{\sqrt{r_f}r_e}$$

$$g_f^{(3),s'}(H) = \sum_{w\in S} b_{s'w}^T L_{(I\setminus D)/S'}^+ b_{s'w} \sum_{e\in \partial_H w} \frac{|b_f^T L_{(H\setminus D)/(S,S')}^+ b_e|}{\sqrt{r_f}r_e}$$

$$\phi^{(3)}(H) := (H\setminus D)/(S,S')$$

$$g_f^{(4),s}(H) = \sum_{w\in S'} \left(\frac{|b_{sw}^T L_{(H\setminus D)/S}^+ b_f|}{\sqrt{r_f}}\right)^2 \sum_{e\in \partial_I w} \frac{b_{ss'}^T L_{(I\setminus D)/(S,S')}^+ b_e}{r_e}$$

$$\phi^{(4),s}(H) := (H\setminus D)/S$$

$$g_f^{(4),s'}(H) = \sum_{w\in S} \left(\frac{|b_{s'w}^T L_{(H\setminus D)/S'}^+ b_f|}{\sqrt{r_f}}\right)^2 \sum_{e\in \partial_I w} \frac{b_{ss'}^T L_{(I\setminus D)/(S,S')}^+ b_e}{r_e}$$

$$\phi^{(4),s'}(H) := (H\setminus D)/S'$$

\subsubsection{Tying the parts together}

Now, we implement $\FastOracle$. This oracle is similar to $\SlowOracle$ but restricts the set $W$ up front using $\VeryStable$ in order to take care of the flexibility of the $g^{()}$ functions.

\begin{algorithm}[H]
\caption{$\FastOracle(I,S,S',D,A,B,W)$ part 1 (everything but return statement)}
\SetAlgoLined
\DontPrintSemicolon

    \KwData{graph $I$, sets $S,S'\subseteq V(I)$ for identification, deleted edges $D\subseteq E(I)$, deferred edges $A,B\subseteq E(I)$, input edges $W\subseteq E(I)$}

    \KwResult{a relatively large subset $Z\subseteq W$ for which objectives remain stable}

    $I'\gets $ graph obtained by splitting each edge of $W$\;

    $W'\gets $ arbitrary copy of each edge in $W$ in $I'$\;

    \tcp{leverage scores}

    $W'\gets \VeryStable_{I'}(I'\setminus D,W')$\;

    $W'\gets \VeryStable_{I'}(I'/(S,S'),W')$\;

    \tcp{controlling flexible functions}

    $W'\gets \VeryStable_{I'}((I'\setminus D)/(S,S'),W')$\;

    $W'\gets \VeryStable_{I'}((I'\setminus D)/S,W')$\;

    $W'\gets \VeryStable_{I'}((I'\setminus D)/S',W')$\;

    $\ApxPreproc(I',S,S',D,A,B,W')$\;

    \ForEach{$e\in W'$}{

        \tcp{$\ApxQuery$ returns multiplicative 2-approximations $h_e^{()}$ to each $g_e^{()}$}

        $h_e^{()}(I')\gets \ApxQuery(g_e^{()}(I'))$

    }

    \tcp{bucketing for deferred, degree, and midpoint objectives}

    \ForEach{$i\in \{0,1,\hdots,i_{max} := \log (n^8\alpha^4)\}$}{

        $W_i'\gets $ edges $f\in W'$ for which

        $$\frac{h_f^{(0)}(I')}{\sqrt{b_{ss'}^T L_{(I'\setminus D)/(S,S')}^+ b_{ss'}}}\in \left[2^{-i-1}, 2^{-i}\right]$$\;

        with no lower bound for $i = i_{max}$

    }

    $W'\gets $ the $W_i'$ with maximum size\; \label{line:energy-bucket}

    \tcp{midpoint objective downsampling}

    $W'\gets $ uniform random sample of $1/(1000 \log^2 (n/p))$ fraction of $W'$\;\label{line:midpoint-downsampling}
\end{algorithm}

\newpage

\begin{algorithm}[H]
\caption{$\FastOracle$ part 2 (return statement)}
\SetAlgoLined
\DontPrintSemicolon

    \tcp{final output}

    $Z\gets $ edges $e\in W'$ with all of the following properties:

    \begin{itemize}
    \item (Midpoint $s$) $h_e^{(0)}(I')h_e^{(1),W',s}(I')\le p (b_{ss'}^T L_{(I\setminus D)/(S,S')}^+ b_{ss'})/(100\log n)$ if ``$s$ narrow potential neighborhood'' input condition is satisfied

    \item (Midpoint $s'$) $h_e^{(0)}(I')h_e^{(1),W',s'}(I')\le p (b_{ss'}^T L_{(I\setminus D)/(S,S')}^+ b_{ss'})/(100\log n)$ if ``$s'$ narrow potential neighborhood'' input condition is satisfied

    \item (Conductance $s-s'$) $h_e^{(0)}(I')h_e^{(3),s}(I')\le \frac{100 (\log (n\alpha))}{|W'|} \delta_{S,S'}(I'\setminus D)$

    \item (Conductance $s'-s$) $h_e^{(0)}(I')h_e^{(3),s'}(I')\le \frac{100 (\log (n\alpha))}{|W'|} \delta_{S',S}(I'\setminus D)$

    \item (Energy $s-s'$) $h_e^{(4),s}(I')\le \frac{100}{|W'|} \delta_{S,S'}(I'\setminus D)$

    \item (Energy $s'-s$) $h_e^{(4),s'}(I')\le \frac{100}{|W'|} \delta_{S',S}(I'\setminus D)$

    \item (Deferred) $h_e^{(0)}(I')h_e^{(2)}(I')\le \frac{100(\log (\alpha n))}{|W'|}\left(\sum_{\{u,v\}\in A} b_{ss'}^T L_{(I'\setminus D)/(S,S')}^+ (b_{su} + b_{sv}) + \sum_{\{u,v\}\in B} b_{ss'}^T L_{(I'\setminus D)/(S,S')}^+ (b_{us'} + b_{vs'})\right) + r_{min}/n^4$

    \item (Main) $(h_e^{(0)}(I'))^2 \le \frac{100}{|W'|} b_{ss'}^T L_{(I\setminus D)/(S,S')}^+ b_{ss'}$
    \end{itemize}

    \Return $Z$\;

\end{algorithm}

Now, we outline the analysis of this algorithm. Simple calculations involving Sherman-Morrison show the following proposition:

\linestart

\begin{restatable}{proposition}{propflexibleobjectives}\label{prop:flexible-objectives}

For any graph $J$ with $S,S'\subseteq V(J)$ and $A,B,D,X\subseteq E(J)$, the families of functions

\begin{itemize}
\item $\mc F_0 := (\{g_e^{(0)}(H)\}_{e\in X},\phi^{(0)})$

\item $\mc F_{1,s} := (\{g_e^{(1),X\cap E(H),s}(H)\}_{e\in X},\phi^{(1)})$

\item $\mc F_{1,s'} := (\{g_e^{(1),X\cap E(H),s'}(H)\}_{e\in X},\phi^{(1)})$

\item $\mc F_2 := (\{g_e^{(2)}(H)\}_{e\in X},\phi^{(2)})$

\item $\mc F_{3,s} := (\{g_e^{(3),s}(H)\}_{e\in X},\phi^{(3)})$

\item $\mc F_{3,s'} := (\{g_e^{(3),s'}(H)\}_{e\in X},\phi^{(3)})$

\item $\mc F_{4,s} := (\{g_e^{(4),s}(H)\}_{e\in X},\phi^{(4),s})$

\item $\mc F_{4,s'} := (\{g_e^{(4),s'}(H)\}_{e\in X},\phi^{(4),s'})$
\end{itemize}

are flexible for the graph $J$.

\end{restatable}

\lineend

We prove the above proposition in Appendix \ref{sec:fast-concentration-appendix}. Proposition \ref{prop:very-stable} therefore implies that the maximum values of the functions in each family do not change by more than a factor of 2 with high probability over the course of $|W'|/2$ edge contractions and deletions. The bucketing lower bound for $h_e^{(0)}(I')$ together with the upper bounds given in bullets implies an upper bound on the values of all of the $h_e^{()}$ functions, which in turn gives upper bounds on the $g_e^{()}$ functions. This essentially completes all bounds besides the ``Midpoint potential stability'' bounds. We prove these bounds using a simple application of Theorem \ref{thm:martingale-2}.

Now, we just have to check the ``Size of $Z$'' guarantee. This proof of this guarantee is similar to the analysis of $\SlowOracle$.

\begin{proposition}\label{prop:fast-oracle}
There is an $(m^{o(1)},|W|m^{-o(1)})$-stable oracle $\FastOracle$ with runtime

$$\tilde{O}(m + \mc T(\VeryStable) + \mc T(\ApxPreproc) + |W|\mc T(\ApxQuery))$$
\end{proposition}

In the following proof, notice that some statements have been reordered from their appearance in the definition of stable oracles. The statements are roughly stated in order of difficulty.

\begin{proof}
\underline{Algorithm well-definedness.} By Proposition \ref{prop:apx-correctness}, the $h$ functions are indeed approximations to the $g$ functions. Therefore, to check that this algorithm is well-defined, it suffices to check that the input condition of $\VeryStable$ described in Proposition \ref{prop:very-stable} is satisfied. The ``Bounded leverage score difference'' input condition and Rayleigh monotonicity imply that all leverage scores $\texttt{lev}_H(e)$ in graphs $H$ used in calls to $\VeryStable$ are within $1/16$ additively of $\texttt{lev}_I(e)$. Therefore, by Proposition \ref{prop:split-cond}, after splitting, all leverage scores are within $1/8$ additively of $\texttt{lev}_{I'}(e)$. Also by Proposition \ref{prop:split-cond}, $\texttt{lev}_{I'}(e)\in [1/4,3/4]$. Therefore, $\texttt{lev}_{H'}(e)\in [1/8,7/8]$, where $H'$ is the graph with $W$ split for each graphs $H$ supplied as an argument to $\VeryStable$. In particular, splitting using $I$ satisfies all of the input conditions to the $\VeryStable$ calls.

\partspace

\textbf{Runtime.} $\ApxQuery$ is called $\tilde{O}(|W|)$ times while $\ApxPreproc$ and $\VeryStable$ are only called a constant number of times. The rest of the algorithm takes $\tilde{O}(m)$ time, as desired.

\textbf{Leverage score stability.} Notice that $\VeryStable$ is applied with both $I'\setminus D$ and $I'/(S,S')$ as arguments. Therefore, the desired result follows directly from the ``Leverage score bound'' guarantee of Proposition \ref{prop:very-stable} applied to $I'\setminus D$ and $I'/(S,S')$.

\partspace

\textbf{Main objective change stability.} By the ``Flexible function bound'' on the family $\{g_e^{(0)}\}_{e\in Z}$,

\begin{align*}
(g_{f_i}^{(0)}(I_i'))^2 &\le 4\max_{e\in Z} (g_e^{(0)}(I'))^2\\
&\le 16\max_{e\in Z} (h_e^{(0)}(I'))^2\\
&\le \frac{1600}{|Z|} (b_{ss'}^T L_{(I\setminus D)/(S,S')}^+ b_{ss'})\\
\end{align*}

for any $i\le |Z|/2$ with high probability, where $I_i'$ is the graph $I_i$ with edges in $Z$ split. By definition of $g_{f_i}^{(0)}(I_i')$ and the fact that the $g$ functions are electrical,

$$g_{f_i}^{(0)}(I_i') = \frac{b_{ss'}^T L_{(I_i'\setminus D)/(S,S')}^+ b_{f_i}}{\sqrt{r_{f_i}}} = \frac{b_{ss'}^T L_{(I_i\setminus D)/(S,S')}^+ b_{f_i}}{\sqrt{r_{f_i}}}$$

so plugging this in completes the stability proof since $1600 < \rho$.

\partspace

\underline{Main objective change.} By ``Main objective change stability,''``Leverage score stability,'' and the fact that $f_i$ is split before contraction or deletion, each contraction or deletion causes a maximum change of

$$|b_{ss'}^T L_{(I_{i+1}\setminus D)/(S,S')}^+ b_{ss'} - b_{ss'}^T L_{(I_i\setminus D)/(S,S')}^+ b_{ss'}| \le 8\frac{1600 b_{ss'}^T L_{(I_i\setminus D)/(S,S')}^+ b_{ss'}}{|Z|}$$

for all $i$. Therefore, since $K = |W|m^{-o(1)}\le |W|/12800$, the main objective can only increase by a factor of 2 over the course of conditioning on the $f_i$s.

\partspace

\textbf{Deferred endpoint potential change stability.} Suppose that $W_a'$ is the bucket in $W'$ that is chosen in Line \ref{line:energy-bucket} in Part 1 of $\FastOracle$. For edges $e\in Z$, $h_e^{(0)}(I') \ge 2^{-a-1}\sqrt{b_{ss'}^T L_{(I'\setminus D)/(S,S')}^+ b_{ss'}}$ by the bucketing lower bound. Therefore, by the ``Deferred'' condition on edges in $Z$,

\begin{align*}
g_e^{(2)}(I')&\le 2h_e^{(2)}(I') \le \frac{2^{a+2}}{\sqrt{b_{ss'}^T L_{(I'\setminus D)/(S,S')}^+ b_{ss'}}}\frac{100(\log(\alpha n))}{|Z|}\\
&\left(\sum_{\{u,v\}\in A} b_{ss'}^T L_{(I'\setminus D)/(S,S')}^+ (b_{su} + b_{sv}) + \sum_{\{u,v\}\in B} b_{ss'}^T L_{(I'\setminus D)/(S,S')}^+ (b_{us'} + b_{vs'})\right)
\end{align*}

By the bucketing upper bound,

$$g_e^{(0)}(I') \le 2^{-a+1}\sqrt{b_{ss'}^T L_{(I'\setminus D)/(S,S')}^+ b_{ss'}}$$

These bounds hold for all $e\in Z$. Therefore, by the ``Flexible function bound'' of Proposition \ref{prop:very-stable} applied to both of the flexible families $\{g_e^{(0)}\}_{e\in Z}$ and $\{g_e^{(2)}\}_{e\in Z}$,

$$g_{f_i}^{(0)}(I_i') \le \max_{e\in Z} 2g_e^{(0)}(I')\le 2^{-a+2}\sqrt{b_{ss'}^T L_{(I'\setminus D)/(S,S')}^+ b_{ss'}}$$

and

\begin{align*}
g_{f_i}^{(2)}(I_i') &\le \max_{e\in Z} 2g_e^{(2)}(I')\\
&\le \frac{2^{a+3}}{\sqrt{b_{ss'}^T L_{(I'\setminus D)/(S,S')}^+ b_{ss'}}}\frac{100(\log(\alpha n))}{|Z|}\\
&\left(\sum_{\{u,v\}\in A} b_{ss'}^T L_{(I'\setminus D)/(S,S')}^+ (b_{su} + b_{sv}) + \sum_{\{u,v\}\in B} b_{ss'}^T L_{(I'\setminus D)/(S,S')}^+ (b_{us'} + b_{vs'})\right)
\end{align*}

with high probability for all $i\le |Z|/2$. Therefore, since the $g$s are electrical functions,

\begin{align*}
&\left(\sum_{\{u,v\}\in A} \frac{|b_{ss'}^T L_{(I_i\setminus D)/(S,S')}^+ b_{f_i}| |b_{f_i}^T L_{(I_i\setminus D)/(S,S')}^+ (b_{su} + b_{sv})|}{r_{f_i}}\right)\\
&+ \left(\sum_{\{u,v\}\in B} \frac{|b_{ss'}^T L_{(I_i\setminus D)/(S,S')}^+ b_{f_i}| |b_{f_i}^T L_{(I_i\setminus D)/(S,S')}^+ (b_{us'} + b_{vs'})|}{r_{f_i}}\right)\\
&= g_{f_i}^{(0)}(I_i) g_{f_i}^{(2)}(I_i)\\
&\le \frac{3200(\log(\alpha n))}{|Z|} \left(\sum_{\{u,v\}\in A} b_{ss'}^T L_{(I\setminus D)/(S,S')}^+ (b_{su} + b_{sv}) + \sum_{\{u,v\}\in B} b_{ss'}^T L_{(I\setminus D)/(S,S')}^+ (b_{us'} + b_{vs'})\right)
\end{align*}

as desired, since $K < |Z|/2$.

\partspace

\textbf{$S-S'$-normalized degree change stability.} We omit the $S'-S$ guarantees, as their proofs are the same with $s$ and $S$ swapped for $s'$ and $S'$ respectively in all places.

\partspace

\textbf{$S-S'$-conductance term stability.} We break this analysis up into two parts.

\underline{Bounding increases in effective resistances.} We show that with high probability, for all $w\in S'$ and all $i\le m^{-o(1)} |Z|$,

$$b_{sw}^T L_{(I_i\setminus D)/S}^+ b_{sw}\le m^{o(1)} b_{sw}^T L_{(I\setminus D)/S}^+ b_{sw}$$

We use Proposition \ref{prop:random-sequence-control} with

\begin{itemize}
\item $\rho\gets \sqrt{\log n}$
\item $\tau\gets 1/n^6$
\item $\gamma\gets 8$
\item $u_e^{(i)} = w_e^{(i)}\gets |b_{sw}^T L_{(I_i\setminus D)/S}^+ b_e|/\sqrt{r_e}$
\item $v^{(i)} \gets b_{sw}^T L_{(I_i\setminus D)/S}^+ b_{sw}$
\item $\eta\gets 1$ (which works because the sum of energies on edges is at most the overall energy)
\item $\ell_i\gets$ number of remaining edges in $Z$ in $I_i$
\end{itemize}

By electrical functions, Sherman-Morrison, and the triangle inequality,

\begin{align*}
v^{(i+1)} &= b_{sw}^T L_{(I_{i+1}'\setminus D)/S}^+ b_{sw}\\
&\le b_{sw}^T L_{(I_i'\setminus D)/S}^+ b_{sw} + \frac{(b_{sw}^T L_{(I_i'\setminus D)/S}^+ b_{f_i})^2/r_{f_i}}{\min(\texttt{lev}_{(I_i'\setminus D)/S}(f_i),1-\texttt{lev}_{(I_i'\setminus D)/S}(f_i))}\\
&\le b_{sw}^T L_{(I_i'\setminus D)/S}^+ b_{sw} + 8(b_{sw}^T L_{(I_i'\setminus D)/S}^+ b_{f_i})^2/r_{f_i}\\
&\le v^{(i)} + \gamma \sum_e u_e^{(i)} Z_e^{(i+1)} w_e^{(i)}\\
\end{align*}

Therefore, Proposition \ref{prop:random-sequence-control} applies, which means that

$$v^{(i)} = b_{sw}^T L_{(I_i\setminus D)/S}^+ b_{sw}\le 16^{\sqrt{\log n}} b_{sw}^T L_{(I\setminus D)/S}^+ b_{sw}$$

for all $i\le |Z|n^{-6/\sqrt{\log n}} = |Z|2^{-6\sqrt{\log n}}$ with high probability (at least $1 - 1/n^6$). This is the desired result for this part since $K < |Z|2^{-6\sqrt{\log n}}$.

\underline{Bounding increases in flows with constant coefficients.} We now show that

$$\sum_{w\in S'} (b_{sw}^T L_{(I\setminus D)/S}^+ b_{sw}) \sum_{e\in \partial_{I_i}w} \frac{|b_{ss'}^T L_{(I_i\setminus D)/(S,S')}^+ b_{f_i}| |b_{f_i}^T L_{(I_i\setminus D)/(S,S')}^+ b_e|}{r_{f_i}r_e} \le \tilde{O}\left(\frac{1}{|Z|}\right)\delta_{S,S'}(I\setminus D)$$

with high probability for all $i\le |Z|/2$. Assume that $W_a'$ is the bucket chosen on Line \ref{line:energy-bucket}. Then, by the lower bound on $h_e^{(0)}(I')$,

$$g_e^{(3),s}(I')\le 2h_e^{(3),s}(I')\le \frac{2^{a+2}}{\sqrt{b_{ss'}^T L_{(I'\setminus D)/(S,S')}^+ b_{ss'}}}\frac{100 \log(\alpha n)}{|Z|}\delta_{S,S'}(I'\setminus D)$$

for any $e\in Z$. By the bucketing upper bound,

$$g_e^{(0)}(I')\le 2h_e^{(0)}(I')\le 2^{-a+1} \sqrt{b_{ss'}^T L_{(I'\setminus D)/(S,S')}^+ b_{ss'}}$$

By the ``Flexible function bound'' of Proposition \ref{prop:very-stable},

\begin{align*}
g_{f_i}^{(0)}(I_i') g_{f_i}^{(3),s}(I_i')&\le 4(\max_{e\in Z} g_e^{(0)}(I'))(\max_{e\in Z} g_e^{(3),s}(I'))\\
&\le \frac{3200}{\log(\alpha n)}{|Z|}\delta_{S,S'}(I'\setminus D)\\
\end{align*}

for all $i\le |Z|/2$ with high probability. Substituting in the definitions of $g_{f_i}^{(0)}(I_i')$ and $g_{f_i}^{(3),s}(I_i')$ yields the desired result for this part.

\underline{Combining the parts.} By the first part,

\begin{align*}
&\sum_{w\in S'} (b_{sw}^T L_{(I_i\setminus D)/S}^+ b_{sw}) \sum_{e\in \partial_{I_i}w} \frac{|b_{ss'}^T L_{(I_i\setminus D)/(S,S')}^+ b_{f_i}| |b_{f_i}^T L_{(I_i\setminus D)/(S,S')}^+ b_e|}{r_{f_i}r_e}\\
&\le m^{o(1)} \sum_{w\in S'} (b_{sw}^T L_{(I\setminus D)/S}^+ b_{sw}) \sum_{e\in \partial_{I_i}w} \frac{|b_{ss'}^T L_{(I_i\setminus D)/(S,S')}^+ b_{f_i}| |b_{f_i}^T L_{(I_i\setminus D)/(S,S')}^+ b_e|}{r_{f_i}r_e}\\
\end{align*}

By the second part,

$$m^{o(1)} \sum_{w\in S'} (b_{sw}^T L_{(I\setminus D)/S}^+ b_{sw}) \sum_{e\in \partial_{I_i}w} \frac{|b_{ss'}^T L_{(I_i\setminus D)/(S,S')}^+ b_{f_i}| |b_{f_i}^T L_{(I_i\setminus D)/(S,S')}^+ b_e|}{r_{f_i}r_e} \le \frac{3200 m^{o(1)}}{|Z|}\delta_{S,S'}(I\setminus D)$$

The desired result follows from substitution since $3200 m^{o(1)} < \rho$.

\partspace

\textbf{$S-S'$ energy term stability.} As in the previous bound, we break the analysis up into two parts:

\underline{Bounding increases in flows to $S'$ vertices.} We show that with high probability, for all $w\in S'$ and all $i\le m^{-o(1)} |Z|$,

$$\sum_{e\in \partial_{I_i}w} \frac{b_{ss'}^T L_{(I_i\setminus D)/(S,S')}^+ b_e}{r_e}\le m^{o(1)} \sum_{e\in \partial_I w} \frac{b_{ss'}^T L_{(I\setminus D)/(S,S')}^+ b_e}{r_e}$$

We use Proposition \ref{prop:random-sequence-control} with

\begin{itemize}
\item $\rho\gets \sqrt{\log n}$
\item $\tau\gets 1/n^6$
\item $\gamma\gets 8$
\item $u_f^{(i)} \gets \frac{|b_{ss'}^T L_{(I_i\setminus D)/(S,S')}^+ b_f|}{\sqrt{r_f}}$
\item $w_f^{(i)}\gets \sum_{e\in \partial_{I_i} w} \frac{|b_f^T L_{(I_i\setminus D)/(S,S')}^+ b_e|}{\sqrt{r_f}r_e}$
\item $v^{(i)} \gets \sum_{e\in \partial_{I_i} w} \frac{b_{ss'}^T L_{(I_i\setminus D)/(S,S')}^+ b_e}{r_e}$
\item $\eta\gets O(\log (n\alpha))$ (which works by Theorem \ref{thm:harsha})
\item $\ell_i\gets$ number of remaining edges in $Z$ in $I_i$
\end{itemize}

By electrical functions, Sherman-Morrison, and the triangle inequality,

\begin{align*}
v^{(i+1)} &= \sum_{e\in \partial_{I_{i+1}'} w} \frac{b_{ss'}^T L_{(I_{i+1}'\setminus D)/(S,S')}^+ b_e}{r_e}\\
&\le \left(\sum_{e\in \partial_{I_i'} w} \frac{b_{ss'}^T L_{(I_i'\setminus D)/(S,S')}^+ b_e}{r_e}\right) + \frac{1}{\min(\texttt{lev}_{(I_i'\setminus D)/(S,S')}(f_i),1-\texttt{lev}_{(I_i'\setminus D)/(S,S')}(f_i))}\\
&\left(\sum_{e\in \partial_{I_i} w} \frac{|b_{ss'}^T L_{(I_i\setminus D)/(S,S')}^+ b_{f_i}|}{\sqrt{r_{f_i}}}\frac{|b_{f_i}^T L_{(I_i\setminus D)/(S,S')}^+ b_e|}{\sqrt{r_{f_i}}r_e}\right)\\
&\le v^{(i)} + \gamma \sum_e u_e^{(i)} Z_e^{(i+1)} w_e^{(i)}\\
\end{align*}

Therefore, Proposition \ref{prop:random-sequence-control} applies, which means that,

$$v^{(i)} = \sum_{e\in \partial_{I_i} w} \frac{b_{ss'}^T L_{(I_i\setminus D)/(S,S')}^+ b_e}{r_e}\le (16\log n)^{\sqrt{\log n}} \sum_{e\in \partial_I w} \frac{b_{ss'}^T L_{(I\setminus D)/(S,S')}^+ b_e}{r_e}$$

for all $i\le |Z|n^{-6/\sqrt{\log n}} = |Z|2^{-6\sqrt{\log n}}$ with high probability (at least $1 - 1/n^6$). This is the desired result for this part since $K < |Z|2^{-6\sqrt{\log n}}$.

\underline{Bounding increases in energies with constant coefficients.} By the ``Flexible function bound'' of Proposition \ref{prop:very-stable} applied to $g^{(4),s}$, with high probability for all $i$,

\begin{align*}
g_{f_i}^{(4),s}(I_i')&\le 2\max_{e\in Z} g_e^{(4),s}(I')\\
&\le 4\max_{e\in Z} h_e^{(4),s}(I')\\
&\le \frac{400}{|Z|}\delta_{S,S'}(I'\setminus D)\\
\end{align*}

as desired.

\underline{Combining the parts.} By the first part,

$$\sum_{w\in S'} \frac{(b_{sw}^T L_{(I_i\setminus D)/S}^+ b_{f_i})^2}{r_{f_i}} \sum_{e\in \partial_{I_i}w} \frac{b_{ss'}^T L_{(I_i\setminus D)/(S,S')}^+ b_e}{r_e}\le m^{o(1)} \sum_{w\in S'} \frac{(b_{sw}^T L_{(I_i\setminus D)/S}^+ b_{f_i})^2}{r_{f_i}} \sum_{e\in \partial_{I_i}w} \frac{b_{ss'}^T L_{(I\setminus D)/(S,S')}^+ b_e}{r_e}$$

By the second part,

$$m^{o(1)} \sum_{w\in S'} \frac{(b_{sw}^T L_{(I_i\setminus D)/S}^+ b_{f_i})^2}{r_{f_i}} \sum_{e\in \partial_{I_i}w} \frac{b_{ss'}^T L_{(I\setminus D)/(S,S')}^+ b_e}{r_e}\le \frac{400 m^{o(1)}}{|Z|}\delta_{S,S'}(I\setminus D)$$

The desired result follows from substitution since $400 m^{o(1)} < \rho$.

\textbf{Midpoint potential stability.} We assume that the ``$s$ narrow potential neighborhood'' input condition is satisfied. The ``$s'$ narrow potential neighborhood'' case is the same with $s$ and $s'$ swapped.

Suppose that bucket $W_a'$ was chosen on Line \ref{line:energy-bucket}. Then for all $e\in Z$,

$$g_e^{(0)}(I')\le 2h_e^{(0)}(I') \le 2^{-a+1} \sqrt{b_{ss'}^T L_{(I'\setminus D)/(S,S')}^+ b_{ss'}}$$

and

\begin{align*}
g_e^{(1),Z,s}(I')&\le 2h_e^{(1),W',s}(I')\\
&\le \frac{2^{a+2}}{\sqrt{b_{ss'}^T L_{(I'\setminus D)/(S,S')}^+ b_{ss'}}}\frac{p(b_{ss'}^T L_{(I\setminus D)/(S,S')}^+ b_{ss'})}{100\log n} + \frac{p}{200(\log n) h_e^{(0)}(I')}
\end{align*}

by the bucketing upper and lower bounds respectively. By the ``Flexible function bound'' of Proposition \ref{prop:very-stable} and the ``Midpoint $s$'' bound, for all $i\le |Z|/2$, with high probability

$$g_{f_i}^{(0)}(I_i')g_{f_i}^{(1),Z\setminus \{f_0,f_1,\hdots,f_{i-1}\},s}(I_i')\le \frac{32}{100(\log n)}p(b_{ss'}^T L_{(I\setminus D)/(S,S')}^+ b_{ss'})$$

For any $\{u,v\} \in Z\setminus \{f_0,f_1,\hdots,f_{i-1}\}$,

\begin{align*}
\frac{|b_{ss'}^T L_{(I_i'\setminus D)/(S,S')}^+ b_{f_i}| |b_{f_i}^T L_{(I_i'\setminus D)/(S,S')} (b_{su} + b_{sv})/2|}{r_{f_i}} &\le g_{f_i}^{(0)}(I_i')g_{f_i}^{(1),Z\setminus \{f_0,f_1,\hdots,f_{i-1}\},s}(I_i')\\
&\le \frac{32}{100(\log n)}p(b_{ss'}^T L_{(I\setminus D)/(S,S')}^+ b_{ss'})\\
\end{align*}

Therefore, Theorem \ref{thm:martingale-2} applies with expectation change $O((\log (n/p)) (3p)(b_{ss'}^T L_{(I\setminus D)/(S,S')}^+ b_{ss'}))/|Z|$, stepwise variance $(O((\log (n/p))(3p)(b_{ss'}^T L_{(I\setminus D)/(S,S')}^+ b_{ss'})))^2/|Z|$ (both by Theorem \ref{thm:harsha} and an upper bound inductive hypothesis with base case that is the upper bound of ``$s$ narrow potential neighborhood''), and maximum change $\frac{32}{100(\log n)}p(b_{ss'}^T L_{(I\setminus D)/(S,S')}^+ b_{ss'})$. Theorem \ref{thm:martingale-2} shows that

$$\Pr[|(b_{ss'}^T L_{(I_i\setminus D)/(S,S')}^+ (b_{su_i} + b_{sv_i})/2) - (b_{ss'}^T L_{(I\setminus D)/(S,S')}^+ (b_{su_i} + b_{sv_i})/2)| > p/8 b_{ss'}^T L_{(I_i\setminus D)/(S,S')}^+ b_{ss'}]\le 1/n^6$$

Therefore, since $2p + p/8 < 3p < 1 - (p/2)$ (because $p\le 1/4$ by the ``Narrow potential neighborhood'' condition), the inductive hypothesis is verified and the ``$s'$ midpoint potential stability'' guarantee is satisfied. Furthermore, by the ``$s$ narrow potential neighborhood'' lower bound, the normalized potential of the midpoint of $f_i$ is at least $p - p/8 > p/2$, as desired.

\partspace

\textbf{Size of $Z$.} We show that $|Z|\ge |W|/\text{polylog}(n)$ with constant probability. This can be boosted to high probability by running $\FastOracle$ $\Theta(\log n)$ times, since all other guarantees are satisfied with high probability.

By the ``Size'' guarantee of Proposition \ref{prop:very-stable}, $|W'|\ge |W|/(\text{polylog}(n))^5$ immediately after all applications of $\VeryStable$. After Line \ref{line:energy-bucket}, 

$$|W'|\ge |W|/((\text{polylog}(n))^5 i_{max})$$

Let $W''$ be the version of $W'$ immediately after Line \ref{line:midpoint-downsampling}. Then

$$|W''| \ge |W|/((\text{polylog}(n))^5 i_{max} 1000\log^2(n/p))$$

Now, we just need to upper how many edges in $W''$ are not in $Z$. We bound this on a constraint-by-constraint basis.

\underline{Midpoint constraints with the ``$s$ narrow potential neighborhood'' input condition.}This part is very similar in spirit to the ``Size'' bound proof in Proposition \ref{prop:very-stable}. This part under the ``$s'$ narrow potential neighborhood'' assumption is the same with $s$ and $s'$ swapped, so we focus on the $s$ case.

Start by exploiting how $W''$ was sampled from $W'$ on Line \ref{line:midpoint-downsampling}, which we denote by $W_{orig}'$ for clarity. By Theorem \ref{thm:harsha}, for any $e\in W_{orig}'$

\begin{align*}
\sum_{e\in W_{orig}'} g_e^{(0)}(I')g_e^{(1),W_{orig}',s}(I') &= \sum_{e\in W_{orig}'} \sum_{f = \{u,v\}\in W_{orig}'\setminus \{e\} } \frac{|b_{ss'}^T L_{(I'\setminus D)/(S,S')}^+ b_e| |b_e^T L_{(I'\setminus D)/(S,S')}^+ (b_{su} + b_{sv})/2|}{r_e}\\
&\le \sum_{f = \{u,v\}\in W_{orig}'} O(\log (n/p))b_{ss'}^T L_{(I'\setminus D)/(S,S')}^+ (b_{su} + b_{sv})/2\\
\end{align*}

By the ``$s$ narrow potential neighborhood'' condition,

$$\sum_{f = \{u,v\}\in W_{orig}'} O(\log (n/p))b_{ss'}^T L_{(I'\setminus D)/(S,S')}^+ (b_{su} + b_{sv})/2 \le O(\log (n/p))(2p b_{ss'}^T L_{(I'\setminus D)/(S,S')}^+ b_{ss'}) |W_{orig}'|$$

Therefore, there is a set $W_{low}'\subseteq W_{orig}'$ with $|W_{low}'|\ge |W_{orig}'|/2$ with the property that

$$g_e^{(0)}(I')g_e^{(1),W_{orig}',s}(I') \le O(\log (n/p))(4p b_{ss'}^T L_{(I'\setminus D)/(S,S')}^+ b_{ss'})$$

for all $e\in W_{low}'$. We now upper bound the expected number of edges in $W_{low}'$ that violate the ``Midpoint $s$'' condition. First, for any $e\in W_{low}'$

\begin{align*}
\textbf{E}_{W''}[g_e^{(0)}(I')g_e^{(1),W'',s}(I') | e\in W''] &= g_e^{(0)}(I')\textbf{E}_{W''}[g_e^{(1),W'',s}(I')]\\
&= \frac{1}{1000\log^2(n/p)} g_e^{(0)}(I')g_e^{(1),W_{orig}',s}(I')\\
&\le \frac{p}{400(\log n)} (b_{ss'}^T L_{(I'\setminus D)/(S,S')}^+ b_{ss'})\\
\end{align*}

The first equality follows from the fact that the sum in $g_e$ does not include $e$. The second equality follows from the definition of $W''$ on Line \ref{line:midpoint-downsampling}. The last line follows from the definition of $W_{low}'$. By Markov's Inequality, for any $e\in W_{low}'$,

\begin{align*}
&\Pr_{W''}[e \text{ satisfies ``Midpoint $s$'' condition and } e\in W'']\\
&\Pr_{W''}[e \text{ satisfies ``Midpoint $s$'' condition} | e\in W'']\Pr_{W''}[e\in W'']\\
&\ge 1 - \Pr_{W''}[g_e^{(0)}(I')g_e^{(1),W'',s}(I') > \frac{p}{200(\log n)} (b_{ss'}^T L_{(I'\setminus D)/(S,S')}^+ b_{ss'}) | e\in W''] \frac{1}{1000\log^2(n/p)}\\
&\ge \frac{1}{2000\log^2(n/p)}\\
\end{align*}

Let $Z_{mid}$ be the set of edges in $W_{orig}'$ satisfying the ``Midpoint $s$'' condition. Then

\begin{align*}
\textbf{E}_{W''}[|Z_{mid}|] &\ge \sum_{e\in W_{low}'}\Pr_{W''}[e\in Z_{mid}]\\
&= \sum_{e\in W_{low}'} \Pr_{W''}[e \text{ satisfies ``Midpoint $s$'' condition and } e\in W'']\\
&\ge \frac{|W_{low}'|}{2000\log^2(n/p)}\\
&\ge \frac{|W_{orig}'|}{4000\log^2(n/p)}\\
\end{align*}

Since $|Z_{mid}|\subseteq |W_{orig}'|$, $|Z_{mid}|\ge \frac{|W_{orig}'|}{8000\log^2(n/p)}$ with probability at least $\frac{1}{8000\log^2(n/p)}$, as desired.

\underline{Conductance $s-s'$ constraint.} We upper bound the number of elements in $Z_{mid}$ that do not satisfy this. By Theorem \ref{thm:harsha},

$$\sum_{e\in Z_{mid}} g_e^{(0)}(I')g_e^{(3),s}(I')\le (\log (n\alpha))\delta_{S,S'}(I'\setminus D)$$

By the approximation lower bound, for all edges $e$ with $h_e^{(0)}(I')h_e^{(3),s}(I')\ge \frac{100(\log (n\alpha))}{|Z_{mid}|}\delta_{S,S'}(I'\setminus D)$,

$$g_e^{(0)}(I')g_e^{(3),s}(I')\ge \frac{25(\log (n\alpha))}{|Z_{mid}|}\delta_{S,S'}(I'\setminus D)$$

By the previous inequality, only $|Z_{mid}|/25$ edges in $Z_{mid}$ can satisfy the above inequality, as desired.

\underline{Conductance $s'-s$ constraint.} Same as above, but with $s$ and $S$ swapped for $s'$ and $S'$ and vice versa.

\underline{Energy $s-s'$ constraint.} Since the sum of energies on edges is at most the overall energy,

$$\sum_{e\in Z_{mid}} g_e^{(4),s}(I') \le \delta_{S,S'}(I'\setminus D)$$

For all edges with $h_e^{(4),s}(I')\ge \frac{100}{|Z_{mid}|}\delta_{S,S'}(I'\setminus D)$,

$$g_e^{(4),s}(I')\ge \frac{50}{|Z_{mid}|}\delta_{S,S'}(I'\setminus D)$$

By the previous inequality, only $|Z_{mid}|/50$ edges in $Z_{mid}$ can satisfy the above inequality, as desired.

\underline{Energy $s'-s$ constraint.} Same as above, but with $s$ and $S$ swapped for $s'$ and $S'$ and vice versa.

\underline{Deferred constraint.} By Theorem \ref{thm:harsha},

\begin{align*}
\sum_{e\in Z_{mid}} g_e^{(0)}(I')g_e^{(2)}(I')&\le (\log (n\alpha))\left(\sum_{\{u,v\}\in A} b_{ss'}^T L_{(I'\setminus D)/(S,S')}^+ (b_{su} + b_{sv}) + \sum_{\{u,v\}\in B} b_{ss'}^T L_{(I'\setminus D)/(S,S')}^+ (b_{us'} + b_{vs'})\right)\\
&+ r_{min}/n^6
\end{align*}

Therefore, only $|Z_{mid}|/25$ edges $e\in Z_{mid}$ can have

\begin{align*}
h_e^{(0)}(I')h_e^{(2)}(I')&\ge \frac{100(\log (\alpha n))}{|W'|}\left(\sum_{\{u,v\}\in A} b_{ss'}^T L_{(I'\setminus D)/(S,S')}^+ (b_{su} + b_{sv}) + \sum_{\{u,v\}\in B} b_{ss'}^T L_{(I'\setminus D)/(S,S')}^+ (b_{us'} + b_{vs'})\right)\\
&+ r_{min}/n^4
\end{align*}

as desired.

\underline{Main constraint.} Since the sum of energies on edges is at most the overall energy,

$$\sum_{e\in Z_{mid}} (g_e^{(0)}(I'))^2 \le b_{ss'}^T L_{(I\setminus D)/(S,S')}^+ b_{ss'}$$

For edges $e$ with $(h_e^{(0)}(I'))^2 \ge \frac{100}{|Z_{mid}|} b_{ss'}^T L_{(I\setminus D)}^+ b_{ss'}$,

$$(g_e^{(0)}(I'))^2 \ge \frac{25}{|Z_{mid}|} b_{ss'}^T L_{(I\setminus D)/(S,S')}^+ b_{ss'}$$

By the previous inequality, this can only occur for $|Z_{mid}|/25$ edges in $Z_{mid}$, as desired.

\underline{Combining the parts.} Adding all of the guarantees from the previous parts about $Z_{mid}$ shows that at most $\frac{6|Z_{mid}|}{25}$ edges in $Z_{mid}$ are removed by those constraints. This leaves a set $Z$ with $|Z|\ge 19|Z_{mid}|/25$. Therefore, with probability at least $1/\text{polylog}(n)$, $|Z|\ge |W|/\text{polylog}(n)$. Repeating $\FastOracle$ $\text{polylog}(n)$ times therefore finds a set $Z$ with the desired size lower bound with high probability.

\end{proof}

\subsection{Efficient approximations to required quantities}

Now, we just need to efficiently compute approximations $h_e^{()}$ to the functions $g_e^{()}$ used in $\FastOracle$. We do this by observing that

\begin{itemize}
\item Each of the functions are weighted $\ell_1$ or $\ell_2$ norms of columns of a matrix $M$ with entries $M_{ij} := v_i^T L_H^+ v_j$.
\item $\ell_1$ and $\ell_2$ norms of columns can be approximated using $O(\log n)$ ``linear queries'' of the form $v^T L_H^+ v_j$ by Theorem \ref{thm:linear-sketch}. All linear queries for all $v_j$s can be computed using one Laplacian solve with demand vector $v$ (preprocessing) and a constant amount of work for $v_j$ (query). One can extract the desired information from the queries by taking the median, which takes $\tilde{O}(1)$ time per query.
\end{itemize}

While this is good enough for some of the functions we need to compute, we also need to leave out elements on the diagonal in some cases. To do this, we use the above idea to efficiently compute row sums for entries in some off-diagonal rectangle in the matrix. We show that all off-diagonal row sums can be approximated using $O(\log n)$ off-diagonal rectangle queries.

\subsubsection{Full row norm approximations}

In this subsection, we describe how to obtain 2-approximations for the functions $g_e^{(0)}$, $g_e^{(2)}$, $g_e^{(3),s}$, and $g_e^{(4),s}$ for arbitrary query edges $e$. $g_e^{(0)}(H)$ for all $e$ can be computed using one Laplacian solve for the demand $b_{ss'}$ on the graph $(H\setminus D)/(S,S')$. $g_e^{(2)}(H)$ and $g_e^{(3),s}(H)$ are each weighted $\ell_1$ norms, while $g_e^{(4),s}(H)$ is a weighted $\ell_2$ norm.

We now implement the first parts of $\ApxPreproc$ and $\ApxQuery$ that compute these functions:

\begin{algorithm}[H]
\caption{$\ApxPreproc(H,S,S',D,A,B,X)$ part 1 (functions with diagonal terms)}
\DontPrintSemicolon
\SetAlgoLined

    \KwIn{a graph $H$, $S,S'\subseteq V(H)$, $D,A,B,X\subseteq E(H)$}

    \KwResult{an implicit data structure for use by $\ApxQuery$}

    \tcp{main terms}

    $X^{(0)}\gets L_{(H\setminus D)/(S,S')}^+ b_{ss'}$\;

    \vspace{.4 in}

    \tcp{deferred terms}

    $C^{(2)}\gets \SketchMatrix(|A|+|B|,1/n^6,1,1/2)$\;

    $D^{(2)}\gets $ the $n\times (|A| + |B|)$ matrix with columns $b_{su} + b_{sv}$ for $\{u,v\}\in A$ and $b_{us'} + b_{vs'}$ for $\{u,v\}\in B$\;

    $X^{(2)}\gets L_{(H\setminus D)/(S,S')}^+ D^{(2)}(C^{(2)})^T$\;

    \vspace{.4 in}

    \tcp{conductance terms}

    $C^{(3),s}\gets \SketchMatrix(|\partial_H S'|,1/n^6,1,1/2)$\;

    $D^{(3),s}\gets $ the $n\times |\partial_H S'|$ matrix with columns $b_{sw}^T L_{(H\setminus D)/S}^+ b_{sw} \sum_{e\in \partial_H w} \frac{b_e}{r_e}$ for edges $e\in \partial_H w$ for some $w\in S'$, where the effective resistance value is $1.1$-approximated using Johnson-Lindenstrauss\;

    $X^{(3),s}\gets L_{(H\setminus D)/(S,S')}^+ D^{(3),s}(C^{(3),s})^T$\;

    $C^{(3),s'}\gets \SketchMatrix(|\partial_H S'|,1/n^6,1,1/2)$\;

    Compute $X^{(3),s'}$ in the same way as $X^{(3),s}$ with $s$,$S$ and $s'$,$S'$ interchanged\;

    \vspace{.4 in}

    \tcp{energy terms}

    $C^{(4),s}\gets \SketchMatrix(|S'|,1/n^6,2,1/2)$\;

    $D^{(4),s}\gets $ the $n\times |S'|$ matrix with columns $\left(\sqrt{\sum_{e\in \partial_H w} \frac{b_{ss'}^T L_{(H\setminus D)/(S,S')}^+ b_e}{r_e}} \right)b_{sw}$ for $w\in S'$, where the flow values are computed using the vector $L_{(H\setminus D)/(S,S')}^+ b_{ss'}$ with one Laplacian solve\;

    $X^{(4),s}\gets L_{(H\setminus D)/S}^+ D^{(4),s} (C^{(4),s})^T$\;

    $C^{(4),s'}\gets \SketchMatrix(|S'|,1/n^6,2,1/2)$\;

    Compute $X^{(4),s'}$ in the same way as $X^{(4),s}$, with $s$,$S$ and $s'$,$S'$ swapped\; 

\end{algorithm}

\newpage

\begin{algorithm}[H]
\caption{$\ApxQuery(g_f^{()}(H))$ for non-$g^{(1)}$ functions}
\DontPrintSemicolon
\SetAlgoLined

    \KwIn{a function $g_f^{()}(H)$}

    \KwOut{a 2-approximation to the value of the function}

    $d\gets $ number of columns of $C^{()}$\;

    $p\gets $ 1 or 2, depending on whether $g^{(4)}$ is being used\;

    \Return{$\RecoverNorm((X^{()})^T (b_f/\sqrt{r_f}),d,1/n^6,p,1/2)$ or this value squared if $p = 2$}
\end{algorithm}

\subsubsection{Row with diagonal element left out}

Now, we compute $g^{(1)}$ and implement $\ColumnApx$. These quantities differ from the ones discussed in the previous subsubsection in that they leave out one ``diagonal'' element in each row. One could try to subtract out this element, but this destroys multiplicative approximation.

Instead, we use sketching to approximate sums of rows of random off-diagonal submatrices. Specifically, the algorithm does the following for $\ColumnApxPreproc$:

\begin{itemize}
\item Do $\Theta(\log n)$ times:
    \begin{itemize}
    \item Pick a random balanced partition $(Z_0,Z_1)$ of $Z$
    \item For each $e\in Z_0$, approximate $a_e\gets \sum_{f\in Z_1} \frac{|b_e^T L_{G'}^+ b_f|}{\sqrt{r_e}\sqrt{r_f}}$ using sketching
    \end{itemize}
\item For each $e\in Z$, average the $a_e$s together and scale up the average by a factor of 4 to obtain an estimate for $\sum_{f\ne e\in Z} \frac{|b_e^T L_{G'}^+ b_f|}{\sqrt{r_e}\sqrt{r_f}}$
\end{itemize}

This algorithm takes $\tilde{O}(m)$ time thanks to the sketching step. To show correctness, think about the indicator variable $Y_{ef}$ that is 1 if and only if $e\in X_0$ and $f\in X_1$. This event happens with probability 1/4 in each of the $\Theta(\log n)$ trials if $e \ne f$ and with probability 0 otherwise. Since the trials are independent, the weight in front of each off-diagonal term concentrates around 1/4. Scaling up the average by a factor of 4 yields the desired result.

Now, we formally implement this idea. We analyze these implementations in the proof of Proposition \ref{prop:apx-correctness}.

\begin{algorithm}[H]
\caption{$\ColumnApxPreproc(G',Z)$}
\DontPrintSemicolon
\SetAlgoLined

    \KwIn{a graph $G'$ and $Z\subseteq E(G')$}

    \KwResult{nothing; implicit data structure}

    $K\gets 100\log n$\;

    $a_e\gets 0$ for each $e\in Z$\;

    \ForEach{$k \gets 1,2,\hdots,K$}{

        $Z_0,Z_1\gets $ uniformly random partition of $Z$ with $|Z_0| - |Z_1| \le 1$\;

        $C\gets \SketchMatrix(|Z_1|,1/n^6,1,1/4)$\;

        $D\gets $ $n\times |Z_1|$ matrix of columns $b_f/\sqrt{r_f}$ for $f\in Z_1$\;

        $U \gets L_{G'}^+ DC^T$\;

        \ForEach{$e\in Z_0$}{

            Increment $a_e$ by $\RecoverNorm(U^T (b_e/\sqrt{r_e}),|Z_1|,1/n^6,1,1/4)$\;

        }

    }

\end{algorithm}

\begin{algorithm}[H]
\caption{$\ColumnApx(e)$}
\DontPrintSemicolon
\SetAlgoLined

    \KwIn{an edge $e\in Z$}

    \KwResult{a 2-approximation to the value of $\sum_{f\ne e\in Z} \frac{|b_e^T L_{G'}^+ b_f|}{\sqrt{r_e}\sqrt{r_f}}$}

    \Return{$4a_e/K$}

\end{algorithm}

\newpage

\begin{algorithm}[H]
\caption{$\ApxPreproc(H,S,S',D,A,B,X)$ part 2 (functions without diagonal terms)}
\DontPrintSemicolon
\SetAlgoLined

    \KwIn{a graph $H$, $S,S'\subseteq V(H)$, $D,A,B,X\subseteq E(H)$}

    \KwResult{an implicit data structure for use by $\ApxQuery$}

    \tcp{code for $s$ only given for clarity; swap $s,S$ for $s',S'$}

    $K\gets 100\log n$\;

    $a_e\gets 0$ for each $e\in X$\;

    \ForEach{$k\gets 1,2,\hdots,K$}{

        $X_0,X_1\gets $ uniformly random partition of $X$ with $|X_0|-|X_1|\le 1$\;

        $C\gets \SketchMatrix(|X_1|,1/n^6,1,1/4)$\;

        $D\gets n\times |X_1|$ matrix of columns $(b_{su} + b_{sv})/2$ for $\{u,v\}\in X_1$\;

        $U\gets L_{(H\setminus D)/(S,S')}^+ DC^T$\;

        \ForEach{$e\in X_0$}{

            Increment $a_e$ by $\RecoverNorm(U^T(b_e/\sqrt{r_e}),|X_1|,1/n^6,1,1/4)$

        }

    }
    
\end{algorithm}

\begin{algorithm}[H]
\caption{$\ApxQuery(g_e^{(1),X,s})$}
\DontPrintSemicolon
\SetAlgoLined

    \KwIn{the function $g_e^{(1),X,s}$}

    \KwResult{a 2-approximation to its value on $H$}

    \tcp{code for $s$ only given for clarity; swap $s,S$ for $s',S'$}

    \Return{$4a_e/K$}
    
\end{algorithm}

\subsubsection{Combining the parts}

\begin{proposition}\label{prop:apx-correctness}

The following guarantees hold for $\ApxQuery$, $\ColumnApx$, $\ApxPreproc$, and $\ColumnApxPreproc$:

\begin{itemize}
\item (Correctness) Each call to $\ApxQuery$ and $\ColumnApx$ returns a 2-approximation to the correct value.
\item (Query runtime) Each call to $\ApxQuery$ and $\ColumnApx$ takes $\tilde{O}(1)$ time.
\item (Preprocessing runtime) Each call to $\ApxPreproc$ and $\ColumnApxPreproc$ takes $\tilde{O}(m)$ time.
\end{itemize}
\end{proposition}

\begin{proof}

\textbf{Correctness for $\ApxQuery$ for all but $g^{(1)}$.} Follows directly from the ``Approximation'' guarantee of Theorem \ref{thm:linear-sketch}.

\textbf{Correctness for $\ApxQuery$ for $g^{(1)}$ and $\ColumnApx$.} We focus on $g_e^{(1),X,s}$, since an intuitive overview for $\ApxQuery$ was given earlier and the proof is very similar in both cases. Let $Y_{ef}^{(k)}$ be the indicator variable of the event $\{e\in X_0\text{ and } f\in X_1\}$. By the ``Approximation'' guarantee of Theorem \ref{thm:linear-sketch}, at the end of the outer foreach loop in $\ApxPreproc$,

$$a_e\in [3/4,5/4] \left(\sum_{f = \{u,v\}\in X} \frac{|b_e^T L_{(H\setminus D)/(S,S')}^+ (b_{su} + b_{sv})/2|}{\sqrt{r_e}}\left(\sum_{k=1}^K Y_{ef}^{(k)}\right)\right)$$

for each $e\in X$. Since $Y_{ee}^{(k)} = 0$ for all $k$ and $e\in X$,

\begin{align*}
&\left(\sum_{f = \{u,v\}\in X} \frac{|b_e^T L_{(H\setminus D)/(S,S')}^+ (b_{su} + b_{sv})/2|}{\sqrt{r_e}}\left(\sum_{k=1}^K Y_{ef}^{(k)}\right)\right)\\
&= \left(\sum_{f = \{u,v\}\in X\setminus e} \frac{|b_e^T L_{(H\setminus D)/(S,S')}^+ (b_{su} + b_{sv})/2|}{\sqrt{r_e}}\left(\sum_{k=1}^K Y_{ef}^{(k)}\right)\right)\\
\end{align*}

Notice that $\textbf{E}[Y_{ef}^{(k)}] = 1/4$ if $e \ne f$ and that for any fixed $e,f$, the collection of random variables $\{Y_{ef}^{(k)}\}_{k=1}^K$ is independent. Therefore, by Chernoff bounds (using the value of $K$) and a union bound over all pairs $e\ne f\in X$,

$$3K/16 \le \sum_{k=1}^K Y_{ef}^{(k)} \le 5K/16$$

Therefore,

$$a_e\in [9K/64,25K/64] \left(\sum_{f = \{u,v\}\in X\setminus e} \frac{|b_e^T L_{(H\setminus D)/(S,S')}^+ (b_{su} + b_{sv})/2|}{\sqrt{r_e}}\right)\\$$

which means that $4a_e/K$ is a 2-approximation, as desired.

\textbf{Query runtime.} $\ColumnApx$ and $\ApxQuery$ for $g^{(1)}$ just return precomputed values, so they both take constant time. $\ApxQuery$ for other functions computes a matrix-vector product with a vector that is supported on only two entries. Therefore, this product only takes time proportional to the number of rows of the matrix, which is $\ell = O(\log n)$ by Theorem \ref{thm:linear-sketch}. $\RecoverNorm$ only take $\text{poly}(\ell) = \text{polylog}(n)$ time by the ``Runtime'' guarantee of Theorem \ref{thm:linear-sketch}. Therefore, all queries take $\text{polylog}(n)$ time, as desired.

\textbf{Preprocessing runtime.} Each $X$ matrix computation takes near-linear time, as it involves a $\ell$ sparse matrix-vector products to compute $DC^T$ and $\ell$ Laplacian solves to compute $L^+ DC$. Each $\SketchMatrix$ call takes $\tilde{O}(m)$ time by Theorem \ref{thm:linear-sketch}. All other operations take $\tilde{O}(m)$ time, so the total precomputation runtime is $\tilde{O}(m)$, as desired.

\end{proof}

\subsection{Proof of Lemma \ref{lem:fast-fix} (an algorithm for the fixing lemma with almost-linear runtime)}

Now, we combine all of the results of this section to prove Lemma \ref{lem:fast-fix}:

\begin{proof}[Proof of Lemma \ref{lem:fast-fix}]
The result follows immediately from Lemma \ref{lem:oracle-fix} and Proposition \ref{prop:fast-oracle} with $\FastOracle$ substituted in for $\Oracle$ in the $\Fix$ algorithm. Call this algorithm $\FastFix$. $\FastOracle$'s runtime is bounded using Proposition \ref{prop:apx-correctness}.
\end{proof}

\newpage

\appendix
\section{Facts about electrical potentials}

\subsection{Bounding potentials using effective resistances}

\begin{lemma}\label{lem:pot-bound}
Consider a graph $I$ with three vertices $u,v,w\in V(I)$. Then

$$\Pr_v[t_u > t_w]\le \frac{\texttt{Reff}_I(u,v)}{\texttt{Reff}_I(u,w)}$$
\end{lemma}

\begin{proof}
This probability can be written in terms of normalized potentials:

$$\Pr_v[t_u > t_w] = \frac{b_{uw}^T L_I^+ b_{uv}}{b_{uw}^T L_I^+ b_{uw}}$$

Since electrical potentials are maximized on the support of the demand vector,

$$\frac{b_{uw}^T L_I^+ b_{uv}}{b_{uw}^T L_I^+ b_{uw}}\le \frac{b_{uv}^T L_I^+ b_{uv}}{b_{uw}^T L_I^+ b_{uw}} = \frac{\texttt{Reff}_I(u,v)}{\texttt{Reff}_I(u,w)}$$

as desired.
\end{proof}

\lemwellsep*

To prove this, we take advantage of the dual formulation of electrical flow given in \cite{KOSZ13}:

\begin{remark}[Equation (2) of \cite{KOSZ13}]
The solution $x$ for the optimization problem $Lx = b$ is also the vector that optimizes

$$\max_{p\in \mathbb{R}^n} 2b^T p - p^T L p$$
\end{remark}

\begin{proof}[Proof of Lemma \ref{lem:well-sep}]
Pick arbitrary $s\in C_1$, $t\in C_2$ in $I$ and use the electrical potentials $p = L_I^+ b_{st}$ to construct a good solution $q$ to the dual formulation of electrical flows in $J$. Construct $J$ from $I$ by adding 0-resistance edges from $s$ and $t$ to all $C_1$ and $C_2$ vertices respectively. Notice that $p_s = \min_v p_v$ and $p_t = \max_v p_v$. For all $v$, let $q_v = \max(p_s + R, \min(p_t - R, p_v))$.

$q$ is feasible, so by the above remark,

$$\texttt{Reff}_J(s,t)\ge 2b_{st}^T q - q^T L_J q$$

Notice that $2b_{st}^T q = 2(p_t - p_s - 2R)$, so we just need to upper bound $q^T L_J q$. Notice that for any $x\in C_1$, $p_x - p_s = b_{st}^T L_I^+ b_{sx}\le b_{sx}^T L_I^+ b_{sx}\le R$. Similarly, for any $y\in C_2$, $p_t - p_y\le R$. Therefore, all vertices $x\in C_1$ and $y\in C_2$ have $q_x = p_s + R$ and $q_y = p_t - R$ respectively. This means that all 0-resistance edges in $J$ have potential drop 0 across them, which means that $q^T L_J q$ is defined.

No other edges were added to create $J$ and potential drops are only smaller using $q$ instead of $p$, so $q^T L_J q\le p^T L_I p = p_t - p_s$. Therefore,

$$\texttt{Reff}_J(s,t)\ge p_t - p_s - 4R\ge (\gamma - 4)R$$

as desired.
\end{proof}

\subsection{Schur complement facts}

\lemmonotoneschur*

\begin{proof}
It suffices to show this when $S_0 = S_0'\cup \{v\}$ for some vertex $v\in V(I)$. By Remark \ref{rmk:schur}, $J' = \texttt{Schur}(J,V(J)\setminus \{v\})$. By the formula in Definition \ref{def:schur}, each edge $\{u,v\}$ for $u\in S_1$ is mapped to a set of edges $\{u,w\}$ with $c_{uw} = c_{uv}^Jc_{vw}^J/c_v^J$, where $c_v^J$ is the total conductance of edges incident with $v$ in $J$. In particular, summing over all $w$ and using the fact that

$$\sum_{w\in S_0'} c_{vw}^J\le c_v^J$$

shows that the sum of the conductances of the edges created when $v$ is eliminated is less than the original $uv$ conductance. This is the desired result.
\end{proof}

\lemmonotonecond*

\begin{proof}
By Proposition \ref{prop:lin-alg-total-cond}, $1/c^I(S_0',S_1) = b_{s_0's_1}^T L_{I/(S_0',S_1)}^+ b_{s_0's_1}$ and $1/c^I(S_0,S_1) = b_{s_0s_1}^T L_{I/(S_0,S_1)}^+ b_{s_0s_1}$. $I/(S_0,S_1)$ can be obtained from $I/(S_0',S_1)$ by identifying vertices. This only reduces the quadratic form by Rayleigh monotonicity, as desired.
\end{proof}

\lempcond*

\begin{proof}
Let $J = I/(C,V(I)\setminus S_C)$, with $C$ and $V(I)\setminus S_C$ identified to $s$ and $t$ respectively. Let $J'$ be the graph obtained by subdividing every edge that crosses the $p$ normalized $st$-potential threshold in $J$, identifying the vertices created through subdivision, and making $J'$ the induced subgraph on vertices with normalized $st$-potential at most $p$. The $st$ electrical flow on $J$ restricted to $J'$ is also an electrical flow in $J'$ with energy exactly $p (b_{st}^T L_J^+ b_{st})$. $S_{C,V(I)\setminus S_C}(p,C)$ contains all edges of $J'$, so by Proposition \ref{prop:lin-alg-total-cond} and Rayleigh monotonicity,

$$c^I(C,V(I)\setminus S_{C,V(I)\setminus S_C}(p,C))\le \frac{1}{p b_{st}^T L_J^+ b_{st}}$$

By \ref{prop:lin-alg-total-cond} again,

$$b_{st}^T L_J^+ b_{st} = \frac{1}{c^I(C,V(I)\setminus S_C)}$$

Substitution yields the desired result.
\end{proof}

\subsection{Carving clusters from shortcutters does not increase conductance much}

\begin{lemma}\label{lem:outside-schur}
Consider a graph $I$ and three disjoint sets $S_0,S_1,S_2\subseteq V(I)$ with $S_2' \subseteq S_2$. Let $J = \texttt{Schur}(I,S_0\cup S_1\cup S_2)$ and $J' = \texttt{Schur}(I,S_0\cup S_1\cup S_2')$. Then

$$c^J(E_J(S_0,S_1))\le c^{J'}(E_{J'}(S_0,S_1))$$
\end{lemma}

\begin{proof}
It suffices to show this result when $S_2 = S_2'\cup \{v\}$ for some vertex $v\in V(I)$, since one can eliminate the vertices of $S_2$ one at a time to get to $S_2'$. In this case, $J' = \texttt{Schur}(J,V(J)\setminus \{v\})$ by Remark \ref{rmk:schur}. By Definition \ref{def:schur},

$$L_{J'} = L_J[V(J)\setminus \{v\},V(J)\setminus \{v\}] - L_J[V(J)\setminus \{v\},\{v\}]1/c_v^J L_J[\{v\},V(J)\setminus \{v\}]$$

where $c_v^J$ is the total conductance of edges incident with $v$ in the graph $J$. $L_J[V(J)\setminus \{v\},\{v\}]$ only consists of nonpositive numbers, so conductances in $J'$ are only larger than they are in $J$. In particular, summing over all edges in $E_J(S_0,S_1)$ shows that

$$c^J(E_J(S_0,S_1))\le c^{J'}(E_{J'}(S_0,S_1))$$

as desired.
\end{proof}

\lemsmallholes*

\begin{proof}
Let $H_0 = \texttt{Schur}(H,C\cup C'\cup (V(H)\setminus S_C))$. By Lemma \ref{lem:well-sep} and Proposition \ref{prop:lin-alg-total-cond},

$$c^H(C,C') \le \frac{1}{(\beta_1 - 4)R}$$

By Lemma \ref{lem:outside-schur} with $I\gets H$, $S_0\gets C$, $S_1\gets C'$, $S_2' \gets \emptyset$, and $S_2\gets V(H)\setminus S_C$,

$$c^{H_0}(E_{H_0}(C,C'))\le c^H(C,C')$$

By Lemma \ref{lem:outside-schur} with $I\gets H$, $S_0\gets C$, $S_1\gets V(H)\setminus S_C$, $S_2' \gets \emptyset$ and $S_2\gets C'$

$$c^{H_0}(E_{H_0}(C,V(H)\setminus S_C))\le c^H(S_C)$$

Therefore,

\begin{align*}
c^H(S_C\setminus C') &= c^H(C,C'\cup (V(H)\setminus S_C))\\
&= c^{H_0}(E_{H_0}(C,C')) + c^{H_0}(E_{H_0}(C,V(H)\setminus S_C))\\
&\le \frac{1}{(\beta_1 - 4)R} + c^H(S_C)\\
\end{align*}

as desired.
\end{proof}

\section{Deferred proofs for Section \ref{sec:compute-shortcutters}}

\subsection{Efficient construction of covering communities ($\CoveringCommunity$ implementation)}\label{sec:covering-community}

We now prove Lemma \ref{lem:covering-community}. Our construction is similar to the construction of sparse covers given in \cite{ABCP99} and is made efficient for low-dimensional $\ell_2^2$ metrics using locality-sensitive hashing (\cite{AI06}). Our use of locality-sensitive hashing is inspired by \cite{HIS13}.

The algorithm maintains a set of uncovered vertices $S$ and builds families of clusters one at a time. We build each family by initializing a set $S'\gets S$ and repeatedly building clusters one-by-one. Each cluster is built by picking a vertex in $S'$ and building a effective resistance ball around it. We repeatedly consider growing the effective resistance radius of the ball by a factor of $\gamma$ if it could contain a much larger number of edges (an $m^{1/z_0}$ factor). This neighborhood can be generated efficiently (in time comparable to its size) using locality-sensitive hashing. After generating this neighborhood, one can grow it slightly in a way that decreases its boundary size using ball-growing (for example \cite{LR99}).

Before giving the $\CoveringCommunity$ algorithm, we start by giving the ball-growing algorithm $\BallGrow$. Ideally, $\BallGrow$ would do standard ball-growing done by sorting vertices $x\in X$ with respect to the values $||D(x_0) - D(x)||_2^2$. Unfortunately, $||D(x) - D(y)||_2^2$ is not necessarily a metric on $X$, despite the fact that it approximates the metric $\texttt{Reff}(x,y)$. Luckily, though, we only need to preserve distances to some vertex $x_0$ along with distances between endpoints of edges. This can be done by modifying $D$ slightly, as described in the first line of $\BallGrow$:

\begin{algorithm}[H]
\DontPrintSemicolon
\SetAlgoLined

    $J\gets $ the graph with $V(J) = V(I)$ and $E(J) = E_I(X)\cup$ edges from $x_0$ to each edge in $X$\;

    $d_J\gets $ the shortest path metric of $J$ with $\{x,y\}$ edge weight $||D(x) - D(y)||_2^2$\;

    $x_0,x_1,\hdots,x_k\gets $ vertices in $X$ in increasing order of $d_J(x_0,x_i)$ value\;

    $i^*\gets $ the value for which $c^I(\partial \{x_0,x_1,\hdots,x_{i^*}\})$ is minimized subject to the constraint that $d_J(x_0,x_i)\in (R_1,R_2]$ or $d_J(x_0,x_{i+1})\in (R_1,R_2]$\;

    \Return{ $\{x_0,x_1,\hdots,x_{i^*}\}$}

\caption{$\BallGrow_D(x_0,X,I,R_1,R_2)$}
\end{algorithm}

\begin{proposition}\label{prop:ball-grow}
$\BallGrow_D(x_0,X,I,R_1,R_2)$, given the Johnson-Lindenstrauss embedding $D$ of the effective resistance metric of $I$ with $\ep = 1/2$, returns a cluster $C$ with the following properties:

\begin{itemize}
\item (Subset of input) $C\subseteq X$.
\item (Large enough) $C$ contains the subset of $X$ with effective resistance distance $2R_1/3$ of $x_0$.
\item (Not too large) The $I$-effective resistance diameter of $C$ is at most $4R_2$.
\item (Boundary size) $c^I(\partial C)\le \frac{2|E_I(X)|}{R_2 - R_1} + c^I(\partial C\cap \partial X)$.
\end{itemize}

Furthermore, it does so in $\tilde{O}(|E_I(X)\cup \partial X|)$ time.
\end{proposition}

\begin{proof}
\textbf{Subset of input.} The $x_i$s are all in $X$ and $C$ only consists of $x_i$s.

\textbf{Large enough.} By the lower bound constraint on $i^*$, $C$ contains all vertices $x\in X$ with $d_J(x_0,x)\le R_1$. Since the edge $\{x_0,x\}\in E(J)$, $d_J(x_0,x)\le ||D(x_0) - D(x)||_2^2\le (3/2)\texttt{Reff}_I(x_0,x)$ by the upper bound of Theorem \ref{thm:jl}. For any vertex $x$ with effective resistance distance at most $2R_1/3$ from $x_0$, $d_J(x_0,x)\le R_1$. Therefore, any such $x$ is in $C$, as desired.

\textbf{Not too large.} Consider any vertex $x\in C$. By the upper bound constraint on $i^*$, $d_J(x_0,x)\le R_2$. By the triangle inequality for the effective resistance metric, to get a diameter bound of $4R_2$ on $C$, it suffices to show that for any $x\in X$, $\texttt{Reff}_I(x_0,x) \le 2d_J(x_0,x)$.

Consider any path $\{y_0 = x_0,y_1,y_2,\hdots,y_k = x\}$ in $J$. The length of this path is $\sum_{i=0}^{k-1} ||D(y_i) - D(y_{i+1})||_2^2$ by definition of $J$. By the lower bound of Theorem \ref{thm:jl},

$$\sum_{i=0}^{k-1} \texttt{Reff}_I(y_i,y_{i+1})/2 \le \sum_{i=0}^{k-1} ||D(y_i) - D(y_{i+1})||_2^2$$

By the triangle inequality for the effective resistance metric,

$$\texttt{Reff}_I(x_0,x)/2 = \texttt{Reff}_I(y_0,y_k)/2\le \sum_{i=0}^{k-1} \texttt{Reff}_I(y_i,y_{i+1})/2$$

so all paths from $x_0$ to $x$ in $J$ have length at least $\texttt{Reff}_I(x_0,x)/2$. Therefore, $\texttt{Reff}_I(x_0,x)/2 \le d_J(x_0,x)$, as desired.

\textbf{Boundary size.} It suffices to bound the conductance of edges in $\partial C$ that have both endpoints in $X$. Consider the quantity

$$Q = \sum_{e = \{x,y\}\in E_I(X): d_J(x_0,x) \le d_J(x_0,y)} c_e^I (d_J(x_0,y) - d_J(x_0,x))$$

with $x$ closer to $x_0$ than $y$ in $d_J$-distance. We start by showing that there is a $d_J$ distance threshold cut with conductance at most $\frac{Q}{R_2 - R_1}$. For any number $a\in \mathbb{R}$, let $\Clamp(a) = \max(\min(a,R_2),R_1)$. Notice that

\begin{align*}
&(R_2 - R_1)\min_{i: d_J(x_0,x_i)\in (R_1,R_2] \text{ or } d_J(x_0,x_{i+1})\in (R_1,R_2]} \sum_{e\in E_I(X)\cap \partial\{x_0,x_1,\hdots,x_i\}} c_e^I\\
&\le\sum_{i=0}^{k-1} (\Clamp(d_J(x_0,x_{i+1})) - \Clamp(d_J(x_0,x_i)))\sum_{e\in E_I(X)\cap \partial\{x_0,x_1,\hdots,x_i\}} c_e^I\\
&= \sum_{e = \{x,y\}\in E_I(X): d_J(x_0,x) \le d_J(x_0,y)} c_e^I(\Clamp(d_J(x_0,y)) - \Clamp(d_J(x_0,x)))\\
&\le Q
\end{align*}

Dividing both sides by $(R_2 - R_1)$ shows that the minimizing cut (for defining $i^*$) has total conductance at most $Q/(R_2 - R_1)$. Now, we upper bound $Q$. By the upper bound of Theorem \ref{thm:jl} and the triangle inequality for $d_J$,

\begin{align*}
Q &\le \sum_{e = \{x,y\}\in E_I(X)} c_e^I d_J(x,y)\\
&\le \sum_{e = \{x,y\}\in E_I(X)} c_e^I ||D(x) - D(y)||_2^2\\
&\le \sum_{e \in E_I(X)} c_e^I (3/2)\texttt{Reff}_I(e)\\
&\le 2|E_I(X)|\\
\end{align*}

This yields the desired conductance bound.

\textbf{Runtime.} Only $|X|-1$ edges are added to $I$ to make $J$. It only takes $\tilde{O}(|E_I(X)\cup \partial X|)$ time to compute all of the distances to $x_0$, which are the only ones that are queried. Finding $i^*$ only takes one linear sweep over the $x_i$s, which takes $\tilde{O}(|E_I(X)\cup \partial X|)$ time.
\end{proof}

Now, we implement $\CoveringCommunity$:

\begin{algorithm}[H]
\DontPrintSemicolon
\SetAlgoLined

    \tcp{the community that we output}
    $\mc F\gets \emptyset$\; 

    \tcp{the set of uncovered vertices}
    $S\gets X$\;

    \While{$S\neq\emptyset$}{

        \tcp{set to form $\mc F_i$s from}
        $S'\gets S$\;

        \For{$i = 0,\hdots,2z_0$}{

            $\mc F_i\gets \emptyset$\;

            $\mc{H}_i\gets (R\gamma^i,R\gamma^{i+1},1/n^{1/\gamma},1/n)$-sensitive hash family for the $\ell_2^2$ metric $||D(x) - D(y)||_2^2$\;

            $\mc{H}_i'\gets $ a sample of $(\log n)n^{1/\gamma}$ hash functions from $\mc{H}_i$\;

            bucket all vertices by hash values for each function in $\mc H_i'$

        }

        \While{$S'\neq\emptyset$}{

            $C\gets \{$ arbitrary vertex $v_0$ in $S'\}$\;

            $C'\gets X$\;

            $C''\gets X$\;

            $i\gets 0$\;

            \While{$|E_I(C')\cup \partial C'| > m^{1/z_0}|E_I(C)\cup \partial C|$}{

                \tcp{update $C'$ to be the set of nearby vertices to $C$}

                $C'\gets $ the set of vertices $v\in X$ with $h(v) = h(v_0)$ for some $h\in \mc{H}_{i+1}'$\; \label{line:hash-lookup}

                $C''\gets $ subset of $C'$ within $D$-distance $(2R/3)\gamma^{i+1}$ of $v_0$ found using a scan of $C'$\;

                $C\gets \BallGrow_D(v_0,C'',I,(3/2)R\gamma^i,2R\gamma^i)$\;

                $i\gets i+1$\;

            }

            $\mc F_i\gets \mc F_i\cup \{C\}$\; \label{line:add-cover}

            $S'\gets S'\setminus C'$\; \label{line:disjoint}

            $S\gets S\setminus C$\; \label{line:sub-cover}

        }

        add all $\mc F_i$s to $\mc C$\;

    }

    \Return $\mc{C}$\;
\caption{$\CoveringCommunity_D(X,I,R)$}
\end{algorithm}

\begin{figure}
\includegraphics[width=1.0\textwidth]{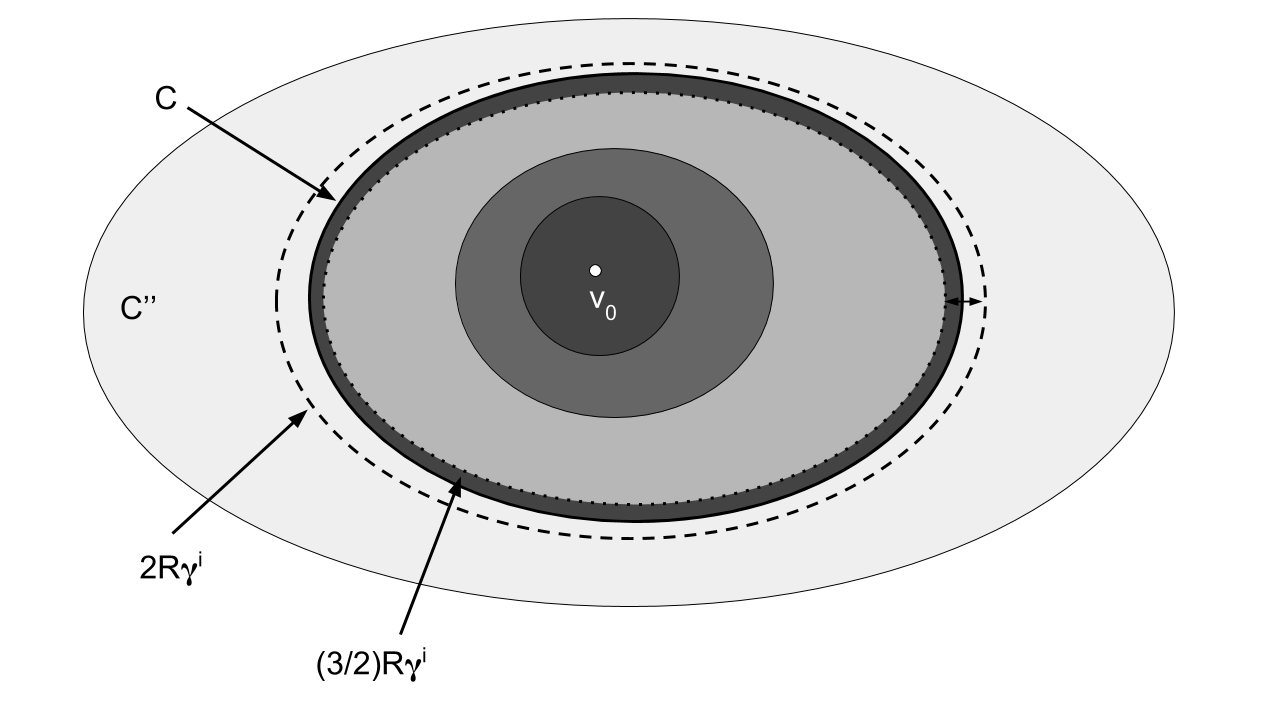}
\caption{How $\CoveringCommunity$ constructs one core. When clustering stops at some $C''$, a ball $C$ with small boundary between the $(3/2)R\gamma^i$ and $2R$ radii balls centered around $C_0$ is found using $\BallGrow$. $C''$ is removed from the set $S'$ of candidates that are well-separated from $v_0$ and $C$ is removed from the set of vertices that need to be covered.}
\label{fig:covering-community}
\end{figure}

\lemcoveringcommunity*

\begin{proof}[Proof of Lemma \ref{lem:covering-community}]

\textbf{$\murad R$-community.} By the ``Not too large'' condition of Proposition \ref{prop:ball-grow}, each cluster added to some family has diameter at most $8R\gamma^{i_{max}}\le \murad R$, where $i_{max}$ is the maximum value of $i$ over the course of the algorithm. To show that $\mc C$ is an $\murad R$-community, we therefore just have to show that $$i_{max}\le 2z_0$$, where $z_0 = (1/10)\log_{\gamma} \murad$.

To do this, we show that the innermost while loop executes at most $2z_0$ times. Let $C_i,C_i',$ and $C_i''$ be the values of $C,C'$, and $C''$ set for the value of $i$ in the innermost while loop. By the second (high distance) condition of locality-sensitive families, $C_i'$ only consists of vertices in $X$ with $D$-distance at most $R\gamma^{i+2}$ from $v_0$ with high probability. By the first (low distance) condition of locality-sensitive families, $C_i''$ contains all of the vertices in $X$ with $D$-distance at most $(2R/3)\gamma^{i+1}$ from $v_0$. Therefore, by the ``Large enough'' condition of Proposition \ref{prop:ball-grow}, $C_i$ consists of all vertices in $X$ with $D$-distance at most $R\gamma^i$ from $v_0$. Therefore, $C_i'\subseteq C_{i+2}$ for all $i$. By the inner while loop condition, $|E(C_{i+2})\cup \partial C_{i+2}|\ge m^{1/z_0}|E(C_i)\cup \partial C_i|$ for all $i$ for which iteration $i+2$ occurs. But $|E(C_{i+2})\cup \partial C_{i+2}|\le n$ for all $i$. This means that $i_{max}\le 2z_0$, as desired.

\textbf{Input constraint.} Each cluster consists of vertices in $S$, which is initialized to $X$ and only decreases in size. Therefore, $\mc C$ is $X$-constrained at the end of the algorithm.

\textbf{Covering.} Vertices are only removed from $S$ in Line \ref{line:sub-cover}. When they are removed from $S$, they are added to some cluster in $\mc F$ in Line \ref{line:add-cover}. The algorithm only terminates when $S$ is empty. Therefore, if $\CoveringCommunity$ terminates, it outputs a covering community.

\textbf{Boundedness.} Let $C$ be a cluster and consider the $C'$ and $C''$ used to generate it. By the ``Boundary size'' condition of Proposition \ref{prop:ball-grow},

$$c^I(\partial C)\le \frac{4|E_I(C'')|}{R\gamma^i} + c^I(\partial C\cap \partial C'')$$

Now, we bound the conductance of the set $\partial C\cap \partial C''$. By the first (upper bound) condition of locality-sensitive families, $C''$ contains all vertices in $X$ within $(4R/9)\gamma^{i+1}\ge 16 R\gamma^i$ $I$-effective resistance distance of $v_0$. By the ``Not too large'' condition of Proposition \ref{prop:ball-grow}, $C$ contains vertices with $I$-effective resistance distance $8R\gamma^i$ of $v_0$. Therefore, by the triangle inequality for effective resistance, each edge in $\partial C\cap \partial C''$ with both endpoints in $X$ has conductance at most $1/(8R\gamma^i)$. This means that

$$c^I(\partial C\cap \partial C'')\le \frac{|\partial C''|}{8R\gamma^i} + c^I(\partial C\cap \partial X)$$

The $C''$s for clusters $C$ in the same family are disjoint by Line \ref{line:disjoint}. Therefore, the total boundary size of all clusters $C$ in one family $\mc F$ is at most

$$\sum_{C\in \mc F} \left(\frac{4|E_I(C'')|}{R} + \frac{4|\partial C''|}{R} + c^I(\partial X\cap \partial C)\right) = c^I(\partial X) + \frac{4|E_I(X)\cup \partial X|}{R}$$

since $\gamma > 1$. This is the desired result.

\textbf{Well-separatedness.} By the ``Not too large'' condition of Proposition \ref{prop:ball-grow} and Theorem \ref{thm:jl}, $C'$ contains all vertices with $I$-effective resistance distance at most $(2R/3)\gamma^{i+1}$ from $v_0$. The corresponding cluster $C$ has $I$-effective resistance diameter at most $8R\gamma^i$. Therefore, if $C$ is added to a family $\mc F_{i+1}$, it is $\gamma/12$-well-separated from any cluster added to $\mc F_{i+1}$ in the future.

Since $\mc F_{i+1}$ only consists of clusters that were added to $\mc F_{i+1}$ with separation at least $(2R/3)\gamma^{i+1}$ from the remaining vertices in $S'$, each cluster in $\mc F_{i+1}$ is also $\gamma/12$-well-separated from any cluster added to $\mc F_{i+1}$ in the past. Therefore, each $F_i$ is a well-separated family.

\textbf{Number of families.} When a cluster $C$ is added to a family $\mc F_i$, the number of edges incident with $S'$ decreases by at most $|E_I(C')\cup \partial C'|\le m^{1/z_0}|E_I(C)\cup \partial C|$. The number of edges incident with $S$ decreases by at least $|E_I(C)\cup \partial C|$. Therefore, when $S'$ is empty, $S$ must have had at least $m^{1 - 1/z_0}$ incident edges removed. Therefore, the outer while loop can only execute $m^{1/z_0}$ times. Each outer while loop iteration adds $2z_0$ families to $\mc C$, so the total number of families is at most $2z_0m^{1/z_0}$, as desired.

\textbf{Runtime.} The innermost while loop runs for $2z_0$ iterations. Line \ref{line:hash-lookup} takes time proportional to the number of elements in the bin $h(v_0)$ for the $n^{1/\gamma} = m^{o(1)}$ hash functions in $\mc H_{i+1}'$. The runtime of the innermost while loop is therefore at most $O(n^{1/\gamma}z_0 |E_I(C')\cup \partial C'|)$ for the final $C'$ that it outputs by the ``Runtime'' condition of Proposition \ref{prop:ball-grow}. This runtime can be charged to the removal of $C'$ (and its incident edges) from $S'$. Therefore, each middle while loop takes $O(n^{1/\gamma} z_0 |E_I(X)\cup \partial X|)$ time. As described in the ``Number of families'' condition, the outermost while loop only executes $2z_0 m^{1/z_0}$ times. Therefore, the total runtime is $O(m^{1/z_0}n^{1/\gamma} z_0^2 |E_I(X)\cup \partial X|)\le |E_I(X)\cup \partial X| m^{o(1)}$, as desired.
\end{proof}

\subsection{Efficient construction of Voronoi diagrams ($\Voronoi$ implementation)}\label{sec:voronoi}

We now prove Lemma \ref{lem:voronoi}. Ideally, the algorithm would just compute the set of vertices from which a random walk has a probability of $7/8$ of hitting one cluster in $\mc F$ before any other. While this satisfies the correctness constraints, doing this would take too long. Computing the set with probability $7/8$ of hitting one cluster before another would take a Laplacian solve on $Z$ for each cluster. This is prohibitive, as we have no bound on the number of clusters in $\mc F$. This algorithm does not take advantage of slack on the lower bound for $S_C$.

Instead, we arbitrarily split the clusters into two groups $(X_i,Y_i)$ in $\log n$ different ways, generate $S_{\{X_i,Y_i\}}(1/(8\log n),X_i)$ and $S_{\{X_i,Y_i\}}(1/(8\log n),Y_i)$, and refine the resulting clusters. This only requires $O(\log n)$ Laplacian solves on $Z$. The lower bound on $S_C$ follows immediately from implication. The upper bound on $S_C$ follows from thinking in terms of random walks. The event in which a random walk hitting $C$ before any other cluster in $\mc F$ is equivalent to the conjunction of the $\log n$ events in which the random walk hits $C$'s half in each partition before the other half. The probability that any of these events can bound bounded satisfactorially with a union bound.

\begin{algorithm}[H]
    \SetAlgoLined
    \DontPrintSemicolon

    $C_1,C_2,\hdots,C_k\gets $ arbitrary ordering of clusters in $\mc{F}$\;

    $S_{C_1},\hdots,S_{C_k}\gets V(I)$\;

    \For{$i=0,1,\hdots,\log k$}{
        $X_i\gets $ the union of clusters $C_j\in \mc F$ with $i$th index digit 0 in binary\;

        $Y_i\gets \{C_1,\hdots,C_k\}\backslash X_i$\;

        \For{$j=1,\hdots,k$}{

            $Z_{ij}\gets $ the member of $\{X_i,Y_i\}$ that contains $C_j$;

            $S_{C_j}\gets S_{C_j}\cap S_{\{X_i,Y_i\}}(1/(8\log k),Z_{ij})$\;
        }
    }

    \Return $\{S_{C_1},\hdots,S_{C_k}\}$\;

\caption{$\Voronoi(I,\mc F)$}
\end{algorithm}

\begin{figure}
\includegraphics[width=1.0\textwidth]{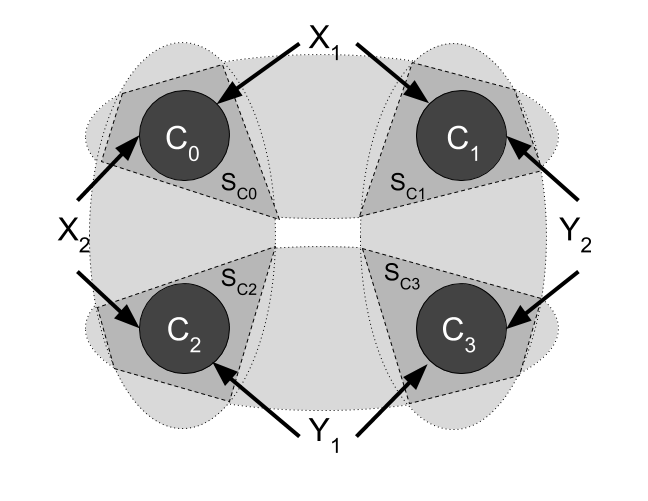}
\caption{How $\Voronoi$ constructs shortcutters. The horizontal and vertical dotted boundary sets are the $S_{\{X_1,Y_1\}}(1/(8\log k),Z_{i1})$ and $S_{\{X_2,Y_2\}}(1/(8\log k),Z_{i2})$ respectively.}
\label{fig:voronoi}
\end{figure}

\lemvoronoi*

\begin{proof}[Proof of Lemma \ref{lem:voronoi}]

\textbf{Lower bound.} We show that $S_{\mc F}(1/(8\log k),C)\subseteq S_C$ for each $C\in \mc F$. Consider a vertex $x\in S_{\mc F}(1/(8\log k),C)$. This vertex has the property that $\Pr_x[t_C > t_{\mc F\setminus C}] \le (1/8\log k)$. Let $C_j := C$. Since $C\subseteq Z_{ij}$, $x\in S_{\{X_i,Y_i\}}(1/(8\log k),Z_{ij})$ for all $i$. Therefore, it is in the intersection of them as well. Since $S_C$ is this intersection, $x\in S_C$, as desired.

\textbf{Upper bound.} We show that $S_C\subseteq S_{\mc F}(1/8,C)$. Consider a vertex $x\in S_C$. Let $C_j := C$. Since some $\{X_i,Y_i\}$ partition separates $C$ from each other cluster in $\mc F$,

$$\Pr_x[t_C > t_{\mc F\setminus C}] = \Pr_x[\exists i \text{ } t_{Z_{ij}} > t_{\mc F\setminus Z_{ij}}]$$

By a union bound,

$$\Pr_x[\exists i \text{ } t_{Z_{ij}} > t_{\mc F\setminus Z_{ij}}]\le \sum_{i=1}^{\log k} \Pr_x[t_{Z_{ij}} > t_{\mc F\setminus Z_{ij}}]\le \frac{\log k}{8\log k}\le 1/8$$

Therefore, $x\in S_{\mc F}(1/8,C)$, as desired.

\textbf{Runtime.} Each $S_{\{X_i,Y_i\}}(1/(8\log k),Z_{ij})$ can be computed in near-linear time in the number of edges incident with $Z$ by Remark \ref{rmk:level-set-lsolve}. Since there are only $\log k$ $i$s and only 2 possible $Z_{ij}$s for each $i$, the algorithm takes near-linear time.
\end{proof}

\subsection{A generalization of $\sum_{e\in E(G)} \texttt{lev}_G(e) = n-1$ to edges between well-separated clusters in the effective resistance metric (Lemma \ref{lem:well-sep-lev-score})}

In this section, we prove Lemma \ref{lem:well-sep-lev-score}:

\lemwellseplevscore*

We will use the Aldous-Broder algorithm for sampling random spanning trees (see Theorem \ref{thm:aldous-broder}) to reduce proving this result to a combinatorial problem. This problem is very closely related to Davenport-Schinzel strings from computational geometry:

\begin{definition}[Davenport-Schinzel strings \cite{DS65}]
Consider a string $S = s_1s_2\hdots s_y$ from an alphabet of size $n$. $S$ is called an $(n,s)$-\emph{Davenport-Schinzel} string if the following two properties hold:

\begin{itemize}
\item $s_i\ne s_{i+1}$ for all $i$
\item There is no subsequence of $S$ of the form $ABA\hdots BA$ of length $s+2$ for two distinct letters $A$ and $B$. We call such a subsequence an \emph{alternating subsequence} with length $s+2$.
\end{itemize}
\end{definition}

Davenport-Schinzel strings are useful because they cannot be long. Here is a bound that is good for nonconstant $s$, but is not optimal for constant $s$:

\begin{theorem}[Theorem 3 and Inequality (34) of \cite{DS65}]\label{thm:davenport-schinzel}
$(n,s)$-Davenport-Schinzel strings cannot be longer than $nC_0(n,s) := n (3s)! \exp(10\sqrt{s\log s \log n})$.
\end{theorem}

We now outline the proof of Lemma \ref{lem:well-sep-lev-score}. Recall that the leverage score of an edge is the probability that the edge is in a random spanning tree (see Theorem \ref{thm:edge-marginal}). Therefore, it suffices to bound the expected number of edges in a random spanning tree between clusters of $\mc F$ in $H$.

We bound this by reasoning about Aldous-Broder. Start by considering the case in which each cluster in $\mc F$ is a single vertex. Then $H$ is a graph with $k$ vertices. All spanning trees of this graph have at most $k$ edges, so the total leverage score of edges between clusters is at most $k$.

When $\mc F$ consists of (possibly large) well-separated clusters, we exploit Lemma \ref{lem:pot-bound} to show that a random walk that goes between any two clusters more than $\log_{\gammads} n$ times covers both clusters with high probability. This ensures that the sequence of destination clusters $(s_i)_i$ for edges added to the tree by Aldous-Broder satisfies the second condition of Davenport-Schnizel sequences for $s = \log_{\gammads} n = (\log n)^{1/3}$ with high probability. If two clusters did alternate more than $s$ times in this sequence, both would be covered, which means that Aldous-Broder would not add any more edges with a desination vertex in either cluster.

The sequence $(s_i)_i$ does not necessarily satisfy the first definition of Davenport-Schinzel because the random walk could go from an uncovered cluster $C$ to a covered one $C'$ and back. This motivates all of the pseudocode in the if statement on Line \ref{line:repeat-visits} of $\InterclusterEdges$. If $C$ is visited more than $\zeta$ times in a row, the random walk must visit many clusters that are very ``tied'' to $C$ in the sense that they have a high probability of returning to $C$ before hitting any other unvisited cluster. These ties are represented in a forest $T$ on the clusters of $\mc F$. Each time a sequence of $\zeta$ $C$s occurs, one can add another edge to $T$. $T$ can only have $k$ edges because it is a forest. These ideas show that $(s_i)_i$ is only $\zeta$ times longer than an $(k,\log_{\gammads} n)$-Davenport-Schinzel after removing $\zeta k$ additional letters that are charged to edges of $T$. Theorem \ref{thm:davenport-schinzel} applies and finishes the proof.

Now, we discuss this intuition more formally. Let $k = |\mc F|$. It suffices to show that a random spanning tree in $H$ contains $O(kC_0(k,\log_{\gammads} n))$ edges between the clusters of $\mc{F}$ with high probability. Generate a random spanning tree using Aldous-Broder and organize the intercluster edges using the following algorithm:

\begin{algorithm}[H]
\SetAlgoLined
\DontPrintSemicolon

    \KwData{a family of disjoint clusters $\mc F$ and a graph $H$ on the vertices in $\cup_{C\in \mc F} C$}

    \KwResult{A sequence of visited clusters $(s_i)_{i=0}^{\ell}$ and a rooted forest $T$ with $V(T) = \mc{F}$ with all edges pointing towards from the root}

    $i\gets 0$\;

    $\zeta \gets 100\log^2 (n\alpha)$\;

    $\beta \gets 1/(2\log n)$\;

    \tcp{start Aldous-Broder}
    $u_0 \gets $ arbitrary vertex of $H$\;

    $c\gets 1$\;

    $T\gets (\mc{F}, \emptyset)$\;

    \While{$H$ is not covered}{

        $u_c\gets $ random neighbor of $u_{c-1}$ in $H$\;

        $u\gets u_{c-1}, v\gets u_c$\;

        $C_u\gets $ cluster in $\mc{F}$ containing $u$\;

        $C_v\gets $ cluster in $\mc{F}$ containing $v$\;

        \If{$C_u\ne C_v$ and $C_v$ is not covered}{

            $s_i \gets C_v$\; \label{line:add}

            $i \gets i + 1$\;

            \label{line:deletions-start}

            \If{$s_j = C_v$ for all $j \in [i - \zeta, i - 1]$, and $i\ge \zeta$}{\label{line:repeat-visits}

                \tcp{in a multiset, members can be counted multiple times}

                $\mc{E}\gets $ the multiset of clusters $A\in \mc F$ for which (1) $A$ is visited immediately after one of the last $\zeta$ occurrences of $C_v$, (2) The root cluster $B_A$ of $A$'s tree in $T$ is visited before the next visit to $C_v$ and (3) $\Pr_x[\text{ hit } C_v \text{ before any other uncovered cluster}] \ge 1 - \beta$ for any $x\in B_A$\; \label{line:e-properties}

                \While{$\mc{E}$ contains a cluster $C$ whose tree in $T$ contains more than half of the clusters in the new tree formed by attaching $\mc{E}$'s trees to $C_v$}{

                    Remove all clusters in $C$'s arborescence that are also in $\mc E$ from $\mc E$\;

                }\label{line:cutting}

                \For{clusters $A\in \mc{E}$}{

                    add an edge from $B_A$ to $C_v$ in $T$ if one does not exist\; \label{line:forest-building}               

                }

                remove the last $\zeta$ occurrences of $C_v$ from $(s_j)_j$\;\label{line:deletions}

                $i\gets i - \zeta$\;

            }

            \label{line:deletions-end}

        }

        $c\gets c + 1$\;

    }

    \Return $(s_i)_i, T$\;

\caption{$\InterclusterEdges(H,\mc{F})$, never executed}
\end{algorithm}

Let $(s_i')_i$ be the sequence $(s_i)_i$ obtained by removing lines \ref{line:deletions-start} to \ref{line:deletions-end} from $\InterclusterEdges$ and removing the statement ``and $C_v$ is not covered'' from the outer if statement; that is $s_i'$ is the list of all cluster visits. We break up the analysis of this algorithm into a few subsections.

\subsubsection{Invariants of the forest $T$}

We start by discussing some invariants that hold for the directed graph $T$ over the entire course of the algorithm:

\begin{invariant}\label{inv:forest}
$T$ is a forest of directed arborescences with all edges pointing towards the root.
\end{invariant}

\begin{proof}
Initially, $T$ just consists of isolated vertices, which trivially satisfy the invariant. Line \ref{line:forest-building} is the only line of $\InterclusterEdges$ that adds edges to $T$. When Line \ref{line:forest-building} adds an edge, it adds an edge from a root of one arborescence to the root of another arborescence. This preserves the invariant, so we are done.
\end{proof}

\begin{invariant}\label{inv:tree-diam}
Each arborescence with size $\kappa$ in $T$ has maximum leaf-to-root path length $\log \kappa$.
\end{invariant}

\begin{proof}
Initially, $T$ consists of isolated vertices, each of which has diameter $\log 1 = 0$. Line \ref{line:cutting} ensures that the arborescence produced by the for loop containing Line \ref{line:forest-building} contains at least twice as many vertices as any of the constituent child arborescences. Furthermore, the length of the path to the root from each leaf increases by at most 1. Therefore, the new maximum path length is at most $1 + \log \kappa\le \log(2\kappa)\le \log \kappa'$, where $\kappa$ is the maximum size of a child arborescence and $\kappa'$ is the size of the combined arborescence created by Line \ref{line:forest-building}. Therefore, Line \ref{line:forest-building} maintains the invariant.
\end{proof}

\begin{invariant}\label{inv:edge-ties}
For each edge $(C_0,C_1)\in E(T)$ and any $x\in C_0$,

$$\Pr_x[\text{hits $C_1$ before any uncovered cluster}] \ge 1 - \beta$$
\end{invariant}

\begin{proof}
Initially, there are no edges, so the invariant is trivially satisfied. When an edge $(C_0,C_1)$ is added to $T$, property (3) on Line \ref{line:e-properties} ensures that the invariant is satisfied. As $\InterclusterEdges$ progresses, clusters that are covered remain covered, so

$\Pr_x[\text{hits $C_1$ before any uncovered cluster}]$ only increases. In particular, the invariant remains satisfied.
\end{proof}

These invariants imply the following proposition, which will effectively allow us to replace occurences of a cluster $A\in \mc F$ in the sequence $\{s_j'\}_j$ with the root $B_A$ of $A$'s arborescence in $T$:  

\begin{proposition}\label{prop:root-ties}
For each cluster $A\in \mc F$ in a component of $T$ with root $B_A\in \mc F$,

$$\Pr_x[\text{hits $B_A$ before any uncovered cluster}] \ge 1/2$$

for any $x\in A$.
\end{proposition}

\begin{proof}
By Invariant \ref{inv:tree-diam}, there is a path from $A$ to $B_A$ in $T$ with length at most $\log\kappa\le \log n$. Let $C_0 = A,C_1,\hdots,C_{\ell} = B_A$ with $\ell\le \log n$ be the path in $A$'s arborescence from $A$ to $B_A$. Notice that

\begin{align*}
&\Pr_x[\text{hits $B_A$ before any uncovered cluster}]\\
&\ge \Pr_x[\wedge_{p=0}^{\ell} \text{hits $C_p$ after $C_{p-1},C_{p-2},\hdots,C_1,C_0$ and before any uncovered cluster}]\\
&= \prod_{p=0}^{\ell-1} \Pr_x[\text{hits $C_{p+1}$ after $C_p,C_{p-1},\hdots,C_0$ and before any uncovered cluster}\\
&| \text{hits $C_p$ after $C_{p-1}$ after $\hdots$ after $C_0$ and before any uncovered cluster}]\\
&= \prod_{p=0}^{\ell-1} \Pr_x[\text{hits $C_{p+1}$ after $C_p$ and before any uncovered cluster}\\
&| \text{hits $C_p$ after $C_{p-1}$ after $\hdots$ after $C_0$ and before any uncovered cluster}]
\end{align*}

By the Markov property,

\begin{align*}
&\prod_{p=0}^{\ell-1} \Pr_x[\text{hits $C_{p+1}$ after $C_p$ and before any uncovered cluster}\\
&| \text{hits $C_p$ after $C_{p-1}$ after $\hdots$ after $C_0$ and before any uncovered cluster}]\\
&= \prod_{p=0}^{\ell-1} \Pr_x[\text{hits $C_{p+1}$ after $C_p$ and before any uncovered cluster} | \text{hits $C_p$ before any uncovered cluster}]
\end{align*}

Furthermore, by Invariant \ref{inv:edge-ties},

\begin{align*}
&\Pr_x[\text{hits $C_{p+1}$ after $C_p$ and before any uncovered cluster} | \text{hits $C_p$ before any uncovered cluster}]\\
&\ge \max_{y\in C_p} \Pr_y[\text{hits $C_{p+1}$ before any uncovered cluster}]\\
&\ge 1-\beta
\end{align*}

Combining these inequalities shows that

\begin{align*}
\Pr_x[\text{hits $B_A$ before any uncovered cluster}]&\ge (1 - \beta)^\ell\\
&\ge (1 - \ell\beta)\\
&\ge (1 - (\log n)\beta)\\
&\ge 1/2
\end{align*}

as desired.
\end{proof}

\subsubsection{Number of possible alternations before coverage}

Next, we show that any two clusters can alternate a small number of times before both being covered. This is useful for bounding both the length of $(s_i)_i$ at the end of the algorithm and the number of edges added to $T$.

Specifically, we show the following:

\begin{proposition}\label{prop:coverage-alternation}
Consider two clusters $C$ and $C'$ in some graph $G$. Suppose that the $C$ and $C'$ have effective resistance diameter $\gamma R$ and that $\min_{u\in C,v\in C'} \texttt{Reff}_G(u,v)\ge \gamma R$. Then, for any $x\in C$ or $C'$ and any $z > 0$,

$$\Pr_x[\text{$C$ or $C'$ is not covered after alternating between them more than $T$ times}] \le n\left(\frac{1}{\gamma-4}\right)^{\tau}$$
\end{proposition}

\begin{proof}
Consider two vertices $y_0,y_1\in C$. By Lemmas \ref{lem:pot-bound} and \ref{lem:well-sep}

$$\Pr_{y_0}[\text{hit $C'$ before visiting $y_1$}]\frac{R}{(\gamma - 4)R} = \frac{1}{\gamma - 4}$$

In particular, by the Markov property,

\begin{align*}
&\Pr_x[\text{$y$ not visited before alternating between $C$ and $C'$ $T$ times}]\\
&\le \prod_{t=0}^{T-1} \max_{y'\in C}\Pr_{y'}[\text{$y$ not visited before alternating between $C$ and $C'$ once}]\\
&\le \left(\frac{1}{\gamma - 4}\right)^{\tau}\\
\end{align*}

If $C$ or $C'$ has not been covered, then there exists a vertex in $C$ or $C'$ that has not yet been visited. Therefore, by a union bound, the probability that there exists an unvisited vertex in $C$ or $C'$ after $T$ alternations is at most $(1/(\gamma-4))^{\tau}$, as desired.
\end{proof}

For all subsequent sections, set $\tau = 5(\log (\alpha m))/(\log (\gamma - 4))$. This ensure that any pair of clusters is covered with probability at least $1 - n(1/(m\alpha)^5) \ge 1 - (m\alpha)^4$ after the random walk alternates between both of them $T$ times. In particular, the following property holds:

\begin{definition}[Fresh sequences]
Consider the two sequences $(s_i)_i$ and the supersequence $(s_j')_j$ defined using the algorithm $\InterclusterEdges$. This pair of sequences is called \emph{fresh} if any cluster $C\in \mc F$'s appearances in $(s_i)_i$ do not alternate with some cluster $C'$ more than $\tau$ times in the supersequence $(s_j')_j$.
\end{definition}

We exploit Proposition \ref{prop:coverage-alternation} to show the following:

\begin{proposition}\label{prop:fresh-hp}
The pair of sequences $(s_i)_i$ and $(s_j')_j$ is fresh with probability at least $1 - 1/(m\alpha)^2$.
\end{proposition}

\begin{proof}
$(s_i)_i$ is a sequence of uncovered clusters. Therefore, by Proposition \ref{prop:coverage-alternation}, the appearances of a cluster $C$ in $(s_i)_i$ cannot alternate with the appearances of some other cluster $C'$ in $(s_j')_j$ more than $\tau$ times with probability at least $1 - 1/(m\alpha)^4$. Union bounding over all pairs of clusters $C$ and $C'$ yields the desired result.
\end{proof}

\subsubsection{Hitting probabilities are similar within a cluster}

We now show the following fact, which will allow us to turn a maximum probability lower bound into a bound on the minimum. In the following proposition, think of $\mc F'$ as being the set of uncovered clusters besides $C_v$:

\begin{proposition}\label{prop:prob-max-to-min}
Let $C,C'\in \mc F$ and $\mc F'\subseteq F$, with $C,C'\notin \mc F'$. Then

$$(1 - 1/(\gamma - 4))\max_{y\in C} \Pr_y[\text{hit $C'$ before any cluster in $\mc F'$}] \le \min_{y\in C} \Pr_y[\text{hit $C'$ before any cluster in $\mc F'$}]$$
\end{proposition}

\begin{proof}
Let $y_0\in C$ be the maximizer of $\Pr_y[\text{hit $C'$ before any cluster in $\mc F'$}]$. By Lemmas \ref{lem:pot-bound} and \ref{lem:well-sep}, $\Pr_y[\text{hit $C'$ before $y_0$}] \le \frac{R}{(\gamma - 4)R} = \frac{1}{\gamma - 4}$ for any $y\in C$. Furthermore,

\begin{align*}
\Pr_y[\text{hit $C'$ before any cluster in $\mc F'$}] &\ge \Pr_y[(\text{hit $C'$ before any cluster in $\mc F'$})\wedge(\text{hit $y_0$ before $C'$})]\\
&= \Pr_y[\text{hit $C'$ before any cluster in $\mc F'$} | \text{hit $y_0$ before $C'$}]\\
&\Pr_y[\text{hit $y_0$ before $C'$}]\\
&= \Pr_{y_0}[\text{hit $C'$ before any cluster in $\mc F'$}]\Pr_y[\text{hit $y_0$ before $C'$}]\\
&\ge \Pr_{y_0}[\text{hit $C'$ before any cluster in $\mc F'$}](1 - 1/(\gamma - 4))\\
\end{align*}

as desired.
\end{proof}

\subsubsection{Line \ref{line:forest-building} always adds an edge}

Here, we assume that the pair of sequences $(s_i)_i,(s_j')_j$ is fresh in light of Proposition \ref{prop:fresh-hp}. Subject to this assumption, we examine the effect of Lines \ref{line:e-properties} and \ref{line:cutting} on $\mc E$. Specifically, we show that passing $\mc E$ through these lines does not make it empty in a sequence of four propositions.

\begin{proposition}\label{prop:e1-size}
Let $\mc E_1$ be the multiset of clusters $A\in \mc F$ for which Property (1) on Line \ref{line:e-properties} is satisfied. Then $|\mc E_1|\ge \zeta - 1$.
\end{proposition}

\begin{proof}
The number of clusters visited immediately after each of the last $\zeta$ visits to $C_v$ is at least $\zeta-1$, since all prior visits to $C_v$ were followed immediately by a visit to another cluster. Therefore, $|\mc E_1|\ge \zeta - 1$, as desired.
\end{proof}

We use the following supermartingale concentration inequality to bound the drop from $\mc E_1$ to $\mc E_2$:

\begin{theorem}[Theorem 29 in \cite{CL06} with $\phi_i = 0$ for all $i$ and $M = 0$]\label{thm:supermartingale}
Suppose a supermartingale $X$ associated with filter $\textbf{F}$ satisfies, for all $1\le i\le n$,

$$\textbf{Var}(X_i|\mc F_{i-1})\le \sigma_i^2$$

and

$$\textbf{E}[X_i|\mc F_{i-1}]\le a_i$$

Then

$$\Pr[X_n\le X_0 - \lambda]\le e^{\frac{-\lambda^2}{2\sum_{i=1}^n (\sigma_i^2 + a_i^2)}}$$
\end{theorem}

\begin{proposition}\label{prop:e2-size}
Let $\mc E_2$ be the submultiset of $\mc E_1$ for which Property (2) on Line \ref{line:e-properties} is satisfied. Suppose that $|\mc E_1| > 256\log (m\alpha)$. Then $|\mc E_2|\ge |\mc E_1|/4$ with probability at least $1 - 1/(m\alpha)^3$.
\end{proposition}

\begin{proof}
Let $A_1,A_2,\hdots,A_{|\mc E_1|}$ be the (possibly nondistinct) clusters in $\mc E_1$ listed in visitation order. For each $i\in \{1,2,\hdots,|E_1|\}$, let $X_i$ be the indicator variable of the event $$\{B_{A_i} \text{ is visited between $A_i$ and the next visit to $C_v$}\}$$ and let $Z_i = (\sum_{j\le i} X_i) - i/2$. By Proposition \ref{prop:root-ties} applied to $A\gets A_i$, $\textbf{E}[X_i]\ge 1/2$ for all $i$. This means that $\{Z_i\}$ is a supermartingale with stepwise variance 1 and change in mean at most 1. By Theorem \ref{thm:supermartingale},

$$\Pr[Z_i = Z_i - Z_0\le -i/4]\le \exp(-(i/4)^2/(4i)) = \exp(-i/64)$$

Since $|\mc E_1| > 256\log (m\alpha)$, at least $|\mc E_1|/4$ $A_i$s visit $B_{A_i}$ before returning to $C_v$ with probability at least $1 - 1/(m\alpha)^4$. Union bounding over all $C_v$s gives the desired result that $|\mc E_2|\ge |\mc E_1|/4$ with high probability.
\end{proof}

\begin{proposition}\label{prop:e3-size}
Let $\mc E_3$ be the submultiset of $\mc E_2$ for which Property (3) on Line \ref{line:e-properties} is satisfied. Then $|\mc E_3|\ge |\mc E_2| - 10(\log^2 (m\alpha))$ with probability at least $1 - 1/(m\alpha)^4$.
\end{proposition}

\begin{proof}
Suppose that

$$\Pr[\text{ all roots for clusters in $\mc E_2$ hit $C_v$ before another uncovered cluster }]\ge 1/(m\alpha)^3$$

Let $x_i$ be a random variable denoting the first vertex that the random walk visits in $B_{A_i}$. Then

\begin{align*}
&\Pr[\text{ all roots for clusters in $\mc E_2$ hit $C_v$ before another uncovered cluster }]\\
&= \Pr[\wedge_{i:A_i\in \mc E_2} \text{random walk starting at $x_i$ hits $C_v$ before another uncovered cluster}]\\
\end{align*}

By the Markov property applied after writing the above wedge as a product of conditional probabilities,

\begin{align*}
&\prod_{i:A_i\in \mc E_2} \max_{y\in B_{A_i}} \Pr_y[\text{hits $C_v$ before another uncovered cluster}]\\
&\ge \Pr[\text{ all clusters in $\mc E_2$ hit $C_v$ before another uncovered cluster }]\\
&\ge 1/(m\alpha)^4\\
\end{align*}

Therefore, for all but $10 \log^2 (m\alpha)$ of the $A_i$s,

$$\max_{y\in B_{A_i}} \Pr_y[\text{hits $C_v$ before another uncovered cluster}] \ge 1 - 1/(4\log n)$$

For each of these $A_i$s, apply Proposition \ref{prop:prob-max-to-min} with $C\gets B_{A_i}$, $C'\gets C_v$,

and $\mc F'\gets \{\text{uncovered clusters besides $C_v$}\}$ to show that

$$\min_{y\in B_{A_i}} \Pr_y[\text{hits $C_v$ before another uncovered cluster}] \ge (1 - 1/(4\log n))(1 - 1/(\gamma - 4)) \ge 1 - \beta$$

Therefore, each of these $A_i$s is also in $\mc E_3$. As a result, $|\mc E_3|\ge |\mc E_2| - 10(\log^2 (m\alpha))$ if

$$\Pr[\text{ all roots for clusters in $\mc E_2$ hit $C_v$ before another uncovered cluster }]\ge 1/(m\alpha)^4$$

The other case happens with probability at most $1/(m\alpha)^4$ for each $C_v$. Union bounding over all clusters in $\mc F$ gives the desired result.
\end{proof}

\begin{proposition}\label{prop:leftover-size}
Let $\mc E_{final}$ be the submultiset of $\mc E_3$ that remains in $\mc E$ after Line \ref{line:cutting}. Suppose that the pair $(s_i)_i,(s_j')_j$ is fresh. Then $|\mc E_{final}|\ge |\mc E_3| - \tau (\log n)$.
\end{proposition}

\begin{proof}
We start by noticing that no arborescence in $T$ can contain more than $\tau$ clusters (possibly repeated) in the multiset $\mc E_3$. Otherwise, by Property (2), $C_v$ would alternate with the root of that arborescence more than $\tau$ times, which contradicts the freshness of the pair $(s_i)_i,(s_j')_j$. Therefore, each iteration of the while loop on Line \ref{line:cutting} can only remove $\tau$ elements from $\mc E_3$. Each iteration reduces the size of the arborescence formed by joining $\mc E$'s arborescences to $C_v$ by at least a factor of 2. Therefore, the while loop can only execute for $\log n$ iterations. Therefore, $|\mc E_{final}|\ge |\mc E_3| - \tau(\log n)$, as desired.
\end{proof}

We now combine all of these propositions into one key observation:

\begin{corollary}\label{cor:e-size}
$\mc E_{final}\ne \emptyset$.
\end{corollary}

\begin{proof}
Since $\zeta\ge 100\log^2(m\alpha)$, the size reductions due to Propositions \ref{prop:e1-size}, \ref{prop:e2-size}, \ref{prop:e3-size}, and \ref{prop:leftover-size} ensure that $|\mc E_{final}| > 0$, as desired.
\end{proof}

As a result, $\InterclusterEdges$ adds an edge to $T$ every time the if statement on Line \ref{line:repeat-visits} is triggered. Therefore, the $\zeta$ elements that are removed from $(s_i)_i$ can be charged to the addition of one edge to $T$. By Invariant \ref{inv:forest}, $T$ can only have $|\mc F| - 1$ edges, which means that only $(|\mc F|-1)\zeta$ letters can be removed from $(s_i)_i$ over the course of the algorithm. 

\subsubsection{Tying the parts together}

The upside to removing elements of $(s_i)_i$ is that the resulting sequence does not contain more than $\zeta$ consecutive occurrences of any cluster. Since $(s_i)_i,(s_j')_j$ is a fresh pair, $(s_i)_i$ with all consecutive duplicates removed is a $(|\mc F|,\tau)$-Davenport-Schinzel string. Therefore, its length can be bounded using Theorem \ref{thm:davenport-schinzel}. Since each edge in a random spanning tree can be charged to a specific visit in $s_i$ or a visit removed from $s_i$, we are done.

\begin{proof}[Proof of Lemma \ref{lem:well-sep-lev-score}]
We start with the high probability regime in which all of the above propositions hold. By Theorem \ref{thm:edge-marginal},

$$\sum_{e\in E(C,C'), C\ne C'\in \mc F} \frac{\texttt{Reff}_H(e)}{r_e^H} = \textbf{E}[\text{random spanning tree edges in $H$ between clusters in $\mc F$}]$$

By Theorem \ref{thm:aldous-broder}, the number of spanning tree edges between clusters in $\mc F$ is at most the number of clusters appended to the sequence $(s_i)_i$. The number of clusters appended to $(s_i)_i$ is equal to the number removed over the course of the algorithm plus the number left at the end of the algorithm. We bound these two quantities separately:

\underline{Clusters removed from $(s_i)_i$.} Only Line \ref{line:deletions} removes elements from $(s_i)_i$. By Corollary \ref{cor:e-size}, every $\zeta$ deletions from $(s_i)_i$ result in the addition of at least one edge to $T$. By Invariant \ref{inv:forest}, only $|\mc F|-1$ edges can be added to $T$, so at most $\zeta |\mc F|$ elements of $(s_i)_i$ are removed over the course of the algorithm.

\underline{Clusters remaning in $(s_i)_i$ at the end.} At the end of $\InterclusterEdges$, no cluster appears more than $\zeta$ times consecutively in $(s_i)_i$ by the if condition on Line \ref{line:repeat-visits}. Therefore, the sequence $(s_i'')_i$ obtained by removing consecutive duplicates of clusters in $(s_i)_i$ is at most $\zeta$ times shorter than $(s_i)_i$. Furthermore, if $(s_i)_i,(s_j')_j$ is a fresh pair (which by Proposition \ref{prop:fresh-hp} occurs with high probability), $(s_i'')_i$ is an $(|\mc F|,\tau)$-Davenport-Schinzel string. Therefore, by Theorem \ref{thm:davenport-schinzel}, the length of $(s_i'')_i$ is at most $|\mc F|C_0(m,\tau)\le |\mc F|(\muapp/(2\zeta))$. Therefore, the length of $(s_i)_i$ at the end of the algorithm is at most $\zeta|\mc F|m^{o(1)}$.

\underline{Combining the parts.} Therefore, with probablity at least $1 - 1/(m\alpha)^2$, the number of random spanning tree edges in $H$ between clusters in $\mc F$ is at most

$$\zeta |\mc F|(\muapp/(2\zeta)) + \zeta |\mc F|$$

The maximum number of edges that can be added is at most $n^2$, as there are at most $n$ vertices in $H$. Therefore, the expected number of edges is at most

$$\zeta |\mc F|(\muapp/(2\zeta)) + \zeta |\mc F| + n^2/(m\alpha)^2\le \muapp |\mc F|$$

as desired.
\end{proof}

\subsection{Proof of Proposition \ref{prop:separated-large}}

\propseparatedlarge*

\begin{proof}
Let $J := \texttt{Schur}(H_{\mc C},(V(H_{\mc C})\setminus X)\cup(\cup_{C\in \mc F} C))$. By Lemma \ref{lem:monotone-schur} with $I\gets H_{\mc C}$, $S_0\gets Y$, $S_0'\gets \cup_{C\in \mc F} C$, and $S_1\gets V(H_{\mc C})\setminus X$, and $S_2\gets \emptyset$,

$$\sum_{C\in \mc F} c^J(E_J(C,(V(H_{\mc C})\setminus X)))\le \xi$$

By Lemma \ref{lem:outside-schur} with $I\gets H_{\mc C}$, $S_0\gets C$, $S_1\gets C'$, $S_2\gets (V(H_{\mc C})\setminus X)\cup(\cup_{C''\in \mc F, C''\ne C,C'} C'')$, and $S_2'\gets S_2\setminus (V(H_{\mc C})\setminus X)$ and Lemma \ref{lem:well-sep}, for any $C\in \mc F$

\begin{align*}
\sum_{C'\in \mc F\setminus \{C\}} c^J(E_J(C,C')) &\le \sum_{C'\in \mc F\setminus \{C\}} c^I(E_I(C,C'))\\
&= \sum_{e\in E_I(C,C')} c_e^I\\
&\le \frac{1}{(\gamma - 4)R} \sum_{e\in E_I(C,C')} \texttt{Reff}_I(e) c_e^I\\
&= \frac{\Delta_{\mc F}(C)}{(\gamma - 4)R}
\end{align*}

By definition of $S_C$, $S_{\mc F'}(p,C)\subseteq S_C$ for all $C\in \mc F$. Therefore, by Lemma \ref{lem:p-cond},

\begin{align*}
\sum_{C\in \mc F} c^{\mc C}(S_C) &\le \sum_{C\in \mc F} c^{\mc C}(S_{\mc F'}(p,C))\\
&\le \sum_{C\in \mc F} \frac{c^J(E_J(C,\cup_{C'\in \mc F', C'\ne C} C'))}{p}
\end{align*}

Substitution shows that

\begin{align*}
\sum_{C\in \mc F} \frac{c^J(E_J(C,\cup_{C'\in \mc F', C'\ne C} C'))}{p} &= \sum_{C\in \mc F} \frac{c^J(E_J(C,\cup_{C'\in \mc F, C'\ne C} C'))}{p} + \sum_{C\in \mc F} \frac{c^J(E_J(C,V(H_{\mc C})\setminus X))}{p}\\
&\le \left(\sum_{C\in \mc F}\frac{\Delta_{\mc F}(C)}{p(\gamma-4) R}\right) + \frac{\xi}{p}\\
\end{align*}

as desired.
\end{proof}

\subsection{Proof of Proposition \ref{prop:well-spaced}}

\propwellspaced*

\begin{proof}
Consider a vertex $x\in C'\cap S_C$. This vertex exists by the ``Intersection'' condition of ties. Let $J$ be the graph obtained by identifying $C$ to a vertex $s$ and all vertices in clusters of $\mc F\setminus \{C\}$ to a vertex $t$. Then $b_{st}^T L_J^+ b_{sx}\le pb_{st}^T L_J^+ b_{st}$ since $S_C\subseteq S_{\mc F}(p,C)$.

Let $y$ be any vertex in $C'$. Let $H_0 = \texttt{Schur}(H,(\cup_{C''\in \mc F} C'')\cup \{x,y\})$ and let $H_1$ be the graph obtained by identifying $C$ to $s$. By Lemma \ref{lem:well-sep} and the ``Well-Separatedness'' condition of ties, $b_{sx}^T L_{H_1}^+ b_{sx}\ge \beta_0 R = 100 R$ and $b_{sy}^T L_{H_1}^+ b_{sy}\ge 90R$. This means that $r_{sx}^{H_1}\ge 100R$ and that $r_{sy}^{H_1}\ge 90R$.

Let $J_0 = \texttt{Schur}(J,\{s,x,y,t\})$. We now reason about the $s-t$ electrical flow on the edges $\{s,x\}$ and $\{x,y\}$. Since $J$ is obtained by identifying $\mc F\setminus \{C\}$ to $t$, the $\{s,x\}$ and $\{s,y\}$ conductances are uneffected in obtaining $J_0$ from $H_1$, so $r_{sx}^{J_0} = r_{sx}^{H_1}\ge 100R$ and $r_{sy}^{J_0} = r_{sy}^{H_1} \ge 90R$. Let $f\in \mathbb{R}^{E(J_0)}$ be the $s-t$ unit electrical flow vector for $J_0$. Since flow times resistance is the potential drop,

$$f_{sx} = \frac{b_{st}^T L_{J_0}^+ b_{sx}}{r_{sx}^{J_0}}\le \frac{b_{st}^T L_{J_0}^+ b_{st}}{100R}$$

By Rayleigh monotonicity, $b_{xy}^T L_{J_0}^+ b_{xy}\le 10R$. Therefore, either $r_{xy}^{J_0}\le 30R$ or $r_{xt}^{J_0} + r_{ty}^{J_0}\le 30R$ since $r_{sx}^{J_0},r_{sy}^{J_0}\ge 90R$. We start by showing that the latter case does not happen. In the latter case, the potential drop from $s$ to $t$ is

\begin{align*}
b_{st}^T L_J^+ b_{st} &= b_{st}^T L_{J_0}^+ b_{st}\\
&\le b_{st}^T L_{J_0}^+ b_{sx} + r_{xt}^{J_0}f_{xt}\\
&\le pb_{st}^T L_J^+ b_{st} + r_{xt}^{J_0}f_{sx}\\
&\le (1/2 + 3/10)b_{st}^T L_J^+ b_{st}\\
&< b_{st}^T L_J^+ b_{st}\\
\end{align*}

which is a contradiction. Therefore, $r_{xy}^{J_0}\le 30R$. In this case, if $y$'s potential is less than $x$'s, we are done. Otherwise, by flow conservation,

$$f_{xy}\le f_{sx}\le \frac{b_{st}^T L_{J_0}^+ b_{st}}{100R}$$

and the $x-y$ potential drop is therefore at most

$$r_{xy}^{J_0}f_{xy}\le (30R)\frac{b_{st}^T L_{J_0}^+ b_{st}}{100R}\le (3/10)b_{st}^T L_{J_0}^+ b_{st}$$

In particular, $y$'s normalized potential is at most $3/10$ greater than $x$'s, as desired.
\end{proof}

\section{Deferred proofs for Section \ref{sec:random-walk}}

\subsection{Bound on the number of random walk steps before covering a neighborhood}\label{sec:cover-time}

We prove Lemma \ref{lem:walk-bound} in this section. We exploit the Matthews trick in this proof \cite{M88}:

\lemwalkbound*

\begin{proof}[Proof of Lemma \ref{lem:walk-bound}]
Let $B$ be the set of vertices in $I$ with $I$-effective resistance distance at most $R$ of $u$. For any two vertices $x,y\in B$, the random walk from $x$ to $y$ in $I$ traverses $f$ from $u$ to $v$ at most $\texttt{Reff}_I(x,y)c_f^I$ times by Theorem \ref{thm:edge-visits}. By the triangle inequality, this is at most $2Rc_f^I$.

Pick a random bijection $\pi:\{1,2,\hdots,|B|-1\}\rightarrow B\setminus \{u\}$ of the vertices in $B\setminus \{u\}$. All traversals across $f$ from $u\rightarrow v$ within distance $R$ of $S$ occur between the first visit to $u$ and the last first visit to any vertex in $B$. Let $\tau_i$ be the random variable denoting the number of $u\rightarrow v$ $f$ traversals before the last first visit to $\{\pi_1,\hdots,\pi_i\}$. Let $\tau_0$ be the the first visit to $u$. Notice that for all $i > 0$,

\begin{align*}
\textbf{E}[\tau_i - \tau_{i-1}]&\le \Pr_{\pi}[t_{\pi_i} < t_{\{\pi_{i-1},\hdots,\pi_1\}}] \max_{x,y\in B} \textbf{E}[\text{ $f$ traversals from $x\rightarrow y$ }]\\
&\le \Pr_{\pi}[t_{\pi_i} < t_{\{\pi_{i-1},\hdots,\pi_1\}}] 2Rc_f^I\\
&\le \frac{1}{i} 2Rc_f^I\\
\end{align*}

Summing over all $i$ shows that

$$\textbf{E}[K] = \textbf{E}[\tau_i] = \tau_0 + \sum_{i=1}^{|B|-1} \textbf{E}[\tau_i - \tau_{i-1}] \le 0 + O(\log n) Rc_f^I$$

as desired.
\end{proof}

\section{Deferred proofs for Sections \ref{sec:choose-parts} and \ref{sec:fix-reduction}}

\subsection{Ball splitting after deletions}

We now prove the following:

\lemballsplit*

Our proof uses Lemma \ref{lem:well-sep-lev-score} to reduce to the case where both (1) $\mc D$ just consists of one cluster $C$ and (2) $F = (\partial C)\cup F'$, where $F'\subseteq E(C)$. We deal with this case using a sparsification technique that peels many low-stretch spanners off of the graph to show that the effective resistance between two particular vertices is low \cite{KAP15}.

\subsubsection{The case where $\mc D$ just contains one cluster and $F = (\partial C)\cup F'$}

Before proving this result, we discuss low-stretch sparse spanners \cite{ADDJS93}:

\begin{theorem}[\cite{ADDJS93}]\label{thm:sp-spanner}
Any unweighted graph $G$ has a subgraph $H$ with the following two properties:

\begin{itemize}
\item (Distortion) For any two vertices $u,v\in V(G)$, $d_H(u,v)\le 2(\log n)d_G(u,v)$, where $d_I$ is the shortest path metric on the vertices of the  graph $I$.
\item (Sparsity) $H$ has at most $100 n\log n$ edges.
\end{itemize}
\end{theorem}

We use this as a tool to sparsify $G$ in a way that ignores the deleted edges in the deleted set $F$.

\begin{proposition}\label{prop:one-ball}
Consider a resistance-weighted graph $G$. Let $C$ be a cluster with effective resistance diameter $R$ in $G$. Let $F\subseteq E(G)$ and suppose that $k = |F|$. Furthermore, suppose that $F$ is the union of $\partial C$ and a set of edges $F'$ with both endpoints in $C$.

Let $H = G\backslash F$. Then $C$ is the union of at most $O(k\log^6 n)$ clusters in $H$ that have effective resistance diameter $O(R\log^6 n)$. 
\end{proposition}

\begin{proof}[Proof of Proposition \ref{prop:one-ball}]
The following iterative algorithm, $\PartitionBall$, partitions $C$ into the desired clusters.

\begin{algorithm}[H]
\SetAlgoLined
\DontPrintSemicolon

    \KwData{A graph $G$, a cluster $C$, and a set of edges $F$ to delete from $G$}

    \KwResult{A partition $C_1\cup\hdots\cup C_{\tau} = C$}

    $H\gets G\backslash F$\;

    \tcp{set of current centers}
    $K\gets V(C)$\;

    $i\gets 0$\;

    \While{$|K| \ge 31k\log n$}{

        $I_i\gets \texttt{Schur}(G,V(F)\cup K)$\;

        $K'\gets K\setminus V(F)$\;

        \ForEach{$j$ from 0 to $2 \log n$}{

            $B_j\gets $ the set of edges $e\in E(I_i)\setminus F$ with $2^j \texttt{Reff}_{I_i}(e)\le r_e^{I_i}\le 2^{j+1}\texttt{Reff}_{I_i}(e)$\;

        }

        $j^*\gets $ the index $j\in \{0,1,\hdots,2\log n\}$ for which the number of vertices incident with leverage score at least $1/(2\log n)$ in $B_j$ is maximized\;

        $I_i'\gets I_i[B_{j^*}]$\;

        \ForEach{of $\lfloor 2^{j^*}/(64\log^3 n)\rfloor$ rounds}{

            $J\gets $ a $2(\log n) - 1$ spanner of $I_i'$ as per Theorem \ref{thm:sp-spanner}\;

            $I_i'\gets I_i'\setminus E(J)$\;

        }

        $K\gets $ the minimum dominating set of the unweighted version of $I_i'$\;

        $i\gets i+1$\;

    }

    \Return Voronoi diagram of $C$ with respect to $K$ in the metric $\texttt{Reff}_{G\setminus F}$\;

\caption{$\PartitionBall(G,C,F)$, never executed}
\end{algorithm}

Now, we analyze this algorithm. We start by showing that the while loop only executes $O(\log^2 n)$ times. Start by noticing that $|K'|\ge 30k\log n\ge |K|/2$ because $|V(F)|\le 2k$ and $|K|\ge 31k\log n$.

For an edge $e\in I_i$, let $z_e = \texttt{lev}_{I_i}(e)$. For a set $S\subseteq E(I_i)$, let $z(S) = \sum_{e\in S} z_e$. The total $z$-weight of edges incident with each vertex in $I_i$ is at least 1 since the total leverage score incident with a vertex in any graph is at least 1. Since there are only $2\log n$ $B_j$s, at least one contributes $1/(2\log n)$ to the $z$-weight incident with each vertex. Therefore, the number of vertices in $K'$ with more than $1/(2\log n)$ incident $z$-weight in $I_i'$ is at least $|K'|/(2\log n)\ge (|K| + k)/(2\log n)$ before the second ForEach loop. Therefore, by the bucketing resistance lower bound, $|E(I_i')|\ge (|K| + k)2^{j^*}/(2\log n)$ before the second ForEach loop.

The spanners created during the second ForEach loop delete at most $(2^{j^*}/(64\log^3 n))(2(|K| + k)\log n) = (|K| + k)2^{j*}/(32\log^2 n)$ edges from $I_i'$. Each of the $(|K| + k)/(2\log n)$ high $z$-weight vertices is incident with at least $2^{j^*}/(2\log n)$ edges in $B_{j^*}$. Each edge is incident with two vertices. Notice that the number of vertices that can have all of their incident edges deleted is at most

$$\frac{2((|K| + k)2^{j*}/(32\log^2 n))}{2^{j*}/(2\log n)} = (|K| + k)/(8\log n)$$

so the number of vertices with an incident edge after the second ForEach loop is at least $(|K| + k)/(4\log n)$. The spanners are edge-disjoint, so together they contain at least $2^{j*}/(64\log^3 n)$ disjoint paths with at most $2\log n$ edges between the endpoints of each edge of $I_i'$. In particular, each path has resistance length at most $(2\log n)2^{j*+1}R$ by the bucketing resistance upper bound. Therefore, the effective resistance between the endpoints of any edge left over in $I_i'\backslash F$ is at most $256(\log^4 n) R$.

Since the while loop only executes when $|K|\ge 31k\log n$, the number of vertices in $K'$ with an incident edge in $I_i'$ after the second ForEach loop is at least $8k$, for a total of at least $4k$ disjoint pairs. Removing one side of each $K$ pair certifies that there is a dominating set with size at most $|K| - (|K| + k)/(8\log n) + 2k\le (1 - 1/(16\log n))|K|$. Therefore, the algorithm will only take $16\log^2 n$ iterations.

At the end of those iterations, for every vertex in $C$ there is a connected sequence of $16\log^2 n$ pairs to some vertex in $K$ with each pair having $G\backslash F$ effective resistance distance at most $256(\log^4 n) R$. By the triangle inequality, this means that every vertex in $C$ is within $G\backslash F$ distance $O(R \log^6 n)$ of some vertex in $K$ at the end of the algorithm. This is the desired result.
\end{proof}

\subsubsection{Sparsification}

We now prove a version of spectral sparsification that allows a large part of the graph to not be changed. We are sure that this is known, but provide a proof for completeness.

\begin{proposition}[Subset sparsification]\label{prop:subset-sparsify}
Let $G$ be a weighted graph and let $F\subseteq E(G)$. Let $\ep\in (0,1)$. Then, there exists a graph $H$ with $V(H) = V(G)$ and the following additional properties:

\begin{itemize}
\item (Only $F$ modified) The weights of edges in $G\setminus F$ are not modified.
\item (Sparsity of $F$) $|E(H)\cap F|\le (8(\log n)/\ep^2)(1 + \sum_{e\in F} \texttt{lev}_G(e))$
\item (Spectral approximation) For any vector $d\in \mathbb{R}^n$,

$$(1 - \ep) d^T L_G^+ d \le d^T L_H^+ d \le (1 + \ep) d^T L_G^+ d$$
\end{itemize}
\end{proposition}

The proof is a simple application of matrix Chernoff bounds, which we state here:

\begin{theorem}[Theorem 1.1 in \cite{T12}]\label{thm:matrix-chernoff}
Consider a finite sequence $\{X_k\}_k$ of independent, random, self-adjoint matrices with dimension $n$. Assume that each random matrix satisfies

$$X_k\succeq 0$$

and

$$\lambda_{\max}(X_k)\le R$$

Define

$$\mu_{\min} := \lambda_{\min}(\sum_k \textbf{E}[X_k])$$

and

$$\mu_{\max} := \lambda_{\max}(\sum_k \textbf{E}[X_k])$$

Then

$$\Pr[\lambda_{\min}(\sum_k X_k) \le (1 - \delta)\mu_{\min}] \le n \left(\frac{e^{-\delta}}{(1-\delta)^{1-\delta}}\right)^{\mu_{\min}/R}$$

and

$$\Pr[\lambda_{\max}(\sum_k X_k) \le (1 + \delta)\mu_{\max}] \le n \left(\frac{e^{\delta}}{(1+\delta)^{1+\delta}}\right)^{\mu_{\max/R}}$$
\end{theorem}

\begin{proof}[Proof of Proposition \ref{prop:subset-sparsify}]
Let $q = 20(\log n)/\ep^2$. For each edge $e\in E(G)$, define $q$ matrices $X_e^{(k)}$ for $k\in \{1,2,\hdots q\}$, with

\begin{itemize}
\item $X_e^{(k)} = \frac{c_e^G}{\texttt{lev}_G(e)q} (L_G^+)^{1/2} b_e b_e^T (L_G^+)^{1/2}$ with probability $\texttt{lev}_G(e)$ and 0 otherwise for $e\in F$
\item $X_e^{(k)} = \frac{c_e^G}{q} (L_G^+)^{1/2} b_e b_e^T (L_G^+)^{1/2}$ deterministically for $e\notin F$.
\end{itemize}

Notice that

$$\sum_{e\in E(G)}\sum_{k=1}^q \textbf{E}[X_e^{(k)}] = I$$

so $\mu_{min} = \mu_{max} = 1$. Apply Theorem \ref{thm:matrix-chernoff} with $\delta = \ep$ and $R = 1/q$ (since $\lambda_{\max}(X_e^{(k)}) = c_e^G(b_e^T L_G^+ b_e)/(q\texttt{lev}_G(e)) = 1/q$ for all $e\in F$) to reason about the random matrix $M_H = \sum_{e\in E(G)}\sum_{k=1}^q X_e^{(k)}$ and $L_H = L_G^{1/2} M_H L_G^{1/2}$. We now analyze the guarantees one by one.

\textbf{Only $F$ modified.} The contribution of edge edge $e\notin F$ to $L_H$ is $c_e^G b_e b_e^T$, which is the same as its weight in $G$.

\textbf{Sparsity of $F$.} By standard Chernoff bounds on each edge, the total number of edges from $F$ in $H$ is at most $8(\log n)/\ep^2$.

\textbf{Spectral approximation.} By Theorem \ref{thm:matrix-chernoff},

\begin{align*}
\Pr[(1 + \ep)I \succeq M_H \succeq (1 - \ep)I] &\ge 1 - n\left(\frac{e^{-\ep}}{(1-\ep)^{1-\ep}}\right)^q - n\left(\frac{e^{\ep}}{(1+\ep)^{1+\ep}}\right)^q\\
&\ge 1 - n\left(1 - \ep^2/4\right)^q - n\left(1 - \ep^2/4\right)^q\\
&\ge 1 - 1/n^4\\
\end{align*}

By standard facts about the Loewner ordering, $(1 + \ep)L_G \succeq L_H \succeq (1 - \ep) L_G$ and $(1 + \ep)L_G^+ \succeq L_H^+ \succeq (1 - \ep)L_G^+$, as desired.
\end{proof}

\subsubsection{Tying the parts together}

Now, we prove Lemma \ref{lem:ball-split}. To do this, we exploit the algorithm $\CoveringCommunity$ to split the clusters of $\mc D$ into a small number of well-separated families. All vertices outside of these clusters can be Schur complemented out, as they are irrelevant. By Lemma \ref{lem:well-sep-lev-score}, the intercluster edges have small total leverage score. Proposition \ref{prop:subset-sparsify} allows us to replace these intercluster edges with $m^{o(1)}|\mc D|$ reweighted edges that separate all of the clusters. Therefore, the clusters in each well-separated family can be handled independently using Proposition \ref{prop:one-ball}.

\begin{proof}[Proof of Lemma \ref{lem:ball-split}]
Start by applying Lemma \ref{lem:covering-community} to produce a community

$\mc G \gets \CoveringCommunity(V(H),H,R)$. Form a new community $\mc G'$ by taking each cluster $C$ in a family of $\mc G$ and replacing it with a cluster $C'$ formed by adding all vertices in clusters of $\mc D$ that intersect $C'$. Furthermore, remove all clusters in $\mc G'$ that do not intersect a cluster in $\mc D$.

Each family in $\mc G'$ is $\gammads/3$-well-separated after doing this by the ``Well-separatedness'' guarantee of $\CoveringCommunity$, the ``$R$-community'' guarantee, and the fact that each cluster in $\mc D$ has effective resistance diameter at most $R$. Furthermore, by the ``Covering'' guarantee of $\mc G$, each cluster in $\mc D$ intersects some cluster in a family of $\mc G$. This means that each cluster in $\mc D$ is completely contained in some cluster in some family of $\mc G'$. Therefore, we can focus on each family of $\mc G'$ independently.

The above paragraph effectively reduced the problem to the case where $\mc D$ is a small number of well-separated families of clusters. Now, we exploit this. By Lemma \ref{lem:well-sep-lev-score}, the total leverage score of the intercluster edges in a family is at most $O(C_0(\ell,\gammads)\ell)\le m^{o(1)}\ell$, where $\ell$ is the number of clusters in all families in $\mc G'$. Since each family consists of clusters that contain at least one cluster in $\mc D$, $\ell\le |\mc G||\mc D|\le m^{o(1)}|\mc D|$. By Proposition \ref{prop:subset-sparsify} applied to all intercluster edges, a graph $H_i$ can be made for each family $\mc G_i'\in \mc G$ that spectrally approximates $H$ and only has $m^{o(1)}|\mc D|$ intercluster edges. Add these intercluster edges to $F$ to obtain $F_i$. Each of these edges is incident with at most two clusters in $\mc G_i$.

Apply Proposition \ref{prop:one-ball} to each cluster $C\in \mc G_i$ in the graph $H_i$ with deleted edge set consisting of the edges of $F_i$ incident with the cluster. This shows that $C$ splits into $\tilde{O}(|F_i\cap (C\cup \partial_{H_i} C)|)$ clusters with radius $\tilde{O}(R)$. Add these clusters to $\mc D'$. We now analyze the guarantees one-by-one:

\textbf{Covering.} Each vertex in a cluster in $\mc D$ is in some cluster of $\mc G$, which is in turn covered by the clusters produced using Proposition \ref{prop:one-ball}, as desired.

\textbf{Diameter.} Since $\mc G$ is $R$-bounded, each cluster in $\mc G$ has $H$-effective resistance diameter at most $m^{o(1)}R$. By the ``Spectral approximation'' guarantee of Proposition \ref{prop:subset-sparsify}, each cluster in $\mc G$ has $H_i$-effective resistance diameter at most $(1 + \ep)m^{o(1)}R = m^{o(1)}R$ as well. Proposition \ref{prop:one-ball} implies that after deleting the edges in $F$, all clusters added to $\mc D'$ have $H_i\setminus F = H\setminus F$-effective resistance diameter $O(\log n)m^{o(1)}R = m^{o(1)}$, as desired.

\textbf{Number of clusters.} Summing up $\tilde{O}(|F_i\cap (C\cup \partial_{H_i} C)|)$ over all clusters in $\mc G_i$ shows that applications of Proposition \ref{prop:one-ball} add $\tilde{O}(|F_i|)\le m^{o(1)}(|\mc D| + |F|)$ clusters to $\mc D'$. Doing this for all $|\mc G'|\le m^{o(1)}$ families in $\mc G$ yields the desired result.
\end{proof}

\section{Deferred proofs for Section \ref{sec:slow-fix}}

\subsection{Stable objective subresults}\label{sec:stable-appendix}

We also use the following folklore fact about leverage scores:

\begin{remark}\label{rmk:lev-sum}
In any graph $G$ with $n$ vertices,

$$\sum_{e\in E(G)} \texttt{lev}_G(e) = n-1$$
\end{remark}

\subsubsection{Preliminaries for first-order terms}

\begin{proposition}\label{prop:elec-flow-cont}
Consider a graph $G$, a vertex $x\in V(G)$, a set $Y\subseteq V(G)$ with $x\notin Y$, and $\gamma\in [0,1]$. Calculate electrical potentials with boundary conditions $p_x = 0$ and $p_w = 1$ for every $w\in Y$. Suppose that there is no edge $\{u,v\}\in E(G)$ with $p_u < \gamma$ and $p_v > \gamma$.

Let $U$ be the set of vertices $u$ with $p_u\le \gamma$. Make a new electrical flow with boundary conditions $q_u = 0$ for $u\in U$ and $q_w = 1$ for all $w\in Y$. Then for any $w\in Y$,

$$\sum_{w'\in N_G(w)} \frac{q_w - q_{w'}}{r_{ww'}} = \frac{1}{1-\gamma}\sum_{w'\in N_G(w)} \frac{p_w - p_{w'}}{r_{ww'}}$$
\end{proposition}

\begin{figure}
\includegraphics[width=1.0\textwidth]{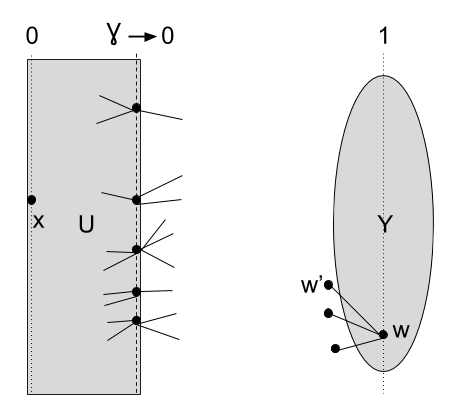}
\caption{Depiction of Proposition \ref{prop:elec-flow-cont}. The proof follows from the observation that identifying the vertices of $U$ to $x$ does not change the $x-Y$ electrical flow (with $y$ identified) on any edge outside of $U$.}
\label{fig:incoming-potential-scale}
\end{figure}

\begin{proof}
For any edge $e=(u,v)$ with $p_v \ge p_u$, let $f_e = \frac{p_v - p_u}{r_e}$,

$$g_e = \begin{cases}\frac{1}{1-\gamma}f_e & \text{ if $p_u\ge \gamma$}\\ 0 &\text{ if $p_v\le \gamma$ } \end{cases}$$

and

$$q_w = \max(0,1-\frac{1-p_w}{1-\gamma})$$

The cases for $g_e$ form a partition of the edges thanks to the condition that no edge can cross the $\gamma$ potential threshold. One can check that $g_e$ is an electrical flow with respect to the potentials $q_w$, $q_w = 1$ for any $w\in Y$, and $q_u = 0$ for all $u\in U$. Therefore, since no edge crosses the $\gamma$ threshold,

$$\sum_{w'\in N_G(w)} \frac{q_w - q_{w'}}{r_{ww'}} = \frac{1}{1-\gamma}\sum_{w'\in N_G(w)} \frac{p_w - p_{w'}}{r_{ww'}}$$

as desired.
\end{proof}

\begin{proposition}\label{prop:cond-pot}
Consider a vertex $x$ and a set $Y$ in a graph $G$. Let $Z$ be a set of edges with both endpoints having electrical potential at most $\gamma$, with $p_x = 0$ and $p_y = 1$ in $G/Y$ with identification $y$. Let $G'$ be a graph obtained by deleting and contracting edges in $Z$ in an arbitrary way. Let $H' = \texttt{Schur}(G',\{x\}\cup Y)$ and $H = \texttt{Schur}(G,\{x\}\cup Y)$. Then for any $w\in Y$,

$$r_{xw}^{H'}\ge (1 - \gamma) r_{xw}^H$$
\end{proposition}

\begin{figure}
\includegraphics[width=1.0\textwidth]{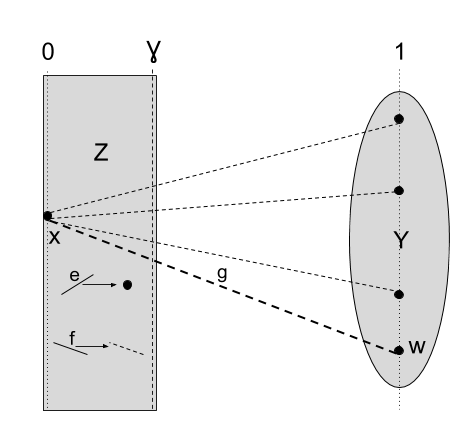}
\caption{Depiction of Proposition \ref{prop:cond-pot}. Contracting edges like $e$ and deleting edges like $f$ can only decrease resistances of edges like $g$ by a small amount.}
\label{fig:schur-comp-change}
\end{figure}

\begin{proof}
Let $G_0$ be a graph with every edge $\{u,v\}$ with $p_u < \gamma$ and $p_v > \gamma$ subdivided at potential $\gamma$. Let $U$ be the set of vertices $u$ with $p_u\le \gamma$. Let $H_0 = \texttt{Schur}(G_0,Y\cup U)$ and $H_0' = \texttt{Schur}(G_0\backslash Z, Y\cup U)$. Since $Z\subseteq E(G[U])$, $H_0' = H_0\backslash Z$, so for any $w\in Y$,

$$\sum_{u\in U} \frac{1}{r_{wu}^{H_0'}} = \sum_{u\in U} \frac{1}{r_{wu}^{H_0}}$$

By Proposition \ref{prop:elec-flow-cont} applied to $H_0$, for any $w\in Y$,

$$\sum_{u\in U} \frac{1}{r_{wu}^{H_0}} = \frac{1}{1-\gamma}\frac{1}{r_{wx}^H}$$

Obtain $H'$ from $H_0'$ by eliminating vertices of $U\backslash \{x\}$ one by one. Let $U_0 = U = \{u_0,u_1,\hdots,u_{\ell},x\}$, $U_i = \{u_i,\hdots,u_{\ell},x\}$, and $H_i'$ be the graph obtained by eliminating $u_0,u_1,\hdots,u_{i-1}$. By the formula for elimination of one vertex,

\begin{align*}
\sum_{u\in U_{i+1}} c_{wu}^{H_{i+1}'} &= \sum_{u\in U_{i+1}} \left(c_{wu}^{H_i'} + \frac{c_{wu_i}^{H_i'}c_{u_iu}^{H_i'}}{\sum_{v\in U_{i+1}\cup Y} c_{u_iv}^{H_i'}}\right)\\
&\le c_{wu_i}^{H_i'} + \sum_{u\in U_{i+1}} c_{wu}^{H_i'}\\
&= \sum_{u\in U_i} c_{wu}^{H_i'}\\
\end{align*}

Chaining these inequalities shows that

$$c_{wx}^{H'} \le \frac{1}{1-\gamma} c_{wx}^H$$

since $U_{\ell+1} = \{x\}$ and $H_{\ell+1}' = H'$. Taking reciprocals shows the desired result.
\end{proof}

\begin{proposition}\label{prop:deg-split}
Consider two sets $X,Y\subseteq V(G)$ with $X\cap Y = \emptyset$. Let $G_0$ be a graph obtained by splitting each edge $e\in \partial_G X$ into a path of length 2 consisting of edges $e_1$ and $e_2$ with $r_{e_1} = r_{e_2} = r_e/2$, with $e_1$ having an endpoint in $X$. Add $X$ and all endpoints of edges $e_1$ to a set $X_0$. Let $x$ and $x_0$ denote the identifications of $X$ and $X_0$ in $G/X$ and $G_0/X_0$ respectively. Then

$$\Delta^{G_0}(X_0,Y)\le 2\Delta^G(X,Y)$$
\end{proposition}

\begin{figure}
\includegraphics[width=1.0\textwidth]{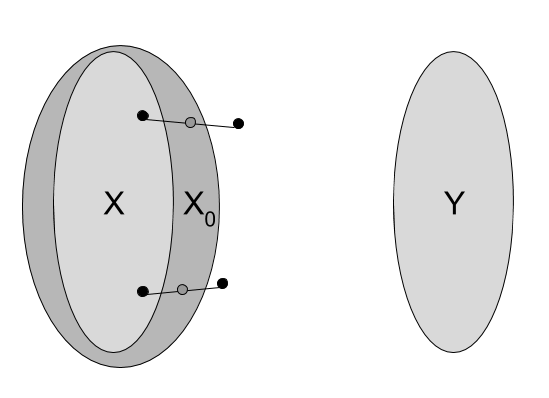}
\caption{Depiction of Proposition \ref{prop:deg-split}.}
\label{fig:degree-split}
\end{figure}

\begin{proof}
Recall that

$$\Delta^{G_0}(X_0,Y) = \sum_{y_i\in Y} b_{x_0y_i}^T L_{G_0/X_0}^+ b_{x_0y_i} c_{x_0y_i}^{H_0}$$

where $H_0 = \texttt{Schur}(G_0,\{x_0\}\cup Y)$. By Rayleigh monotonocity and $X\subseteq X_0$, $b_{x_0y_i}^T L_{G_0/X_0}^+ b_{x_0y_i}\le b_{xy_i}^T L_{G/X}^+ b_{xy_i}$. By Proposition \ref{prop:cond-pot} and the fact that all vertices in $X_0$ have normalized potential at most 1/2 in $G$ with $p_x = 0$ and $p_w = 1$ for all $w\in Y$,

$$c_{x_0y_i}^{H_0}\le 2c_{xy_i}^H$$

where $H = \texttt{Schur}(G,\{x\}\cup Y)$. Therefore,

$$\Delta^{G_0}(X_0,Y)\le 2\sum_{y_i\in Y} (b_{xy_i}^T L_{G/X}^+ b_{xy_i}) c_{xy_i}^H = 2\Delta^G(X,Y)$$

as desired.
\end{proof}

\lemfirstorderdeg*

\begin{proof}[Proof of Lemma \ref{lem:first-order-deg}]
Split every edge in $\partial_G X$ into a path of two edges with half the resistance. Let $X_0$ be the set of edges with endpoints in $X$ or one of the copies incident with $X$ and let $G_0$ be this graph. We start by showing that

$$\sum_{f\in G_0[X_0]} (\texttt{lev}_{G_0}(f) - \texttt{lev}_{G_0/Y}(f)) \le \Delta_0$$

where $\Delta_0 = \Delta^{G_0}(X_0,Y)$, $G_0' = G_0/X_0$, and $x_0$ is the identification of $X_0$. First, by Remark \ref{rmk:lev-sum}

$$\sum_{f\in H_0'[Y]} \texttt{lev}_{H_0'}(f) = |Y| - \Delta_0$$

where $H_0' = \texttt{Schur}(G_0',\{x_0\}\cup Y)$. By Rayleigh monotonicity,

$$\sum_{f\in H_0[Y]} \texttt{lev}_{H_0}(f) \ge |Y| - \Delta_0$$

where $H_0 = \texttt{Schur}(G_0,X_0\cup Y)$. By Remark \ref{rmk:lev-sum},

$$\sum_{f\in E(H_0)\setminus H_0[Y]} \texttt{lev}_{H_0}(f) \le |X_0| + \Delta_0 - 1$$

But by Remark \ref{rmk:lev-sum},

$$\sum_{f\in E(H_0)\setminus H_0[Y]} \texttt{lev}_{H_0/Y}(f) = |X_0|$$

so therefore

$$\sum_{f\in E(H_0)\setminus H_0[Y]} (\texttt{lev}_{H_0}(f) - \texttt{lev}_{H_0/Y}(f)) \le \Delta_0 - 1$$

By Rayleigh monotonicity, the summand of the above sum is always nonnegative. Therefore, the above inequality also applies for any subset of $E(H_0)\setminus H_0[Y]$. In particular, it applies for $G_0[X_0]$, completing the proof that

$$\sum_{f\in G_0[X_0]} (\texttt{lev}_{G_0}(f) - \texttt{lev}_{G_0/Y}(f)) \le \Delta_0$$

Consider an edge $f$ that is one of the two copies of an edge $e$ resulting from subdivision with $\texttt{lev}_G(e)\ge \frac{1}{4}$. Consider the three-vertex Schur complement of $G$ with respect to $V(e)\cup V(f)$. Identifying $Y$ only changes the Schur complement resistance of the edge with endpoints $V(e)$. Let $f'$ be the remaining of the three edges in this Schur complement. $\texttt{lev}_{G_0}(f) = \texttt{lev}_{G_0}(f')$ and $\texttt{lev}_{G_0/Y}(f) = \texttt{lev}_{G_0/Y}(f')$ since the edge weights of $f$ and $f'$ are the same and they are part of the same path with length 2. Therefore, by Remark \ref{rmk:lev-sum} applied to the three-vertex Schur complement,

\begin{align*}
\texttt{lev}_G(e) - \texttt{lev}_{G/Y}(e) &= (1 - (3 - 1 - 2\texttt{lev}_G(f))) - (1 - (3 - 1 - 2\texttt{lev}_{G/Y}(f)))\\
&= 2\texttt{lev}_G(f) - 2\texttt{lev}_{G/Y}(f)\\
\end{align*}

so

$$\sum_{f\in G[X]\cup \partial_G X} (\texttt{lev}_G(f) - \texttt{lev}_{G/Y}(f)) \le 2\Delta_0$$.

By Proposition \ref{prop:deg-split}, $\Delta_0\le 2\Delta$. Combining this with Remark \ref{rmk:lev-cng} proves the lemma.  
\end{proof}

\lemdelfirstorderdeg*

\begin{proof}
By Remark \ref{rmk:lev-cng}, it suffices to show that

$$\sum_{e\in E(G)\setminus D} (\texttt{lev}_{G\setminus D}(e) - \texttt{lev}_G(e))\le |D|$$

Furthermore, we may assume that $D$ consists of just one edge $f$, because deleting one edge at a time suffices for proving the above inequality for general $D$. By Sherman-Morrison,

\begin{align*}
\sum_{e\in E(G)\setminus f} (\texttt{lev}_{G\setminus\{f\}}(e) - \texttt{lev}_G(e)) &= \sum_{e\in E(G)\setminus f} \frac{(b_e^T L_G^+ b_f)^2}{r_e (r_f - b_f^T L_G^+ b_f)}\\
&= \frac{1}{r_f - b_f^T L_G^+ b_f} (b_f^T L_G^+ b_f - \frac{(b_f^T L_G^+ b_f)^2}{r_f})\\
&= \texttt{lev}_G(f)\\
&\le 1
\end{align*}

as desired.
\end{proof}

\subsubsection{Preliminaries for second-order terms}

\begin{proposition}\label{prop:edge-cut-bound}
Consider a graph $G$ with a set $Y\subseteq V(G)$, $X := V(G) \setminus Y$, $e = \{a,b\}\in E(X)$, $s\in X$, and $t\in Y$. Then

$$|b_{st}^T L_G^+ b_e| \le \sum_{f = \{p,q\}\in E(Y,X)} |b_{sq}^T L_G^+ b_e|\frac{|b_{st}^T L_{G/X}^+ b_f|}{r_f}$$
\end{proposition}

\begin{proof}
We start by showing that

$$b_{st}^T L_G^+ b_{su} = \sum_{f = \{p,q\}\in E(Y,X)} b_{sq}^T L_G^+ b_{su} \frac{|b_{st}^T L_{G/X}^+ b_f|}{r_f}$$

for any vertex $u\in X$. To show this, interpret the statement probabilistically. Notice that

$$\Pr_t[t_s > t_u] = \frac{b_{st}^T L_G^+ b_{su}}{b_{su}^T L_G^+ b_{su}}$$

This probability can be factored based on which edge is used to exit $Y$ first:

\begin{align*}
\Pr_t[t_s > t_u] &= \sum_{f = \{p,q\}\in E(Y,X)} \Pr_t[t_s > t_u, f \text{ used to exit $Y$ for the first time}]\\
&= \sum_{f = \{p,q\}\in E(Y,X)} \Pr_t[t_s > t_u | f \text{ used to exit $Y$ for the first time}]\\
&\Pr_t[ f \text{ used to exit $Y$ for the first time}]\\
&= \sum_{f = \{p,q\}\in E(Y,X)} \Pr_q[t_s > t_u] \Pr_t[ f \text{ used to exit $Y$ for the first time}]\\
&= \sum_{f = \{p,q\}\in E(Y,X)} \frac{b_{sq}^T L_G^+ b_{su}}{b_{su}^T L_G^+ b_{su}} \frac{|b_{st}^T L_{G/X}^+ b_f|}{r_f}\\
\end{align*}

Multiplying both sides by $b_{su}^T L_G^+ b_{su}$ shows that

$$b_{st}^T L_G^+ b_{su} = \sum_{f = \{p,q\}\in E(Y,X)} b_{sq}^T L_G^+ b_{su} \frac{|b_{st}^T L_{G/X}^+ b_f|}{r_f}$$

For any two vertices $u,v\in X$, subtracting the resulting equations shows that

$$b_{st}^T L_G^+ b_{uv} = \sum_{f = \{p,q\}\in E(Y,X)} b_{sq}^T L_G^+ b_{uv} \frac{|b_{st}^T L_{G/X}^+ b_f|}{r_f}$$

Letting $b_{uv} := b_e$ and orienting the edge in the positive direction shows the desired inequality.
\end{proof}

\begin{proposition}\label{prop:high-drop-flow}
Consider a graph $G$ with a set $Y\subseteq V(G)$, $X := V(G)\setminus Y$, $u\in X$, $s\in X$, and $t\in Y$. Let $S_{\gamma}$ be the set of edges $f = \{p,q\} \in E(Y,X)$ with $|b_{sq}^T L_G^+ b_{su}| \ge \gamma |b_{st}^T L_G^+ b_{st}|$. Then

$$\sum_{f\in S_{\gamma}} \frac{|b_{st}^T L_{G/X}^+ b_f|}{r_f} \le \frac{3}{\gamma}$$
\end{proposition}

\begin{figure}
\includegraphics[width=1.0\textwidth]{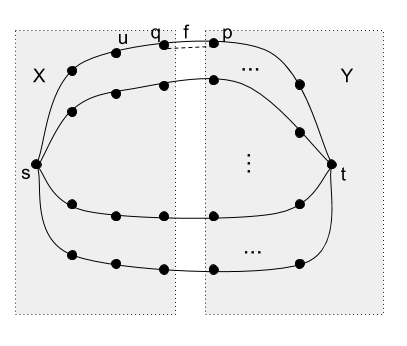}
\caption{Depiction of Proposition \ref{prop:high-drop-flow}. In this example, there are $k$ parallel paths with length $k$ from $s$ to $t$ and the $X-Y$ cut cuts the middle edge of each path. $u$ is close to the cutedge $f$. $f$ is the one $X-Y$ cutedge for which $q$ has roughly $k$ times the $s-u$ potential of $t$. The proposition implies that $f$ has at most $O(1/k)$ flow.}
\label{fig:potential-flow-bound}
\end{figure}

\begin{proof}
Split each edge $e\in S_{\gamma}$ with resistance $r_e$ into a path with two edges $e_1$ and $e_2$ with resistance $r_e/2$. Let the resulting graph be $H$. For an edge $\{p,q\} \in E(Y,X)$, suppose that it is split into $\{p,w\}$ and $\{w,q\}$. Notice that $|b_{st}^T L_{H/X}^+ b_{e_2}|/r_{e_2}^H = |b_{st}^T L_H^+ b_e|/r_e^G$, so it suffices to bound the $st$-flow on the edges $e_2$ for each edge $e\in S_{\gamma}$.

Let $S_{\gamma}'$ be the set of edges $e_2$ for edges $e\in S_{\gamma}$.Notice that $|b_{su}^T L_H^+ b_{sz}| \ge \gamma b_{st}^T L_G^+ b_{st} / 2 = \gamma |b_{st}^T L_H^+ b_{st}| / 2$ for each $z\in V(S_{\gamma}')$. Let $I = \texttt{Schur}(H, V(S_{\gamma}')\cup X\cup\{t\})$. By the harmonic property of vertex potentials,

\begin{align*}
\left(\sum_{x\in V(S_{\gamma}')\cup X} c_{xt}^I\right) b_{st}^T L_I^+ b_{st} &\ge \left(\sum_{x\in V(S_{\gamma}')\cup X} c_{xt}^I\right) b_{su}^T L_I^+ b_{st}\\
&= \sum_{x\in V(S_{\gamma}') \cup X} c_{xt}^I b_{su}^T L_I^+ b_{sx}\\
&\ge \sum_{x\in V(S_{\gamma}')\setminus X} c_{xt}^I b_{su}^T L_I^+ b_{sx}\\
&\ge \frac{\gamma}{2} b_{st}^T L_I^+ b_{st} \sum_{x\in V(S_{\gamma}')\setminus X} c_{xt}^I
\end{align*}

so at least a $1 - 2/\gamma$ fraction of the conductance incident with $t$ is also incident with $X$. By Rayleigh monotonicity, $1 / \left(\sum_{x\in V(S_{\gamma}') \cup X} c_{xt}^I\right) \le b_{st}^T L_{I/X}^+ b_{st}$. Therefore, the total flow in $I/X$ incident with $t$ that does not come directly from $s$ is

\begin{align*}
\sum_{x\in V(S_{\gamma}') \setminus X} c_{xt}^{I/X} b_{st}^T L_{I/X}^+ b_{xt} &= 1 - c_{st}^{I/X} b_{st}^T L_{I/X}^+ b_{st}\\
&= 1 - \left(\sum_{x\in X} c_{xt}^I\right) b_{st}^T L_{I/X}^+ b_{st}\\
&\le 1 - \frac{\sum_{x\in X} c_{xt}^I}{\sum_{x\in V(S_{\gamma}') \cup X} c_{xt}^I}\\
&= \frac{\sum_{x\in V(S_{\gamma}')\setminus X} c_{xt}^I}{\sum_{x\in V(S_{\gamma}') \cup X} c_{xt}^I}\\
&\le \frac{2}{\gamma}\\
\end{align*}

Notice that $I$ contains the edges in $S_{\gamma}'$ and that all $s-t$ flow on these edges in $I/X$ enters $t$ through some vertex besides $s$. Since all edges of $S_{\gamma}'$ are incident with $X$, flow only travels across one edge in $S_{\gamma}'$. This means that

$$\sum_{f\in S_{\gamma}'} c_f^I |b_{st}^T L_{I/X}^+ b_f| \le \frac{2}{\gamma}$$

The desired result follows from the fact that $|b_{st}^T L_{I/X}^+ b_f| = |b_{st}^T L_{G/X}^+ b_f|$ by definition of the Schur complement for all $f\in S_{\gamma}'$ and the fact that $c_f^I \ge c_f^G$.
\end{proof}

\begin{proposition}\label{prop:pot-energy-ub}
Consider a graph $I$ with two vertices $s,t\in V(I)$. Let $F$ be the set of edges $e = \{a,b\}\in E(I)$ with $\max_{c\in \{a,b\}} b_{st}^T L_I^+ b_{sc}\le \rho$ for some $\rho > 0$. Then

$$\sum_{e\in F} (b_{st}^T L_I^+ b_e)^2/r_e\le \rho$$
\end{proposition}

\begin{proof}
Write energy as current times potential drop. Doing this shows that

\begin{align*}
\sum_{e\in F} (b_{st}^T L_I^+ b_e)^2/r_e &\le \int_0^{\rho} \sum_{e\in x \text{ threshold cut}} (|b_{st}^T L_I^+ b_e|/r_e)dx\\
&= \int_0^{\rho} dx\\
&= \rho
\end{align*}

as desired.
\end{proof}

\lemsecondorderdeg*

\begin{proof}[Proof of Lemma \ref{lem:second-order-deg}]
We start by reducing to the case where all edges have both endpoints in $X$. Subdividing all edges in $\partial_G X$ to obtain a graph $G_0$. Let $X_0$ be the set of vertices in $G_0$ that are endpoints of edges in $E_G(X)\cup \partial_{G_0} X$. Let $x_0$ be the identification of $X_0$ in $G_0/X_0$. Notice that

\begin{align*}
\sum_{e\in E_G(X)\cup \partial_G X} \max_{t\in Y} \frac{(b_{st}^T L_G^+ b_e)^2}{(b_{st}^T L_G^+ b_{st})r_e^G} &= \sum_{e\in E_{G_0}(X)} \max_{t\in Y} \frac{(b_{st}^T L_{G_0}^+ b_e)^2}{(b_{st}^T L_{G_0}^+ b_{st})r_e^{G_0}}\\
&+ \sum_{e\in \partial_{G_0} X} \max_{t\in Y} \frac{4(b_{st}^T L_{G_0}^+ b_e)^2}{(b_{st}^T L_{G_0}^+ b_{st})2r_e^{G_0}}\\
&\le 2\sum_{e\in E_{G_0}(X_0)} \max_{t\in Y} \frac{(b_{st}^T L_{G_0}^+ b_e)^2}{(b_{st}^T L_{G_0}^+ b_{st})r_e^{G_0}}\\
\end{align*}

with all equalities and inequalities true termwise. By Proposition \ref{prop:deg-split},

$$\Delta^{G_0}(X_0,Y)\le 2\Delta^G(X,Y)$$

so to prove the lemma it suffices to show that

$$\sum_{e\in E_{G_0}(X_0)} \max_{t\in Y} \frac{(b_{st}^T L_{G_0}^+ b_e)^2}{(b_{st}^T L_{G_0}^+ b_{st})r_e^{G_0}}\le 6\xibuc^2\Delta^{G_0}(X_0,Y)$$

Let $H = \texttt{Schur}(G_0,X_0\cup Y)$. The Schur complement only adds edges to $E(X_0)$, so

$$\sum_{e\in E_{G_0}(X_0)} \max_{t\in Y} \frac{(b_{st}^T L_{G_0}^+ b_e)^2}{(b_{st}^T L_{G_0}^+ b_{st})r_e^{G_0}}\le \sum_{e\in E_H(X_0)} \max_{t\in Y} \frac{(b_{st}^T L_H^+ b_e)^2}{(b_{st}^T L_H^+ b_{st})r_e^H}$$

For an edge $g\in E_H(X_0,Y)$, let $x_g$ and $y_g$ denote its endpoints in $X_0$ and $Y$ respectively, $S_{g,i}\subseteq E_H(X_0)$ be the set of edges $e = \{a,b\}$ with $\max_{c \in \{a,b\}} |b_{sx_g}^T L_G^+ b_{sc}| \in [2^i, 2^{i+1}]$ for integers $i\in (\log (r_{min}/m),\log r_{max}]$ and $\max_{c \in \{a,b\}} |b_{sx_g}^T L_G^+ b_{sc}| \in [0, 2^{i+1}]$ for $i = \log (r_{min}/m)$. By Proposition \ref{prop:edge-cut-bound},

$$|b_{st}^T L_H^+ b_e|\le \sum_{g\in E_H(X_0,Y)} |b_{sx_g}^T L_H^+ b_e| \frac{|b_{st}^T L_{H/X_0}^+ b_g|}{r_g^H}$$

Applying Cauchy-Schwarz twice and Proposition \ref{prop:high-drop-flow} shows that

\begin{align*}
&\sum_{e\in E_H(X_0)} \max_{t\in Y} \frac{(b_{st}^T L_H^+ b_e)^2}{(b_{st}^T L_H^+ b_{st})r_e^H} \le \sum_{e\in E_H(X_0)} \max_{t\in Y} \frac{1}{(b_{st}^T L_H^+ b_{st})r_e^H} \left(\sum_{i=\log (r_{min}/m)}^{\log r_{max}}\sum_{g:e\in S_{g,i}} |b_{sx_g}^T L_H^+ b_e| \frac{|b_{st}^T L_{H/X_0}^+ b_g|}{r_g^H}\right)^2\\
&\le \xibuc\sum_{i=\log (r_{min}/m)}^{\log r_{max}} \sum_{e\in E_H(X_0)} \max_{t\in Y} \frac{1}{(b_{st}^T L_H^+ b_{st})r_e^H} \left(\sum_{g:e\in S_{g,i}} |b_{sx_g}^T L_H^+ b_e| \frac{|b_{st}^T L_{H/X_0}^+ b_g|}{r_g^H}\right)^2\\
&\le \xibuc\sum_{i=\log (r_{min}/m)}^{\log r_{max}} \sum_{e\in E_H(X_0)} \max_{t\in Y} \frac{1}{(b_{st}^T L_H^+ b_{st})r_e^H} \left(\sum_{g:e\in S_{g,i}} (b_{sx_g}^T L_H^+ b_e)^2 \frac{|b_{st}^T L_{H/X_0}^+ b_g|}{r_g^H}\right)\left(\sum_{g:e\in S_{g,i}} \frac{|b_{st}^T L_{H/X_0}^+ b_g|}{r_g^H}\right)\\
&\le \xibuc\sum_{i=\log (r_{min}/m)+1}^{\log r_{max}} \sum_{e\in E_H(X_0)} \max_{t\in Y} \left(\sum_{g:e\in S_{g,i}} \frac{(b_{sx_g}^T L_H^+ b_e)^2}{(b_{st}^T L_H^+ b_{st})r_e^H} \frac{|b_{st}^T L_{H/X_0}^+ b_g|}{r_g^H}\right) \frac{3b_{st}^T L_H^+ b_{st}}{2^i}\\
&+ \xibuc \sum_{e\in E_H(X_0)} \max_{t\in Y} \left(\sum_{g:e\in S_{g,\log (r_{min}/m)}} \frac{(b_{sx_g}^T L_H^+ b_e)^2}{(b_{st}^T L_H^+ b_{st})r_e^H} \frac{|b_{st}^T L_{H/X_0}^+ b_g|}{r_g^H}\right) m\\
&\le 3 \xibuc\sum_{i=\log (r_{min}/m)}^{\log r_{max}} \sum_{g\in E_H(X_0,Y)} \sum_{e\in S_{g,i}} \frac{(b_{sx_g}^T L_H^+ b_e)^2}{2^ir_e^H} \left(\max_{t\in Y}\frac{|b_{st}^T L_{H/X_0}^+ b_g|}{r_g^H}\right)\\
\end{align*}

By Proposition \ref{prop:pot-energy-ub},

$$\sum_{e\in S_{g,i}} \frac{(b_{sx_g}^T L_H^+ b_e)^2}{r_e^H} \le 2^{i+1}$$

so

\begin{align*}
\sum_{e\in E_H(X_0)} \max_{t\in Y} \frac{(b_{st}^T L_H^+ b_e)^2}{(b_{st}^T L_H^+ b_{st})r_e^H} &\le 6 \xibuc \sum_{i=\log (r_{min}/m)}^{\log r_{max}} \sum_{g\in E_H(X_0,Y)}\left(\max_{t\in Y}\frac{|b_{st}^T L_{H/X_0}^+ b_g|}{r_g^H}\right)\\
&\le 6 \xibuc \sum_{i=\log (r_{min}/m)}^{\log r_{max}} \sum_{g\in E_H(X_0,Y)}\frac{b_g^T L_{H/X_0}^+ b_g}{r_g^H}\\
&= 6 \xibuc^2 \Delta^{G_0}(X_0,Y)\\
\end{align*}

as desired.
\end{proof}

\subsubsection{Other preliminaries}

\propdegpotbound*

\begin{proof}
Let $H_z = \texttt{Schur}(G/Z, \{z\}\cup Y)$. By definition,

$$\Delta^G(Z,Y) = \sum_{v\in Y} \frac{b_{zv}^T L_{G/Z}^+ b_{zv}}{r_{zv}^{H_z}}$$

By Rayleigh monotonicity and the fact that $X\subseteq Z$,

$$b_{zv}^T L_{G/Z}^+ b_{zv} \le b_{xv}^T L_{G/X}^+ b_{xv}$$

By Proposition \ref{prop:elec-flow-cont} applied to $U = Z$ and $G'$,

$$r_{zv}^{H_z} \ge \frac{1}{1-\gamma} r_{xv}^{H_x}$$

since the conductance in the Schur complement for a vertex $v\in Y$ is proportional to the incoming flow by Proposition \ref{prop:lin-alg-norm-cond}. Substituting these inequalities and using the definition of $\Delta^G(X,Y)$ gives the desired result.
\end{proof}

\propcondpotbound*

\begin{proof}
Apply Proposition \ref{prop:cond-pot} twice, first with $Z\gets A$ and second with $Z\gets B$. Do this for all $w\in X\cup Y$. Each application increases the $X-Y$ conductance by at most a factor of $1/(1-\gamma)$, proving the desired result. 
\end{proof}

\subsection{Objectives that are stable thanks to stable oracles}\label{sec:alg-stable-appendix}

Before proving these stability propositions, we prove one fact that is common to all of them: that splitting can be done in the same direction (parallel or series) for all of the graphs $I_k$:

\begin{proposition}\label{prop:common-split}
Consider some $k\in \{0,1,\hdots,K-1\}$ and any $e\not\in \{f_0,f_1,\hdots,f_{k-1}\}$. Let $(I_0'',\{e^{(0)},e^{(1)}\}) \gets \Split(I_0,e)$ and obtain $I_k''\gets I_{k-1}''[[f_{k-1}]]$. Then,

\begin{itemize}
\item (Splitting later) $I_k''$ is equivalent in distribution to the graph obtained by splitting $e$ in $I_k$ in the same direction as $e$ was split in $I_0$ to obtain $I_0''$
\item (Leverage score bounds) With high probability,

$$1/8\le \texttt{lev}_{I_k''/(S,S')}(e^{(0)})\le \texttt{lev}_{I_k''\setminus D}(e^{(0)})\le 7/8$$

\end{itemize}
\end{proposition}

\begin{proof}

\textbf{Splitting later.} Effective resistances are electrical functions, so order of splitting does not matter.

\textbf{Leverage score bounds.} By the ``Bounded leverage score difference'' condition in Definition \ref{def:oracle}, the ``Leverage score stability'' guarantees in Definition \ref{def:oracle}, and the triangle inequality,

$$|\texttt{lev}_{I_k\setminus D}(e) - \texttt{lev}_{I_0}(e)|\le 3/16$$

and

$$|\texttt{lev}_{I_k/(S,S')}(e) - \texttt{lev}_{I_0}(e)|\le 3/16$$

If $\texttt{lev}_{I_0}(e)\ge 1/2$, then $e$ is split in parallel in $I_0$ to obtain $I_0''$. Furthermore, $\texttt{lev}_{I_k/(S,S')}(e)\ge 1/2 - 3/16 > 1/4$. This means that $1/2\ge \texttt{lev}_{I_k''/(S,S')}(e^{(0)})\ge (1/4)/2 = 1/8$ since $\texttt{lev}_{I_k''/(S,S')}(e^{(0)}) = \texttt{lev}_{I_k/(S,S')}(e)/2$. By Rayleigh monotonicity, $\texttt{lev}_{I_k''\setminus D}(e) \ge \texttt{lev}_{I_k''/(S,S')}(e)\ge 1/8$. As a result, $1/2\ge \texttt{lev}_{I_k''\setminus D}(e^{(0)})\ge 1/8$ as well.

If $\texttt{lev}_{I_0}(e)\le 1/2$, then $e$ is split in series in $I_0$ to obtain $I_0''$. Furthermore, $\texttt{lev}_{I_k\setminus D}(e)\le 1/2 + 3/16 \le 3/4$. This means that $1/2\le \texttt{lev}_{I_k''\setminus D}(e^{(0)})\le (1/2)(3/4) + 1/2 = 7/8$ since $\texttt{lev}_{I_k''\setminus D}(e^{(0)})\le (1/2)\texttt{lev}_{I_k\setminus D}(e) + 1/2$. By Rayleigh monotonicity, $\texttt{lev}_{I_k''/(S,S')}(e)\le 7/8$. As a result, $1/2\le \texttt{lev}_{I_k''/(S,S')}(e^{(0)})\le 7/8$ as well.

This completes the desired result in both cases.
\end{proof}

Now, we proceed with the proofs of the stability propositions:

\propdeltastability*

\begin{proof}[Proof of Proposition \ref{prop:delta-stability}]

We focus on $\delta_{S,S'}(H\setminus D)$, as the argument for $\delta_{S',S}(H\setminus D)$ is the same, with $S$ and $S'$ swapped.

\underline{Well-definedness.} Recall that $\delta_{S,S'}(H\setminus D)$ is an electrical function of $H$ since $\delta_{S,S'}(H\setminus D)$ is preserved under Schur complementation of vertices besides the ones referenced in the summand of $\delta_{S,S'}(H\setminus D)$. Therefore, stability is well-defined.

\underline{Degree of midpoints.} Let $I_k''$ be the graph obtained from $I_k$ by splitting all edges of $W\setminus \{f_0,f_1,\hdots,f_{k-1}\}$ in series. By the ``Midpoint potential stability'' guarantee of $\Oracle$, the midpoints of these edges are contained in the sets $U_S,U_{S'}\subseteq V(I_k'')$ defined to be the sets of vertices with $s=0-s'=1$ normalized potential less than $1-p/2$ and greater than $p/2$ respectively. By Proposition \ref{prop:deg-pot-bound},

$$\Delta^{I_k''\setminus D}(U_S,S')\le 2\Delta^{I_k\setminus D}(S,S')/p$$

and

$$\Delta^{I_k''\setminus D}(U_{S'},S)\le 2\Delta^{I_k\setminus D}(S',S)/p$$

\textbf{Lipschitz contractions.} For simplicity of notation in the following proof, let $I_k$ denote the graph in which the edge $f_k$ has been split and let $f_k$ denote an arbitrary copy. By Sherman-Morrison,

\begin{align*}
&\delta_{S,S'}((I_k\setminus D)/f_k) - \delta_{S,S'}(I_k\setminus D)\\
&= \sum_{w\in S'} \sum_{e\in \partial_{I_k} w} \left(b_{sw}^T L_{(I_k\setminus D)/S}^+ b_{sw} - \frac{(b_{sw}^T L_{(I_k\setminus D)/S}^+ b_{f_k})^2}{b_{f_k}^T L_{(I_k\setminus D)/S}^+ b_{f_k}}\right)\\
&\left(\frac{b_{ss'}^T L_{(I_k\setminus D)/(S,S')}^+ b_e}{r_e} - \frac{(b_{ss'}^T L_{(I_k\setminus D)/(S,S')}^+ b_{f_k})(b_{f_k}^T L_{(I_k\setminus D)/(S,S')}^+ b_e)}{(b_{f_k}^T L_{(I_k\setminus D)/(S,S')}^+ b_{f_k})r_e}\right)\\
&- \sum_{w\in S'} \sum_{e\in \partial_{I_k} w} \left(b_{sw}^T L_{(I_k\setminus D)/S}^+ b_{sw}\right)\left(\frac{b_{ss'}^T L_{(I_k\setminus D)/(S,S')}^+ b_e}{r_e}\right)\\
&= -\sum_{w\in S'} \sum_{e\in \partial_{I_k} w} \left(\frac{(b_{sw}^T L_{(I_k\setminus D)/S}^+ b_{f_k})^2}{b_{f_k}^T L_{(I_k\setminus D)/S}^+ b_{f_k}}\right)\left(\frac{b_{ss'}^T L_{(I_k\setminus D)/(S,S')}^+ b_e}{r_e}\right)\\
&- \sum_{w\in S'} \sum_{e\in \partial_{I_k} w} \left(b_{sw}^T L_{(I_k\setminus D)/S}^+ b_{sw}\right)\left(\frac{(b_{ss'}^T L_{(I_k\setminus D)/(S,S')}^+ b_{f_k})(b_{f_k}^T L_{(I_k\setminus D)/(S,S')}^+ b_e)}{(b_{f_k}^T L_{(I_k\setminus D)/(S,S')}^+ b_{f_k})r_e}\right)\\
&+ \sum_{w\in S'} \sum_{e\in \partial_{I_k} w} \left(\frac{(b_{sw}^T L_{(I_k\setminus D)/S}^+ b_{f_k})^2}{b_{f_k}^T L_{(I_k\setminus D)/S}^+ b_{f_k}}\right)\left(\frac{(b_{ss'}^T L_{(I_k\setminus D)/(S,S')}^+ b_{f_k})(b_{f_k}^T L_{(I_k\setminus D)/(S,S')}^+ b_e)}{(b_{f_k}^T L_{(I_k\setminus D)/(S,S')}^+ b_{f_k})r_e}\right)\\
\end{align*}

Therefore, by Proposition \ref{prop:common-split}, Rayleigh monotonicity, and the fact that the total energy is an upper bound on the energy of an edge,

\begin{align*}
&|\delta_{S,S'}((I_k\setminus D)/f_k) - \delta_{S,S'}(I_k\setminus D)|\\
&\le 8\sum_{w\in S'} \sum_{e\in \partial_{I_k} w} \left(\frac{(b_{sw}^T L_{(I_k\setminus D)/S}^+ b_{f_k})^2}{r_{f_k}}\right)\left(\frac{b_{ss'}^T L_{(I_k\setminus D)/(S,S')}^+ b_e}{r_e}\right)\\
&+ 8\sum_{w\in S'} \sum_{e\in \partial_{I_k} w} \left(b_{sw}^T L_{(I_k\setminus D)/S}^+ b_{sw}\right)\left(\frac{|b_{ss'}^T L_{(I_k\setminus D)/(S,S')}^+ b_{f_k}||b_{f_k}^T L_{(I_k\setminus D)/(S,S')}^+ b_e|}{r_{f_k}r_e}\right)\\
&+ 64\sum_{w\in S'} \sum_{e\in \partial_{I_k} w} \left(b_{sw}^T L_{(I_k\setminus D)/S}^+ b_{sw}\right)\left(\frac{|b_{ss'}^T L_{(I_k\setminus D)/(S,S')}^+ b_{f_k}| |b_{f_k}^T L_{(I_k\setminus D)/(S,S')}^+ b_e|}{r_{f_k}r_e}\right)
\end{align*}

By the ``$S-S'$ normalized degree change stability'' guarantees of $\Oracle$,

\begin{align*}
|\delta_{S,S'}((I_k\setminus D)/f_k) - \delta_{S,S'}(I_k\setminus D)| &\le \frac{8\rho}{|W|}\delta_{S,S'}(I\setminus D) + \frac{8\rho}{|W|}\delta_{S,S'}(I\setminus D) + \frac{64\rho}{|W|}\delta_{S,S'}(I\setminus D)\\
&\le \frac{80\rho}{|W|}\delta_{S,S'}(I\setminus D)\\
\end{align*}

as desired.

\textbf{Lipschitz deletions.} This bound is very similar to the contraction case. By Sherman-Morrison,

\begin{align*}
&\delta_{S,S'}((I_k\setminus D)\backslash f_k) - \delta_{S,S'}(I_k\setminus D)\\
&= \sum_{w\in S'} \sum_{e\in \partial_{I_k} w} \left(b_{sw}^T L_{(I_k\setminus D)/S}^+ b_{sw} + \frac{(b_{sw}^T L_{(I_k\setminus D)/S}^+ b_{f_k})^2}{r_{f_k} - b_{f_k}^T L_{(I_k\setminus D)/S}^+ b_{f_k}}\right)\\
&\left(\frac{b_{ss'}^T L_{(I_k\setminus D)/(S,S')}^+ b_e}{r_e} + \frac{(b_{ss'}^T L_{(I_k\setminus D)/(S,S')}^+ b_{f_k})(b_{f_k}^T L_{(I_k\setminus D)/(S,S')}^+ b_e)}{(r_{f_k} - b_{f_k}^T L_{(I_k\setminus D)/(S,S')}^+ b_{f_k})r_e}\right)\\
&- \sum_{w\in S'} \sum_{e\in \partial_{I_k} w} \left(b_{sw}^T L_{(I_k\setminus D)/S}^+ b_{sw}\right)\left(\frac{b_{ss'}^T L_{(I_k\setminus D)/(S,S')}^+ b_e}{r_e}\right)\\
&= \sum_{w\in S'} \sum_{e\in \partial_{I_k} w} \left(\frac{(b_{sw}^T L_{(I_k\setminus D)/S}^+ b_{f_k})^2}{r_{f_k} - b_{f_k}^T L_{(I_k\setminus D)/S}^+ b_{f_k}}\right)\left(\frac{b_{ss'}^T L_{(I_k\setminus D)/(S,S')}^+ b_e}{r_e}\right)\\
&+ \sum_{w\in S'} \sum_{e\in \partial_{I_k} w} \left(b_{sw}^T L_{(I_k\setminus D)/S}^+ b_{sw}\right)\left(\frac{(b_{ss'}^T L_{(I_k\setminus D)/(S,S')}^+ b_{f_k})(b_{f_k}^T L_{(I_k\setminus D)/(S,S')}^+ b_e)}{(r_{f_k} - b_{f_k}^T L_{(I_k\setminus D)/(S,S')}^+ b_{f_k})r_e}\right)\\
&+ \sum_{w\in S'} \sum_{e\in \partial_{I_k} w} \left(\frac{(b_{sw}^T L_{(I_k\setminus D)/S}^+ b_{f_k})^2}{r_{f_k} - b_{f_k}^T L_{(I_k\setminus D)/S}^+ b_{f_k}}\right)\left(\frac{(b_{ss'}^T L_{(I_k\setminus D)/(S,S')}^+ b_{f_k})(b_{f_k}^T L_{(I_k\setminus D)/(S,S')}^+ b_e)}{(r_{f_k} - b_{f_k}^T L_{(I_k\setminus D)/(S,S')}^+ b_{f_k})r_e}\right)\\
\end{align*}

Therefore, by Proposition \ref{prop:common-split}, Rayleigh monotonicity, and the fact that the total energy is an upper bound on the energy of an edge,

\begin{align*}
&|\delta_{S,S'}((I_k\setminus D)\backslash f_k) - \delta_{S,S'}(I_k\setminus D)|\\
&\le 8\sum_{w\in S'} \sum_{e\in \partial_{I_k} w} \left(\frac{(b_{sw}^T L_{(I_k\setminus D)/S}^+ b_{f_k})^2}{r_{f_k}}\right)\left(\frac{b_{ss'}^T L_{(I_k\setminus D)/(S,S')}^+ b_e}{r_e}\right)\\
&+ 8\sum_{w\in S'} \sum_{e\in \partial_{I_k} w} \left(b_{sw}^T L_{(I_k\setminus D)/S}^+ b_{sw}\right)\left(\frac{|b_{ss'}^T L_{(I_k\setminus D)/(S,S')}^+ b_{f_k}||b_{f_k}^T L_{(I_k\setminus D)/(S,S')}^+ b_e|}{r_{f_k}r_e}\right)\\
&+ 64\sum_{w\in S'} \sum_{e\in \partial_{I_k} w} \left(b_{sw}^T L_{(I_k\setminus D)/S}^+ b_{sw}\right)\left(\frac{|b_{ss'}^T L_{(I_k\setminus D)/(S,S')}^+ b_{f_k}| |b_{f_k}^T L_{(I_k\setminus D)/(S,S')}^+ b_e|}{r_{f_k}r_e}\right)
\end{align*}

By the ``$S-S'$ normalized degree change stability'' guarantees of $\Oracle$,

\begin{align*}
|\delta_{S,S'}((I_k\setminus D)\backslash f_k) - \delta_{S,S'}(I_k\setminus D)| &\le \frac{8\rho}{|W|}\delta_{S,S'}(I\setminus D) + \frac{8\rho}{|W|}\delta_{S,S'}(I\setminus D) + \frac{64\rho}{|W|}\delta_{S,S'}(I\setminus D)\\
&\le \frac{80\rho}{|W|}\delta_{S,S'}(I\setminus D)\\
\end{align*}

as desired.

\textbf{Change in expectation.} Write down the change in the expectation:

\begin{align*}
&\textbf{E}_{H'\sim I_k[f_k]}[\delta_{S,S'}(H'\setminus D) | f_k] - \delta_{S,S'}(I_k\setminus D)\\
&= \texttt{levcng}_{I_k\rightarrow (I_k\setminus D)/S}(f_k)\sum_{w\in S'} \sum_{e\in \partial_{I_k} w} \left(\frac{(b_{sw}^T L_{(I_k\setminus D)/S}^+ b_{f_k})^2}{r_{f_k}}\right)\left(\frac{b_{ss'}^T L_{(I_k\setminus D)/(S,S')}^+ b_e}{r_e}\right)\\
&\texttt{levcng}_{I_k\rightarrow (I_k\setminus D)/(S,S')}(f_k)\sum_{w\in S'} \sum_{e\in \partial_{I_k} w} \left(b_{sw}^T L_{(I_k\setminus D)/S}^+ b_{sw}\right)\left(\frac{(b_{ss'}^T L_{(I_k\setminus D)/(S,S')}^+ b_{f_k})(b_{f_k}^T L_{(I_k\setminus D)/(S,S')}^+ b_e)}{r_{f_k}r_e}\right)\\
&+ \left(\frac{\texttt{lev}_{I_k}(f_k)}{\texttt{lev}_{(I_k\setminus D)/S}(f_k)\texttt{lev}_{(I_k\setminus D)/(S,S')}(f_k)} + \frac{1 - \texttt{lev}_{I_k}(f_k)}{(1 - \texttt{lev}_{(I_k\setminus D)/S}(f_k))(1 - \texttt{lev}_{(I_k\setminus D)/(S,S')}(f_k))}\right)\\
&\sum_{w\in S'} \sum_{e\in \partial_{I_k} w} \left(\frac{(b_{sw}^T L_{(I_k\setminus D)/S}^+ b_{f_k})^2}{r_{f_k}}\right)\left(\frac{(b_{ss'}^T L_{(I_k\setminus D)/(S,S')}^+ b_{f_k})(b_{f_k}^T L_{(I_k\setminus D)/(S,S')}^+ b_e)}{r_{f_k}r_e}\right)\\
\end{align*}

By Proposition \ref{prop:common-split},

\begin{align*}
&|\textbf{E}_{H'\sim I_k[f_k]}[\delta_{S,S'}(H'\setminus D) | f_k] - \delta_{S,S'}(I_k\setminus D)|\\
&\le |\texttt{levcng}_{I_k\rightarrow (I_k\setminus D)/S}(f_k)|\sum_{w\in S'} \sum_{e\in \partial_{I_k} w} \left(\frac{(b_{sw}^T L_{(I_k\setminus D)/S}^+ b_{f_k})^2}{r_{f_k}}\right)\left(\frac{b_{ss'}^T L_{(I_k\setminus D)/(S,S')}^+ b_e}{r_e}\right)\\
&+ |\texttt{levcng}_{I_k\rightarrow (I_k\setminus D)/(S,S')}(f_k)|\sum_{w\in S'} \sum_{e\in \partial_{I_k} w} \left(b_{sw}^T L_{(I_k\setminus D)/S}^+ b_{sw}\right)\left(\frac{|b_{ss'}^T L_{(I_k\setminus D)/(S,S')}^+ b_{f_k}| |b_{f_k}^T L_{(I_k\setminus D)/(S,S')}^+ b_e|}{r_{f_k}r_e}\right)\\
&+ 128\sum_{w\in S'} \sum_{e\in \partial_{I_k} w} \left(\frac{(b_{sw}^T L_{(I_k\setminus D)/S}^+ b_{f_k})^2}{r_{f_k}}\right)\left(\frac{|b_{ss'}^T L_{(I_k\setminus D)/(S,S')}^+ b_{f_k}| |b_{f_k}^T L_{(I_k\setminus D)/(S,S')}^+ b_e|}{r_{f_k}r_e}\right)\\
\end{align*}

Let $\alpha_{f_k} = \max_{w\in S'} \frac{(b_{sw}^T L_{(I_k\setminus D)/S}^+ b_{f_k})^2}{(b_{sw}^T L_{(I_k\setminus D)/S}^+ b_{sw})r_{f_k}}$. By the ``$S-S'$-normalized degree change stability'' guarantees of $\Oracle$,

\begin{align*}
|\textbf{E}_{H'\sim I_k[f_k]}[\delta_{S,S'}(H'\setminus D) | f_k] - \delta_{S,S'}(I_k\setminus D)| &\le |\texttt{levcng}_{I_k\rightarrow (I_k\setminus D)/S}(f_k)|\left(\frac{\rho}{|W|}\delta_{S,S'}(I\setminus D)\right)\\
&+ \left(|\texttt{levcng}_{I_k\rightarrow (I_k\setminus D)/(S,S')}(f_k)| + 128\alpha_{f_k}\right)\left(\frac{\rho}{|W|}\delta_{S,S'}(I\setminus D)\right)
\end{align*}

Now, we exploit the fact that $f_k$ is chosen randomly from all remaining edges in $W$. $K \le |W|/4$, so there are always at least $|W|/4$ remaining edges in $W$ by the ``Size of $Z$'' guarantee. By Lemma \ref{lem:first-order-deg} applied twice, \ref{lem:del-first-order-deg} applied once, and \ref{lem:second-order-deg} (for $\alpha_{f_k}$) applied to the graph $(I_k\setminus D)/S$ with $X\gets U_S$ and $Y\gets S'$,

\begin{align*}
&|\textbf{E}_{H'\sim I_k[f_k]}[\delta_{S,S'}(H'\setminus D)] - \delta_{S,S'}(I_k\setminus D)|\\
&\le \textbf{E}_{f_k}\left[|\textbf{E}_{H'\sim I_k[f_k]}[\delta_{S,S'}(H'\setminus D) | f_k] - \delta_{S,S'}(I_k\setminus D)|\right]\\
&\le \frac{2}{|W|}\sum_{f_k\in W} \left(|\texttt{levcng}_{I_k\rightarrow (I_k\setminus D)/S}(f_k)| + |\texttt{levcng}_{I_k\rightarrow (I_k\setminus D)/(S,S')}(f_k)| + 128\alpha_{f_k}\right)\left(\frac{\rho}{|W|}\delta_{S,S'}(I\setminus D)\right)\\
&\le \frac{2\rho(2\Delta^{I_k''\setminus D}(U_S,S') + 2\Delta^{I_k''\setminus D}(U_{S'},S) + 2|D| + 128\Delta^{I_k''\setminus D}(U_S,S'))}{|W|^2}\delta_{S,S'}(I\setminus D)\\
\end{align*}

Applying the result of the ``Degree of midpoints'' part shows that

$$|\textbf{E}_{H'\sim I_k[f_k]}[\delta_{S,S'}(H'\setminus D)] - \delta_{S,S'}(I_k\setminus D)|\le \frac{520\rho \Delta_k}{p|W|^2}\delta_{S,S'}(I\setminus D)$$

as desired.
\end{proof}

\propdeferredstability*

\begin{proof}[Proof of Proposition \ref{prop:deferred-stability}]

\underline{Well-definedness.} The function given is an electrical function of $H$.

\underline{Degree of midpoints.} Same as in the proof of Proposition \ref{prop:deferred-stability}.

\textbf{Lipschitz contractions.} As in the previous stability proof, let $I_k$ denote the graph after splitting an edge, one of whose copies is $f_k$. By Sherman-Morrison, the triangle inequality, Proposition \ref{prop:common-split} with Rayleigh monotonicity, and the ``Deferred endpoint potential change stability'' guarantee of $\Oracle$ in that order,

\begin{align*}
&\Biggr|\left(\sum_{\{u,v\}\in A} b_{ss'}^T L_{(I_k\setminus D)/(S,S')/f_k}^+ (b_{su} + b_{sv}) \right) + \left(\sum_{\{u,v\}\in B} b_{ss'}^T L_{(I_k\setminus D)/(S,S')/f_k}^+ (b_{us'} + b_{vs'})\right)\\
&- \left(\sum_{\{u,v\}\in A} b_{ss'}^T L_{(I_k\setminus D)/(S,S')}^+ (b_{su} + b_{sv}) \right) - \left(\sum_{\{u,v\}\in B} b_{ss'}^T L_{(I_k\setminus D)/(S,S')}^+ (b_{us'} + b_{vs'})\right)\Biggr|\\
&= \Biggr|\sum_{\{u,v\}\in A} \frac{(b_{ss'}^T L_{(I_k\setminus D)/(S,S')}^+ b_{f_k})(b_{f_k}^T L_{(I_k\setminus D)/(S,S')}^+ (b_{su} + b_{sv}))}{b_{f_k}^T L_{(I_k\setminus D)/(S,S')}^+ b_{f_k}}\\
&+ \sum_{\{u,v\}\in B} \frac{(b_{ss'}^T L_{(I_k\setminus D)/(S,S')}^+ b_{f_k})(b_{f_k}^T L_{(I_k\setminus D)/(S,S')}^+ (b_{us'} + b_{vs'}))}{b_{f_k}^T L_{(I_k\setminus D)/(S,S')}^+ b_{f_k}}\Biggr|\\
&\le \sum_{\{u,v\}\in A} \frac{|b_{ss'}^T L_{(I_k\setminus D)/(S,S')}^+ b_{f_k}| |b_{f_k}^T L_{(I_k\setminus D)/(S,S')}^+ (b_{su} + b_{sv})|}{b_{f_k}^T L_{(I_k\setminus D)/(S,S')}^+ b_{f_k}}\\
&+ \sum_{\{u,v\}\in B} \frac{|b_{ss'}^T L_{(I_k\setminus D)/(S,S')}^+ b_{f_k}| |b_{f_k}^T L_{(I_k\setminus D)/(S,S')}^+ (b_{us'} + b_{vs'})|}{b_{f_k}^T L_{(I_k\setminus D)/(S,S')}^+ b_{f_k}}\\
&\le 8\sum_{\{u,v\}\in A} \frac{|b_{ss'}^T L_{(I_k\setminus D)/(S,S')}^+ b_{f_k}| |b_{f_k}^T L_{(I_k\setminus D)/(S,S')}^+ (b_{su} + b_{sv})|}{r_{f_k}}\\
&+ 8\sum_{\{u,v\}\in B} \frac{|b_{ss'}^T L_{(I_k\setminus D)/(S,S')}^+ b_{f_k}| |b_{f_k}^T L_{(I_k\setminus D)/(S,S')}^+ (b_{us'} + b_{vs'})|}{r_{f_k}}\\
&\le \frac{8\rho}{|W|}\left(\sum_{\{u,v\}\in A} b_{ss'}^T L_{(I\setminus D)/(S,S')}^+ (b_{su} + b_{sv}) + \sum_{\{u,v\}\in B} b_{ss'}^T L_{(I\setminus D)/(S,S')}^+ (b_{us'} + b_{vs'})\right) + \frac{8r_{min}}{n^4}
\end{align*}

as desired.

\textbf{Lipschitz deletions.} This part is similar to the contraction part. By Sherman-Morrison, the triangle inequality, Proposition \ref{prop:common-split} with Rayleigh monotonicity, and the ``Deferred endpoint potential change stability'' guarantee of $\Oracle$ in that order,

\begin{align*}
&\Biggr|\left(\sum_{\{u,v\}\in A} b_{ss'}^T L_{(I_k\setminus D)/(S,S')\backslash f_k}^+ (b_{su} + b_{sv}) \right) + \left(\sum_{\{u,v\}\in B} b_{ss'}^T L_{(I_k\setminus D)/(S,S')\backslash f_k}^+ (b_{us'} + b_{vs'})\right)\\
&- \left(\sum_{\{u,v\}\in A} b_{ss'}^T L_{(I_k\setminus D)/(S,S')}^+ (b_{su} + b_{sv}) \right) - \left(\sum_{\{u,v\}\in B} b_{ss'}^T L_{(I_k\setminus D)/(S,S')}^+ (b_{us'} + b_{vs'})\right)\Biggr|\\
&= \Biggr|\sum_{\{u,v\}\in A} \frac{(b_{ss'}^T L_{(I_k\setminus D)/(S,S')}^+ b_{f_k})(b_{f_k}^T L_{(I_k\setminus D)/(S,S')}^+ (b_{su} + b_{sv}))}{r_{f_k} - b_{f_k}^T L_{(I_k\setminus D)/(S,S')}^+ b_{f_k}}\\
&+ \sum_{\{u,v\}\in B} \frac{(b_{ss'}^T L_{(I_k\setminus D)/(S,S')}^+ b_{f_k})(b_{f_k}^T L_{(I_k\setminus D)/(S,S')}^+ (b_{us'} + b_{vs'}))}{r_{f_k} - b_{f_k}^T L_{(I_k\setminus D)/(S,S')}^+ b_{f_k}}\Biggr|\\
&\le \sum_{\{u,v\}\in A} \frac{|b_{ss'}^T L_{(I_k\setminus D)/(S,S')}^+ b_{f_k}| |b_{f_k}^T L_{(I_k\setminus D)/(S,S')}^+ (b_{su} + b_{sv})|}{r_{f_k} - b_{f_k}^T L_{(I_k\setminus D)/(S,S')}^+ b_{f_k}}\\
&+ \sum_{\{u,v\}\in B} \frac{|b_{ss'}^T L_{(I_k\setminus D)/(S,S')}^+ b_{f_k}| |b_{f_k}^T L_{(I_k\setminus D)/(S,S')}^+ (b_{us'} + b_{vs'})|}{r_{f_k} - b_{f_k}^T L_{(I_k\setminus D)/(S,S')}^+ b_{f_k}}\\
&\le 8\sum_{\{u,v\}\in A} \frac{|b_{ss'}^T L_{(I_k\setminus D)/(S,S')}^+ b_{f_k}| |b_{f_k}^T L_{(I_k\setminus D)/(S,S')}^+ (b_{su} + b_{sv})|}{r_{f_k}}\\
&+ 8\sum_{\{u,v\}\in B} \frac{|b_{ss'}^T L_{(I_k\setminus D)/(S,S')}^+ b_{f_k}| |b_{f_k}^T L_{(I_k\setminus D)/(S,S')}^+ (b_{us'} + b_{vs'})|}{r_{f_k}}\\
&\le \frac{8\rho}{|W|}\left(\sum_{\{u,v\}\in A} b_{ss'}^T L_{(I\setminus D)/(S,S')}^+ (b_{su} + b_{sv}) + \sum_{\{u,v\}\in B} b_{ss'}^T L_{(I\setminus D)/(S,S')}^+ (b_{us'} + b_{vs'})\right) + \frac{8r_{min}}{n^4}
\end{align*}

as desired.

\textbf{Change in expectation.} Compute the expected change using Sherman-Morrison:

\begin{align*}
&\textbf{E}_{H'\sim I_k[f_k]}\left[\sum_{\{u,v\}\in A} b_{ss'}^T L_{(H'\setminus D)/(S,S')}^+ (b_{su} + b_{sv}) + \sum_{\{u,v\}\in B} b_{ss'}^T L_{(H'\setminus D)/(S,S')}^+ (b_{us'} + b_{vs'})\Biggr| f_k\right]\\
&- \sum_{\{u,v\}\in A} b_{ss'}^T L_{(I_k\setminus D)/(S,S')}^+ (b_{su} + b_{sv}) - \sum_{\{u,v\}\in B} b_{ss'}^T L_{(I_k\setminus D)/(S,S')}^+ (b_{us'} + b_{vs'})\\
&= \texttt{levcng}_{I_k\rightarrow (I_k\setminus D)/(S,S')}(f_k) \Biggr(\sum_{\{u,v\}\in A} \frac{(b_{ss'}^T L_{(I_k\setminus D)/(S,S')}^+ b_{f_k})(b_{f_k}^T L_{(I_k\setminus D)/(S,S')}^+ (b_{su} + b_{sv}))}{r_{f_k}}\\
&+ \sum_{\{u,v\}\in B} \frac{(b_{ss'}^T L_{(I_k\setminus D)/(S,S')}^+ b_{f_k})(b_{f_k}^T L_{(I_k\setminus D)/(S,S')}^+ (b_{us'} + b_{vs'}))}{r_{f_k}}\Biggr)
\end{align*}

By the ``Deferred endpoint potential change stability'' guarantee,

\begin{align*}
&\Biggr|\textbf{E}_{H'\sim I_k[f_k]}\left[\sum_{\{u,v\}\in A} b_{ss'}^T L_{(H'\setminus D)/(S,S')}^+ (b_{su} + b_{sv}) + \sum_{\{u,v\}\in B} b_{ss'}^T L_{(H'\setminus D)/(S,S')}^+ (b_{us'} + b_{vs'})\Biggr| f_k\right]\\
&- \sum_{\{u,v\}\in A} b_{ss'}^T L_{(I_k\setminus D)/(S,S')}^+ (b_{su} + b_{sv}) - \sum_{\{u,v\}\in B} b_{ss'}^T L_{(I_k\setminus D)/(S,S')}^+ (b_{us'} + b_{vs'})\Biggr|\\
&\le |\texttt{levcng}_{I_k\rightarrow (I_k\setminus D)/(S,S')}(f_k)| \Biggr(\sum_{\{u,v\}\in A} \frac{|b_{ss'}^T L_{(I_k\setminus D)/(S,S')}^+ b_{f_k}| |b_{f_k}^T L_{(I_k\setminus D)/(S,S')}^+ (b_{su} + b_{sv})|}{r_{f_k}}\\
&+ \sum_{\{u,v\}\in B} \frac{|b_{ss'}^T L_{(I_k\setminus D)/(S,S')}^+ b_{f_k}| |b_{f_k}^T L_{(I_k\setminus D)/(S,S')}^+ (b_{us'} + b_{vs'})|}{r_{f_k}}\Biggr)\\
&\le |\texttt{levcng}_{I_k\rightarrow (I_k\setminus D)/(S,S')}(f_k)|\Biggr(\frac{\rho}{|W|}\Biggr(\sum_{\{u,v\}\in A} b_{ss'}^T L_{(I\setminus D)/(S,S')}^+ (b_{su} + b_{sv})\\
&+ \sum_{\{u,v\}\in B} b_{ss'}^T L_{(I\setminus D)/(S,S')}^+ (b_{us'} + b_{vs'})\Biggr) + \frac{r_{min}}{n^4}\Biggr)\\
\end{align*}

By Lemmas \ref{lem:first-order-deg} and \ref{lem:del-first-order-deg} and the ``Size of $Z$'' guarantee,

\begin{align*}
&\Biggr|\textbf{E}_{H'\sim I_k[f_k]}\Biggr[\sum_{\{u,v\}\in A} b_{ss'}^T L_{(H'\setminus D)/(S,S')}^+ (b_{su} + b_{sv}) + \sum_{\{u,v\}\in B} b_{ss'}^T L_{(H'\setminus D)/(S,S')}^+ (b_{us'} + b_{vs'})\\
&- \sum_{\{u,v\}\in A} b_{ss'}^T L_{(I_k\setminus D)/(S,S')}^+ (b_{su} + b_{sv}) - \sum_{\{u,v\}\in B} b_{ss'}^T L_{(I_k\setminus D)/(S,S')}^+ (b_{us'} + b_{vs'})\Biggr]\Biggr|\\
&\le \textbf{E}_{f_k}\Biggr[\Biggr|\textbf{E}_{H'\sim I_k[f_k]}\Biggr[\sum_{\{u,v\}\in A} b_{ss'}^T L_{(H'\setminus D)/(S,S')}^+ (b_{su} + b_{sv}) + \sum_{\{u,v\}\in B} b_{ss'}^T L_{(H'\setminus D)/(S,S')}^+ (b_{us'} + b_{vs'})\\
&- \sum_{\{u,v\}\in A} b_{ss'}^T L_{(I_k\setminus D)/(S,S')}^+ (b_{su} + b_{sv}) - \sum_{\{u,v\}\in B} b_{ss'}^T L_{(I_k\setminus D)/(S,S')}^+ (b_{us'} + b_{vs'})\Biggr| f_k\Biggr]\Biggr|\Biggr]\\
&\le \frac{2}{|W|} \left(\sum_{f_k\in W} |\texttt{levcng}_{I_k\rightarrow (I_k\setminus D)/(S,S')}(f_k)|\right)\Biggr(\frac{\rho}{|W|}\Biggr(\sum_{\{u,v\}\in A} b_{ss'}^T L_{(I\setminus D)/(S,S')}^+ (b_{su} + b_{sv})\\
&+ \sum_{\{u,v\}\in B} b_{ss'}^T L_{(I\setminus D)/(S,S')}^+ (b_{us'} + b_{vs'})\Biggr) + \frac{r_{min}}{n^4}\Biggr)\\
&\le \frac{2}{|W|} \left(\Delta^{I_k''\setminus D}(U_S,S') + \Delta^{I_k''\setminus D}(U_{S'},S) + |D|\right)\Biggr(\frac{\rho}{|W|}\Biggr(\sum_{\{u,v\}\in A} b_{ss'}^T L_{(I\setminus D)/(S,S')}^+ (b_{su} + b_{sv})\\
&+ \sum_{\{u,v\}\in B} b_{ss'}^T L_{(I\setminus D)/(S,S')}^+ (b_{us'} + b_{vs'})\Biggr) + \frac{r_{min}}{n^4}\Biggr)\\
\end{align*}

By ``Degree of midpoints,'' $\Delta^{I_k''\setminus D}(U_S,S') + \Delta^{I_k''\setminus D}(U_{S'},S) + |D|\le (2/p)\Delta_k$, so

\begin{align*}
&\Biggr|\textbf{E}_{H'\sim I_k[f_k]}\Biggr[\sum_{\{u,v\}\in A} b_{ss'}^T L_{(H'\setminus D)/(S,S')}^+ (b_{su} + b_{sv}) + \sum_{\{u,v\}\in B} b_{ss'}^T L_{(H'\setminus D)/(S,S')}^+ (b_{us'} + b_{vs'})\\
&- \sum_{\{u,v\}\in A} b_{ss'}^T L_{(I_k\setminus D)/(S,S')}^+ (b_{su} + b_{sv}) - \sum_{\{u,v\}\in B} b_{ss'}^T L_{(I_k\setminus D)/(S,S')}^+ (b_{us'} + b_{vs'})\Biggr]\Biggr|\\
&\le \frac{4\rho\Delta_k}{p|W|^2}\left(\sum_{\{u,v\}\in A} b_{ss'}^T L_{(I\setminus D)/(S,S')}^+ (b_{su} + b_{sv}) + \sum_{\{u,v\}\in B} b_{ss'}^T L_{(I\setminus D)/(S,S')}^+ (b_{us'} + b_{vs'})\right) + \frac{r_{min}}{n^4}
\end{align*}

as desired.

\end{proof}

\propmainstability*

\begin{proof}[Proof of Proposition \ref{prop:main-stability}]

\underline{Well-definedness.} $b_{ss'}^T L_{H\setminus D} b_{ss'}$ is an electrical function of $H$.

\underline{Degree of midpoints.} Same as in the proof of Proposition \ref{prop:deferred-stability}.

\textbf{Lipschitz contractions.} As in the previous stability proofs, let $I_k$ denote the graph after splitting an edge, one of whose copies is $f_k$. By Sherman-Morrison, Proposition \ref{prop:common-split} with Rayleigh monotonicity, and the ``Main objective change stability'' guarantee of $\Oracle$ in that order,

\begin{align*}
|b_{ss'}^T L_{(I_k\setminus D)/(S,S')/f_k}^+ b_{ss'} - b_{ss'}^T L_{(I_k\setminus D)/(S,S')}^+ b_{ss'}| &= \frac{(b_{ss'}^T L_{(I_k\setminus D)/(S,S')}^+ b_{f_k})^2}{b_{f_k}^T L_{(I_k\setminus D)/(S,S')}^+ b_{f_k}}\\
&\le 8\frac{(b_{ss'}^T L_{(I_k\setminus D)/(S,S')}^+ b_{f_k})^2}{r_{f_k}}\\
&\le \frac{8\rho}{|W|}b_{ss'}^T L_{(I\setminus D)/(S,S')}^+ b_{ss'}\\
\end{align*}

as desired.

\textbf{Lipschitz deletions.} This is similar to the contraction proof. By Sherman-Morrison, Proposition \ref{prop:common-split} with Rayleigh monotonicity, and the ``Main objective change stability'' guarantee of $\Oracle$ in that order,

\begin{align*}
|b_{ss'}^T L_{(I_k\setminus D)/(S,S')\backslash f_k}^+ b_{ss'} - b_{ss'}^T L_{(I_k\setminus D)/(S,S')}^+ b_{ss'}| &= \frac{(b_{ss'}^T L_{(I_k\setminus D)/(S,S')}^+ b_{f_k})^2}{r_{f_k} - b_{f_k}^T L_{(I_k\setminus D)/(S,S')}^+ b_{f_k}}\\
&\le 8\frac{(b_{ss'}^T L_{(I_k\setminus D)/(S,S')}^+ b_{f_k})^2}{r_{f_k}}\\
&\le \frac{8\rho}{|W|}b_{ss'}^T L_{(I\setminus D)/(S,S')}^+ b_{ss'}\\
\end{align*}

as desired.

\textbf{Change in expectation.} The change in expectation conditioned on a choice of $f_k$ can be written as

\begin{align*}
\textbf{E}_{H'\sim I_k[f_k]}[b_{ss'}^T L_{(H'\setminus D)/(S,S')}^+ b_{ss'}|f_k] - b_{ss'}^T L_{(I_k\setminus D)/(S,S')}^+ b_{ss'} &= \texttt{levcng}_{I_k\rightarrow (I_k\setminus D)/(S,S')}(f_k) \frac{(b_{ss'}^T L_{(I_k\setminus D)/(S,S')}^+ b_{f_k})^2}{r_{f_k}}
\end{align*}

Therefore, by Lemma \ref{lem:first-order-deg} applied twice, Lemma \ref{lem:del-first-order-deg} applied once, and the ``Size of $Z$'' guarantee,

\begin{align*}
&|\textbf{E}_{H'\sim I_k[f_k]}[b_{ss'}^T L_{(H'\setminus D)/(S,S')}^+ b_{ss'}] - b_{ss'}^T L_{(I_k\setminus D)/(S,S')}^+ b_{ss'}|\\
&\le \textbf{E}_{f_k}\left[|\textbf{E}_{H'\sim I_k[f_k]}[b_{ss'}^T L_{(H'\setminus D)/(S,S')}^+ b_{ss'}|f_k] - b_{ss'}^T L_{(I_k\setminus D)/(S,S')}^+ b_{ss'}|\right]\\
&\le \frac{2}{|Z|} \sum_{f_k\in W} |\texttt{levcng}_{I_k\rightarrow (I_k\setminus D)/(S,S')}(f_k)|\frac{(b_{ss'}^T L_{(I_k\setminus D)/(S,S')}^+ b_{f_k})^2}{r_{f_k}}\\
&\le \frac{4\rho }{|W|^2} (\Delta^{I_k''\setminus D}(U_S,S') + \Delta^{I_k''\setminus D}(U_{S'},S) + |D|) b_{ss'}^T L_{(I\setminus D)/(S,S')}^+ b_{ss'}\\
\end{align*}

By ``Degree of midpoints,''

$$|\textbf{E}_{H'\sim I_k[f_k]}[b_{ss'}^T L_{(H'\setminus D)/(S,S')}^+ b_{ss'}] - b_{ss'}^T L_{(I_k\setminus D)/(S,S')}^+ b_{ss'}| \le \frac{4\rho \Delta_k}{p|W|^2} b_{ss'}^T L_{(I_k\setminus D)/(S,S')}^+ b_{ss'}$$

as desired.

\end{proof}

\section{Deferred proofs for Section \ref{sec:fast-fix}}

\subsection{Fast stability concentration inequalities}\label{sec:fast-concentration-appendix}

\proprandommatixcontrol*

\begin{proof}[Proof of Proposition \ref{prop:random-matrix-control}]
Inductively assume that for all $i\notin S^{(k)}$ and for all $k\le n/2$, $\sum_{j\ne i, j\notin S^{(k)}} M_{ij}^{(k)} \le \sigma_1$. We now use Theorem \ref{thm:martingale-2} to validate this assumption. The third given condition along with the inductive assumption shows that

\begin{align*}
\sum_{j\ne i,j\notin S^{(k+1)}} M_{ij}^{(k+1)} &\le \left(\sum_{j\ne i,j\notin S^{(k+1)}} M_{ij}^{(k)}\right) + \gamma \left(\sum_{l=1}^n M_{il}^{(k)} Z_l^{(k+1)} \left(\sum_{j\ne i,j\notin S^{(k+1)}}M_{lj}^{(k)}\right)\right)\\
&= \left(\sum_{j\ne i,j\notin S^{(k+1)}} M_{ij}^{(k)}\right) + \gamma \left(\sum_{l=1}^n M_{il}^{(k)} Z_l^{(k+1)} \left(\sum_{j\ne i,j\ne w^{(k+1)}, j\notin S^{(k)}}M_{lj}^{(k)}\right)\right)\\
&= \left(\sum_{j\ne i,j\notin S^{(k+1)}} M_{ij}^{(k)}\right) + \gamma \left(\sum_{l=1}^n M_{il}^{(k)} Z_l^{(k+1)} \left(\sum_{j\ne i,j\ne l, j\notin S^{(k)}}M_{lj}^{(k)}\right)\right)\\
&\le \left(\sum_{j\ne i,j\notin S^{(k)}} M_{ij}^{(k)}\right) + \gamma \sigma_1 \left(\sum_{l=1}^n M_{il}^{(k)} Z_l^{(k+1)}\right)
\end{align*}

for $i\notin S^{(k+1)}$. To apply Theorem \ref{thm:martingale-2}, we need bounds on the mean, variance, and maximum deviation of each increment. We start with the mean:

\begin{align*}
\textbf{E}\left[\sum_{j\ne i, j\notin S^{(k+1)}} M_{ij}^{(k+1)}\mid S^{(k)}, i\notin S^{(k+1)}\right] &\le \sum_{j\ne i,j\notin S^{(k)}} M_{ij}^{(k)} + \gamma \sigma_1 \textbf{E}\left[\sum_{l=1}^n M_{il}^{(k)} Z_l^{(k+1)}\mid S^{(k)}, i\notin S^{(k+1)}\right]\\
&\le \sum_{j\ne i,j\notin S^{(k)}} M_{ij}^{(k)} + \gamma \sigma_1 \textbf{E}\left[\sum_{l\ne i, l\notin S^{(k)}}  M_{il}^{(k)} Z_l^{(k+1)}\mid S^{(k)}, i\notin S^{(k+1)}\right]\\
&= \sum_{j\ne i,j\notin S^{(k)}} M_{ij}^{(k)} + \gamma \sigma_1 \sum_{l\ne i, l\notin S^{(k)}}  M_{il}^{(k)} \textbf{E}\left[Z_l^{(k+1)}\mid S^{(k)}, i\notin S^{(k+1)}\right]\\
&= \left(1 + \frac{\gamma \sigma_1}{n-k-1}\right)\sum_{j\ne i,j\notin S^{(k)}} M_{ij}^{(k)}\\
&\le \sum_{j\ne i,j\notin S^{(k)}} M_{ij}^{(k)} + \frac{2\gamma \sigma_1^2}{n}\\
\end{align*}

Next, we bound the variance:

\begin{align*}
\textbf{Var}\left(\sum_{j\ne i, j\notin S^{(k+1)}} M_{ij}^{(k+1)}\mid S^{(k)}, i\notin S^{(k+1)}\right) &= \textbf{Var}\left(\sum_{j\ne i, j\notin S^{(k+1)}} M_{ij}^{(k+1)} - \sum_{j\ne i, j\notin S^{(k)}} M_{ij}^{(k)}\mid S^{(k)}, i\notin S^{(k+1)}\right)\\
&\le \textbf{E}\left[\left(\sum_{j\ne i, j\notin S^{(k+1)}} M_{ij}^{(k+1)} - \sum_{j\ne i, j\notin S^{(k)}} M_{ij}^{(k)}\right)^2 \mid S^{(k)}, i\notin S^{(k+1)}\right]\\
&\le \gamma^2\sigma_1^2\textbf{E}\left[\left(\sum_{l\ne i, l\notin S^{(k)}} M_{il}^{(k)} Z_l^{(k+1)}\right)^2 \mid S^{(k)}, i\notin S^{(k+1)}\right]\\
&= \gamma^2\sigma_1^2\textbf{E}\left[\sum_{l\ne i, l\notin S^{(k)}} (M_{il}^{(k)})^2 Z_l^{(k+1)} \mid S^{(k)}, i\notin S^{(k+1)}\right]\\
&\le \frac{\gamma^2\sigma_1^2}{n-k-1}\sum_{l\ne i, l\notin S^{(k)}} (M_{il}^{(k)})^2\\
&\le \frac{2\gamma^2\sigma_1^4}{n}\\
\end{align*}

Finally, we bound the maximum change:

\begin{align*}
\sum_{j\ne i,j\notin S^{(k+1)}} M_{ij}^{(k+1)} - \sum_{j\ne i,j\notin S^{(k)}} M_{ij}^{(k)} &\le \gamma \sigma_1 \left(\sum_{l=1}^n M_{il}^{(k)} Z_l^{(k+1)}\right)\\
&\le \gamma \sigma_1 \left(\sum_{l\ne i, l\notin S^{(k)}} M_{il}^{(k)}\right)\\
&\le \gamma \sigma_1^2
\end{align*}

conditioned on $i\notin S^{(k+1)}$. Applying Theorem \ref{thm:martingale-2} to the random variables $\{\sum_{j\ne i, j\notin S^{(k)}} M_{ij}^{(k)}\}_k$ before the stopping time $\{k : i\in S^{(k)}\}$ shows that

\begin{align*}
&\Pr\left[\sum_{j\ne i, j\notin S^{(k)}} M_{ij}^{(k)} - \textbf{E}\left[\sum_{j\ne i, j\notin S^{(k)}} M_{ij}^{(k)}\mid i\notin S^{(k)}\right] \ge \lambda \mid i\notin S^{(k)}\right]\\
&\le \exp\left(-\frac{\lambda^2}{(n/2)(2\gamma^2\sigma_1^4/n) + (\lambda/3)\gamma\sigma_1^2}\right)\\
&\le \exp\left(-\frac{\lambda^2}{\gamma^2\sigma_1^4 + (\lambda/3)\gamma\sigma_1^2}\right)\\
\end{align*}

given the inductive assumption and $k\le n/2$. Substituting $\lambda = (8\log n)\gamma\sigma_1^2 \le \sigma_1/4$ gives a probability bound of $1/n^8$. Now, we just have to bound the change in the expectation. Do this by summing up increments:

\begin{align*}
\textbf{E}\left[\sum_{j\ne i, j\notin S^{(k)}} M_{ij}^{(k)}\mid i\notin S^{(k)}\right] - \sum_{j\ne i} M_{ij}^{(0)} &= \sum_{r=0}^{k-1} \textbf{E}\left[\sum_{j\ne i, j\notin S^{(r+1)}} M_{ij}^{(r+1)} - \sum_{j\ne i, j\notin S^{(r)}} M_{ij}^{(r)}\mid i\notin S^{(k)}\right]\\
&= \sum_{r=0}^{k-1} \textbf{E}\left[\sum_{j\ne i, j\notin S^{(r+1)}} M_{ij}^{(r+1)} - \sum_{j\ne i, j\notin S^{(r)}} M_{ij}^{(r)}\mid i\notin S^{(r+1)}\right]\\
&= \sum_{r=0}^{k-1}\\
&\textbf{E}_{S^{(r)}}\left[\textbf{E}_{w^{(r+1)}}\left[\sum_{j\ne i, j\notin S^{(r+1)}} M_{ij}^{(r+1)} - \sum_{j\ne i, j\notin S^{(r)}} M_{ij}^{(r)}\mid i\notin S^{(r+1)},S^{(r)}\right]\right]\\
&\le 2k\gamma\sigma_1^2/n\\
&\le \gamma\sigma_1^2\\
&\le \sigma_1/4\\
\end{align*}

By the first given condition and the fact that $\sigma_0\le \sigma_1/2$,

$$\sum_{j\ne i, j\notin S^{(k)}} M_{ij}^{(k)}\le \sigma_1/4 + \sigma_1/4 + \sigma_1/2 \le \sigma_1$$

with probability at least $1 - 1/n^8$. This completes the verification of the inductive hypothesis and the desired result.
\end{proof}

\proprandomvectorcontrol*

\begin{proof}[Proof of Proposition \ref{prop:random-vector-control}]
Inductively assume that for all $i\notin S^{(k)}$ and for all $k\le n/2$, $v_i^{(k)} \le \tau_1$. We now use Theorem \ref{thm:martingale-2} to validate this assumption. The fourth given condition along with the inductive assumption shows that

\begin{align*}
v_i^{(k+1)} &\le v_i^{(k)} + \gamma \left(\sum_{l=1}^n M_{il}^{(k)} Z_l^{(k+1)} (v_l^{(k)} + v_i^{(k)})\right)\\
&\le v_i^{(k)} + 2\gamma \tau_1 \left(\sum_{l=1}^n M_{il}^{(k)} Z_l^{(k+1)}\right)\\
\end{align*}

for $i\notin S^{(k+1)}$. To apply Theorem \ref{thm:martingale-2}, we need bounds on the mean, variance, and maximum deviation of each increment. We start with the mean:

\begin{align*}
\textbf{E}\left[v_i^{(k+1)}\mid S^{(k)}, i\notin S^{(k+1)}\right] &\le v_i^{(k)} + 2\gamma \tau_1 \textbf{E}\left[\sum_{l=1}^n M_{il}^{(k)} Z_l^{(k+1)}\mid S^{(k)}, i\notin S^{(k+1)}\right]\\
&\le v_i^{(k)} + 2\gamma \tau_1 \textbf{E}\left[\sum_{l\ne i, l\notin S^{(k)}}  M_{il}^{(k)} Z_l^{(k+1)}\mid S^{(k)}, i\notin S^{(k+1)}\right]\\
&= v_i^{(k)} + 2\gamma \tau_1 \sum_{l\ne i, l\notin S^{(k)}}  M_{il}^{(k)} \textbf{E}\left[Z_l^{(k+1)}\mid S^{(k)}, i\notin S^{(k+1)}\right]\\
&\le v_i^{(k)} + \frac{4\gamma \tau_1\sigma_1}{n}\\
\end{align*}

Next, we bound the variance:

\begin{align*}
\textbf{Var}\left(v_i^{(k+1)}\mid S^{(k)}, i\notin S^{(k+1)}\right) &= \textbf{Var}\left(v_i^{(k+1)} - v_i^{(k)}\mid S^{(k)}, i\notin S^{(k+1)}\right)\\
&\le \textbf{E}\left[\left(v_i^{(k+1)} - v_i^{(k)}\right)^2 \mid S^{(k)}, i\notin S^{(k+1)}\right]\\
&\le 4\gamma^2\tau_1^2\textbf{E}\left[\left(\sum_{l\ne i, l\notin S^{(k)}} M_{il}^{(k)} Z_l^{(k+1)}\right)^2 \mid S^{(k)}, i\notin S^{(k+1)}\right]\\
&= 4\gamma^2\tau_1^2\textbf{E}\left[\sum_{l\ne i, l\notin S^{(k)}} (M_{il}^{(k)})^2 Z_l^{(k+1)} \mid S^{(k)}, i\notin S^{(k+1)}\right]\\
&\le \frac{4\gamma^2\tau_1^2}{n-k-1}\sum_{l\ne i, l\notin S^{(k)}} (M_{il}^{(k)})^2\\
&\le \frac{8\gamma^2\tau_1^2\sigma_1^2}{n}\\
\end{align*}

Finally, we bound the maximum change:

\begin{align*}
v_i^{(k+1)} - v_i^{(k)} &\le 2\gamma \tau_1 \left(\sum_{l=1}^n M_{il}^{(k)} Z_l^{(k+1)}\right)\\
&\le 2\gamma \tau_1 \left(\sum_{l\ne i, l\notin S^{(k)}} M_{il}^{(k)}\right)\\
&\le 2\gamma \tau_1\sigma_1
\end{align*}

conditioned on $i\notin S^{(k+1)}$. Applying Theorem \ref{thm:martingale-2} to the random variables $\{v_i^{(k)}\}_k$ before the stopping time $\{k : i\in S^{(k)}\}$ shows that

\begin{align*}
\Pr\left[v_i^{(k)} - \textbf{E}\left[v_i^{(k)}\mid i\notin S^{(k)}\right] \ge \lambda \mid i\notin S^{(k)}\right]&\le \exp\left(-\frac{\lambda^2}{(n/2)(8\gamma^2\sigma_1^2\tau_1^2/n) + (\lambda/3)2\gamma\sigma_1\tau_1}\right)\\
&\le \exp\left(-\frac{\lambda^2}{4\gamma^2\sigma_1^2\tau_1^2 + (\lambda/3)2\gamma\sigma_1\tau_1}\right)\\
\end{align*}

given the inductive assumption and $k\le n/2$. Substituting $\lambda = (16\log n)\gamma\sigma_1\tau_1 \le \tau_1/4$ gives a probability bound of $1/n^8$. Now, we just have to bound the change in the expectation. Do this by summing up increments:

\begin{align*}
\textbf{E}\left[v_i^{(k)}\mid i\notin S^{(k)}\right] - v_i^{(0)} &= \sum_{r=0}^{k-1} \textbf{E}\left[v_i^{(r+1)} - v_i^{(r)}\mid i\notin S^{(k)}\right]\\
&= \sum_{r=0}^{k-1} \textbf{E}\left[v_i^{(r+1)} - v_i^{(r)}\mid i\notin S^{(r+1)}\right]\\
&= \sum_{r=0}^{k-1}\textbf{E}_{S^{(r)}}\left[\textbf{E}_{w^{(r+1)}}\left[v_i^{(r+1)} - v_i^{(r)}\mid S^{(r)}, i\notin S^{(r+1)}\right]\right]\\
&\le 4k\gamma\sigma_1\tau_1/n\\
&\le 2\gamma\sigma_1\tau_1\\
&\le \tau_1/4\\
\end{align*}

By the first given condition and the fact that $\tau\le \tau_1/2$,

$$v_i^{(k)}\le \tau_1/4 + \tau_1/4 + \tau_1/2 \le \tau_1$$

with probability at least $1 - 1/n^8$. This completes the verification of the inductive hypothesis and the desired result.
\end{proof}

\proprandomsequencecontrol*

\begin{proof}
Let $K = m_0 \min(\frac{1}{(\log (M_0^2/\tau))\eta^2\gamma^2},\frac{1}{200\eta \gamma^2 \log (M_0^2/\tau)}(\tau/M_0^2)^{1/\rho})$. Let $\{k_i\}_i$ be the (random variable) subsequence of superscripts for which $v^{(k_i)} - v^{(k_i-1)} > \frac{v^{(k_{i-1})}}{100\log (M_0^2/\tau)}$, with $k_0 := 0$. For all $k\in [k_{i-1},k_i-1]$, by Theorem \ref{thm:martingale-2},

$$\Pr[v^{(k)} - v^{(k_{i-1})} \ge \lambda ] \le e^{-\frac{\lambda^2}{K\eta^2\gamma^2 (v^{(k_{i-1})})^2 / m_0 + \lambda v^{(k_{i-1})} / (100\log (M_0^2/\tau))}}$$

so each such interval can only increase $v^{(k)}$ by a factor of 2 with probability at least $1 - \tau/(M_0^2)$. Therefore, to show the desired concentration bound, we just need to bound the number of $k_i$s. For any $k\in [k_{i-1},k_i]$, only $100\log(M_0^2/\tau) \eta \gamma^2$ different $u_j^{(k)}w_j^{(k)}$ products can be greater than $v^{(k-1)}/(200\gamma\log(M_0^2/\tau))\le v^{(k_{i-1})}/(100\gamma\log(M_0^2/\tau))$. Therefore, the probability that $v_k - v_{k-1}\ge \frac{v^{(k_{i-1})}}{100\log (M_0^2/\tau)}$ is at most 

$$\frac{200\log(M_0^2/\tau) \eta \gamma^2}{m_0}$$

This means that the probability that more than $\rho$ $k_i$s is at most

\begin{align*}
\sum_{a=(\rho+1)}^K \binom{K}{a} \left(\frac{200\log(M_0^2/\tau) \eta \gamma^2}{m_0}\right)^a &\le 2K^{\rho} \left(\frac{200\log(M_0^2/\tau) \eta \gamma^2}{m_0}\right)^\rho\\
&\le 2\tau/M_0^2\\
\end{align*}

For each $i$, we know that

$$v^{(k_i)} \le \eta \gamma v^{(k_i-1)}$$

With probability at least $1 - \frac{2\tau}{M_0^2}$, there are at most $\rho$ $i$s, as discussed above. The value can increase by a factor of at most a $(2\gamma\eta)^{\rho}$ factor in total with probability at least $1 - M_0^2 \frac{2\tau}{M_0^2} = 1 - 2\tau$, as desired.
\end{proof}

\propflexibleobjectives*

\begin{proof}
Consider any minor $H$ of $J$ and consider some edge $f\in X\cap E(H)$ as described in Definition \ref{def:flexible-families}. Fix an edge $e\in X\cap E(H)$. Throughout all proofs, we focus on the contraction case, as the deletion case is exactly the same except with $\texttt{lev}_{\phi(H)}(f)$ replaced with $1 - \texttt{lev}_{\phi(H)}(f)$. 

\textbf{$g_e^{(0)}$ flexiblity.} By Sherman-Morrison (Propositions \ref{prop:deletion-update} and \ref{prop:contraction-update}) and the triangle inequality,

\begin{align*}
g_e^{(0)}(H/f) &= \frac{|b_{ss'}^T L_{(H/f\setminus D)/(S,S')}^+ b_e|}{\sqrt{r_e}}\\
&\le \frac{|b_{ss'}^T L_{(H\setminus D)/(S,S')}^+ b_e|}{\sqrt{r_e}} + \frac{|b_{ss'}^T L_{(H\setminus D)/(S,S')}^+ b_f| |b_f^T L_{(H\setminus D)/(S,S')}^+ b_e|}{\texttt{lev}_{(H\setminus D)/(S,S')}(f) r_f\sqrt{r_e}}\\
&= g_e^{(0)}(H) + \frac{1}{\texttt{lev}_{\phi^{(0)}(H)}(f)} \frac{|b_e^T L_{\phi^{(0)}(H)}^+ b_f|}{\sqrt{r_e}\sqrt{r_f}} g_f^{(0)}(H)\\
\end{align*}

as desired.

\textbf{$g_e^{(1),X\cap E(H),s}$ and $g_e^{(1),X\cap E(H),s'}$ flexibility.} The proof for $g_e^{(1),X,s'}$ is the same as the proof for $g_e^{(1),X,s}$ with $s,S$ and $s',S'$ swapped. Therefore, we focus on $g_e^{(1),X,s}$. By Sherman-Morrison and the triangle inequality,

\begin{align*}
g_e^{(1),X\cap E(H/f),s}(H/f) &= \sum_{\{u,v\}\in (X\cap E(H/f))\setminus \{e\}} \frac{|b_e^T L_{(H/f\setminus D)/(S,S')}^+ (b_{su} + b_{sv})/2|}{\sqrt{r_e}}\\
&= \sum_{\{u,v\}\in (X\cap E(H))\setminus \{e,f\}} \frac{|b_e^T L_{(H/f\setminus D)/(S,S')}^+ (b_{su} + b_{sv})/2|}{\sqrt{r_e}}\\
&\le \sum_{\{u,v\}\in (X\cap E(H))\setminus \{e,f\}} \Biggr(\frac{|b_e^T L_{(H\setminus D)/(S,S')}^+ (b_{su} + b_{sv})/2|}{\sqrt{r_e}}\\
&+ \frac{|b_e^T L_{(H\setminus D)/(S,S')}^+ b_f| |b_f^T L_{(H\setminus D)/(S,S')}^+ (b_{su} + b_{sv})/2|}{\texttt{lev}_{(H\setminus D)/(S,S')}(f) r_f\sqrt{r_e}}\Biggr)\\
&\le g_e^{(1),X\cap E(H),s}(H) + \frac{1}{\texttt{lev}_{\phi^{(1)}(H)}(f)} \frac{|b_e^T L_{\phi^{(1)}(H)}^+ b_f|}{\sqrt{r_e}\sqrt{r_f}} g_f^{(1),X\cap E(H),s}(H)\\
\end{align*}

as desired.

\textbf{$g_e^{(2)}$ flexibility.} By Sherman-Morrison and the triangle inequality,

\begin{align*}
g_e^{(2)}(H/f) &= \sum_{\{u,v\}\in A} \frac{|b_e^T L_{(H/f\setminus D)/(S,S')}^+ (b_{su} + b_{sv})|}{\sqrt{r_e}} + \sum_{\{u,v\}\in B} \frac{|b_e^T L_{(H/f\setminus D)/(S,S')}^+ (b_{us'} + b_{vs'})|}{\sqrt{r_e}}\\
&\le g_e^{(2)}(H) + \sum_{\{u,v\}\in A} \frac{|b_e^T L_{(H\setminus D)/(S,S')}^+ b_f| |b_f^T L_{(H\setminus D)/(S,S')}^+ (b_{su} + b_{sv})|}{\texttt{lev}_{(H\setminus D)/(S,S')}(f)r_f\sqrt{r_e}}\\
&+ \sum_{\{u,v\}\in B} \frac{|b_e^T L_{(H\setminus D)/(S,S')}^+ b_f| |b_f^T L_{(H\setminus D)/(S,S')}^+ (b_{us'} + b_{vs'})|}{\texttt{lev}_{(H\setminus D)/(S,S')}(f)r_f\sqrt{r_e}}\\
&= g_e^{(2)}(H) + \frac{1}{\texttt{lev}_{\phi^{(2)}(H)}(f)} \frac{|b_e^T L_{\phi^{(2)}(H)}^+ b_f|}{\sqrt{r_e}\sqrt{r_f}} g_f^{(2)}(H)\\
\end{align*}

as desired.

\textbf{$g_e^{(3),s}$ and $g_e^{(3),s'}$ flexibility.} We focus on $g_e^{(3),s}$, as $g_e^{(3),s'}$ is the same except that $s',S'$ is swapped with $s,S$. By Sherman-Morrison and the triangle inequality,

\begin{align*}
g_e^{(3),s}(H/f) &= \sum_{w\in S'} b_{sw}^T L_{(J\setminus D)/S}^+ b_{sw} \sum_{e'\in \partial_H w} \frac{|b_e^T L_{(H/f\setminus D)/(S,S')}^+ b_{e'}|}{\sqrt{r_e}r_{e'}}\\
&\le g_e^{(3),s}(H) + \sum_{w\in S'} b_{sw}^T L_{(J\setminus D)/S}^+ b_{sw} \sum_{e'\in \partial_H w} \frac{|b_e^T L_{(H\setminus D)/(S,S')}^+ b_f| |b_f^T L_{(H\setminus D)/(S,S')}^+ b_{e'}|}{\sqrt{r_e}\texttt{lev}_{(H\setminus D)/(S,S')}(f) r_fr_{e'}}\\
&= g_e^{(3),s}(H) + \frac{1}{\texttt{lev}_{\phi^{(3)}(H)}(f)} \frac{|b_e^T L_{\phi^{(3)}(H)}^+ b_f|}{\sqrt{r_e}\sqrt{r_f}} g_f^{(3),s}(H)\\
\end{align*}

as desired.

\textbf{$g_e^{(4),s}$ and $g_e^{(4),s'}$ flexibility.} We focus on $g_e^{(4),s}$, as $g_e^{(4),s'}$ is the same except that $s',S'$ is swapped with $s,S$. By Sherman-Morrison and the triangle inequality,

\begin{align*}
g_e^{(4),s}(H/f) &= \sum_{w\in S'} \left(\frac{|b_{sw}^T L_{(H/f\setminus D)/S}^+ b_e|}{\sqrt{r_e}}\right)^2 \sum_{e'\in \partial_J w} \frac{b_{ss'}^T L_{(J\setminus D)/(S,S')}^+ b_{e'}}{r_{e'}}\\
&\le \sum_{w\in S'} \left(\frac{|b_{sw}^T L_{(H\setminus D)/S}^+ b_e|}{\sqrt{r_e}} + \frac{|b_{sw}^T L_{(H\setminus D)/S}^+ b_f| |b_f^T L_{(H\setminus D)/S}^+ b_e|}{\texttt{lev}_{(H\setminus D)/S}(f) r_f \sqrt{r_e}}\right)^2 \sum_{e'\in \partial_J w} \frac{b_{ss'}^T L_{(J\setminus D)/(S,S')}^+ b_{e'}}{r_{e'}}\\
&= g_e^{(4),s}(H) + \frac{|b_f^T L_{(H\setminus D)/S}^+ b_e|}{\texttt{lev}_{(H\setminus D)/S}(f)\sqrt{r_f}\sqrt{r_e}}\sum_{w\in S'} \Biggr(2\frac{|b_{sw}^T L_{(H\setminus D)/S}^+ b_e| |b_{sw}^T L_{(H\setminus D)/S}^+ b_f|}{\sqrt{r_e}\sqrt{r_f}}\\
&+ \left(\frac{|b_{sw}^T L_{(H\setminus D)/S}^+ b_f|}{\sqrt{r_f}}\right)^2\frac{|b_f^T L_{(H\setminus D)/S}^+ b_e|}{\texttt{lev}_{(H\setminus D)/S}(f) \sqrt{r_f}\sqrt{r_e}}\Biggr)\sum_{e'\in \partial_J w} \frac{b_{ss'}^T L_{(J\setminus D)/(S,S')}^+ b_{e'}}{r_{e'}}\\
&\le g_e^{(4),s}(H) + \frac{3|b_f^T L_{\phi^{(4),s}(H)}^+ b_e|}{(\texttt{lev}_{\phi^{(4),s}(H)}(f))^2\sqrt{r_f}\sqrt{r_e}}(g_e^{(4),s}(H) + g_f^{(4),s}(H))\\
\end{align*}

where the last inequality follows from AM-GM and $\frac{|b_e^T L_{(H\setminus D)/S}^+ b_f|}{\sqrt{r_e}\sqrt{r_f}} \le 1$. This is the desired result.

\end{proof}

\section{Parameters}\label{sec:constants}

Throughout this paper, there are many parameters and constants that are hidden in $m^{o(1)}$ and $\alpha^{o(1)}$ dependencies. For any values of $\sigma_0$ and $\sigma_1$, our main algorithm takes

$$f(\sigma_0,\sigma_1)m^{1 + 1/\sigma_1}\alpha^{1/(\sigma_0+1)}$$

time, where $f(\sigma_0,\sigma_1)\le 2^{\sqrt{\log n}\sigma_1^{100\sigma_0}}$. Setting $\sigma_0 = \sigma_1 = (\log \log n)/(400\log \log \log n)$ yields an algorithm with runtime $m^{o(1)}m^{1 + (\log \log \log n)/(400\log \log n)}\alpha^{1/(\log \log n)} = m^{1 + o(1)}\alpha^{o(1)}$, as desired.

To establish this runtime, we need bounds on the values of each of the following parameters that do not depend on $m^{1/\sigma_1}$ or $\alpha^{1/(\sigma_0+1)}$:

\begin{itemize}
\item The conductivity parameter $\zeta$
\item The modifiedness parameter $\tau$
\item The boundedness parameter $\kappa$
\item The number of clans $\ell$
\item The carving parameter $\mucarve = (2000\muapp)^{\sigma_1^{4\sigma_0}}$
\item Smaller parameters:
    \begin{itemize}
    \item The appendix bounding parameter $\muapp = 2^{100(\log n)^{3/4}}$
    \item The fixing bounding parameter $\mucon = 2^{100\sqrt{\log n}}$
    \item The empire radius bounding parameter $\murad = 2^{1000(\log n)^{3/4}}$
    \item Lemma \ref{lem:clan-fix} parameter: $\mumod = 200\mucon\muapp\mucarve$
    \end{itemize}
\end{itemize}

The parameter $\muapp$ is used for simplicity to bound all $m^{o(1)}\alpha^{o(1)}$ terms that appear in the appendix. $\mucon$ is used in the ``Size'' bound of Lemma \ref{lem:fast-fix}. The carving parameter $\mucarve$ arises from growing a ball around each $\mc Q_i$ in the conditioning digraph with radius $\sigma_1^i\le \sigma_1^{\sigma_0}$. By the ``Temporary condition'' of edges in conditioning digraphs, a part $Q$ with an edge from $P$ is only $8\muapp$ times farther from $\mc Q_i$ than $P$ is, so all parts within distance $\sigma_1^{\sigma_0}$ in the conditioning digraph must be within distance $(8\muapp)^{\sigma_1^{\sigma_0}}\R{i}\le \mucarve\R{i}$. For more details, see Section \ref{sec:choose-parts}.

The conductivity, modifiedness, boundedness, and number of clan parameters increase during each iteration of the while loop in $\ExactTree$. We designate maximum values $\zetamax$, $\taumax$, $\kappamax$, and $\ellmax$ for each. $\ell$ does not depend on $\zeta$, $\tau$, and $\kappa$ because it only increases by more than a $(\log m)$-factor in $\CreateEmpire$, which increases it by a factor of $\muapp$. $\kappa$ increases by at most a factor of $\ell\le \ellmax$ during each iteration. $\tau$ and $\zeta$ depend on one another. $\tau$ additively increases by $\tau\muapp + \zeta\mucarve\muapp$ during each iteration, while $\zeta$ multiplicatively increases by a constant factor and additively increases by $\ell\muapp$. Since the while loop in $\ExactTree$ only executes for $\sigma_1^{\sigma_0}$ iterations, values of

\begin{itemize}
\item $\zetamax = ((\log m)\muapp)^{(2\sigma_1)^{8\sigma_0}}$
\item $\taumax = ((\log m)\muapp)^{(2\sigma_1)^{8\sigma_0}}$
\item $\kappamax = ((\log m)\muapp)^{(2\sigma_1)^{4\sigma_0}}$
\item $\ellmax = ((\log m)\muapp)^{(2\sigma_1)^{2\sigma_0}}$
\end{itemize}

suffice. Therefore, the runtime of the partial sampling step (the bottleneck) is at most $$2^{\sqrt{\log n}\sigma_1^{100\sigma_0}}m^{1 + 1/\sigma_1}\alpha^{1/(\sigma_0+1)}\le m^{1 + o(1)}\alpha^{o(1)}$$ as desired.

\section{List of Definitions}

\vspace*{-10mm}

\listoftheorems[ignoreall,show={definition}]

\section{Almost-linear time, $\ep$-approximate random spanning tree generation for arbitrary weights}\label{sec:apx-tree}

In this section, we obtain a $m^{1+o(1)}\ep^{-o(1)}$-time algorithm for generating an $\ep$-approximate random spanning tree. Specifically, we output each spanning tree according to a distribution with total variation distance at most $\ep$ from the real one.

We do this by reducing the problem to our main result; an $m^{1+o(1)}\alpha^{1/(\sigma_0+1)}$-time algorithm for generating an exact random spanning tree. The reduction crudely buckets edges by their resistance and computes connected components of low-resistance edges. Deleting edges with very high resistances does not affect the marginal of the random spanning tree distribution with respect to low-resistance edges very much. Therefore, sampling a random spanning tree in the graph with high-resistance edges deleted yields an $\ep$-approximation random tree for that marginal in the entire graph.

To make this approach take $m^{1+o(1)}\ep^{-o(1)}$ time, we just need to make sure that conditioning on a connected component $F$ of low-resistance edges does not take much longer longer than $O(|F|m^{o(1)})$ time. In other words, at most $O(|F|m^{o(1)})$ medium-resistance edges should exist. This can be assured using ball growing on $\sqrt{\sigma_0}$ different scales. This results in a graph being conditioned on with $\alpha\le m^{10\sqrt{\sigma_0}}$ and size at most $|F|m^{1/\sqrt{\sigma_0}}$, where $F$ is the set being conditioned on. Therefore, to condition on $|F|$ edges, the algorithm only needs to do $|F|^{1+o(1)}m^{1/\sqrt{\sigma_0}}m^{10\sqrt{\sigma_0}/(\sigma_0+1)}$ work, as desired.

\begin{figure}
\begin{center}
\includegraphics[width=0.7\textwidth]{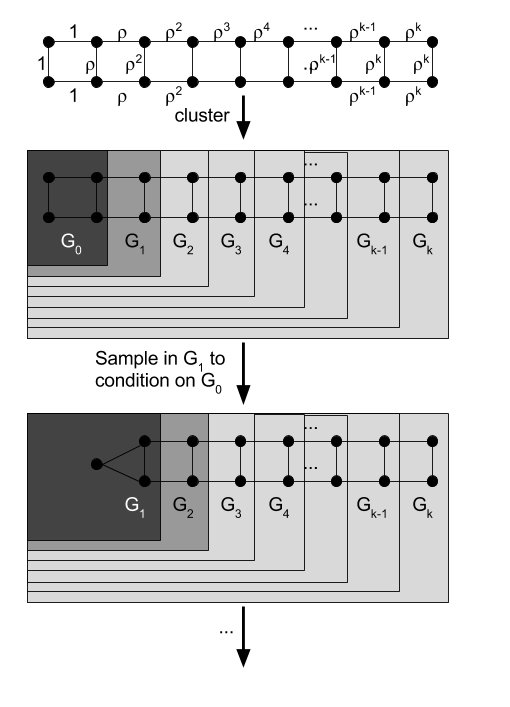}
\end{center}
\caption{The first iteration of the while loop of the $\ApxTree$ algorithm. The amount of work done during each iteration is not much larger than the amount of progress (reduction in graph size) made.}
\label{fig:weight-reduction}
\end{figure}

Now, we implement this approach as follows:

\begin{algorithm}[H]
\SetAlgoLined
\DontPrintSemicolon
\caption{$\ApxTree(G)$}

    $\rho\gets (m/\ep)^{10}$\;

    $T_{return}\gets \emptyset$\;

    \tcp{can compute connected components of all $G_i$s in near-linear time with union-find}

    $G_i\gets$ the subgraph of $G$ consisting of edges with resistance at most $\rho^i r_{min}$\;

    $i^*\gets 0$\;

    \While{$E(G)\ne \emptyset$}{

        \While{$|E(G_{i^*+1})| > m^{1/\sqrt{\sigma_0}} |E(G_{i^*})|$}{ \label{line:size}

            $i^*\gets i^*+1$\;

        }

        \ForEach{ connected component $C$ of $G_{i^*+1}$}{

            $T\gets \ExactTree(G_{i^*+1}[C])$\;\label{line:exact-sample}

            contract edges in $G$ in the intersection of $T$ with $E(G_{i^*})$ and add these edges to $T_{return}$\;\label{line:apx-cont}

            delete edges from $G$ that are in $E(G_{i^*})$ but not $T$\;\label{line:apx-del}

        }

    }

    \Return $T_{return}$

\end{algorithm}

The correctness of this algorithm relies on the following claim:

\begin{proposition}\label{prop:del-high-res}
Consider a graph $I$ and an edge $e\in E(I)$. Let $I'$ be the subgraph of $I$ of edges with resistance at most $\rho r_e^I$. Then

$$\texttt{Reff}_I(e)\le \texttt{Reff}_{I'}(e)\le (1+2\ep/m^9) \texttt{Reff}_I(e)$$
\end{proposition}

\begin{proof}
The lower bound follows immediately from Rayleigh monotonicity, so we focus on the upper bound. Let $f$ be the unit electrical flow between the endpoints of $e$. It has energy $\texttt{Reff}_{I'}(e)$. Each edge $g\in E(I)\setminus E(I')$ has

$$f_g \le \frac{b_e^T L_I^+ b_e}{r_g^I}\le r_e^I/(\rho r_e^I)\le \ep/m^{10}$$

Therefore, deleting edges in $I\setminus I'$ and removing the associated flow results in a flow between the endpoints of $e$ with energy at most $b_e^T L_I^+ b_e$ that ships at least $1 - m(\ep/m^{10}) = 1 - \ep/m^9$ units of flow. Scaling the flow up by $1/(1 - \ep/m^9)$ yields the desired result.
\end{proof}

\thmmainresultapply*

\begin{proof}[Proof of Theorem \ref{thm:main-result-apply} given that $\ExactTree$ satisfies Theorem \ref{thm:main-result-aspect}]

\textbf{Correctness.} At each step, the algorithm contracts a forest in $G$, so it maintains the fact that $G$ is connected. Since the algorithm terminates only when $E(G) = \emptyset$, $T_{return}$ is a tree. We now compute the total variation distance between the output distribution and the real spanning tree distribution for $G$.

To do this, we set up a hybrid argument. Let $(e_i)_{i=1}^m$ be an ordering of the edges in $G$ in increasing order of resistance. Notice that the subset $E(G_j)\subseteq E(G)$ is a prefix of this list for all $j$, which means that conditioning on a $G_j$ eliminates a prefix of $(e_i)_{i=1}^m$.

For $k \in \{0,1,\hdots,m\}$, let $\mc D_k$ denote the distribution over spanning trees of $G$ obtained by sampling a tree $T$ using $\ApxTree$, conditioning on its intersection with the set $\{e_1,e_2,\hdots,e_k\}$ to obtain a graph $G'$, and sampling the rest of $T$ using the real uniform random spanning tree distribution of the graph $G'$. $\mc D_0$ is the uniform distribution over spanning trees of $G$, while $\mc D_m$ is the distribution over trees output by $\ApxTree$. Therefore, to complete the correctness proof, we just need to show that the total variation distance between $\mc D_0$ and $\mc D_m$ is at most $\ep$. We do this by bounding the total variation distance between $\mc D_k$ and $\mc D_{k+1}$ for all $k\le m-1$.

Consider a spanning tree $T'$ of $G$. We now compute its probability mass in the distributions $\mc D_k$ and $\mc D_{k+1}$. Consider the $i^*$ in the $\ApxTree$ algorithm for which $e_{k+1}\in E(G_{i^*})\setminus E(G_{i_{prev}^*})$, where $i_{prev}^*$ is the previous value for which $\ApxTree$ conditions on $G_{i_{prev}^*}$ in Lines \ref{line:apx-cont} and \ref{line:apx-del}. Line \ref{line:exact-sample} extends the partial sample from $G_{i_{prev}^*}$ to $G_{i^*}$ by sampling from the true uniform distribution of $G_{i^*+1}$ conditioned on the partial sample for $G_{i_{prev}^*}$ matching $T'$. By Theorem \ref{thm:edge-marginal}, one could equivalently sample from this distribution by sampling edges with probability equal to their leverage score in $G_{i^*+1}$ conditioned on the samples for the previous edges. Applying this reasoning shows that

\begin{align*}
\Pr_{T_a\sim \mc D_k}[T_a = T'] &= \Pr_{T_b\sim \ApxTree(G)}[E(T_b)\cap \{e_i\}_{i=1}^k = E(T')\cap \{e_i\}_{i=1}^k]\\
&\Pr_{T_c\sim H_k}[E(T_c)\cap \{e_{k+1}\} = E(T')\cap \{e_{k+1}\}]\\
&\Pr_{T_d\sim H_{k+1}}[E(T_d)\cap \{e_i\}_{i=k+2}^m = E(T')\cap \{e_i\}_{i=k+2}^m]
\end{align*}

and that

\begin{align*}
\Pr_{T_a\sim \mc D_{k+1}}[T_a = T'] &= \Pr_{T_b\sim \ApxTree(G)}[E(T_b)\cap \{e_i\}_{i=1}^k = E(T')\cap \{e_i\}_{i=1}^k]\\
&\Pr_{T_c\sim H_k'}[E(T_c)\cap \{e_{k+1}\} = E(T')\cap \{e_{k+1}\}]\\
&\Pr_{T_d\sim H_{k+1}}[E(T_d)\cap \{e_i\}_{i=k+2}^m = E(T')\cap \{e_i\}_{i=k+2}^m]
\end{align*}

where $H_k$ is the graph obtained by conditioning on the partial sample of $T'$ in $\{e_1,e_2,\hdots,e_k\}$ and $H_k'$ is obtained from $H_k$ by removing all edges of $H_k$ that are not in $G_{i^*+1}$. By definition of $G_{i^*+1}$ and the fact that $e_k\in E(G_{i^*})$, $H_k'$ contains all edges in $H_k$ with resistance at most $\rho r_{e_{k+1}}$. Therefore, by Proposition \ref{prop:del-high-res} applied after using Theorem \ref{thm:edge-marginal} to write probabilities as leverage and nonleverage scores,

$$|\Pr_{T_c\sim H_k}[E(T_c)\cap \{e_{k+1}\} = E(T')\cap \{e_{k+1}\}] - \Pr_{T_c\sim H_k'}[E(T_c)\cap \{e_{k+1}\} = E(T')\cap \{e_{k+1}\}]|\le 2\ep/m^9$$

Therefore, the total variation distance between $\mc D_k$ and $\mc D_{k+1}$ is

\begin{align*}
&\sum_{\text{spanning trees $T'$ of $G$}} |\Pr_{T_a\sim \mc D_k}[T_a = T'] - \Pr_{T_a\sim \mc D_k}[T_a = T']|\\
&= \sum_{\text{spanning trees $T'$ of $G$}} \Biggr(|\Pr_{T_c\sim H_k}[E(T_c)\cap \{e_{k+1}\} = E(T')\cap \{e_{k+1}\}] - \Pr_{T_c\sim H_k'}[E(T_c)\cap \{e_{k+1}\} = E(T')\cap \{e_{k+1}\}]|\\
&\Pr_{T_b\sim \ApxTree(G)}[E(T_b)\cap \{e_i\}_{i=1}^k = E(T')\cap \{e_i\}_{i=1}^k]\\
&\Pr_{T_d\sim H_{k+1}}[E(T_d)\cap \{e_i\}_{i=k+2}^m = E(T')\cap \{e_i\}_{i=k+2}^m]\Biggr)\\
&\le \frac{2\ep}{m^9}\sum_{\text{spanning trees $T'$ of $G$}} \Biggr(\Pr_{T_b\sim \ApxTree(G)}[E(T_b)\cap \{e_i\}_{i=1}^k = E(T')\cap \{e_i\}_{i=1}^k]\\
&\Pr_{T_d\sim H_{k+1}}[E(T_d)\cap \{e_i\}_{i=k+2}^m = E(T')\cap \{e_i\}_{i=k+2}^m]\Biggr)\\
&\le \frac{4\ep}{m^9}\sum_{\text{restrictions of spanning trees $T'$ of $G$ to $\{e_i\}_{i=1}^k$}} \Biggr(\Pr_{T_b\sim \ApxTree(G)}[E(T_b)\cap \{e_i\}_{i=1}^k = E(T')\cap \{e_i\}_{i=1}^k]\Biggr)\\
&\le \frac{4\ep}{m^9}\\
\end{align*}

Therefore, the total variation distance between $\mc D_0$ and $\mc D_m$ is at most $m(4\ep/m^9)\le \ep$, as desired.

\textbf{Runtime.} When the algorithm samples a tree in $G_{i^*+1}$, it contracts or deletes (removes) all edges in $G_{i^*}$. Line \ref{line:size} ensures that $|E(G_{i^*+1})|\le m^{1/\sqrt{\sigma_0}}|E(G_{i^*})|$ whenever conditioning occurs. Furthermore, since $G_i\subseteq G_{i+1}$ for all $i$, the innermost while loop only executes $\sqrt{\sigma_0}$ times for each iteration of the outer while loop. This means that $G_{i^*+1}$ has $\beta\le \rho^{\sqrt{\sigma_0}}$ since all lower edges where in a prior $G_{i^*}$. By Theorem \ref{thm:main-result-aspect}, $\ExactTree$ takes at most

$$|E(G_{i^*+1})|^{1+o(1)} (\rho^{\sqrt{\sigma_0}})^{1/\sigma_0}\le |E(G_{i^*})|(m/\ep)^{o(1) + 11/\sqrt{\sigma_0}}$$

time. Since the $G_{i^*}$s that are conditioned on each time are edge-disjoint (as they are contracted/deleted immediately), the total size of the $G_{i*}$s over the course of the algorithm is $m$, which means that the total runtime of the outer while loop is at most $m^{1+o(1)+11/\sqrt{\sigma_0}}\ep^{-o(1)} = m^{1+o(1)}\ep^{-o(1)}$, as desired.
\end{proof}

\subsection*{Acknowledgments} We would like to thank Nikhil Srivastava and Satish Rao, and Cheng Xiang for many helpful conversations. We would like to thank Aleksander Madry, Jim Pitman, Akshay Ramachandran, Alex Rusciano, Hong Zhou, and Qiuyi Zhang for helpful conversations. We thank Nima Anari, Antares Chen, Tselil Schramm, Di Wang, and Jason Wu for helpful edits.

\bibliographystyle{alpha}
\bibliography{new-rst}

\end{document}